\documentclass{amsart}

	\usepackage{amsmath,amssymb,bbm}
	\usepackage{amsthm}
	\usepackage{latexsym}
	\usepackage{graphicx}
	\usepackage{color,wrapfig}
	\usepackage{multirow}
	\usepackage{bm,url,comment}
	\usepackage{mathtools,mathrsfs}
	\usepackage{cool,algorithmic}
	\usepackage{bm}
	\usepackage{array}
\usepackage{algorithm}
\usepackage{subcaption}

\DeclareMathOperator{\OO}{O}
\DeclareMathOperator{\oo}{o}

\DeclareMathOperator{\bv}{\mathbf{v}}
\DeclareMathOperator{\bu}{\mathbf{u}}
\DeclareMathOperator{\tr}{Tr}
\DeclareMathOperator{\re}{Re}
\DeclareMathOperator{\im}{Im}

\theoremstyle{plain}

\theoremstyle{remark}

\newcommand{\ub}{\mathbf{u}}
\newcommand{\Ub}{\mathbf{U}}
\newcommand{\ctG}{\widetilde{\mathcal{G}}}

\newcommand{\db}{\mathbf{d}}
\newcommand{\Db}{\mathbf{D}}
\newcommand{\vb}{\mathbf{v}}
\newcommand{\wb}{\mathbf{w}}

\newcommand{\T}{\mathrm{T}}
\newcommand{\ri}{\mathrm{i}}

\newcommand{\wtG}{\widetilde{G}}

\newcommand{\norm}[1]{\left\lVert#1\right\rVert}

\newcommand{\wt}{\widetilde}

\newcommand{\bxi}{\bm{\xi}}
\newcommand{\bzeta}{\bm{\zeta}}

\newcommand{\fa}{{\mathfrak a}}
\newcommand{\fb}{{\mathfrak b}}

\theoremstyle{plain} 
\newtheorem{theorem}{Theorem}[section]
\newtheorem*{theorem*}{Theorem}
\newtheorem{lemma}[theorem]{Lemma}
\newtheorem{assumption}[theorem]{Assumption}
\newtheorem*{lemma*}{Lemma}

\newtheorem*{corollary*}{Corollary}
\newtheorem{proposition}[theorem]{Proposition}
\newtheorem*{proposition*}{Proposition}

\newtheorem{definition}[theorem]{Definition}
\newtheorem*{definition*}{Definition}
\theoremstyle{remark}

\newtheorem*{example*}{Example}

\newtheorem*{remark*}{Remark}
\newtheorem*{remarks*}{Remarks}

\DeclareMathOperator*{\argmin}{arg\,min}

\begin{document}

\title[OS under noise with separable covariance structure]{Data-Driven optimal shrinkage of singular values under high-dimensional noise with separable covariance structure}
	
	\author{Pei-Chun Su}

	\address{Department of Mathematics, Duke University, Durham, NC, USA}

		\author{Hau-Tieng Wu}

	\address{Department of Mathematics and Department of Statistical Science, Duke University, Durham, NC, USA}

	\maketitle
	
\begin{abstract}
We develop a data-driven optimal shrinkage algorithm for matrix denoising in the presence of high-dimensional noise with a separable covariance structure; that is, the noise is colored and dependent across samples. The algorithm, coined {\em extended OptShrink} (eOptShrink) depends on the asymptotic behavior of singular values and singular vectors of the random matrix associated with the noisy data. Based on the developed theory, including the sticking property of non-outlier singular values and delocalization of the non-outlier singular vectors associated with weak signals with a convergence rate, and the spectral behavior of outlier singular values and vectors, we develop three estimators, each of these has its own interest. 
First, we design a novel rank estimator, based on which we provide 
an estimator for the spectral distribution of the pure noise matrix, and hence the optimal shrinker called eOptShrink. In this algorithm we do not need to estimate the separable covariance structure of the noise. A theoretical guarantee of these estimators with a convergence rate is given. 
On the application side, in addition to a series of numerical simulations with a comparison with various state-of-the-art optimal shrinkage algorithms, we apply eOptShrink to extract maternal and fetal electrocardiograms from the single channel trans-abdominal maternal electrocardiogram.
\end{abstract}

\noindent%

\section{Introduction}
\label{sec:intro}
We {aim to denoise a $p \times n$ data matrix $\widetilde{S}$, comprised of $n\in\mathbb{N}$ noisy samples of dimension $p\in\mathbb{N}$. The data matrix} is modeled as:
\begin{equation}\label{eq_model}
\widetilde{S}=S+Z=\sum_{i=1}^r d_i \ub_i \vb_i^\top+Z\in \mathbb{R}^{p\times n},
\end{equation}
where $Z$ is a noise-only random matrix, potentially with a {dependence structure that will be detailed later}, $S$ denotes a low-rank signal matrix with the singular value decomposition (SVD) $\sum_{i=1}^r d_i \ub_i \vb_i^\top$, where $r\geq 1$ is assumed to be small compared with $p$ and $n$, $\ub_i\in\mathbb{R}^p$ and $\vb_i\in\mathbb{R}^n$ are left and right singular vectors respectively prescribing the signal, and $d_i>0$ are the associated singular values describing signal strength that may depend on $n$. {To simplify the discussion, we call the pair $\ub_i$ and $\vb_i$ the $i$-th signal and $d_i$ the $i$-th signal strength hereafter.}
{A weighting approach to recovering $S$ from $\wt S$ by SVD, widely known in the literature as the {\em singular value shrinkage}, for recovering $S$ has been actively studied, which was first mentioned, to the best of our knowledge, in \cite{eckart1936approximation,mirsky1960symmetric,golub1965calculating}.} %
The idea is to select a proper function $\varphi: [0, \infty) \rightarrow [0,\infty)$, {often nonlinear, and construct}  
\begin{equation}\label{eq_shrink}
\widehat{S}_{\varphi}=\sum_{i=1}^{p \wedge n} \varphi(\wt\lambda_i) \wt{\bxi}_i \wt{\bzeta}_i^\top
\end{equation} 
{as an estimate of $S$},
where ${\wt\lambda_1} \geq {\wt\lambda_2} \geq \cdots \geq {\wt\lambda_{p \wedge n}}\geq 0$ are common eigenvalues of $\widetilde{S}\widetilde{S}^\top$ and $\widetilde{S}^\top\widetilde{S}$ and $\{\wt{\bxi}_i\}$ and $\{\wt{\bzeta}_i\}$ are the left and right singular vectors of $\widetilde{S}$ respectively. $\varphi$ is termed as a shrinker.
{ By employing a loss function $L_n:\mathbb{R}^{p\times n}\times \mathbb{R}^{p\times n}\to \mathbb{R}_+$ to} quantify the discrepancy between $\widehat S_{\varphi}$ and $S$, the associated {\em optimal shrinker}, if exists, is defined as $\varphi^* := \argmin_{\phi} \lim_{n\to \infty} L_n(\widehat S_{\phi}, S)$. {Common loss functions include the Frobenius norm and operator norm of $\widehat S_{\varphi}-S$.} 
{This approach is termed as \textit{Optimal Shrinkage} (OS), as named in previous literature \cite{donoho2018optimal, GD}.}

{In this paper, we investigate the matrix denoising problem and develop the associated OS under the model \eqref{eq_model} in} the {\em high dimensional} setup, where $p=p(n)$ and $p/n\to \beta\in (0,\infty)$ as $n\to\infty$. 
We begin by considering the white noise as a special example, where $Z = X$ in \eqref{eq_model}, and the entries of $X$ are i.i.d. with zero mean, variance $\sigma^2/n$ with $\sigma>0$, and finite fourth moment. Under this setup, asymptotically the empirical spectral distribution (ESD) of $ZZ^\top$ follows the Marchenko-Pastur (MP) law \cite{marvcenko1967distribution}. 
A phase transition occurs when the signal strength $d_i$ exceeds a critical value, usually referred to as the \textit{BBP (Baik-Ben Arous-P\`ech\`e) phase transition} named after the authors of \cite{baik2005phase}. Due to these peculiar behaviors, we need to modify the traditional SVD truncation scheme \cite{golub1965calculating} to recover $S$ from $\wt S$. 
In \cite{GD}, the optimal shrinker $\varphi^*$ with various loss functions is derived using the closed form of the rightmost bulk edge $\lambda_+$, and $\varphi^*$ is determined solely by $\wt\lambda_i$, $\sigma$ and $\beta$. {See an earlier work \cite{shabalin2013reconstruction} as well. Note that we do not mention the similar but different OS for the covariance structure. Overall,} such OS approach under the white noise assumption (referred to as TRAD hereafter) has found wide applications, such as fetal electrocardiogram (fECG) extraction from the trans-abdominal maternal ECG (ta-mECG) \cite{su2019recovery}, ECG T-wave quality evaluation \cite{su2020robust}, otoacoustic emission signal denoising \cite{liu2021denoising}, stimulation artifact removal from intracranial electroencephalogram (EEG) \cite{alagapan2019diffusion}, and cardiogenic artifact removal from EEG \cite{chiu2022get}. 

We illustrate the application of { TRAD} to the fECG extraction problem in Figure \ref{Figure1}. The ta-mECG recorded from the mother's abdomen during pregnancy is shown in Figure \ref{Figure1}(a). It is truncated into pieces marked by the red boxes and aligned based on the maternal R peaks. The fetal R peaks are labeled by blue crosses. Figure \ref{Figure1}(b) illustrates the associated data matrix $\widetilde S$, where the maternal ECG (mECG) is considered the signal and saved as matrix $S$, while the fECG and inevitable noise are jointly considered as the noise and saved as matrix $Z$. The results of TRAD, represented by $\widehat{S}$, are shown in red curves in Figure \ref{Figure1}(c), representing the recovered mECG after denoise. In \ref{Figure1}(d), $S-\widehat{S}$ is depicted, representing the recovered fECG. The decomposition of mECG and fECG is accomplished, allowing for clearer visualization of the fetal R peaks in Figure \ref{Figure1}(d). However, there are noticeable poorly recovered segments, as indicated by the black box, which is magnified in Figure \ref{Figure1}(e). For further details, a literature review of relevant algorithms, and clinical applications, readers are referred to \cite{su2019recovery}.

While TRAD has been successfully applied to {various scenarios, the assumption of white noise is overly restrictive}. For instance, physiologically, the fECG (considered noise) typically deviates from white characteristics. {Within a sample, it is evident that the fECG exhibits dependencies due to electrophysiology. Moreover, across different samples, there is dependence due to the long-range dependence in fECG recordings. In essence, in the fECG problem, noise manifests both intra-sample and inter-sample dependencies. Refer to Figure \ref{Figure1}(e) for an illustration, and Figure \ref{Figure1}(f) for the nontrivial estimated covariance structure of the fECG.} Another source of non-white noise arises from the application of different filters. Even if we assume the noise to be white, commonly employed filters can disrupt this white structure. Therefore, a modification of the white noise model and TRAD to handle scenarios with more complicated noise is needed.

{A natural generalization is considering the high dimensional setting with a dependent noise structure in \eqref{eq_model}.} This model is a generalization of the traditional high dimensional spiked model \cite{bai2008central,bai2012sample, baik2006eigenvalues}. It has been shown in \cite{limitpaper} that if the ESD of singular values from $Z$ converges to a non-random compactly supported probability measure, the top singular values of $\widetilde{S}$ and associated left and right singular vectors are all biased from those of $S$. These biases converge to a closed form depending on the $D$-transform \cite{limitpaper} and the Stieltjes transform \cite{widder1938stieltjes} of the limiting singular value distribution of the noise only matrix $Z$. {Moreover, we also have the BBP transition. Under this setup,} the optimal shrinker when $L_n$ is the Frobenius norm loss is derived in \cite{nadakuditi2014optshrink} (See Proposition \ref{prop_optimal_shrinker} below) under the following assumptions. First, the ESD of $Z$ asymptotically converges to a non-random compactly supported probability measure $\mu_Z$ so that the derivative of the $D$-transform of $\mu_Z$ is $-\infty$ at the right most edge of the compact support. Second, a critical delocalization conjecture of singular vectors holds. An OS algorithm named \textit{OptShrink} is provided in \cite{nadakuditi2014optshrink} by approximating the $D$-transform with respect to the truncated singular value distribution of $\wt S$ assuming the knowledge of rank. OptShrink is very general, but to our knowledge, it is challenging to directly apply it to real-world problems, unless some assumptions are imposed and we know the rank.

{Inspired by the OptShrink algorithm and driven by practical requirements and common dependence structures encountered, such as the fECG extraction problem described above,} this paper concentrates on noise exhibiting a separable covariance structure. \cite{paul2009no}:
\begin{equation}\label{colored and dependent noise model}
Z = A^{1/2}XB^{1/2}\,, 
\end{equation}
where {$X$ is a random matrix with independent entries and some moment conditions that will be specified later in Section \ref{subsec_ass}}, and $A$ and $B$ are respectively $p \times p$ and $n \times n$ deterministic positive-definite matrices that describe the colorness and dependence structure of noise, respectively. A random matrix satisfying the separable covariance structure has been studied and applied to many problems, like spatiotemporal analysis, wireless communication and several recent applications \cite{efron2009set,hoff2016limitations,gerard2015equivariant,drton2021existence}. When $B=I_n$, it is known that the ESD of $ZZ^\top$ converges to the {\em deformed MP} law \cite{marvcenko1967distribution}, and the distribution of the largest eigenvalue follows the Tracy-Widom law \cite{TW,TW1},  or commonly referred to as the {\em edge universality}. The edge universality has been proved for $A = I_p$ \cite{johnstone2001distribution,
pillai2014universality} and for general $A$ \cite{el2007tracy, onatski2008tracy, bao2015universality, lee2016tracy, ding2018necessary, Knowles2017} under various moment assumptions on the entries for $X$. For the sample singular vectors, the delocalization \cite{Knowles2017, o2016eigenvectors} and ergodicity \cite{principal} have been {shown}. When the general separable covariance structure is assumed, the convergence of the ESD to a limiting law is shown in \cite{paul2009no, wang2014limiting, lixin2007spectral}, the edge universality and delocalization of eigenvectors is established in \cite{ding2018necessary, yang2019edge}, and the local law (local estimates of the resolvents or the Green's functions) of $ZZ^\top$ is recently proved in \cite{xi2020convergence, yang2019edge}.

{Building upon the aforementioned theoretical results under various random matrix assumptions, significant efforts have been dedicated to designing matrix denoising algorithms tailored to the separable covariance noise model}. The special $B=I$ is discussed in \cite{leeb2021optimal} and the general $B$ is discussed in  \cite{leeb2021matrix}. They propose applying the whitening technique before applying TRAD; that is, performing TRAD on the whitened matrix $\widehat A^{-1/2}\wt S \widehat B^{-1/2}$, where $\widehat A$ and $\widehat B$ are the estimated $A$ and $B$ respectively. {This process yields $\tilde Q$, from which $S$ is estimated as $\widehat A^{1/2}\tilde Q\widehat B^{1/2}$, which is unwhitening.} This approach is effective when accurate estimates of $A$ and $B$ are obtainable, such as when both are diagonal. A fast numerical algorithm for $\lambda_+$ and the Stieltjes transform is available \cite{leeb2021rapid} when $Z=A^{1/2}XB^{1/2}$ and accurate estimates of the asymptotic spectral distribution of $A$ and $B$ are feasible. In \cite{donoho2020screenot}, if the upper bound of $r$ is known, $d_i$'s are distinct and several assumptions about the asymptotic ESD of $Z$ hold, an algorithm named {\em ScreeNOT} {({\em Scree} reminds the commonly used Scree plot \cite{cattell1966scree} in practice and {\em NOT} stands for Noise-adaptive Optimal Thresholding) is introduced, which evaluates the {\em optimal threshold} $\vartheta_{\texttt{SN}}>0$ via optimizing the Frobenius norm of the recovered data matrix and the clean data matrix and a technique termed \textit{imputation}. The optimal threshold could be applied to
discard weak components that are very close to the noise bulk edge, and approximates the underlying signal matrix,} or denosing the matrix, by retaining the top $\widehat r_{\texttt{SN}}$ singular values via $\widehat S_{\texttt{SN}}:=\sum_{i=1}^{\widehat r_{\texttt{SN}}}\sqrt{\widetilde{\lambda}_i}\widetilde{\boldsymbol\xi_i}\widetilde{\boldsymbol\eta_i}^\top$ \cite[(1.2)]{donoho2020screenot}, where
\begin{equation}\label{eq_SN}
\widehat r_{\texttt{SN}} = | \{\wt \lambda_i|\wt\lambda_i>\vartheta_{\texttt{SN}}\}|
\end{equation}
could be viewed as an estimated rank of the clean data matrix. 
The authors in \cite{donoho2020screenot} also provide a quantitatively interpretation of the Scree plot heuristic \cite{cattell1966scree}.  

\begin{figure}[hbt!]
\centering
\includegraphics[trim=120 20 140 60, clip, width=0.96\linewidth]{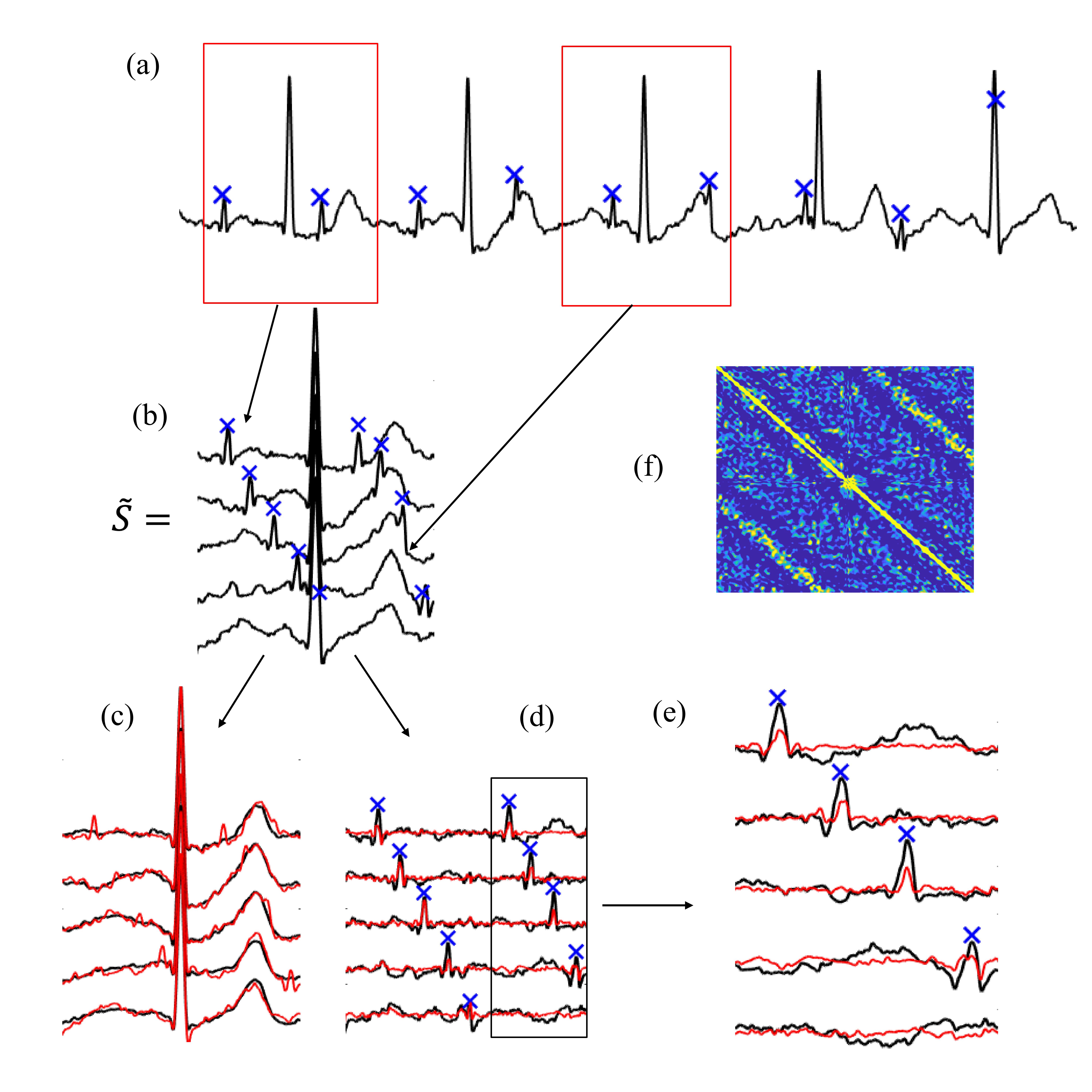}
\caption{\small  An illustration of extracting fECG from the ta-mECG shown in (a), where the fetal R peaks labeled by experts are marked as blue crosses. (b) is the data matrix including the pieces truncated from the ta-mECG and aligned by the maternal R peaks. The results of TRAD and our eOptShrink (both with the operator norm) as the estimated mECG are shown as red and black curves shown in (c) respectively.  By subtracting the estimated mECG from the ta-mECG, we obtain the estimated fECG shown in (d). The mECG estimation error could be visualized in the estimated fECG indicated by the black box, zoomed in in (e). The covariance structure of the fECG as noise is shown in (f).} \label{Figure1}
\end{figure}

Inspired by the success and limitation of existing algorithms and their broad application in data science, 
{within the separable covariance noise model \eqref{colored and dependent noise model}, in this paper} we introduce a novel {OS algorithm for matrix denoisng tailored} for real world data. This algorithm extends OptShrink \cite{nadakuditi2014optshrink} and is coined {\em extended OptShrink} (eOptShrink). Our contributions are multifold. 
First, we extend the results in \cite{benaych2011fluctuations} and derive the asymptotic behavior of the outlier singular values and associated biased singular vectors of the noisy data matrix $\wt S = S + A^{1/2}XB^{1/2}$. This includes the BBP phase transition, the sticking properties of  non-outlier  singular values, and the delocalization of the non-outlier singular vectors, with established convergence rate. These results {align with} those for deformed Wigner matrices \cite{knowles2013isotropic, knowles2014outliers}, deformed rectangular matrices \cite{limitpaper, ding2020high} and spiked covariance matrices \cite{paul2007asymptotics, principal, DY2019,ding2021spiked}, although they have not yet been established for our setup to the best of our knowledge. Recall that the delocalization is an essential condition needed for OptShrink \cite{nadakuditi2014optshrink}. 
Second, based on the developed theorems, we propose a fully data-driven OS algorithm leveraging these results. {
While OptShrink \cite{nadakuditi2014optshrink} typically requires spike rank information or an overestimated rank, our approach substitutes this requirement with a} data-driven {\em { effective} rank} estimation from $\wt S$. {Numerically we show that the proposed effective rank estimation outperforms $\widehat r_{\texttt{SN}}$ in \eqref{eq_SN} by ScreeNOT, and the rank estimation has its own interest. Moreover, whereas OptShrink estimates the $D$-transform based on the ``truncated'' singular value distribution,
our approach replaces this process with} a novel method to estimate the spectral density distribution of $Z$. This method relies on accurately recovering the singular values of $Z$ from those of $\wt S$ that are perturbed by signals. With these two components, we develop a more precise estimate of the Stieltjes transform and $D$-transform compared to existing imputation \cite{donoho2020screenot} and truncation \cite{nadakuditi2014optshrink} methods. Additionally, we extend OptShrink \cite{nadakuditi2014optshrink} to accommodate various loss functions beyond the Frobenius norm. %
Numerically, we evaluate the proposed algorithm by comparing it with existing algorithms \cite{nadakuditi2014optshrink,GD,donoho2020screenot} in simulated datasets and a real-world ta-mECG database. Overall, eOptShrink outperforms other algorithms. 

{
The paper is structured as follows. Section \ref{section prelim} provides essential mathematical background and summarizes optimal shrinkage, along with outlining the precise model assumptions. In Section \ref{section main result}, we elaborate on the main theoretical results presented in this study. Section \ref{section proposed algo} introduces our proposed algorithm, eOptShrink. Section \ref{section numreical eva} conducts a series of numerical evaluations of eOptShrink, including comparisons with existing algorithms using both simulated and real databases. Additional numerical simulations, technical details, and proofs of theorems can be found in the appendix.
}

NOTATION: For any random variable ${X}$, denote ${\wt X}$ as the perturbed ${X}$, and ${\widehat X}$ as the estimator of ${X}$. 
$C$ denotes a generic positive constant, whose value may change from one line to another. For sequences $\{a_n\}$ and $\{b_n\}$ indexed by $n$, $a_n = \OO(b_n)$ means that $|a_n| \le C|b_n|$ for some constant $C>0$ as $n\to \infty$ and $a_n=\oo(b_n)$ means that $|a_n| \le c_n |b_n|$ for some positive sequence $c_n\downarrow 0$ as $n\to \infty$. We also use $a_n \lesssim b_n$ if $a_n = \OO(b_n)$, $|a_n|\ll|b_n|$ if $a_n=\oo(b_n)$, and $a_n \asymp b_n$ if $a_n = \OO(b_n)$ and $b_n = \OO(a_n)$. For a matrix $M$, $\|M\|$ means its operator norm and we may abuse the notation and write $M=\OO(a)$ if $\|M\|=\OO(a)$. $I$ is reserved for the identity matrix of any dimension. For $a,b\in \mathbb{R}$, $a\vee b$ and $a\wedge b$ mean the maximal and minimal value of $a$ and $b$ respectively. We  reserve $E+i\eta$ to indicate a point in $\mathbb{C}^+$, the upper half plan of $\mathbb{C}$. See Tables \ref{tab:notations} and \ref{tab:notations2} for a list of notations. 

\section{Preliminaries}\label{section prelim}

\subsection{Background on random matrix theory}
{In this subsection, we review the following results necessary for our exploration.}
The Stieltjes transform of a probability measure $\nu$ {on $\mathbb{R}$} is $m_\nu(z):=\int{\frac{1}{\lambda-z}}d\nu(\lambda)$, where $z \in \mathbb C^+$. The ESD of an $n \times n$ symmetric matrix $H$ is defined as $\pi_H:=\frac{1}{n} \sum_{i=1}^n \delta_{\ell_i}$, where $\ell_1 \geq \ell_2 \geq \cdots \geq \ell_n$ are the eigenvalues of $H$ and $\delta$ means the Dirac delta measure. Denote the common eigenvalues of $ZZ^\top$ and $Z^\top Z$ (respectively $\widetilde{S}\widetilde{S}^\top$ and $\widetilde{S}^\top\widetilde{S}$) by $\{\lambda_i\}_{i=1}^{p \wedge n}$ (respectively $\{\wt\lambda_i\}_{i=1}^{p \wedge n}$).
For $z \in \mathbb{C}^+,$ denote the Green functions {(or resolvents)} of  $ZZ^\top$ and $Z^\top Z$ as 
\begin{equation}\label{eq_g1g2}
\mathcal{G}_1(z):=(ZZ^\top-z{I_p})^{-1} \ \mbox{ and }\  \mathcal{G}_2(z):=(Z^\top Z-z{I_n})^{-1}
\end{equation}
respectively.
The Stieltjes transforms of ESDs of $ZZ^\top$ and $Z^\top Z$ are formulated respectively by the Green functions $\mathcal{G}_1(z)$ and $\mathcal{G}_2(z)$ via
\begin{equation}\label{eq_green}
m_1(z):=\frac{1}{p} \text{Tr} \mathcal{G}_1(z)\ \mbox{ and }\  m_2(z):=\frac{1}{n} \text{Tr} \mathcal{G}_2(z). 
\end{equation}
Denote $\beta_n:=\frac{p}{n}.$ By a direct calculation, we have the relationship $m_2(z)=-\frac{1-\beta_n}{z}+\beta_n m_1(z)$. 
Similarly, we denote the Green functions of  $\wt S \wt S^\top$ and $\wt S^\top \wt S$ as 
\begin{equation}\label{eq_wtg1g2}
\mathcal{\wt G}_1(z):=(\wt S \wt S^\top-z{I_p})^{-1}\ \mbox{ and }\  \mathcal{\wt G}_2(z):=(\wt S ^\top \wt S-z{I_n})^{-1}
\end{equation}
respectively and the Stieltjes transforms of ESDs of $\wt S \wt S^\top$ and $\wt S^\top \wt S$ as 
\begin{equation}\label{eq_wtgreen}
\wt m_1(z):=\frac{1}{p} \text{Tr} \mathcal{\wt G}_1(z)\ \mbox{ and }\  \wt m_2(z):=\frac{1}{n} \text{Tr} \mathcal{\wt G}_2(z) \,.
\end{equation}

Next we summarize {existing knowledge about the behavior of $m_1(z)$ and $m_2(z)$ under mild assumptions. Assume} when $n \to \infty$, if $\beta_n \to \beta \in (0,\infty)$ and $\pi_A$ and $\pi_B$ converge weakly to deterministic probability distributions {$\rho_{A\infty}\neq \delta_0$ and $\rho_{B\infty}\neq \delta_0$}, and the entries of $X$ are independent with mean $0$, the same variance and finite fourth moments. {The first result is from \cite[Chapter 4]{lixin2007spectral} (See also \cite[Proposition 1.1]{couillet2014analysis} for a summary), which says that when $n\to\infty$, almost surely $\pi_{ZZ^\top}$ and $\pi_{Z^\top Z}$ converge to deterministic distributions, denoted as $\rho_{1 \infty}$ and $\rho_{2 \infty}$, respectively {with desirable properties. Without loss of generality, we detail this result assuming $\tau\leq \beta_n\leq 1$ for all $n$ and $\beta_n\to \beta$, where $\tau\in (0,1)$, and the case $1\leq \beta_n\leq \tau^{-1}$ holds similarly by taking a transpose.} This result comes from analyzing the Stieltjes transforms of available ESDs. {For any $z\in \mathbb{C}^+$,} denote $(\mathsf{M}_{1\infty}(z),\mathsf{M}_{2\infty}(z))\in \mathbb C^+\times \mathbb C^+$ as} the unique solution to the following system of self-consistent equations  \cite[Chapter 4]{lixin2007spectral}
{
\begin{equation}
\mathsf{M}_{1\infty}(z) = \beta\int\frac{x}{-z[1+x\mathsf{M}_{2\infty}(z) ]} \rho_{A\infty}(d x) \ \mbox{ and } \  \mathsf{M}_{2\infty}(z) =  \int\frac{x}{-z[1+x\mathsf{M}_{1\infty}(z)]} \rho_{B\infty}(d x) \,.
\end{equation}
Instead of working on the above two equations, it is sometimes more convenient to} consider a function $f_\infty$ on $\mathbb{C}^+\times \mathbb{C}^+$ satisfying
\begin{equation}\label{separable_MP}
f_\infty(z,m):=- m + \int\frac{x}{-z+x\beta \int\frac{t}{1+tm} \rho_{A\infty}(d t)} \rho_{B\infty}(d x)\,,
\end{equation}
{so that $\mathsf{M}_{2\infty}(z)$ is the unique solution to the equation $f_\infty(z,m)=0$; that is,}
$f_\infty(z,\mathsf{M}_{2\infty}(z))=0$ 
for $z\in \mathbb C^+$. {
With the above quantities, we define 
\begin{equation}\label{eq_m1infm2inf}
m_{1\infty}(z):= \int\frac{1}{-z[1+x\mathsf{M}_{2\infty}(z) ]} \rho_{A\infty}(d x)\ \mbox{ and }\ 
m_{2\infty}(z):= \int\frac{1}{-z[1+x\mathsf{M}_{1\infty}(z)]} \rho_{B\infty}(d x)\,,
\end{equation}
{where $m_{1\infty}$ is the Stieltje's transform of a deterministic probability measure $\mu_{1\infty}$ on $\mathbb{R}$ with the support in $[0,\infty)$ and $m_{2\infty}$ is the Stieltje's transform of a deterministic probability measure $\mu_{2\infty}=\beta\mu_{1\infty}+(1-\beta)\delta_0$  on $\mathbb{R}$ with the support in $[0,\infty)$. Moreover, for all $E\in \mathbb{R}\backslash\{0\}$, $\lim_{\eta\downarrow 0} m_{1 \infty}(E+i \eta)$ and $\lim_{\eta\downarrow 0} m_{2 \infty}(E+i \eta)$ exist, and if we denote 
\begin{equation}\label{extension of m1c and m2c to real}
m_{1\infty}(E):=\lim_{\eta\downarrow 0} m_{1 \infty}(E+i \eta)\ \mbox{ and }\ m_{2\infty}(E):=\lim_{\eta\downarrow 0} m_{2 \infty}(E+i \eta)\,, 
\end{equation}
it is shown in \cite[Theorem 3.1]{couillet2014analysis} that $\operatorname{Im}\, m_{1\infty}(E)$ and $\operatorname{Im}\, m_{2\infty}(E)$ are continuous on $E\in \mathbb{R}\backslash\{0\}$, and the probability measure $\mu_{1 \infty}$ has a continuous density function, which is denoted as $\rho_{1\infty}$ on $E\in (0,\infty)$. It can be further shown that $\rho_{1\infty}$ is analytic for every $E\in (0,\infty)$ for which $\rho_{1\infty}(E)>0$ \cite[Theorem 3.2]{couillet2014analysis}.  Since $\mu_{2\infty}=\beta\mu_{1\infty}+(1-\beta)\delta_0$ on $\mathbb{R}$, we similarly have the continuous density function $\rho_{2\infty}$ on $\mathbb{R}\backslash\{0\}$ associated with the probability measure $\mu_{2\infty}$.} 
See \cite[Section 2]{couillet2014analysis} for other properties of $\mu_{1\infty}$. We also know that when $n\to \infty$, almost surely $\pi_{ZZ^\top}$ and $\pi_{Z^\top Z}$ converge weakly to $\rho_{1\infty}(x)dx$ and $\rho_{2\infty}(x)dx$ respectively \cite[Chapter 4]{lixin2007spectral} (See also \cite[Proposition 1.1]{couillet2014analysis}).
Moreover, we have the edge universality and the averaged local law \cite{yang2019edge}}.

{
Several properties of $\mathsf{M}_{1\infty}$ and $\mathsf{M}_{2\infty}$ have been established as well. $\mathsf{M}_{1\infty}$ ($\mathsf{M}_{2\infty}$ respectively) is the Stieltje's transform of a Radon positive measure $q_{1\infty}$ ($q_{2\infty}$ respectively) on $\mathbb{R}$ with the support in $[0,\infty)$ \cite[Proposition 1.2]{couillet2014analysis}. For all $E\in \mathbb{R}\backslash\{0\}$, $\lim_{\eta\downarrow 0} \mathsf{M}_{1\infty}(E+i \eta)$ and $\lim_{\eta\downarrow 0} \mathsf{M}_{2\infty}(E+i \eta)$ exist, and if we denote 
\[
\mathsf{M}_{1\infty}(E):=\lim_{\eta\downarrow 0} \mathsf{M}_{1\infty}(E+i \eta)\ \mbox{ and }\ \mathsf{M}_{2\infty}(E):=\lim_{\eta\downarrow 0} \mathsf{M}_{2\infty}(E+i \eta)\,, 
\]
$\operatorname{Im}\, \mathsf{M}_{1\infty}(E)$ and $\operatorname{Im}\,\mathsf{M}_{1\infty}(E)$ are continuous on $E\in \mathbb{R}\backslash\{0\}$ \cite[Theorem 3.1]{couillet2014analysis}. Also, the measures $q_{1\infty}$ and  $q_{2\infty}$ have continuous derivatives on $\mathbb{R}\backslash\{0\}$, which are denoted as $\wp_{1\infty}$ and $\wp_{2\infty}$ on $E\in \mathbb{R}\backslash\{0\}$, and we have \cite[Theorem 3.1]{couillet2014analysis}
$\texttt{supp}(\rho_{1\infty})\cap(0,\infty)=\texttt{supp}(\wp_{1\infty})\cap(0,\infty)=\texttt{supp}(\wp_{2\infty})\cap(0,\infty)$. 

The above results have parallel results under the finite $n$ setup \cite[Theorem 2.3, Theorem 2.4, Lemma 2.5] {yang2019edge}, using tools from \cite{Knowles2017,10.1214/17-ECP97,10.1214/17-EJP42}. We will develop our algorithm and theory based on these results, which we summarize here.
}
Similar to the asymptotic statement above, {for any $z\in \mathbb{C}^+$,} define $(\mathsf{M}_{1c}(z),\mathsf{M}_{2c}(z))\in \mathbb C^+\times \mathbb C^+$ as the unique solution to the following system of self-consistent equations \cite[Theorem 2.3]{yang2019edge}
\begin{equation*}
\mathsf{M}_{1c}(z) = \beta_n \int\frac{x}{-z[1+x\mathsf{M}_{2c}(z) ]} \pi_A(d x) \ \mbox{ and } \  \mathsf{M}_{2c}(z) =  \int\frac{x}{-z[1+x\mathsf{M}_{1c}(z)]} \pi_B(d x) \,.
\end{equation*}
Similarly, $\mathsf{M}_{2c}(z)$ is the unique solution to the function $f$ defined on $\mathbb{C}^+\times \mathbb{C}^+$ by
\[
f(z,m):=- m + \int\frac{x}{-z+x\beta_n \int\frac{t}{1+tm} \pi_A(d t)} \pi_B(d x)
\]
for $z\in \mathbb C^+$. 
Define 
\begin{equation}\label{eq_m1cm2c}
m_{1c}(z):= \int\frac{1}{-z[1+x\mathsf{M}_{2c}(z) ]} \pi_A(d x)\ \mbox{ and }\ 
m_{2c}(z):= \int\frac{1}{-z[1+x\mathsf{M}_{1c}(z)]} \pi_B(d x)\,,
\end{equation}
which are the Stieltjes transform of probability measures $\mu_{1c}$ and $\mu_{2c}$ with continuous density functions $\rho_{1c}$ and $\rho_{2c}$ \cite[Theorem 2.3]{yang2019edge}.
Moreover, $\mathsf{M}_{1c}(z)$ and $\mathsf{M}_{2c}(z)$ are the Stieltjes transforms of two Radon positive measures with the densities $\wp_{1c}$ and $\wp_{2c}$ with the same support of $\rho_{1c}$ on $\mathbb{R}\backslash \{0\}$. 
It has been known that when $n\to \infty$, $m_{1c}(z)-m_{1}(z)$ and $m_{2c}(z)-m_{2}(z)$ converge to zero uniformly {over properly chosen $z$} \cite[Theorem S.3.9(b)]{DY2019}.

{Next, we summarize relevant results for $\rho_{1c}$, $\rho_{2c}$, $\wp_{1c}$ and $\wp_{2c}$ when $\pi_{A}$ and $\pi_{B}$ are compactly supported with mild assumptions in the following lemma \cite[Lemma 2.5] {yang2019edge}.}

\begin{lemma}\label{lemma compact support of mu result}
Let $A$ and $B$ be deterministic symmetric matrices with eigenvalues $\sigma_1^a \geq \sigma_2^a \geq \ldots \geq \sigma_p^a$ and  $\sigma_1^b \geq \sigma_2^b \geq \ldots \geq \sigma_n^b$ respectively and satisfy
$\sigma_1^a\vee  \sigma_1^b \le \tau^{-1}$  and $\pi_A([0,\tau])\vee  \pi_B([0,\tau]) \le 1 - \tau$ for a small $\tau>0$.
Then, $\texttt{supp}(\rho_{1c})$ is compact with 
\begin{equation}
\texttt{supp}(\wp_{1c})=\texttt{supp}(\wp_{2c})=\texttt{supp}(\rho_{1c})=\bigcup_{k=1}^{L} [e_{2k}, e_{2k-1}]\,,
\end{equation}
where $e_1>e_2>\ldots>e_{2L}$.
Moreover, $(x,m)=(e_k, \mathsf{M}_{2c}(e_k))$ are the real solutions to the equations 
\[
f(x,m)=0\ \ \mbox{ and }\ \ \frac{\partial f}{\partial m}(x,m) = 0\,.
\] 
Finally, $e_1$ is bounded, $\mathsf{M}_{1c}(e_1) \in (-(\sigma_1^b)^{-1}, 0)$ and $\mathsf{M}_{2c}(e_1) \in (-(\sigma_1^a)^{-1}, 0)$.
\end{lemma}

Note that this is the parallel result of that shown in \cite[Proposition 3.3,\ Proposition 3.4,\ Theorem 3.3]{couillet2014analysis}, which says that if $\texttt{supp}(\rho_{A\infty})\cap (0,\infty)$ and $\texttt{supp}(\rho_{B\infty})\cap (0,\infty)$ consist of $K$ and $\tilde{K}$ connected components respectively, then $\texttt{supp}(\rho_{1\infty})\cap (0,\infty)$ consists of $L$ connected components, where $L\leq K\tilde K$ depends only on $\rho_{A\infty}$ and $\rho_{B\infty}$. 
Second, $\texttt{supp}(\rho_{A\infty})$ and $\texttt{supp}(\rho_{B\infty})$ are compact if and only if $\texttt{supp}(\rho_{1\infty})$ is compact with 
$\texttt{supp}(\rho_{1\infty})=\bigcup_{k=1}^{L} [e_{2k,\infty}, e_{2k-1,\infty}]$,
where $e_{1,\infty}>e_{2,\infty}>\ldots>e_{2L,\infty}$.
Moreover, $(x,m)=(e_k, \mathsf{M}_{2\infty}(e_k))$ are the real solutions to the equations 
$f_\infty(x,m)=0\ \ \mbox{ and }\ \ \frac{\partial f_\infty}{\partial m}(x,m) = 0$.
Finally, $e_{1,\infty}$ is bounded, $\mathsf{M}_{1\infty}(e_{1,\infty}) \in (-(\max\texttt{supp}(\rho_{B\infty}))^{-1}, 0)$ and $\mathsf{M}_{2\infty}(e_{1,\infty}) \in (-(\max\texttt{supp}(\rho_{A\infty}))^{-1}, 0)$.

{From now on, we will carry out analysis following \cite{Knowles2017,principal,yang2019edge,DY2019}.} We call $e_k$ spectral edges {or bulk edges}. In this paper, we focus on the rightmost edge 
\[
\lambda_+ := e_1\,. 
\]
$[e_{2L},\,e_1]$ could be viewed as the ``spectral spreading'' of the noise, and intuitively, the signal should be sufficiently strong compared with $\lambda_+$ ``in some way'' so that the signal can be observed. This ``strength'' will be precisely described below.

For $z \in \mathbb C^+$, denote the $D$-transform of $\rho_{1c}$ as 
\begin{equation}\label{eq_defnt}
 \mathcal{T}(z):=zm_{1c}(z) m_{2c}(z)\,.
 \end{equation}
{The $D$-transform is the key tool we count on to design our algorithm, and here are some of its basic properties. Similar to the asymptotic case in \eqref{extension of m1c and m2c to real},} $m_{1 c}$ and $m_{2 c}$ can be extended {from $\mathbb{C}^+$} to $x> \lambda_+$ via
\begin{equation}\label{eq_defm1m2}
    m_{1 c}(x) = \int_0^{\lambda_+} \frac{\rho_{1 c}(t)}{t-x}dt\  \mbox{ and }\  m_{2c}(x) = \int_0^{\lambda_+} \frac{\rho_{2 c}(t)}{t-x}dt\,.
\end{equation}
Hence, when $x>\lambda_+$, $m_{1c}(x)$, $m_{2c}(x)$ and $\mathcal{T}(x)$ are well-defined. Moreover, by a direct calculation, both $m_{1c}(x)$ and $m_{2c}(x)$ are negative and monotonically increasing  {over  $x>\lambda_+$} and hence $\mathcal{T}(x)$ is monotonically decreasing {over $x>\lambda_+$}, such that $m_{1c}(x)$, $m_{2c}(x)$ and $\mathcal{T}(x)$ are invertible {over  $x>\lambda_+$}. This is the key property we count on to study the signal deformation under the high dimensional noise. {See \cite[Section 2.5]{limitpaper} for more discussion of the $D$-transform.} 
Denote the $j$-th {\em classical location} of the probability density $\rho_{2c}$, where $j=1,\ldots,n$, as
\begin{equation}\label{eq_classical} 
    \gamma_j := \sup_x\left\{\int_x^{+\infty}\rho_{2c}(x)dx > \frac{j-1}{n}\right\}\,.
\end{equation}
In particular, we have $\gamma_1 = \lambda_+$.

\subsection{Model Assumption}\label{subsec_ass}
In this subsection we impose assumptions on the model \eqref{eq_model}.
We follow \cite{yang2019edge,ding2021spiked} and consider the following bounded support condition to simplify the following discussion.

\begin{definition}\label{bsc}(Bounded support condition).
A random matrix $X=[x_{ij}]_{\tiny\substack{i=1,\ldots,p\\ j=1,\ldots,n}}$ is said to have a bounded support $\phi_n>0$ if $\max_{i,j}|x_{i,j}| \leq \phi_n$, where $\phi_n= n^{2/a-1/2}$  for some constant $a>4$ is a deterministic parameter. 
\end{definition}

The bounded support assumption is introduced to simplify the discussion and it can be easily removed. {Recall that for a random matrix $X$ whose entries are independent random variables fulfilling $\mathbb{E}x_{i,j}=0$, $\mathbb{E}|x_{i,j}|^2=n^{-1}$ and $\max_{i,j}\mathbb{E}|\sqrt{n}x_{i,j}|^a\leq C$ for $C>0$ and $a > 4$, it can be truncated to a random matrix $\widetilde X$ with bounded support $\phi_n = n^{2/a-1/2+\varepsilon'}$, where $\varepsilon'>0$ is a constant that can be arbitrarily small, fulfilling Assumption \ref{assum_main}(i), so that $\mathbb{P}(\widetilde X=X)=1 - O(n^{-a\varepsilon'})$. Refer to \cite[Section 4.3.2]{bai2010spectral} and \cite[Corollary 3.19]{ding2021spiked} for details. Consequently, all arguments and theoretical results with the bounded support condition can be extended to the setup without the bounded support condition with probability $1-o(1)$.}

\begin{assumption}\label{assum_main}
Fix a small constant $0<\tau<1$. We need the following assumptions for the model \eqref{eq_model}:
\begin{itemize}
\item[(i)] (Assumption on $x_{ij}$).  Suppose $X=[x_{ij}]_{\tiny\substack{i=1,\ldots,p\\ j=1,\ldots,n}}\in \mathbb{R}^{p\times n}$ has a bounded support $\phi_n$ and its entries are independent and satisfy
\begin{align}
\max_{i,j}\left|\mathbb{E} x_{ij}\right|   \le n^{-2-\tau} \,, \   \max_{i,j}\left|\mathbb{E} | x_{ij} |^2  - n^{-1}\right|   \le n^{-2-\tau} \label{entry_assm1}
\end{align}
and
\begin{equation}
\max_{i,j}\mathbb{E}|\sqrt{n} x_{ij}|^4 \leq C_4
\end{equation}
for a constant $C_4>0$.

\item[(ii)] (Assumptions on $p/n$). $p=p(n)$ satisfies
\begin{equation}
\tau<\frac{p}{n} <\tau^{-1}\,.   
\end{equation}

\item[(iii)] (Assumption on $A$ and $B$). Assume $A$ and $B$ are deterministic symmetric matrices with eigendecompositions
\begin{equation}\label{eigen}
A= Q^a\Sigma^a (Q^a)^\top\ \mbox{ and }\  B= Q^b \Sigma^b (Q^b)^\top 
\end{equation}
respectively, where
$\Sigma^a=\textup{diag}(\sigma_1^a, \ldots, \sigma_p^a)$,  $\Sigma^b=\textup{diag}( \sigma_1^b, \ldots,  \sigma_n^b)$,
$\sigma_1^a \geq \sigma_2^a \geq \ldots \geq \sigma_p^a$,  $\sigma_1^b \geq \sigma_2^b \geq \ldots \geq \sigma_n^b,$
$Q^a= ({\bf q}^a_1, \cdots, {\bf q}^a_p)\in O(p)$ and $Q^b= ({\bf q}^b_1, \cdots, {\bf q}^b_n)\in O(n)$.
We assume that {for all sufficiently large $n$,} 
\begin{equation}\label{ass3_eq1}
1+\mathsf{M}_{1c}(\lambda_+) \sigma_1^b \geq \tau\ \mbox{ and } \ 1+\mathsf{M}_{2c}(\lambda_+)  \sigma_1^a \geq \tau
\end{equation}
and 
\begin{equation}\label{ass3_eq2}
\sigma_1^a\vee  \sigma_1^b \le \tau^{-1}\ \mbox{ and } \ \pi_A([0,\tau])\vee  \pi_B([0,\tau]) \le 1 - \tau .
\end{equation}

\item[(iv)] (Assumption on the signal strength).  
  We assume \begin{equation}
      d_1 \geq d_2 \geq \cdots \geq d_r > 0
  \end{equation}
for some $r\geq 1$ and $d_1<\tau^{-1}$. {We allow the signal strength $d_i$ to depend on $n$ but $r$ is independent of $n$.} Denote a fixed value $\alpha>0$ as 
\begin{equation}\label{eq_defnalpha}
\alpha:=1/\sqrt{\mathcal{T}(\lambda_+)}\,.
\end{equation} 
We allow the singular values $d_k$ to depend on $n$ under the condition that there exists an integer $1\leq r^+\leq r${, called the effective rank of $\wt S$,} such that
\begin{equation}\label{eq_assm_41}
d_k -\alpha > \phi_n + n^{-1/3} \ \mbox{ if and only if }\ 1\leq k\leq r^+\,,
\end{equation}

\item[(v)] (Assumption on the distribution of singular vectors). 
Let $G_\ub^p \in \mathbb{R}^{p \times r}$ and $G_\vb^n \in \mathbb{R}^{n \times r}$ be two independent matrices with i.i.d entries distributed according to a fixed probability measure $\nu$ on $\mathbb{R}$ with mean zero and variance one, and satisfy the log-Sobolev inequality (see below). We assume that the left and right singular vectors, $\ub_i \in \mathbb{R}^p$ and $\vb_i \in \mathbb{R}^n$, {are either 
\begin{enumerate}
\item[a.]{ the $i$-th column of $\frac{1}{\sqrt{p}} G_\ub^p$ and $\frac{1}{\sqrt{n}} G_\vb^n$ respectively, or}
\item[b.]{ the $i$-th columns obtained from the Gram-Schmidt (or QR factorization) of $G_\ub^p$ and $G_\vb^n$ respectively.} 
\end{enumerate}
}
\end{itemize}
\end{assumption}

{We comment on these assumptions. First, in \textit{(iii)}}, \eqref{ass3_eq1} guarantees a regular square-root behavior of the spectral densities $\rho_{1c}$ and $\rho_{2c}$ near $\lambda_+$ (See Lemma \ref{lambdar_sqrt}){, which is based on the behavior of $\mathsf M_{1c}(e_1)$ and $\mathsf M_{2c}(e_1)$ in Lemma \ref{lemma compact support of mu result},}
and rules out the existence of outlier eigenvalues of $ZZ^\top$. Therefore, it also makes sure that $\mathcal{T}(\lambda_+)$ is well-defined so as $\alpha$ in \textit{(iv)}. \eqref{ass3_eq2} means that the spectra of $A$ and $B$ cannot concentrate at zero.

In \textit{(iv)}, $d_i$ fulfilling \eqref{eq_assm_41} is considered to be the singular value associated with a  {``sufficiently strong'' signal. Recall that the signal is described by singular vectors in \eqref{eq_model}}. Below, we will show that when a signal is sufficiently strong, it leads to a singular value {outside the spectral bulk} of $Z$. It is the reason that it is possible to study the signal. If $d_i$ does not satisfy \eqref{eq_assm_41}, they are considered ``weak'', associated with a ``weak signal''. In Theorem \ref{thm_value}, we will show that weak signals will be {``sticked'' to the bulk edge of $Z$}, which is related to the \textit {BBP phase transition}. This is why $r^+$ is called the {{\em effective rank}} of $\wt S$ in \eqref{eq_assm_41}. See Figure \ref{IMG_9221} for an illustration of the relationship of $\mathcal{T}$, $\alpha$ and $\lambda_+$. {From now on, we call the top $r^+$ singular values of $\wt S$ outlier signal strengths or simply {\em outliers}, and the corresponding singular vectors the {\em outlier signals}, and call the $r^++1,\ldots,r$-th singular values of $\wt S$ non-outlier signal strengths or simply {\em non-outliers}, and the corresponding singular vectors the {\em non-outlier signals}. Similarly, we call the top $r^+$ eigenvalues of $\wt S\wt S^\top$ {\em outliers}, and call the $r^++1,\ldots,r$-th eigenvalues of $\wt S\wt S^\top$ {\em non-outliers}.}

{The random signal setup in \textit{(v)} is the same as that in \cite{limitpaper} and others.
The log-Sobolev inequality (see \cite[Section 2.3.2]{anderson2010introduction} for the definition) implies that entries of $\ub_i$ and $\vb_i$ have sub-Gaussian tails, which leads to the desired concentration property  \cite[Proposition 6.2]{benaych2011fluctuations} that is summarized in Lemma \ref{hanson}.  
 If $\bu_i$ and $\bv_i$ have deterministic entries and $Z$ is random but has a bi-unitarily invariant distribution, then we are in the same setting as the second model of (v) \cite[Remark 2.6]{limitpaper}. However, we know that $Z$ is bi-unitarily invariant if $A$ and $B$ are both unitarily invariant, which is not applicable to our case. How to handle the deterministic signal setup will be explored in our future work.}

\subsection{Optimal shrinkers for various loss functions}
Recall the definition of asymptotic loss and optimal shrinker provided in \cite[Definitions 1 and 2]{GD} . 
\begin{definition}[\bf Asymptotic loss]\label{def_1}
Let $\mathcal{L}:=\{L_{p,n}| p, n \in \mathbb{N}\}$ be a family of loss functions, where each $L_{p,n}: M_{p \times n} \times M_{p \times n} \rightarrow [0,\infty)$ is a loss function obeying { 
that $\widehat S \to  L_{p,n}(S, \widehat S)$ is continuous and
$L_{p,n}(S,S)=0.$} Suppose $p=p(n)$ and $\lim_{n \to \infty} p(n)/n \to \beta>0$.  Let $\varphi: [0, \infty) \rightarrow [0, \infty)$ be a nonlinear function and consider $\widehat{S}_{\varphi}$ to be the singular shinkage estimate \eqref{eq_shrink}. When $\lim_{n \rightarrow \infty} L_{p,n}$ exists, we define the asymptotic loss of the shrinker $\varphi$ with respect to $L_{p, n}$ with the signal $\db=(d_1, \cdots, d_r)$ as 
$L_{\infty}(\varphi |\db) = \lim_{n \rightarrow \infty} L_{p,n}(S, \widehat{S}_{\varphi})$, where $S$ is defined in \eqref{eq_model}.
\end{definition}
\begin{definition}[\bf Optimal shrinker]\label{def_2}  Let $L_{\infty}$ and $\varphi$ as defined in Definition \ref{def_1}. If a shrinker $\varphi^*$ has an asymptotic loss that satisfies 
$L_{\infty} (\varphi^*| \db) \leq L_{\infty} (\varphi | \db)$
for any other shrinker $\varphi$, any $r \geq 1$, and any $\db \in \mathbb{R}^r,$ then we say that $\varphi^*$ is unique asymptotically admissible (or simply ``optimal'') for the loss family $\mathcal{L}$ and that class of shrinkers. \end{definition}

{From now on, we denote the optimal shrinker of $\wt\lambda_i$ as 
\begin{equation}
\varphi^*_i :=\varphi^*(\wt\lambda_i) 
\end{equation}
for convenience.} In \cite[Sections IV. A and C]{GD}, the optimal shrinkers under different loss functions were computed. 
When the loss function is the operator norm, the optimal shrinker was proved to be $\varphi^*_i=d_i$ when $p=n$ in \cite[Section IV.B]{GD}. Later in \cite[Lemma 5.1]{leeb2020optimal}, it is shown that $\varphi^*_i=d_i\sqrt{\frac{a_{1,i}\wedge a_{2,i}}{a_{1,i}\vee a_{2,i}}}$ when $p \neq n$, where $a_{1,i} := \lim_{n\to\infty}\langle \ub_i, \wt\bxi_{i}  \rangle^2$ and $a_{2,i} :=   \lim_{n\to\infty}\langle \vb_i, \wt\bzeta_{i}  \rangle^2$.
We list the results here for readers' convenience. 

\begin{proposition}[\cite{GD,leeb2020optimal}]\label{prop_optimal_shrinker}
When $d_i \geq \alpha$, the optimal shrinker is 
$\varphi^*_i=d_i \sqrt{a_{1,i} a_{2,i}}$, $\varphi^*_i=d_i\sqrt{\frac{a_{1,i}\wedge a_{2,i}}{a_{1,i}\vee a_{2,i}}}$ and $\varphi^*_i=d_i\big(\sqrt{a_{1,i} a_{2,i}}-\sqrt{(1-a_{1,i})(1-a_{2,i})}\ \big)$ when the
Frobenius norm, operator norm and nuclear norm are considered in the loss function respectively. When $d_i < \alpha$, for any loss function, we have 
$\varphi^*_i=0$.
\end{proposition}

With Proposition \ref{prop_optimal_shrinker}, if we can estimate $d_i$, $a_{1,i}$ and $a_{2,i}$ using the eigenstructure of the noisy matrix $\wt S\wt S^\top$, we could obtain the desired optimal shrinkers. For example, as shown in \cite{limitpaper}, if the ESD of $Z$ converges almost surely weakly to a compactly supported probability measure, then when $d_i>\alpha$, we have that $d_i=\frac{1}{\sqrt{\mathcal{T}(\wt\lambda_i)}}$, $a_{1,i}=\frac{m_{1c}(\wt\lambda_i)}{d_i^2 \mathcal{T}'(\wt\lambda_i)}$ and $a_{2,i}=\frac{m_{2c}(\wt\lambda_i)}{d_i^2 \mathcal{T}'(\wt\lambda_i)}$, based on which the optimal shrinker with respect to the Frobenius norm is derived by replacing $d_i\sqrt{a_{1,i}a_{2,i}}$ with the corresponding values in \cite{nadakuditi2014optshrink}. Moreover, as shown in \cite{GD}, when the noise is white with $\sigma>0$, and $X$ have independent entries with zero mean, unit variance, and finite fourth moment, if $d_i>\beta^{1/4}$, we have $\wt\lambda_i>\sigma(1+\sqrt{\beta})$  as $n \to \infty$. Denote $y_i = \wt\lambda_i/\sigma^2$. We have $d_i = \ell_i$, $a_{1,i} = {\frac{\ell_i^4-\beta}{\ell_i^4+\beta\ell_i^2}}$ and $a_{2,i} = {\frac{\ell_i^4-\beta}{\ell_i^4+\ell_i^2}}$,
where $\ell_i:= \frac{1}{\sqrt{2}} \sqrt{y_i-\beta-1+\sqrt{({y_i}-\beta-1)^2-4\beta}}.$
In this special case, when $\wt\lambda_i>\sigma(1+\sqrt{\beta})$, the optimal shrinker has the closed form 
${\varphi}_i= \frac{\sigma}{y_i} \sqrt{(y_i^2-\beta-1)^2-4\beta}$, 
${\varphi}_i=\sigma\ell_i$ and ${\varphi}_i= \frac{\sigma}{\ell_i^2 y_i}(\ell_i^4-\beta-\sqrt{\beta}\ell_i y_i)$ respectively when the Frobenius norm, operator norm, and nuclear norm respectively is considered,
and when $\wt\lambda_i \leq \sigma(1+\sqrt{\beta})$, $\varphi^*_i = 0$ for any loss function. The rank can be decided by how many $\wt\lambda_i>\sigma(1+\sqrt{\beta})$. If $\sigma$ is unknown, it is suggested in \cite{donoho2013optimal}, {among others, e.g., \cite{shabalin2013reconstruction},} to estimate $\sigma$ by 
$\check\sigma(\wt S) := s_{med}/\sqrt{\mu_{\beta}}$,
where $s_{med}$ is a median singular value of $\wt S$ and $\mu_\beta$ is the median of the MP distribution. To extend OptShrink under our noise setup \eqref{colored and dependent noise model}, we study the asymptotic behavior of singular values and vectors and their relation with $d_i$, $a_{1,i}$ and $a_{2,i}$ with a convergence rate.

\section{Main Results}\label{section main result}
In this section we state the main theoretical results about the biased singular values and vectors, including the limiting behavior and the associated convergence rate. These results form the foundation of the proposed eOptShrink algorithm {and could be of its own interest}.
To simplify the presentation of our results and their proofs, we apply the notion of stochastic domination, which is a systematic framework to state results of the form ``$\xi$ is bounded by $\zeta$ with high probability up to a small power of $n$" \cite{erdHos2013averaging}. %

\begin{definition}[\bf Stochastic domination]\label{stoch_domination}
Let $\xi=\left(\xi^{(n)}(u):n\in\mathbb{N}, u\in U^{(n)}\right)$ and  $\zeta=\left(\zeta^{(n)}(u):n\in\mathbb{N}, u\in U^{(n)}\right)$
be two families of nonnegative random variables, where $U^{(n)}$ is a possibly $n$-dependent parameter set. We say $\xi$ is stochastically dominated by $\zeta$, uniformly in $u$, if for any fixed (small) $\epsilon>0$ and (large) $D>0$, 
$\sup_{u\in U^{(n)}}\mathbb{P}\left[\xi^{(n)}(u)>n^\epsilon\zeta^{(n)}(u)\right]\le n^{-D}$
for large enough $n \ge n_0(\epsilon, D)$, and we shall use the notation $\xi\prec\zeta$ or $\xi=O_\prec(\zeta)$. Throughout this paper, the stochastic domination will always be uniform in all parameters that are not explicitly fixed, such as matrix indices, and $z$ that takes values in some compact set. Note that $n_0(\epsilon, D)$ may depend on quantities that are explicitly constant, such as $\tau$ in Assumption \ref{assum_main}. 
Moreover, we say an event $\Xi$ holds with high probability if for any constant $D>0$, $\mathbb P(\Xi)\ge 1- n^{-D}$, when $n$ is sufficiently large.
\end{definition}

We denote $\{\bxi_i\}_{i=1}^p$ and $\{\bzeta_i\}_{i=1}^n$ as respectively the left and right singular vectors of the noise matrix $Z$, and $\{\widetilde{\bxi}_i\}_{i=1}^p$ and $\{\widetilde{\bzeta}_i\}_{i=1}^n$ as respectively the left and right singular vectors of $\widetilde{S}$, which can be viewed as a perturbation of $Z$ by adding $S$.  
For $x > \alpha$, denote $\theta(x)$, $a_1(x)$ and $a_2(x)$ by 
\begin{equation}\label{eq_functions}
\theta(x):=\mathcal{T}^{-1}(x^{-2}), \  a_1(x)=\frac{m_{1c}(\theta(x))}{x^2 \mathcal{T}'(\theta(x))}, \   a_2(x)=\frac{m_{2c}(\theta(x))}{x^2 \mathcal{T}'(\theta(x))}\,.
\end{equation} 
{Clearly, on $(\lambda_+,\infty)$, $\mathcal{T}(x)$ is monotonically decreasing as $x\to \infty$, so $\theta(x)$ is monotonically increasing as $x\to \infty$.} These terms are used to estimate the signal strength $d_i$ and inner products of the clean and noisy left singular vectors $ \langle \ub_i, \wt \bxi_j \rangle$ and clean and noisy right singular vectors $\langle \vb_i, \wt \bzeta_j \rangle$, which we will detail in the following theorems.

\subsection{Results of singular values}

We first state the results for singular values. Define 
\begin{equation}
\mathbb O_+ = \{1,\cdots,r^+\}
\ \mbox{ and } \  
    \Delta(d_i) := |d_i-\alpha|^{1/2}\,.
\end{equation}
{In other words, the outlier singular values are indexed by $\mathbb O_+$}. 
The followings are our main theorems, and their proofs are postponed to  the supplementary. Part of the proof of the following theorems is motivated by \cite{yang2019edge} and \cite{DY2019}, where the focus is the covariance matrix analysis. The main difference comes from the fact that in general, the covariance model in \cite{yang2019edge,DY2019} cannot be directly decomposed into the summation of two independent matrices, like $S$ and $Z = A^{1/2}XB^{1/2}$ in our model, so the results cannot be directly applied.
We first state the location of the outlier eigenvalues and the first a few non-outlier eigenvalues of $\wt S\wt S^\top$.
\begin{theorem} \label{thm_value} Suppose Assumption \ref{assum_main} holds. Then we have for $1 \leq i \leq r^+,$
\begin{equation}\label{eq_evoutlier}
|\widetilde{\lambda}_i-\theta(d_i)| \prec  \phi_n \Delta(d_i)^2+n^{-1/2}\Delta(d_i) \,. 
\end{equation} 
Furthermore, for a fixed integer $\varpi>r^+$, we have for $r^++1 \leq i \leq \varpi$,
\begin{equation}\label{eq_evbulk}
|\widetilde{\lambda}_i-\lambda_+| \prec \phi_n^2+n^{-2/3}\,.
\end{equation}
\end{theorem}

{The above theorem gives the absolute error bounds} for the locations of the outlier eigenvalues and the top finite non-outlier eigenvalues of $\wt S \wt S^\top$ and implies the occurrence of the BBP transition \cite{baik2005phase}. When $\phi_n \leq n^{-1/3}$, the right-hand side of \eqref{eq_evbulk} becomes $n^{-2/3}$; that is, the non-outliers associated with weak signals deviate from the bulk edge with the order of the Tracy-Widom law \cite{tracy2002distribution}.

Next, we study the  \emph{eigenvalue sticking,} which states how the non-outlier eigenvalues of $\wt S\wt S^\top$ ``stick" to the eigenvalues of $ZZ^\top$.

\begin{theorem}\label{thm_eigenvaluesticking}
Suppose Assumption \ref{assum_main} holds and assume
\begin{equation}\label{alpha+}
\alpha_+:= \min_{1\leq i \leq r} \Delta(d_i)^2
\ge n^{\varepsilon} (\phi_n+n^{-1/3})
\end{equation}
for a small positive constant $\varepsilon<\frac{1}{2}-\frac{2}{a}$. Fix a small constant $c>0$. 
We have
\begin{equation} \label{eq_stickingeq}
\left|\wt\lambda_{r^+ +i}-\lambda_i\right| \prec \frac{1}{n \alpha_+} + n^{-3/4} + i^{1/3}n^{-5/6} + n^{-1/2}\phi_n +  i^{-2/3}n^{-1/3}\phi_n^2  
\end{equation}
for all $ 1 \leq i \leq c n$.
If either (a) 
$\mathbb E x_{ij}^3 = 0$ for all $1\le i \le p$ and $1\le j \le n$,
or (b) either $A$ or $B$ is diagonal, we have a stronger bound 
\begin{equation} \label{eq_stickingeq_strong}
\left|\wt\lambda_{r^++i}-\lambda_i\right| \prec \frac{1}{n \alpha_+} 
\end{equation}
for all $1 \leq i \leq \tau n$.
\end{theorem}

{Theorem \ref{thm_eigenvaluesticking} establishes the absolute error bound} for the non-outlier
eigenvalues of $\wt S \wt S ^\top$ with respect to the eigenvalues of $ZZ^\top$. The parameter $\alpha_+$ impacts how strongly non-outlier eigenvalues of $\wt S\wt S^\top$ stick to eigenvalues of $ZZ^\top$. {We shall mention that when \eqref{alpha+} holds, for $r^++1 \leq i \leq \varpi$, where $\varpi>r^+$ is a fixed integer, the right-hand side of \eqref{eq_stickingeq} is consistently better than that of \eqref{eq_evbulk}, particularly when $\phi_n\gg n^{-1/3}$; that is, we obtain a sharper sticking bound for non-outliers.
When \eqref{alpha+} fails, for example, when $d_i=\alpha$ for some $i$, we cannot get this sharper bound and we can only control} the behavior of $\wt\lambda_i$ for $i = r^++1, \ldots, r$ by \eqref{eq_evbulk}. 
When $\alpha_+\gg n^{-1/3}$ and $\phi_n \leq n^{-1/6}$, the right-hand side of \eqref{eq_stickingeq} and \eqref{eq_stickingeq_strong} are much smaller than $n^{-2/3}$.

\subsection{Results of singular vectors}

Next we discuss the singular vectors. We first show that when $i \notin \mathbb{O}_+$, the non-outlier singular vectors, $\wt\bxi_i$ and $\wt\bzeta_i$, are distributed roughly uniformly and approximately perpendicular to the space spanned by $\{\ub_j\}_{j=1}^{r}$ and $\{\vb_j\}_{j=1}^{r}$ respectively. This {\em delocalization} is the key step toward designing eOptShrink. {Moreover, for $i \in \mathbb{O}_+$, if the signal is strong but not strong enough, a similar behavior appears.} 

\begin{theorem}\label{thm_noneve} Suppose Assumption \ref{assum_main} holds. For $i=1,\ldots,n$, denote 
$\eta_i:=n^{-3/4}+n^{-5/6} i^{1/3} +  n^{-1/2}\phi_n$ and $ \varkappa_i:=i^{2/3}n^{-2/3}$. For any sufficiently small constant $c>0$ and $r^++1\leq i \leq cn$, we have for $j = 1, \ldots, r$,
\begin{equation}\label{eq_evebulka}
 |\langle \ub_j, \wt\bxi_{i}  \rangle|^2 \vee |\langle \vb_j, \wt\bzeta_{i}  \rangle|^2 \prec  \frac{n^{-1} +\phi_n^3+  \eta_i \sqrt{\varkappa_i}}{\Delta(d_j)^4 + \phi_n^2+\varkappa_i}\,.
\end{equation}
If either $A$ or $B$ is diagonal, the right hand side of \eqref{eq_evebulka} becomes $\frac{n^{-1}+\phi_n^3}{ \Delta(d_j)^4+\phi_n^2+\varkappa_i}$. 
Moreover, fix a constant $\wt\tau \in (0,1/9)$ such that $n^{\wt\tau}(\phi_n+n^{-1/3}) \to 0$ as $n \to \infty$. If $i \in \mathbb O_+$ satisfies
\begin{equation}\label{eq_propohold}
   \Delta(d_i)^2 \leq n^{\wt\tau}(\phi_n + n^{-1/3})\,,
\end{equation}
we have for $j = 1, \ldots, r$, 
\begin{equation}\label{eq_nonspikeeg2}
  |\langle \ub_j, \wt\bxi_{i}  \rangle|^2 \vee
    |\langle \vb_j, \wt\bzeta_{i}  \rangle|^2 \prec n^{3\wt\tau}\left( \frac{n^{-1}+\phi_n^3+\eta_i\sqrt{\varkappa_i}}{ \Delta(d_j)^4+\phi_n^2 + \varkappa_{i}}\right). 
\end{equation}
\end{theorem}
To appreciate this theorem, assume $\phi_n\le n^{-1/3}$ and $d_j$, $1\leq j \leq r^+$, satisfies $d_j - \alpha\gtrsim 1$. Then, for all $r^++1\leq i\leq n^{1/4}$, $\wt{\bxi}_{i}$ and $\wt{\bzeta}_i$ are delocalized with $|\langle \ub_j, \wt\bxi_{i}  \rangle|^2 \vee |\langle \vb_j, \wt\bzeta_{i}  \rangle|^2 \prec n^{-1}$.
On the other hand, when $d_j$ is close to the threshold $\alpha$, for a fixed finite $i$ so that $\wt\lambda_i$ is near the edge $\lambda_+$, \eqref{eq_evebulka} gives $|\langle \ub_j, \wt\bxi_{i}  \rangle|^2 \vee |\langle \vb_j, \wt\bzeta_{i}  \rangle|^2  \prec \frac{n^{-1}}{|d_j-\alpha|^2+n^{-2/3}}$; that is, the delocalization bound changes from the optimal order $n^{-1}$ to $n^{-1/3}$ as $d_j$ approaches $\alpha$.

Next, we state the behavior of the outlier singular vectors. For any $\textbf{A} \subset \mathbb O_+$, define
\begin{equation}
\nu_i(\textbf{A}):= \begin{cases}
   \min_{j \notin \textbf{A}} |d_j - d_i| & \textup{ if } i \in \textbf{A} \\
   \min_{j \in \textbf{A}} |d_j - d_i| & \textup{ if } i \notin \textbf{A}
\end{cases}
\end{equation}
and two projections,
\begin{equation} 
\mathcal{P}_{\textbf{A}}:=\sum_{k\in \textbf A} \wt{\bxi}_k\wt{\bxi}_k^\top \ \mbox{ and } \ 
\mathcal{P}'_{\textbf A}:=\sum_{k\in \textbf A} \wt{\bzeta}_k\wt{\bzeta}_k^\top.
\end{equation}

\begin{theorem}\label{thm_vector} Suppose Assumption \ref{assum_main} holds. Fix any $\textbf{A} \subset \mathbb O_+$, we have  
\begin{align}\label{eq_spikedvector}
\big| \langle \ub_i, \mathcal{P}_\textbf{A}\bu_j\rangle- \delta_{ij}\mathbbm 1(i\in \textbf{A}) a_1(d_i)\big|\, \vee \, &
\big| \langle \vb_i, \mathcal P'_\textbf{A}\vb_j\rangle- \delta_{ij}\mathbbm 1(i\in \textbf{A})a_2(d_i)\big| \\
& \prec  Q(i,j,\textbf{A},n) \nonumber
\end{align}
for all $i,j=1, \cdots, r$, where 
\begin{align*}
    Q(i,j,\textbf{A},n) =&\, \mathbbm 1(i\in \textbf{A},j\notin \textbf{A}) \Delta(d_i)\left[  \sqrt{\frac{\phi_n^2}{\nu_{j}(\textbf{A})}}+\frac{\psi_{1}(d_j)\Delta(d_j) }{\nu_{j}(\textbf{A})} \right]  \nonumber\\
& +  \mathbbm 1(i\notin \textbf{A},j\in \textbf{A}) \Delta(d_j)\left[\sqrt{\frac{\phi_n^2}{\nu_{i}(\textbf{A})}}+\frac{\psi_{1}(d_i)\Delta(d_i) }{\nu_{i}(\textbf{A})} \right]\\
&+\sqrt{R(i,\textbf{A}) R(j,\textbf{A})}  \,,
\end{align*}
$\psi_1(d_i) := \phi_n + \frac{n^{-1/2}}{\Delta(d_i)}$, and
\begin{align}
   R(i,\textbf{A})  := \mathbbm 1(i\in \textbf{A}) \psi_{1}(d_i) + \mathbbm 1(i \notin \textbf{A})   \frac{\phi_n^2}{\nu_{i}(\textbf{A})}  +  \frac{\psi_{1}^2(d_i)\Delta^2(d_i) }{\nu^2_{i}(\textbf{A})}\,.
\end{align}
\end{theorem}

{Theorem \ref{thm_vector} establishes the relationship between outlier singular vectors of $\wt S$ and $S$. To read this seemingly complicated bound, we take a look at all quantities one by one. Note that the quantity $\nu_i(\textbf{A})$ encodes the spectral gap in that $\nu_i(\{i\})$ is the spectral gap of $d_i$. Also, $\nu_i(\{j\})=\nu_j(\{i\})$ is the distance between $d_i$ and $d_j$. Next, the quantity $\psi_1(d_i)$ could be understood as the ``minimally required spectral gap'' that we need to recover a single singular vector even if it is strong. We have $\psi_1(d_i)\to 0$ as $n\to \infty$, which implies $R(i,\textbf{A})\to 0$ and hence $Q(i,j,\textbf{A},n)\to 0$ as $n\to \infty$ when $\nu_i(\textbf{A})$ and $\nu_j(\textbf{A})$ are both away from zero since we assume singular values are bounded in Assumption \ref{assum_main}.} 

Next, take the left singular vectors $\ub_i$ and $\wt \bxi_i$ as an example. Let $\textbf{A} = \{i\}$. From \eqref{eq_spikedvector}, {since $\langle \ub_i, \mathcal{P}_{\{i\}}\bu_i\rangle=|\langle \ub_i, \wt\bxi_i \rangle|^2$ and $Q(i,i,\{i\},n)=R(i,\{i\})=\psi_1(d_i)+\frac{\psi_{1}^2(d_i)\Delta^2(d_i)}{ \nu_{i}^2(\{i\})}$}, we have that 
\begin{equation}\label{eq_projectrueJJ1}
\big||\langle \ub_i, \wt\bxi_i \rangle|^2-a_1(d_i)\big|\vee \big||\langle \vb_i, \wt\bzeta_i \rangle|^2-a_2(d_i)\big| \prec \psi_1(d_i)+\frac{\psi_{1}^2(d_i)\Delta^2(d_i)}{ \nu_{i}^2(\{i\})}  \,.
\end{equation}
When $d_i$ is well-separated from other signals in that $\nu_i(\{i\}) \gg \psi_1(d_i)=\phi_n+\frac{n^{-1/2}}{\Delta(d_i)}$,
the error term converges to 0 {and hence the biases of $\wt\bxi_i$ recovering $\ub_i$ and $\wt\bzeta_i$ recovering $\vb_i$ are determined by $\sqrt{a_1(d_i)}$.
If further $d_i$ is sufficiently strong, a direct calculation using \eqref{eq_defnt}, \eqref{eq_defm1m2}, and \eqref{eq_evoutlier} shows that $a_1(d_i) = \frac{m_{1c}(\theta(d_i))}{d_i^2 \mathcal{T}'(\theta(d_i))} \approx 1$ and $a_2(d_i)  \approx 1$, which leads to $|\langle \ub_i,\wt\xi_i \rangle|^2 \approx 1$ and $|\langle \vb_i,\wt\zeta_i \rangle|^2 \approx 1$.
Similarly, we can set $\textbf{A} = \{k\}$, where $k\neq i$. In this situation, $Q(i,i,\{k\},n)=R(i,\{k\})=\frac{\phi_2^2}{\nu_i(\{k\})}+\frac{\psi_{1}^2(d_i)\Delta^2(d_i)}{ \nu_{i}^2(\{k\})}$, and hence we have
\begin{equation}\label{eq_projectrueJJ2}
|\langle \ub_i, \wt\bxi_k \rangle|^2\vee |\langle \vb_i, \wt\bzeta_k \rangle|^2 \prec \frac{\phi_2^2}{\nu_i(\{k\})}+\frac{\psi_{1}^2(d_i)\Delta^2(d_i)}{ \nu_{i}^2(\{k\})} \,,
\end{equation}
which is sufficiently small when $\nu_i(\{k\})=|d_i-d_k| \gg \psi_1(d_i)=\phi_n+\frac{n^{-1/2}}{\Delta(d_i)}$; that is, $\wt\bxi_k$ and $\ub_i$, as well as $\wt\bzeta_k$ and $\vb_i$, are almost perpendicular. In the last example, we assume there are two consecutive singular values, $d_i\geq d_{i+1}$, that are close, or even the same, but away from other singular values; that is, $\nu_i(\{i\})\ll\psi_1(d_i)$ and $\nu_i(\{i,i+1\})\gg\psi_1(d_i)$. In this case, we expect difficulties to recover $\ub_i$ and $\ub_{i+1}$. Consider $\textbf{A}:=\{i,i+1\}$. In this case, we have $Q(i,i,\textbf{A},n)=R(i,\textbf{A})=\psi_1(d_i)+\frac{\psi_{1}^2(d_i)\Delta^2(d_i)}{\nu_{i}^2(\textbf{A})}$ and hence the desired result that
\begin{equation}\label{eq_projectrueJJ3}
\big| \langle \ub_i, \mathcal{P}_\textbf{A}\bu_i\rangle-a_1(d_i)\big|\vee 
\big| \langle \vb_i, \mathcal P'_\textbf{A}\vb_i\rangle- a_2(d_i)\big|\prec \psi_1(d_i)+\frac{\psi_{1}^2(d_i)\Delta^2(d_i)}{ \nu_{i}^2(\textbf{A})}\,,
\end{equation}
which the error term converges to 0. This means that $\ub_i$ can be well recovered by a vector in the vector space spanned by $\wt\bxi_i$ and $\wt\bxi_{i+1}$.

Finally, we shall compare \eqref{eq_projectrueJJ1} in Theorem \ref{thm_noneve} and \eqref{eq_nonspikeeg2} when $i \in \mathbb O_+$ and $\Delta(d_i)^2 = d_i-\alpha \leq n^{\wt\tau}(\phi_n + n^{-1/3})$ holds. 
To simplify the discussion, we assume again $d_i$ is well-separated from other signals in that $\nu_i(\{i\}) \gg \psi_1(d_i)$, and let $\phi_n \lesssim n^{-1/3}$. With the lower bound of $d_i-\alpha$ given in \eqref{eq_assm_41} when $i \in \mathbb{O}_+$, now we have $n^{-1/3} < \Delta(d_i)^2 \leq  n^{-1/3+\wt\tau}$. Without loss of generallity, we let $\Delta(d_i)^2 = n^{-1/3+\wt\tau}$ and compare the error bounds in \eqref{eq_projectrueJJ1} and \eqref{eq_nonspikeeg2} for $\wt\tau \in (0,1/9)$. The above conditions give us $\psi_1(d_i) \asymp \phi_n + n^{-1/3-\wt\tau/2}$. Also, by the definition of $a_1(x)$ and $a_2(x)$ in \eqref{eq_functions} and the approximations of integral transforms in \eqref{eq_estimm}, \eqref{eq_s36_1}, and \eqref{eq_s36_2}, we have $|a_1(d_i)| \asymp |a_2(d_i)| \asymp n^{-1/3+\wt\tau}$. Above approximations of $\Delta(d_i)$, $\psi(d_i)$, $a_1(d_i)$, and $a_2(d_i)$ and \eqref{eq_projectrueJJ1} give us the following bound $
|\langle \ub_i, \wt\bxi_i \rangle|^2\vee |\langle \vb_i, \wt\bzeta_i \rangle|^2 \prec n^{-1/3+\tau} + \phi_n$.
Moreover, together with $i = j \in \mathbb{O}_+$, \eqref{eq_nonspikeeg2} gives us
$|\langle \ub_i, \wt\bxi_{i}  \rangle|^2  
\vee |\langle \vb_i, \wt\bzeta_{i}  \rangle|^2  \prec n^{-1/3+\wt\tau}.$ Therefore, as mentioned above in Theorem \ref{thm_noneve}, for $i\in \mathbb{O}_+$, if the signal is strong but not strong enough, then  \eqref{eq_nonspikeeg2} gives us a smaller absolute error bound compared to \eqref{eq_projectrueJJ1}.
}

\section{Proposed eOptShrink algorithm}\label{section proposed algo}
We are ready to introduce a data-driven algorithm to estimate the optimal shrinker $\varphi$. We call this algorithm {\em extended OptShrink} (eOptShrink). There are three main steps in eOptShrink.
Based on the delocalization and bias estimate of singular vectors in Theorems \ref{thm_noneve} and \ref{thm_vector} and the sticking result in Theorem \ref{thm_eigenvaluesticking}, we show that under a mild condition, we could accurately estimate $\lambda_{k}$ for the top finite $k$ eigenvalues of $ZZ^\top$ when $n$ is sufficiently large and when $r^+$ is known, and hence a more precise estimate of $\pi_{ZZ^\top}$. 
However, in practice $r$ and $r^+$ are unknown, and based on the established theory, estimating $r$ might be challenging. We show that we could accurately estimate $r^+$ via estimating $\lambda_+$. 
The pseudocode of eOptShrink is summarized in Algorithm \ref{alg}. The Matlab implementation can be found in \url{https://github.com/PeiChunSu/eOptShrink}. Below we detail the algorithm and its associated theoretical support with a asymptotic convergence rate.

\begin{algorithm}[hbt!]
\begin{algorithmic}
\caption{eOptShrink}\label{alg}
\STATE \textbf{Input}: $\wt S = \sum_{i=1}^{ p\wedge n}  \sqrt{\wt\lambda_i} \wt\bxi_i \wt\bzeta_i^\top$, {a constant $c = \min(\frac{1}{2.01}, \frac{1}{\log(\log n)})$}, and the desired loss function.
\STATE \textbf{Compute}: \\
(i) $\widehat \lambda_+$ in \eqref{eq_hatlanbda_+} and $\widehat r^+$ in \eqref{eq_adrk0}.\\
(ii) $\widehat F_{\texttt{e}}(x)$ in \eqref{eq_imp} and $\widehat d_{\texttt{e},j}$,  $\widehat a_{\texttt{e},1,j}$ and $\widehat a_{\texttt{e},2,j}$ for $1\leq j\leq \widehat r^+$ in \eqref{eq_a1a2}. \\ 
(iii) $\widehat\varphi_{\texttt{e},i}$ for $1\leq i \leq \widehat r^+$ in \eqref{eq_adsh} for the associated norm. \\
\STATE \textbf{Output}: The estimator of the clean data matrix $\widehat S_{\widehat\varphi} = \sum_{i=1}^{\widehat r^+} \widehat \varphi_{\texttt{e},i} \wt\bxi_i \wt\bzeta_i^\top$.
\end{algorithmic}
\end{algorithm}

\subsection{Existing imputation approach}
We review an imputation scheme proposed in \cite{donoho2020screenot}. 
With the square root behavior that $\rho_{1c}(x)\asymp (\lambda_+-x)^{1/2}$ as $x\to \lambda_+$, where $x\in \mathbb{R}$, when $Z=A^{1/2}X$ \cite{silverstein1995analysis}, for a fixed large integer $\varpi$, when $p$ and $n$ are sufficiently large, we have
\begin{equation}\label{expansion of 1/2 edge}
    \frac{\ell-1}{p} =\int_{\gamma_{\ell}}^{\lambda_+} \rho_{1c}(z)dz =C'\int_{\gamma_{\ell}}^{\lambda_+} (\lambda_+-z)^{1/2}dz = \frac{2C'}{3}(\lambda_+-\gamma_{\ell})^{3/2}
\end{equation}
for $1\leq \ell \leq \varpi$ and $C'>0$, where $\gamma_\ell$ is the classical location defined in \eqref{eq_classical}. Since $\gamma_\ell$ can be approximated by $\lambda_\ell$ (see Lemma \ref{lem_rigidty} for details), 
this leads to an estimate of the distance between the $j$-th and $\ell$-th  eigenvalues, where $1\leq j,\ell \leq \varpi$,
\begin{equation}\label{expansion of 1/2 edge2}
    \lambda_{\ell}-\lambda_{j} \approx C'' \left[ \left( \frac{j-1}{p} \right)^{2/3} - \left( \frac{\ell-1}{p} \right)^{2/3} \right] \,,
\end{equation}
for some unknown $C''>0$.
In \cite{donoho2020screenot}, by fixing an integer $k$ so that $r\leq k$ and $2k+1<\varpi$ for a large constant $\varpi$, $C''$ is estimated by 
\[
\check C'':=\frac{\wt\lambda_{k+1}-\wt\lambda_{2k+1}}{(2k/p)^{2/3}-(k/p)^{2/3}}
\]
and the $j$-th eigenvalue, where $j=1,\ldots,k$, as a missing value is imputed by 
\[
\check\lambda_{j}:=\wt\lambda_{k+1} + \frac{1-(\frac{j-1}{k})^{2/3}}{2^{2/3}-1}( \wt\lambda_{k+1}-\wt\lambda_{2k+1} )\,, 
\]
where {$k=4r$ is suggested in \cite[p33]{donoho2020screenot}}. The cumulative distribution function (CDF) of $\pi_{ZZ^\top}$ is estimated by 
\begin{equation}\label{eq_imp0}
    \widehat F_{\texttt{imp}}(x) := \frac1p \sum_{j=k+1}^p \mathbbm 1(\wt\lambda_{j}\leq x) + \frac1p \sum_{j=1}^{k} \mathbbm 1(\check\lambda_j \le x) \,.
\end{equation}
{With $\widehat F_{\texttt{imp}}$, the matrix denoising algorithm, ScreeNOT, is given in \eqref{eq_SN}} 

\subsection{Proposed data-driven optimal shrinker, eOptShrink}\label{sec_imputation}

\subsubsection{Step 1: estimate $r^+$}\label{sec_rank}
We estimate $\lambda_+$ first and use it to estimate $r^+$. With the discussion for \eqref{expansion of 1/2 edge2} and the sticking behavior in \eqref{eq_stickingeq}, we modify the estimator $\check C''$ for the constant $C''$ by constructing 
\[
\widehat C: =\frac{\wt\lambda_{k+r^++1}-\wt\lambda_{2k+r^++1}}{(2k/p)^{2/3}-(k/p)^{2/3}}\,.
\] 
Since $r^+$ is unknown but fixed, we set $k = \lfloor n^c\rfloor \gg r^+$, where $\lfloor \cdot\rfloor $ is the floor function, for a small fixed constant $c>0$ 
and construct an estimator of $\lambda_+$ as 
\begin{equation}\label{eq_hatlanbda_+}
    \widehat{\lambda}_{+}  := \wt\lambda_{\lfloor n^c\rfloor +1} + \frac{1}{2^{2/3}-1}\left( \wt\lambda_{\lfloor n^c\rfloor +1}-\wt\lambda_{2\lfloor n^c\rfloor +1} \right) \,.
\end{equation}
Then we estimate $r^+$ based on Theorem \ref{thm_value} and set
\begin{equation}\label{eq_adrk0}
\widehat r^+ = \left|\{\wt\lambda_i|\wt \lambda_i>\widehat \lambda_{+}+n^{-1/3}\}\right|\,.
\end{equation}
The following theorems guarantee the performance of $\widehat \lambda_+$ and $\widehat r^+$. 
\begin{theorem}\label{thm_numbershrink}
Suppose Assumption \ref{assum_main} and \eqref{alpha+} hold and $c\in (0,1)$.
We have
\begin{equation}\label{lambda_+_comv}
   |\widehat\lambda_+ - \lambda_+| \prec \frac{1}{n\alpha^+}+n^{-\min\{\frac{2}{3}+\frac{c}{3},\frac{3}{4},\frac{5}{6}-\frac{c}{3},\frac{4}{3}-\frac{4c}{3}\}}+ n^{-1/2}\phi_n+n^{-1/3-2c/3}\phi_n^2.
\end{equation}
\end{theorem}

\begin{theorem}\label{thm_numbershrink2}
Suppose Assumption \ref{assum_main} and \eqref{alpha+} hold true for some $\varepsilon$ such that $n^{\varepsilon}(\phi_n+n^{-1/3})>n^{-1/6}$. Denote the event $\Xi(r^+):=\{\widehat{r}^+=r^+\}$ and let $0<c<1/2$. Then $\Xi(r^+)$ is an event with high probability. 
\end{theorem}

\subsubsection{Step 2: estimate the CDF of the eigenvalue distribution of $ZZ^\top$}\label{sec_rank}

We achieve this goal by modifying the eigenvalue distribution of $\wt S \wt S^\top$. 
Similar to the idea of the truncated spectrum in \cite{nadakuditi2014optshrink}, we omit the first $\widehat{r}^+$ eigenvalues of $\wt S\wt S^T$.
Moreover, as the discussion in the last section, 
we estimate $C''$ via $\widehat C'' := \frac{\wt\lambda_{\lfloor n^c\rfloor+\widehat r^++1}-\wt\lambda_{2\lfloor n^c\rfloor +\widehat r^++1}}{(2\lfloor n^c\rfloor /p)^{2/3}-(\lfloor n^c\rfloor /p)^{2/3}}$,
and thus we reconstruct the $(\widehat r^++1)$th to the $(\widehat r^++\lfloor n^c\rfloor)$th eigenvalues  by
\begin{equation}\label{eq_hatlanbda2}
\widehat{\lambda}_{j}  := \wt\lambda_{\lfloor n^c\rfloor+ \widehat r^++1} + \frac{1-\big(\frac{j-\widehat r^+-1}{\lfloor n^c\rfloor}\big)^{2/3}}{2^{2/3}-1}\left( \wt\lambda_{\lfloor n^c\rfloor+ \widehat r^++1}-\wt\lambda_{2\lfloor n^c\rfloor+\widehat r^++1} \right), 
\end{equation}
where $1\leq j \leq \lfloor n^c\rfloor$ and $c$ is a small fixed positive constant.
By \eqref{eq_stickingeq} and the fact that $|\gamma_{\lfloor n^c\rfloor}-\gamma_{\lfloor n^c\rfloor+r^++1}| \asymp n^{-2/3}$, 
we immediately have the following comparison for $\check \lambda_j$ and $\widehat \lambda_{j}$.
\begin{theorem}\label{thm_compareimpe}
Suppose Assumption \ref{assum_main} and \eqref{alpha+} hold true and $c\in(0,1/2)$. 
With high probability, for $i =  1,\ldots, \lfloor n^c \rfloor$, we have
\begin{align}
   & |\check{\lambda}_{i} - \lambda_{i}|\vee |\check{\lambda}_{i+\widehat r^+} - \lambda_{i}| \prec  \frac{1}{n \alpha_+} + n^{-2/3} + n^{-1/2}\phi_n + n^{-1/3-2c/3}\phi_n^2\,, \label{eq_imp_bd1}\\
   & |\widehat{\lambda}_{i+\widehat r^+} - \lambda_{i}|\prec \frac{1}{n\alpha^+}+n^{-\min\{\frac{2}{3}+\frac{c}{3},\frac{3}{4},\frac{5}{6}-\frac{c}{3},\frac{4}{3}-\frac{4c}{3}\}}+ n^{-1/2}\phi_n+n^{-1/3-2c/3}\phi_n^2 \label{eq_imp_bd2}.
\end{align}
\end{theorem}

{Recall that $\phi_n = n^{2/a-1/2}$ with $a>4$, which corresponds to the entries of random matrix $X$ before truncation having finite $a$-moment.} When $4<a\leq 6$, the right-hand side of both \eqref{eq_imp_bd1} and \eqref{eq_imp_bd2} are both dominated by $n^{-1/2}\phi_n+n^{-1/3}\phi_n^2$. {When $a>6$, our estimator $\widehat{\lambda}_{j} $ exhibits a lower upper bound on absolute error compared with $\check{\lambda}_j$.}
Finally, we propose to estimate the CDF of $\pi_{ZZ^\top}$ by 
\begin{equation}\label{eq_imp}
    \widehat F_{\texttt{e}}(x) := \frac{1}{p-\widehat r^+} \left( \sum_{j=\widehat r^++1}^{\lfloor n^c\rfloor+\widehat r^+} \mathbbm 1(\widehat\lambda_{j} \le x)+\sum_{j=\lfloor n^c\rfloor+\widehat r^++1}^p \mathbbm 1(\wt\lambda_{j}\leq x)\right) \,.
\end{equation}

\subsubsection{Step 3: estimate the optimal shrinker}

With $\widehat F_{\texttt{e}}$, we now state our optimal shrinker with the proposed rank estimator. 
The CDF of ESD of $ZZ^T$ is estimated by \eqref{eq_imp} with $r^+$ estimated by $\widehat r^+$ in \eqref{eq_adrk0}.
Consider the following ``discretization'' of the associated quantities. 
For $1\leq i \leq \widehat r^+$, 
denote the estimators of $m_{1c}(\wt\lambda_i)$ and $m_{2c}(\wt\lambda_i)$ as
\begin{align}\label{eq_mhat1mhat2}
& \widehat{m}_{e,1,i}:= \int \frac{d\widehat F_e(x)}{x-\wt\lambda_i} =  \frac{1}{p-\widehat r^+}\left( \sum_{j=\widehat r^++1}^{\lfloor n^c \rfloor+\widehat r^+} \frac{1}{\widehat\lambda_j-\wt\lambda_i} + \sum_{j=\lfloor n^c\rfloor+\widehat r^++1}^p \frac{1}{\wt\lambda_j-\wt\lambda_i}\right), \ 
\nonumber\\
& \widehat{m}_{e,2,i} := \frac{1-\beta_n}{\wt\lambda_i}+\beta_n \widehat{m}_{e,1,i}. 
\end{align}
Similarly, denote the discretization of $m'_{1c}(\wt\lambda_i)$ and $m'_{2c}(\wt\lambda_i)$ as
\begin{equation*}
\widehat{m}'_{e,1,i}=\frac{1}{p-\widehat r^+}\left( \sum_{j=\widehat r^++1}^{\lfloor n^c \rfloor+\widehat r^+} \frac{1}{(\widehat\lambda_j-\wt\lambda_i)^2} + \sum_{j=\lfloor n^c \rfloor+\widehat r^++1}^p \frac{1}{(\wt\lambda_j-\wt\lambda_i)^2}\right), \ \widehat{m}'_{e,2,i}=\frac{1-\beta_n}{\wt\lambda_i^2}+\beta_n \widehat{m}'_{e,1,i}. 
\end{equation*}
For $1\leq i \leq \widehat r^+$, the estimators of the $D$-transform $\mathcal{T}(\wt\lambda_i)$ and its derivative $\mathcal{T}'(\wt\lambda_i)$ are
\begin{equation}
\widehat{\mathcal{T}}_{e,i}= \wt\lambda_i \widehat{m}_{e,1,i} \widehat{m}_{e,2,i}, \quad \widehat{\mathcal{T}}'_{e,i}=\widehat{m}_{e,1,i} \widehat{m}_{e,2,i}+\wt\lambda_i \widehat{m}'_{e,1,i} \widehat{m}_{e,2,i}+\wt\lambda_i \widehat{m}'_{e,2,i} \widehat{m}_{e,1,i},
\end{equation}
and the estimators of $d_i$, $a_{1,i} =   \langle \ub_i, \wt\bxi_{i}  \rangle^2$, and $a_{2,i} =   \langle \vb_i, \wt\bzeta_{i}  \rangle^2$ are
 \begin{equation}\label{eq_a1a2}
  \widehat{d}_{e,i}=\sqrt{\frac{1}{\widehat{\mathcal{T}}_{e,i}}}, \quad \widehat{a}_{e,1,i}=\frac{\widehat{m}_{e,1,i}}{\widehat{d}_{e,i}^2 \widehat{\mathcal{T}}'_{e,i}}, \ \mbox{ and } \  \widehat{a}_{e,2,i}=\frac{\widehat{m}_{e,2,i}}{ \widehat{d}_{e,i}^2\widehat{\mathcal{T}}'_{e,i}}.
 \end{equation}
As a result, we estimate the optimal shrinker, $\varphi^*_i$ in Proposition \ref{prop_optimal_shrinker}, by
\begin{equation}\label{eq_adsh}
\begin{array}{lr}
\displaystyle\widehat{\varphi}_{\texttt{e},i}=\widehat{d}_{\texttt{e},i} \sqrt{\widehat{a}_{\texttt{e},1,i}\widehat{a}_{\texttt{e},2,i}}, &\mbox{(Frobenius norm)}  \\
\displaystyle{\widehat{\varphi}_{\texttt{e},i}=\widehat d_{e,i}\sqrt{\frac{\widehat a_{e,1,i}\wedge \widehat a_{e,2,i}}{\widehat a_{e,1,i}\vee \widehat a_{e,2,i}}}}, &\mbox{(Operator norm)}\\
\displaystyle\widehat{\varphi}_{\texttt{e},i}=\widehat{d}_{\texttt{e},i}\Big(\sqrt{\widehat{a}_{\texttt{e},1,i} \widehat{a}_{\texttt{e},2,i}}-\sqrt{(1-\widehat{a}_{\texttt{e},1,i})(1-\widehat{a}_{\texttt{e},2,i})}\Big) \,, &\mbox{(Nuclear norm)}
\end{array}
\end{equation} 
for $1\leq i\leq \widehat r^+$, and $\widehat{\varphi}_{\texttt{e},i}=0$ otherwise. 
The following theorem gives the convergence guarantee of the proposed estimator.

\begin{theorem}\label{thm_shrink} 
Suppose Assumption \ref{assum_main} and \eqref{alpha+} {hold true for some $\varepsilon$ such that $n^{\varepsilon}(\phi_n+n^{-1/3})>n^{-1/6}$, and $c\in (0,1/2)$}. 
For all three types of loss functions mentioned in Proposition \ref{prop_optimal_shrinker}, for $1\leq i \leq \widehat r^+$, conditional on $\Xi(r^+)$, we have
$|\varphi^*_i-\widehat{\varphi}_{\texttt{e},i}| \prec \phi_n + n^{-1/2}/\Delta(d_i)$.
\end{theorem}

If we replace $\widehat F_e(x)$ by
\begin{equation}\label{eq_trun}
\widehat F_{\texttt{T}}(x):=
\frac{1}{p- \widehat r^+}\sum_{i=\widehat r^+ +1}^{p}\mathbbm 1(\wt\lambda_{i}\leq x)\,,
\end{equation}
then by Theorem \ref{thm_numbershrink2}, eOptShrink with the Frobenius norm loss is reduced to OptShrink proposed in \cite{nadakuditi2014optshrink}.
It is possible to estimate $\varphi^*_i$ by $\widehat F_{\texttt{T}}(x)$ or $\widehat F_{\texttt{imp}}(x)$. We denote that resulting estimates of $\varphi^*_i$ as $\widehat{\varphi}_{\texttt{T},i}$ and $\widehat{\varphi}_{\texttt{imp},i}$ respectively.
In the next section, we numerically show that using $\widehat F_{\texttt{e}}(x)$ results in a lower estimation error compared to either using $\widehat F_{\texttt{T}}(x)$ or $\widehat F_{\texttt{imp}}(x)$.

{
\subsubsection{Selection of $c$ in practice}\label{sec_selectionC}

Recall that in \textit{Step 1}, we necessitate a constant $c$ such that $\lfloor n^c \rfloor \gg r^+$ to compute the estimator of bulk edge $\widehat \lambda_+$ and the estimated effective rank $\widehat r^+$. Conclusions from Theorems \ref{thm_numbershrink} and \ref{thm_numbershrink2} necessitate $c \in (0,1/2)$ for our estimators to converge to the ground truth when $n$ is sufficiently large. Similarly, in \textit{Step 2}, based on Theorems \ref{thm_numbershrink}, \ref{thm_numbershrink2}, and \ref{thm_compareimpe}, we require $c\in (0,1/2)$ for $\widehat \lambda_{i+\widehat r^+}$ to accurately estimate $\lambda_i$ for $i = 1,\ldots,\lfloor n^c \rfloor$ when $n$ is sufficiently large. In both cases, smaller error rates occur with smaller $c>0$. However, in practical scenarios where $n$ is not significantly large, e.g., when $n=300$, a small $c$ may fail to guarantee $\lfloor n^c \rfloor \gg r^+$. To be safe, it is natural to consider a $c$ close to $1/2$. With this approach, computing $\widehat \lambda_{i+\widehat r^+}$ for $\lfloor n^c \rfloor$ times becomes time-consuming for large $n$. To balance various practical situations, it is logical to opt for a small $c>0$ when $n$ is large and a large $c<1/2$ when $n$ is small. Hence, we propose considering $c = \frac{1}{2.01}\wedge \frac{1}{\log(\log(n))}$, ensuring the convergence conditions mentioned above while computing $\widehat \lambda_{i+\widehat r^+}$ for $\lfloor n^{1/2.01\wedge 1/\log(\log(n))} \rfloor$ iterations. It is important to emphasize that this choice of $c$ is not necessarily optimal, but in practice, we found it to work effectively.

}

\section{Numerical evaluation}\label{section numreical eva}
We assess the performance of eOptShrink through numerical simulations involving various types of noise and the single-channel fECG extraction problem. In all our findings, we use interquartile range error bars or present means with standard deviations to establish performance superiority between methods. Paired t-tests are conducted, and when multiple testing is involved, we implement the Bonferroni correction. We consider $p<0.005$ as statistically significant.

{Regarding $c$, it is worth noting that $\frac{1}{\log(\log(n))} > \frac{1}{2.01}$ for $n < 1743$, while $\frac{1}{\log(\log(n))} < \frac{1}{2.01}$ when $n \geq 1743$. Hence, in the subsequent numerical simulations, we vary $n$ within the range of $300$ to $2100$ to illustrate the performance of this choice of constant $c$, showcasing scenarios where either $1/2.01$ or $1/\log(\log (n))$ dominates.}

\subsection{Simulated signals} \label{section simulation signal noise}
We consider different types of noises.
Suppose $X\in \mathbb{R}^{p\times n}$ has i.i.d. entries with Student's t-distribution with $10$ degrees of freedom followed by a proper normalization that $\mathbb{E}X^2_{ij}=1/n$.
Set $A = \frac{1}{L_A} 
 Q_AD_AQ_A^{T}\in \mathbb{R}^{p\times p}$, where $D_A = \textup{diag} \{\ell_1,\ell_2,\ldots, \ell_p\}$, $Q_A\in O(p)$ is generated by the QR decomposition of a $p \times p$ random matrix independent of $X$, and $L_A = \sum_{i = 1}^p \ell_i$ is a normalizing factor. 
The same method is applied to generate $B=\frac{1}{L_B} Q_BD_BQ_B^{T}\in \mathbb{R}^{n\times n}$, which is assumed to be independent of $A$ and $X$. Here we consider {three types of noise}. The first one is the white noise (called TYPE1 below); that is, $D_A = I_p$ and $D_B = I_n$. The second one has a separable covariance structure (called TYPE2 below) with a gap in the limiting distribution; that is, $D_A=\textup{diag}\Big\{\sqrt{1+9\times \frac{1}{p}}, \sqrt{1+9\times \frac{2}{p}}, \cdots, \sqrt{1+9\times \frac{p-1}{p}},\sqrt{10}\Big\}$ and \\
$D_B=\textup{diag}\Big\{ \sqrt{10+\frac{1}{n}},\sqrt{10+\frac{2}{n}},\cdots, \sqrt{10+\frac{\lfloor n/4\rfloor}{n}},$
$ \sqrt{0.3}, \cdots, \sqrt{0.3}, \sqrt{0.3} \Big\}$.
{
The third one (called TYPE3 below) has a more complicated separable covariance structure with $D_A=\textup{diag}\Big\{\exp(\frac{1}{p}), \exp(\frac{2}{p}) \cdots, \exp(\frac{p-1}{p}),$ 
$ \exp{(1)}\Big\}$ and 
$D_B=\textup{diag}\Big\{ 1.1+\sin(4\pi(\frac{1}{n})),1.1+\sin(4\pi(\frac{2}{n})),\cdots, 1.1+\sin(4\pi(\frac{n-1}{n})),$
$ 1.1+\sin(4\pi) \Big\}$.}
The signal matrix is designed to be $S= \sum_{i=1}^r d_i \ub_i\vb_i^\top$, 
where $r=15$, $d_i$ are i.i.d. sampled uniformly from $[0,4]$ and ordered so that $d_1\geq d_2\ldots\geq d_{15}$, 
and the left and right singular vectors are generated by the QR decomposition of two independent random matrices.  Below, we independently realize $\tilde S=S+A^{1/2}XB^{1/2}$ for $100$ times for different $n$, different noise types and $p/n = 0.5$ or $1$, and report the comparison of different algorithms from different angles.  More simulations are provided in Section \ref{section supp more numerical simulation} in the supplementary material. 

To compare the performance of different algorithms, we need to calculate $\alpha$ \eqref{eq_defnalpha}. For Type I noise,  $\alpha$ is determined by $\widehat \alpha:= (p/n)^{1/4}$ as described in the paragraph after Proposition \ref{prop_optimal_shrinker}. For {TYPE2 and TYPE3 noises}, it is challenging to directly calculate $\alpha$ from its definition, so we apply the following numerical calculation to determine $\alpha$. For the chosen $p/n$, construct $Z=A^{1/2}XB^{1/2}\in \mathbb{R}^{p'\times n'}$ with $n'$ large and $p'/n'=p/n$, and denote the eigenvalues of $Z^\top Z$ as $\{\lambda_i\}_{i=1}^{n'}$. Denote 
   $\widehat \alpha := 1/\sqrt{\lambda_1 \check m_1(\lambda_1)\check m_2(\lambda_1)}$, 
where $\check m_1(x) =  \frac{1}{p'-1}\left( \sum_{j=2}^{p'} \frac{1}{\lambda_j-x}\right)$ and $\check{m}_{2}(x) =  \frac{1}{n'-1}\left( \sum_{j=2}^{n'} \frac{1}{\lambda_j-x}\right)$. By Lemma \ref{lem_rigidty} and Theorem \ref{LEM_SMALL}, $|\widehat \alpha - \alpha| \lesssim (n')^{-1/3}$, which is sufficiently small when $n'$ is large. We set $n' = 10000$ and independently construct $\widehat \alpha$ for $100$ times.
{For TYPE2 noise, we have $\widehat\alpha = 1.3495\pm 0.0290$ when $p/n = 0.5$, and $\widehat\alpha = 1.6515 \pm 0.0180$ when $p/n= 1$, and 
for TYPE3 noise, we have $\widehat\alpha = 1.5242\pm 0.0320$ when $p/n = 0.5$, and $\widehat\alpha = 1.8115 \pm 0.0348$ when $p/n= 1$, where we show the mean $\pm$ standard deviation.}
We take the mean of the 100 constructions to determine $\alpha$, and for simplicity we still denote the mean as $\widehat \alpha$ afterwards.

\subsubsection{Visualization of ESDs and thresholds}\label{sec_compareTSE}
For TRAD, the noise level is estimated by $\check\sigma(\wt S)$ for a fair comparison. 
For ScreeNOT, we use the ground truth rank and set {$k = 4r$} and call the estimated rank by the hard threshold {$\vartheta_{\texttt{SN}}$} in \eqref{eq_SN} the {\em ScreeNOT rank}.
Figure \ref{fig_GAP_noise} illustrates ESDs of $\wt S = S+ A^{1/2}XB^{1/2}$ for {all three types of noises} when $n = 2000$, where there is an obvious gap in the bulks associated with the TYPE2 noise when $p/n = 1$. 
The black line indicates the the estimated bulk edge from \eqref{eq_hatlanbda_+}, which separates the noise and signal for {all three types} of noises. The yellow line is the estimated bulk edge by TRAD using $\check{\sigma}(\wt S)(1+\sqrt{p/n})$, which separate the noise and the signal well for TYPE1 noise but not for {TYPE2 and TYPE3 noises}. The red line is the ScreeNOT rank, which {discards weak components of signals that are very close to the noise bulkedge.} Note that with TYPE1 noise, the black and yellow lines are exactly around $1+\sqrt{p/n}$. 

\begin{figure}[hbt!]
\begin{minipage}{1\textwidth}
\includegraphics[width=0.32\linewidth]{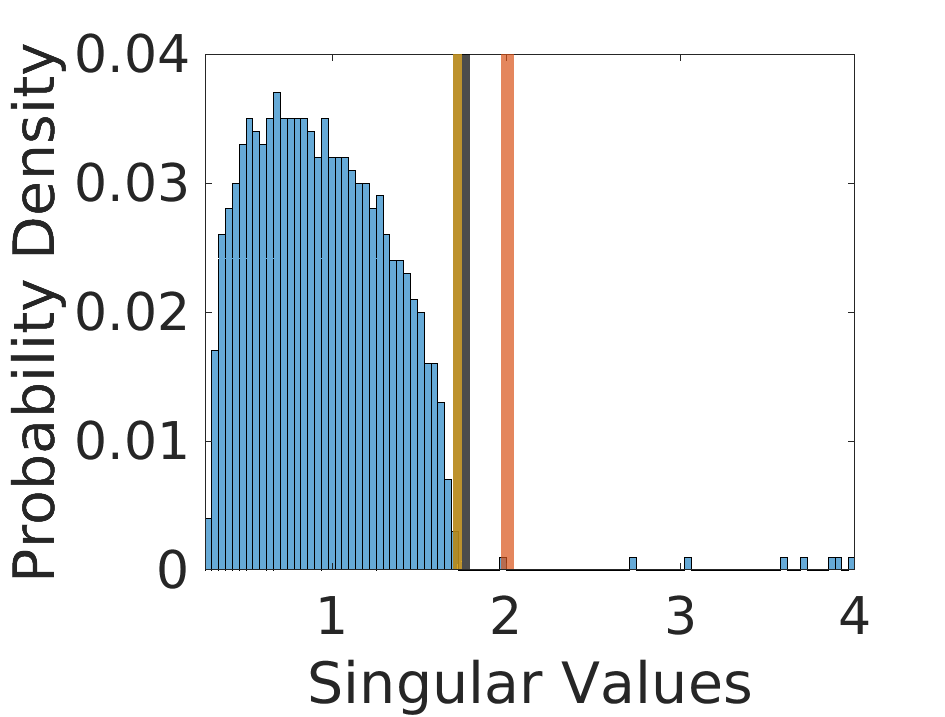}
\includegraphics[width=0.32\linewidth]{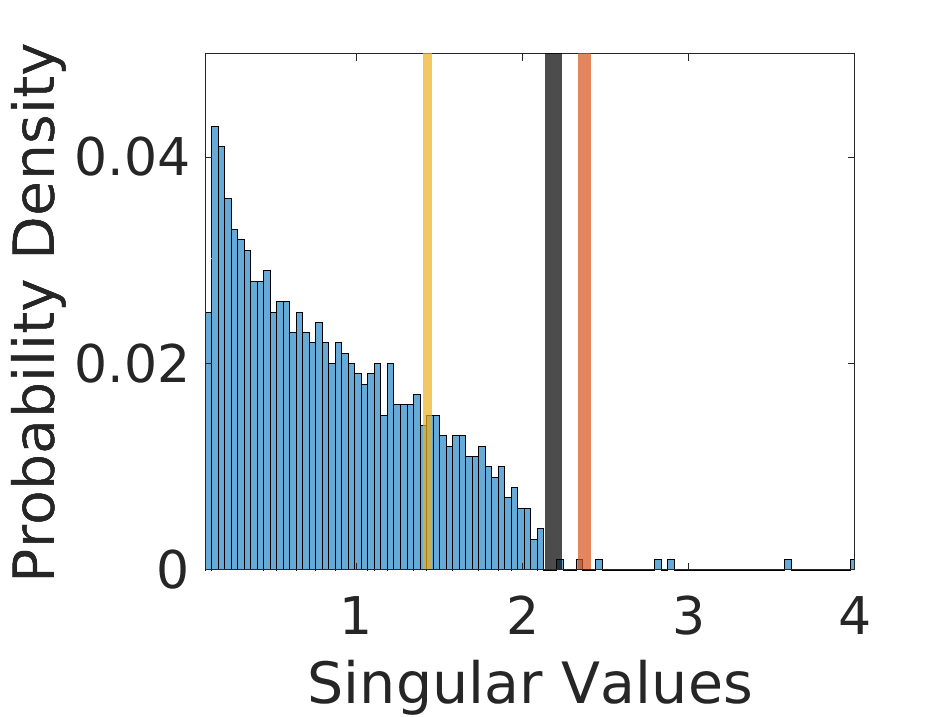}
\includegraphics[width=0.32\linewidth]{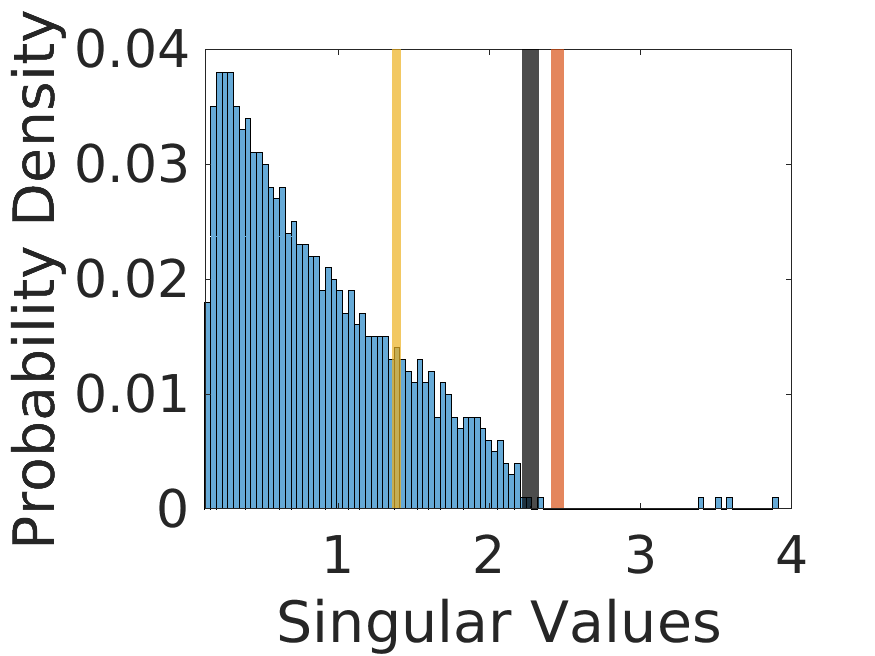}
\end{minipage}

\begin{minipage}{1\textwidth}
\includegraphics[ width=0.32\linewidth]{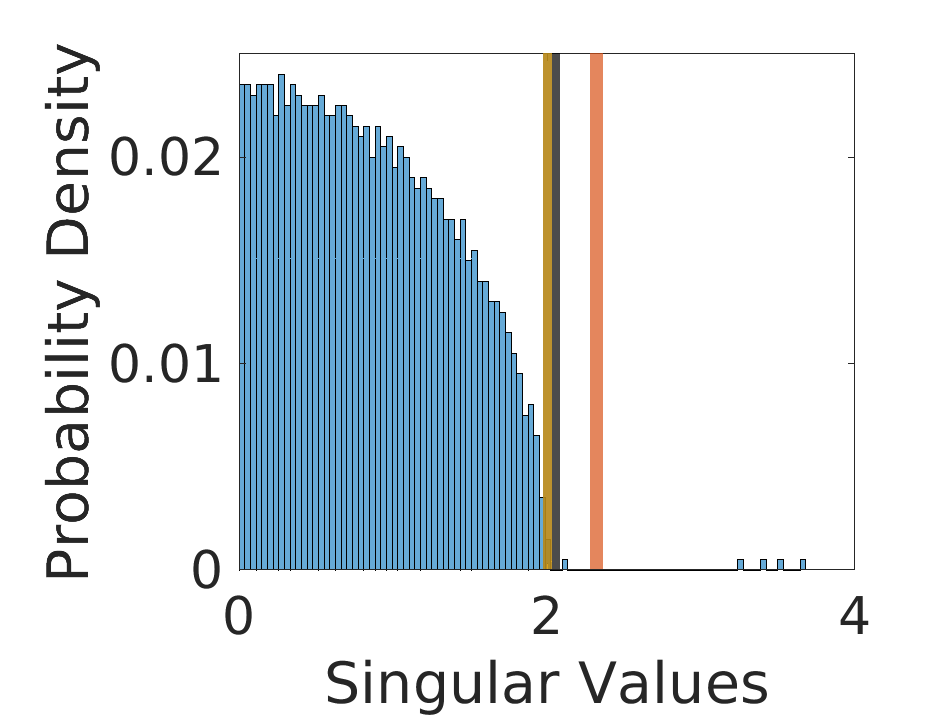}
\includegraphics[width=0.32\linewidth]{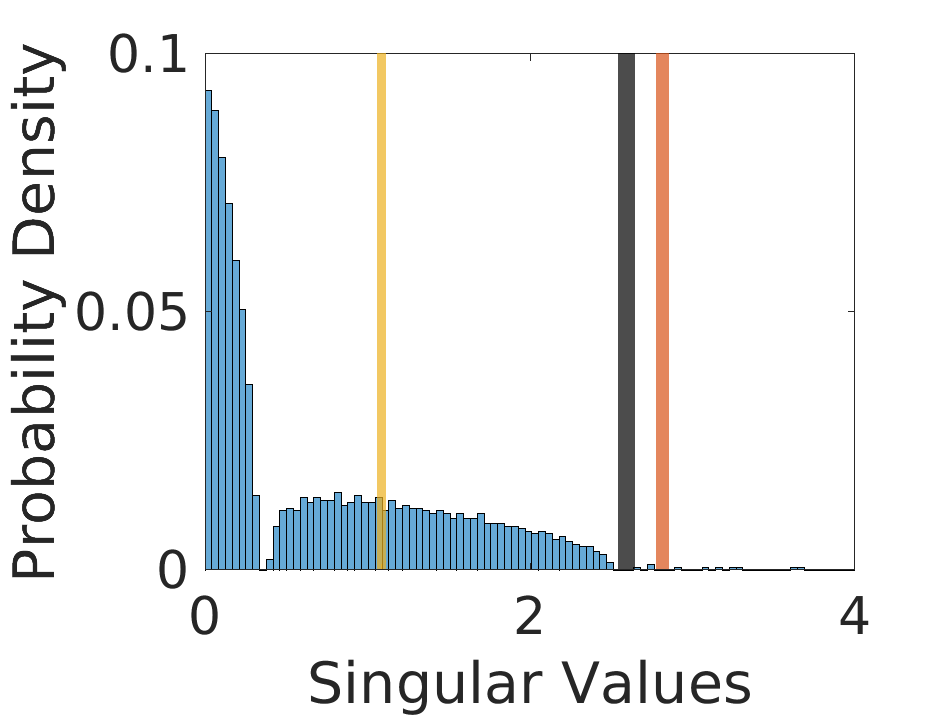}
\includegraphics[width=0.32\linewidth]{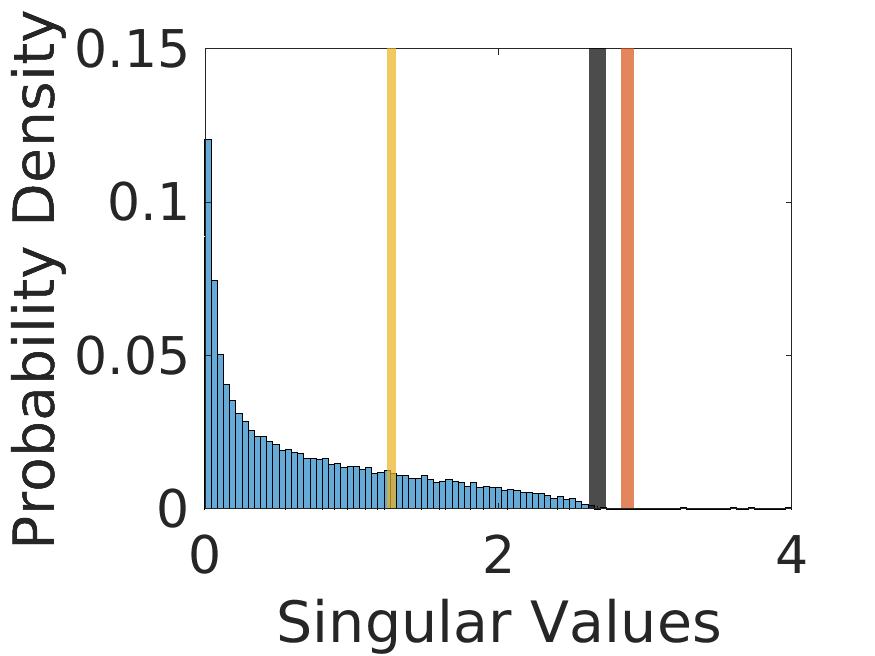}
 \end{minipage}

 \caption{\small ESD of $\wt S$ when $n=2000$. The black and yellow lines are the estimated bulk edge from \eqref{eq_hatlanbda_+} and TRAD respectively. The red line is $\vartheta$ in \eqref{eq_SN} by ScreeNOT. {The first row has $p/n = 0.5$ and the second row has $p/n = 1$.
 The first, second, and third columns are TYPE1, TYPE2, and TYPE 3 noise respectively.}} \label{fig_GAP_noise} 
\end{figure}

\subsubsection{Rank estimation}\label{sec_compareRANKest}
To evaluate the rank estimation, the ground truth is determined with $\widehat\alpha$; that is $r^+$ is determined by
$r^+_*:= \#\{i|d_i-\widehat\alpha>n^{-2/5}\}$, where $n^{-2/5}$ comes from the fact that the entries of $X$ are of 10 degrees of freedom. Note that since $\widehat\alpha$ is sufficiently close to $\alpha$, we can safely assume that $r^+ = r^+_*$ and check the performance of our rank estimator $\widehat r^+$.
In Figure \ref{fig:rankesterror}, we compare the estimated rank using \eqref{eq_adrk0}, the ScreeNOT rank, and the rank estimated by TRAD, which is the number of eigenvalues of $\wt S\wt S^\top$ larger than $\check{\sigma}(\wt S)(1+\sqrt{p/n})$. 
TRAD always overestimates the rank for {TYPE2 and TYPE3 noises} and sometimes overestimates the rank for TYPE1 noise. 
ScreeNOT rank {purposely discards weak components and thus} often underestimating the rank, with a larger error compared to our approach, and it does not seem to improve when $n$ grows.
Our approach outperforms others. Note that since there is {a gap} of order $n^{-1/3}$ as in \eqref{eq_adrk0} to rule out eigenvalues of $\wt S\wt S^\top$ that do not contain information, when $n$ is small, like $300$, our rank estimator sometimes underestimates the rank.

\begin{figure}[hbt!]
\begin{minipage}{1\textwidth}
\includegraphics[width=0.32\linewidth]{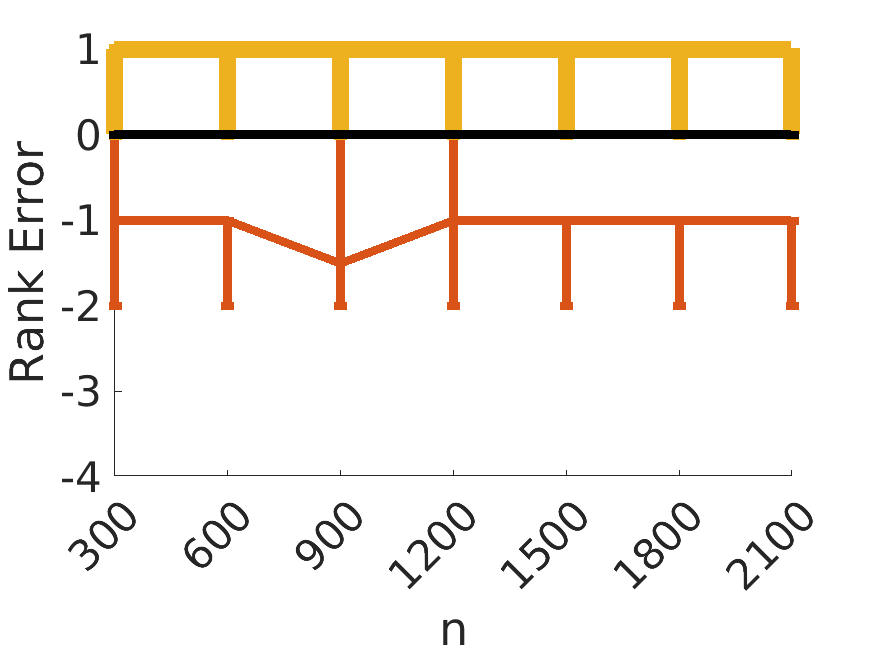}
\includegraphics[width=0.32\linewidth]{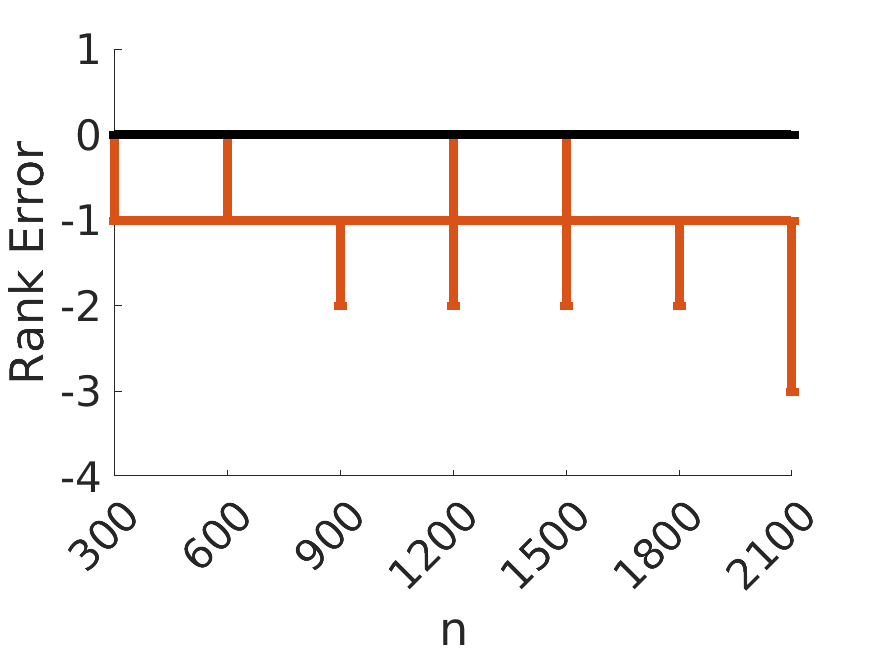}
\includegraphics[width=0.32\linewidth]{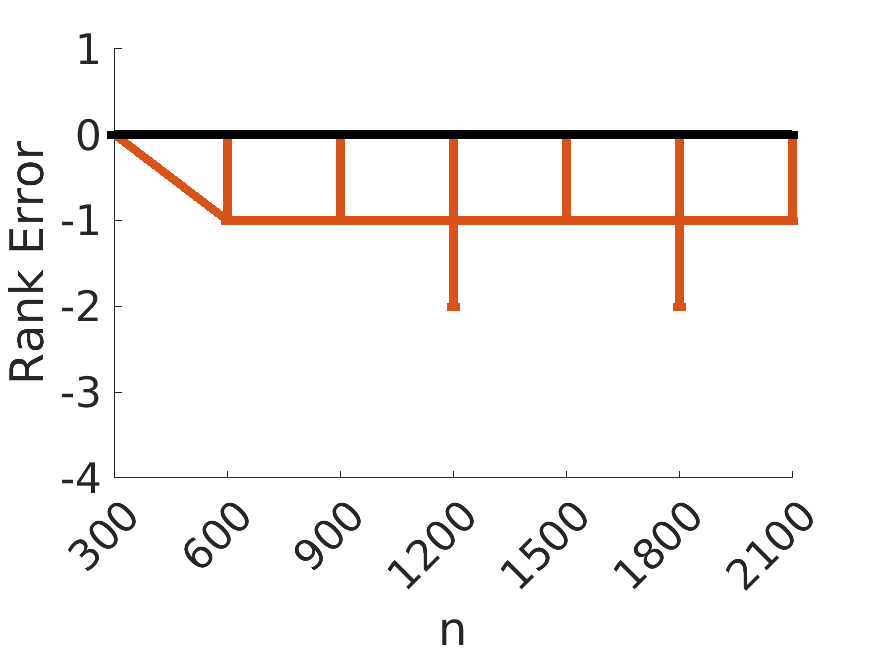}
\end{minipage}

\begin{minipage}{1\textwidth}
\includegraphics[width=0.32\linewidth]{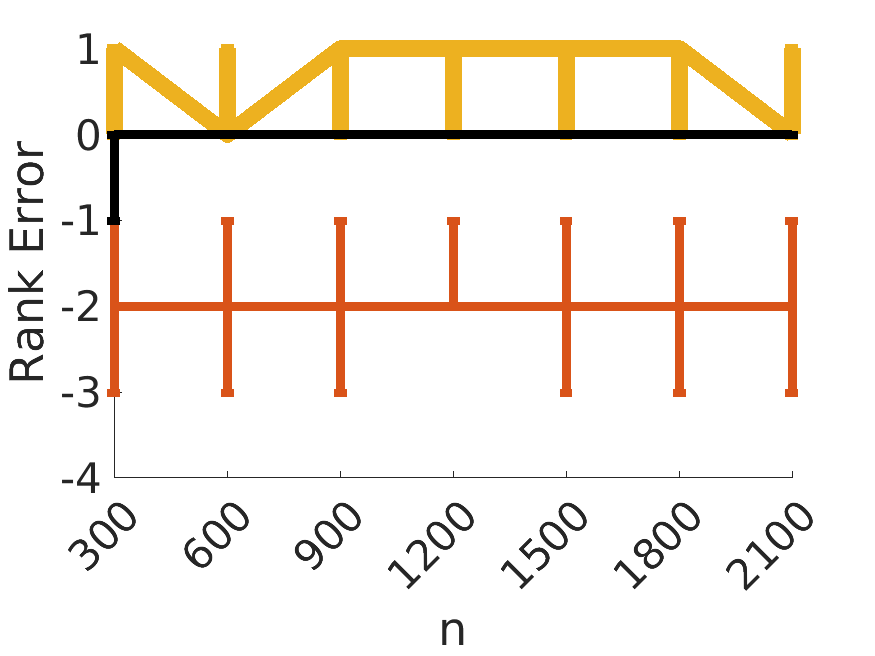}
\includegraphics[width=0.32\linewidth]{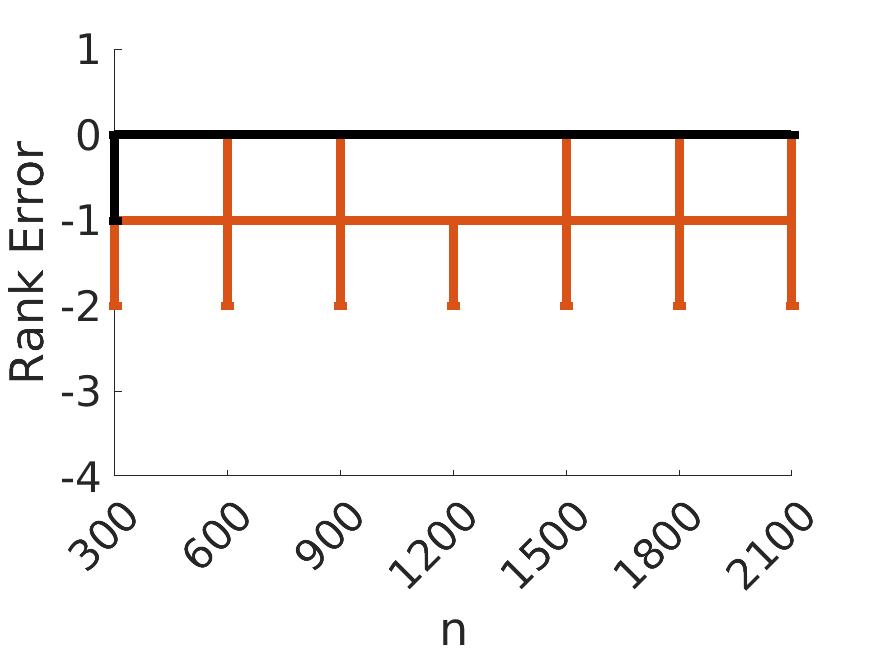}
\includegraphics[width=0.32\linewidth]{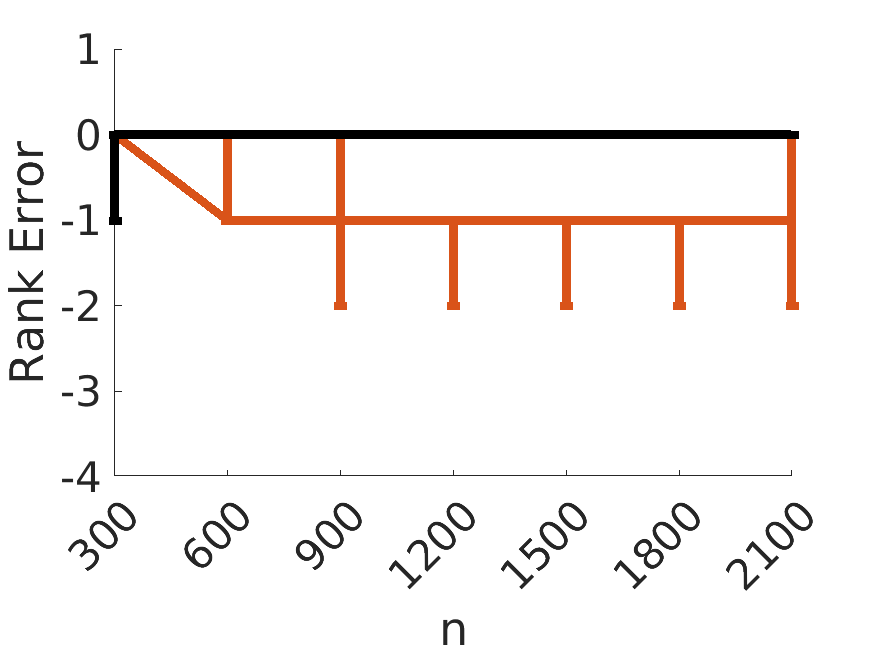}
\end{minipage}

\caption{\small A comparison of different rank estimators, where we show $\widehat r^+ - r^{+}$. The black (red and yellow respectively) lines are errors of rank estimator from our estimator \eqref{eq_adrk0} (ScreeNot and TRAD respectively). {The first row has $p/n = 0.5$ and the second row has $p/n = 1$.
 The first, second, and third columns are TYPE1, TYPE2, and TYPE 3 noise respectively.}}\label{fig:rankesterror}
\end{figure}

Results in Sections { \ref{sec_compareTSE} and \ref{sec_compareRANKest}} show that when TYPE1 noise is considered, TRAD always has the best performance. This is because TRAD has the closed form of $\lambda_+$ and optimal shrinkers.
However, when the noise is {of TYPE2 or TYPE3}, the closed form is invalid and TRAD leads to a large error. 
ScreeNOT cannot detect weak signals and does not perform shrinkage, thus it always has high errors for both types of noise. eOptShrink has a similar performance as TRAD when $n$ is large under TYPE1 noise and has the best performance under {TYPE2 and TYPE3} noises. In summary, these simulations guarantee that \eqref{eq_adrk0} is a precise estimator of $r^+$. This finding supports the application of $\widehat r^+$ \eqref{eq_adrk0} in the proposed eOptShrink.

\subsubsection{Optimal shrinkage estimation via $\widehat F_{\texttt{T}}(x)$, $\widehat F_{\texttt{imp}}(x)$ and $\widehat F_{\texttt{e}}(x)$}\label{section:subsection:compare F}

We compare how well we could use the proposed pseudo-distribution $\widehat F_{\texttt{e}}$ and existing $\widehat F_{\texttt{T}}(x)$ and $\widehat F_{\texttt{imp}}(x)$ to estimate the optimal shrinkage in Proposition \ref{prop_optimal_shrinker} via estimating $d_{r^+}$ and $\sqrt{a_{1,r^+}a_{2,r^+}}$. 
Set $k = 2r$ for $\widehat F_{\texttt{imp}}(x)$ in \eqref{eq_imp0} by using the oracle rank information. For a fair comparison, we set $c$ so that $\lfloor n^c \rfloor = 2r$ when we modify the top $2\widehat r^+$ eigenvalues of $\wt S \wt S^\top$ in $\widehat F_{\texttt{e}}(x)$ \eqref{eq_imp}.
As shown in the last simulation, our rank estimator is precise when $n$ is large, but errors happen when $n$ is small. Since $\widehat F_{\texttt{T}}(x)$ and $\widehat F_{\texttt{e}}(x)$ depend on the estimated rank $\widehat r^+$, we consider the case when $\widehat r^+$ satisfies $\widehat r^+ - r^+ =-2,\ldots, +2$; that is, we study how badly the erroneous rank estimate could impact the final result.
In Figures \ref{fig_d+_Compare_Fmut} and \ref{fig_d+_Compare_Fmut2}, the error ratio of estimating $d_{\min\{r^+,\widehat r^+\}}$ and $\sqrt{a_{1,\min\{r^+,\widehat r^+\}}a_{2,\min\{r^+,\widehat r^+\}}}$ respectively for $n=300$ and $n = 600$ with different pseudo-distributions. Note that the $r^+$th singular value is the smallest ``strong'' one, which is the most challenging one to recover. For each $n$, we repeat the simulation for $100$ times. 
Clearly, our approach has the lowest error ratio with a statistical significance, and the error decreases when $n$ grows from $300$ to $600$. This result shows that compared with ScreeNot, eOptShrink is robust to a slightly erroneous rank estimation. We thus have a theoretical guarantee for the application of $\widehat F_{\texttt{e}}(x)$ in eOptShrink.

\begin{figure}[hbt!]
\begin{minipage}{1\textwidth}
\includegraphics[width=0.24\linewidth]{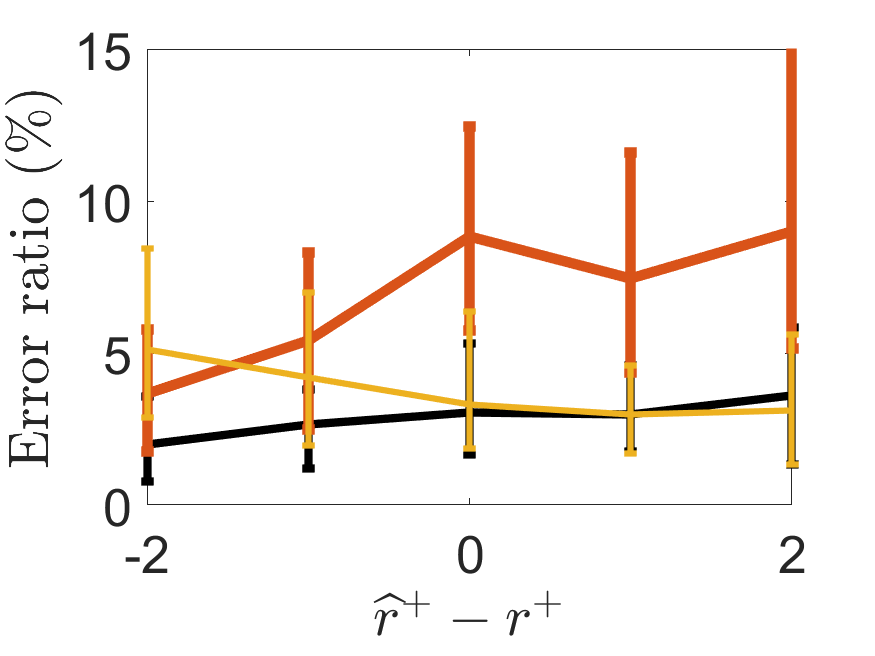}
\includegraphics[width=0.24\linewidth]{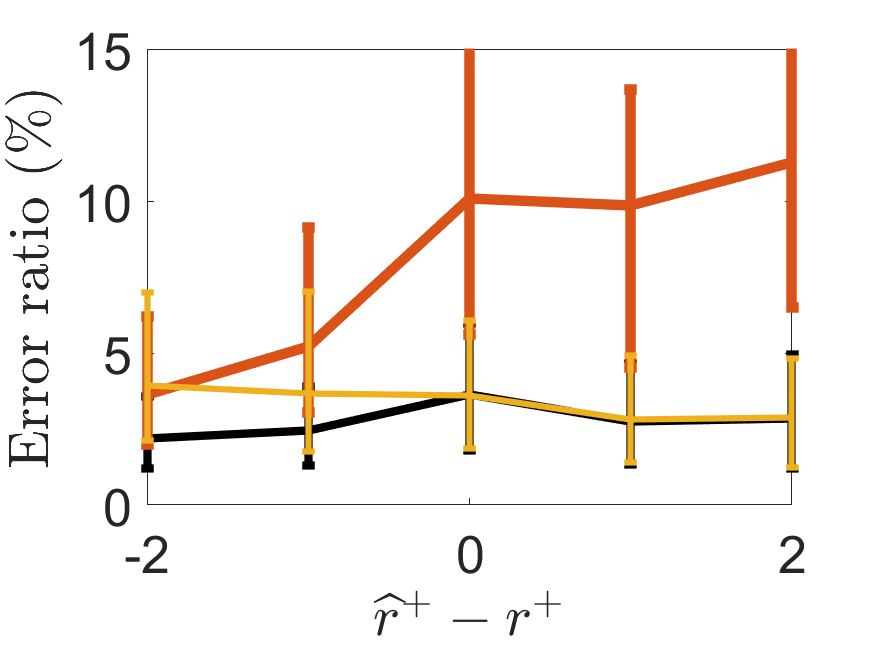}
\includegraphics[width=0.24\linewidth]{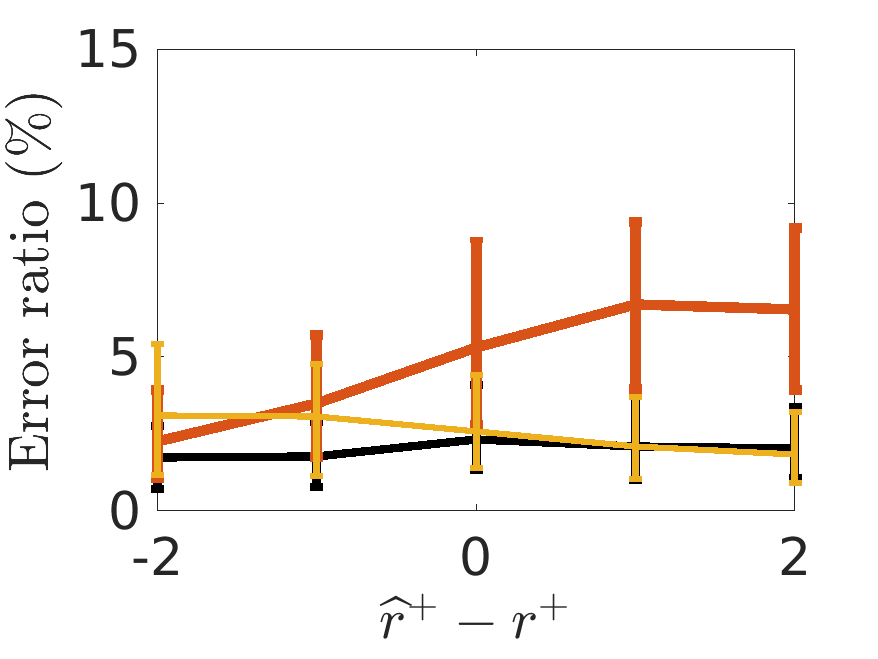}
\includegraphics[width=0.24\linewidth]{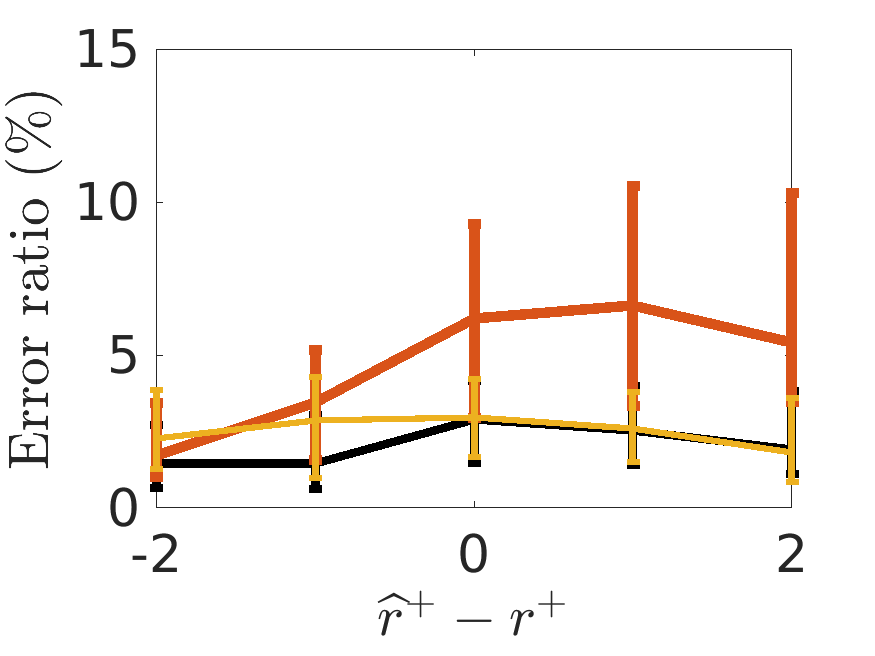}
\end{minipage}

\begin{minipage}{1\textwidth}
\includegraphics[width=0.24\linewidth]{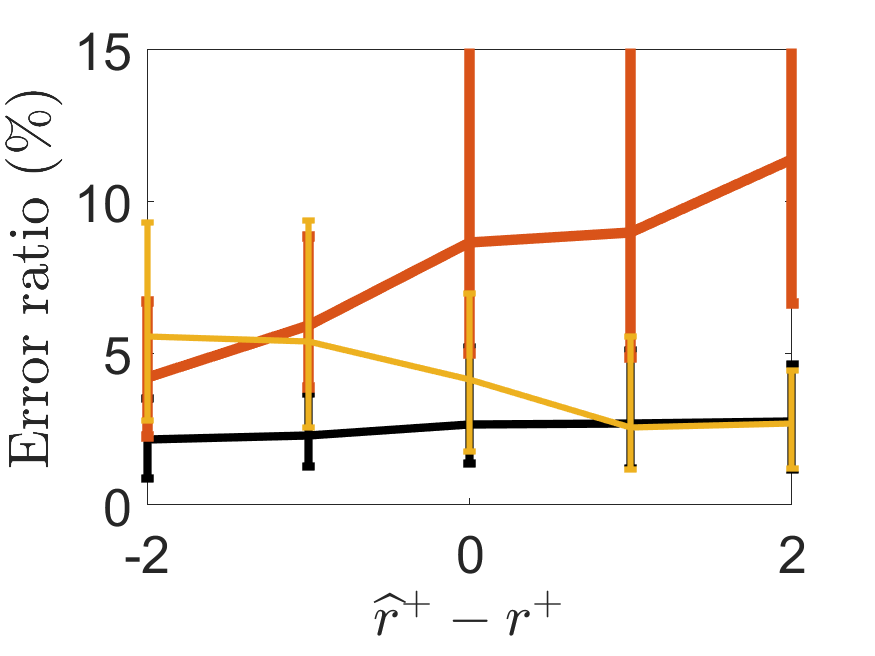}
\includegraphics[width=0.24\linewidth]{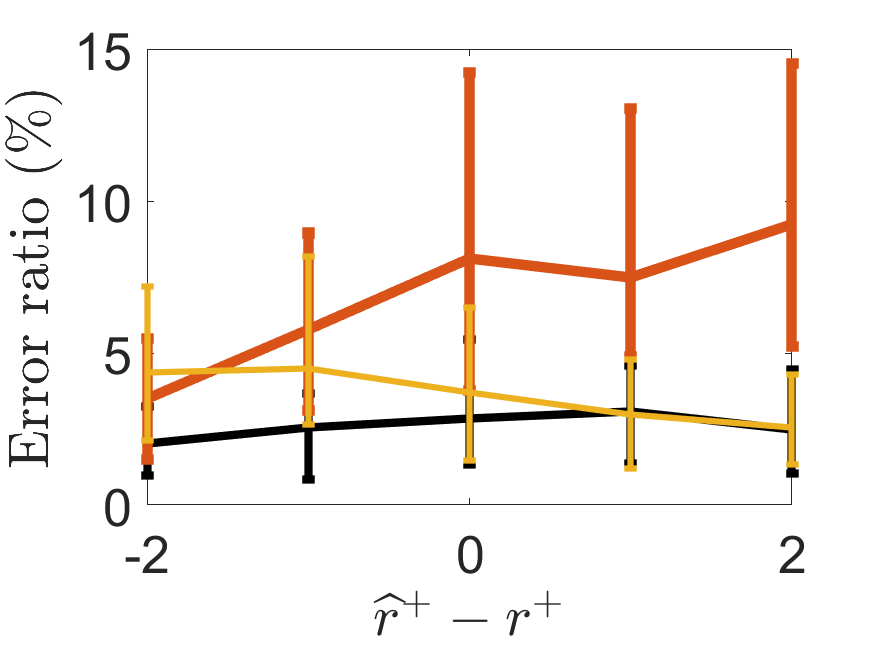}
\includegraphics[width=0.24\linewidth]{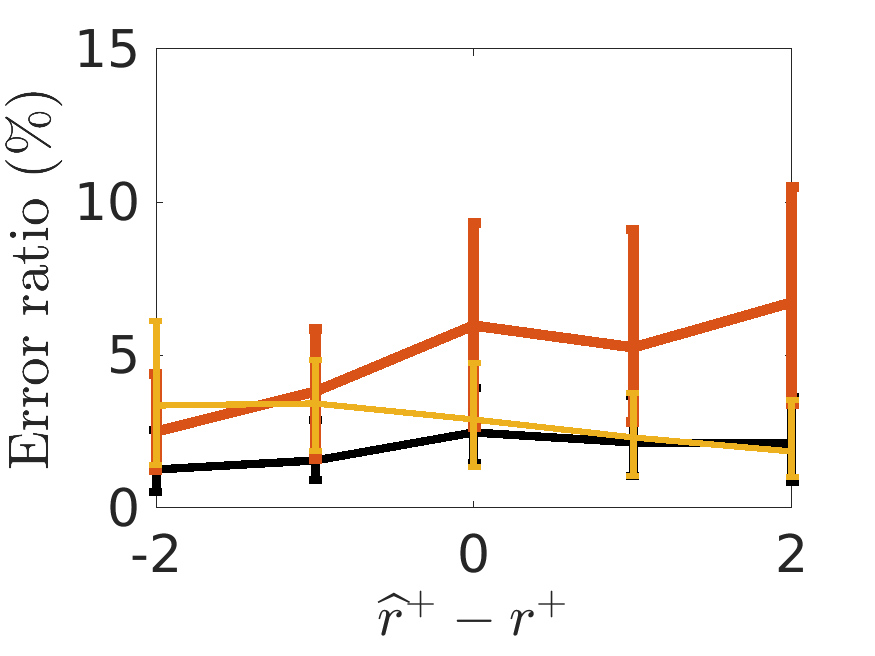}
\includegraphics[width=0.24\linewidth]{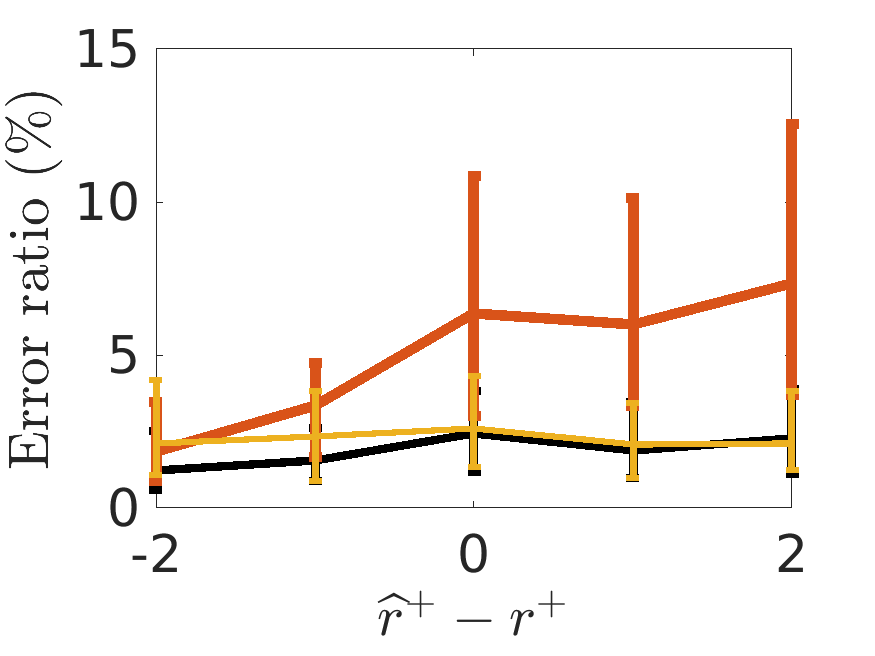}
\end{minipage}

\begin{minipage}{1\textwidth}
\includegraphics[width=0.24\linewidth]{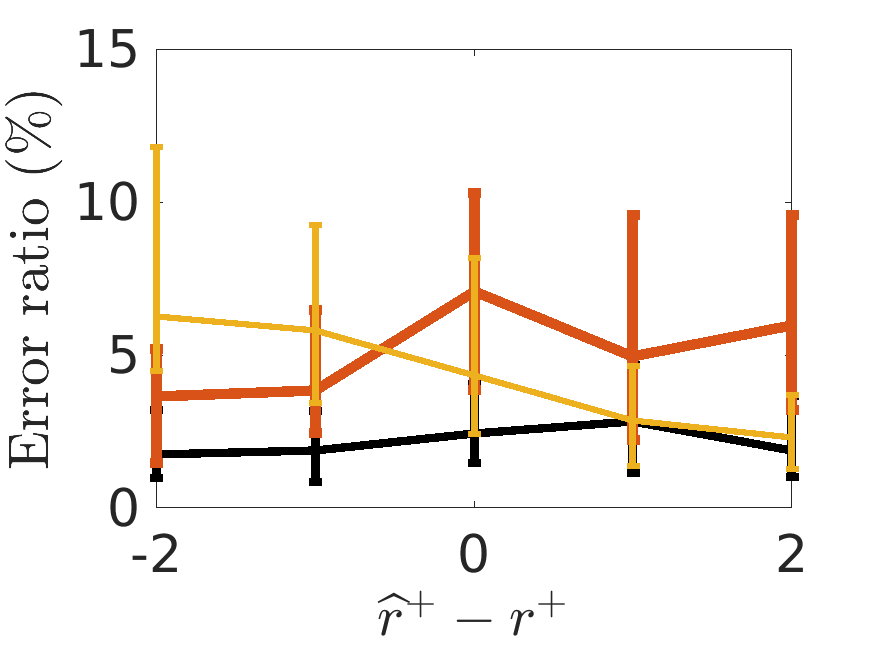}
\includegraphics[width=0.24\linewidth]{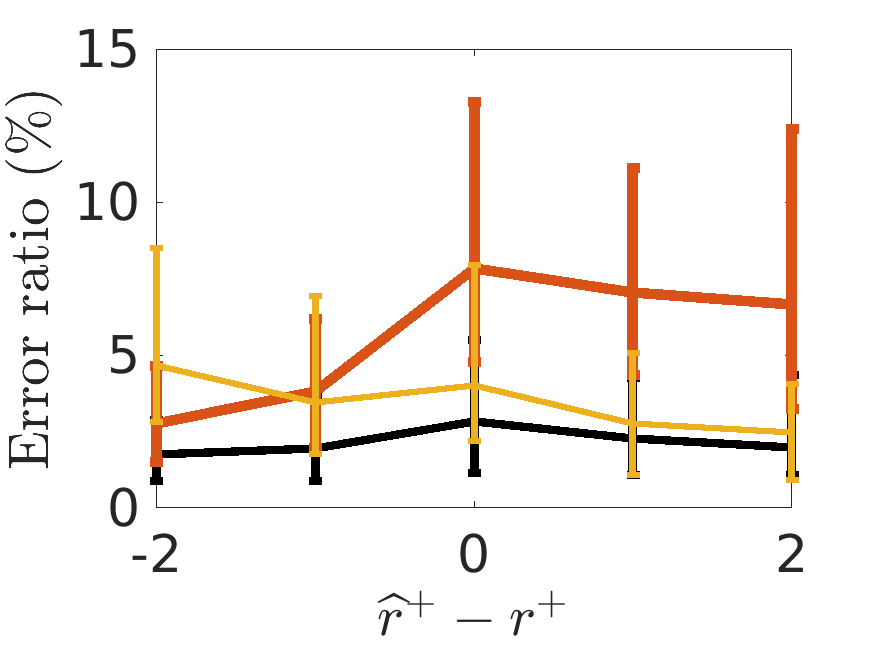}
\includegraphics[width=0.24\linewidth]{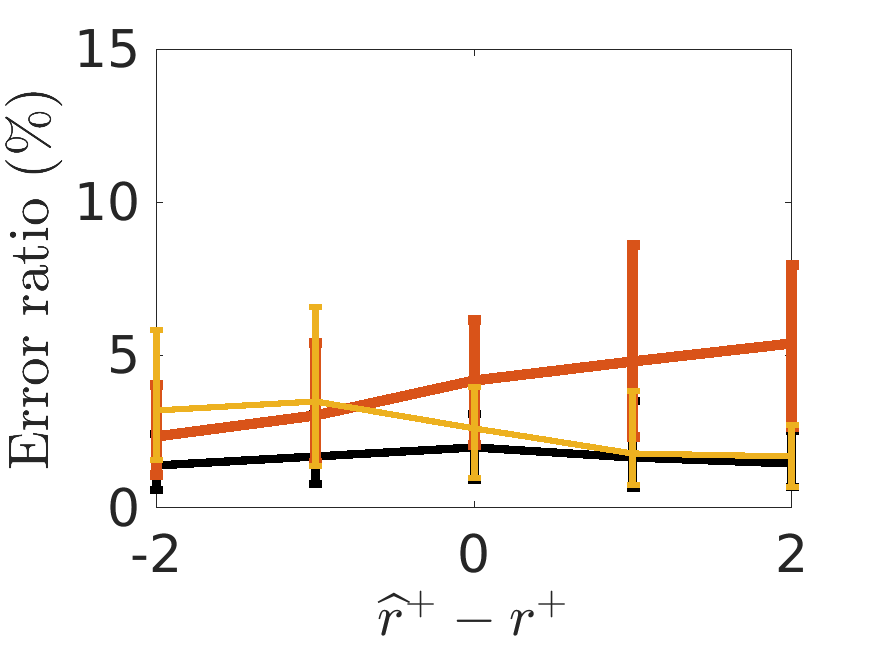}
\includegraphics[width=0.24\linewidth]{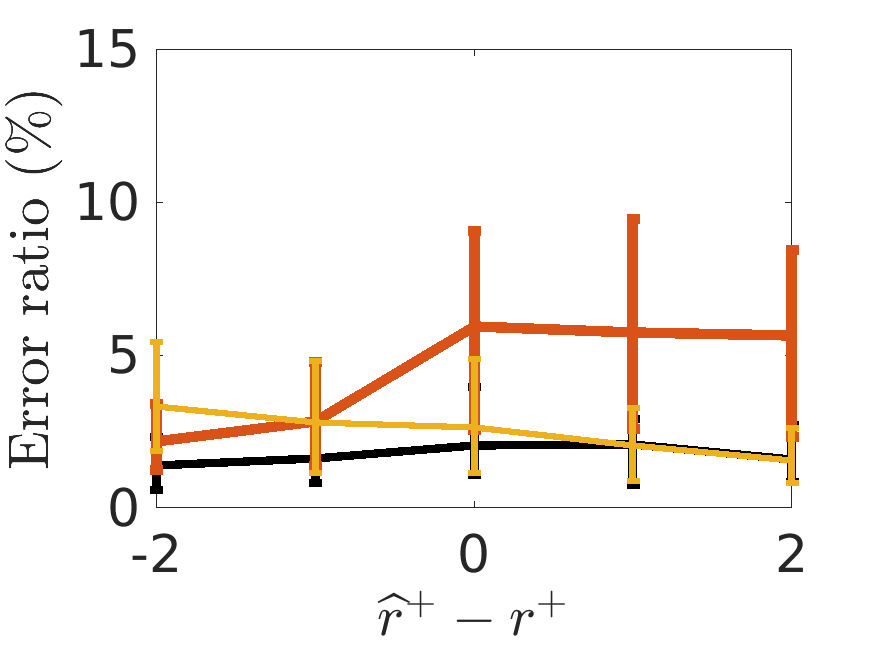}
\end{minipage}

 \caption{\small Interquartile errorbars of error ratios of estimating $d_{\min\{r^+,\widehat r^+\}}$ by using $\widehat F_{\texttt{e}}(x)$, $\widehat F_{\texttt{T}}(x)$ and $\widehat F_{\texttt{imp}}(x)$, shown in black, yellow, and red lines respectively. If the corresponding error ratio is too high, the associated curve is not totally plotted to enhance the visualization. 
 {From the first row to the third row: TYPE1, TYPE2, and TYPE3 noise. From the first column to the fourth column: $p/n = 0.5$ and $n = 300$, $p/n = 1$ and $n = 300$, $p/n = 0.5$ and $n = 600$, and $p/n = 1$ and $n = 600$.}
}\label{fig_d+_Compare_Fmut}
\end{figure}

\begin{figure}[hbt!]
\begin{minipage}{1\textwidth}
\includegraphics[width=0.24\linewidth]{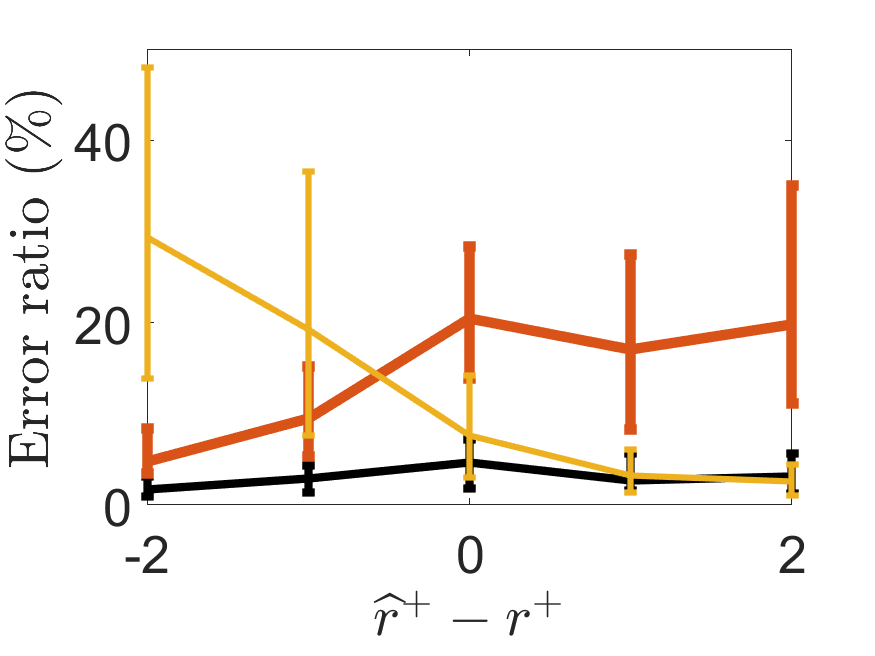}
\includegraphics[width=0.24\linewidth]{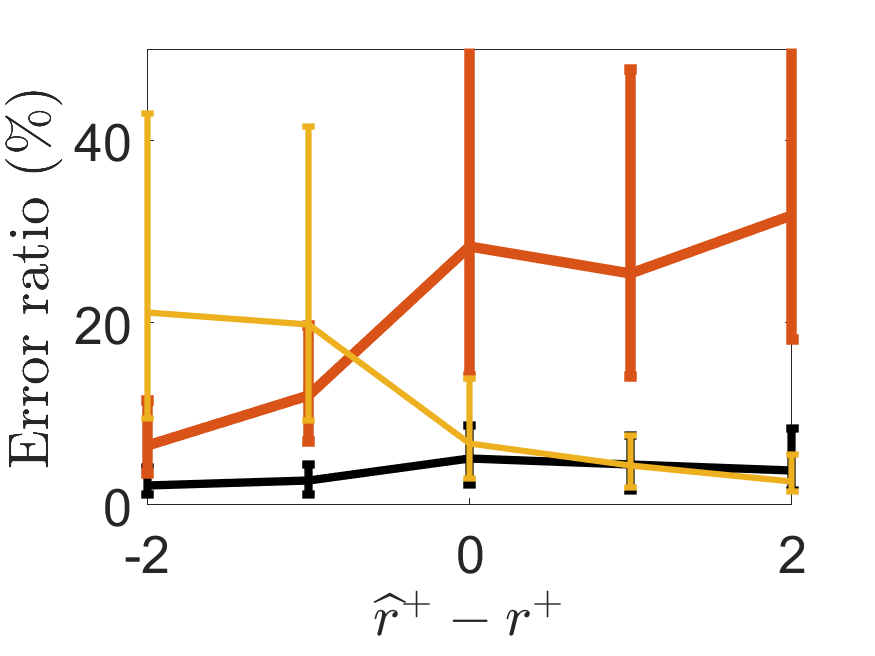}
\includegraphics[width=0.24\linewidth]{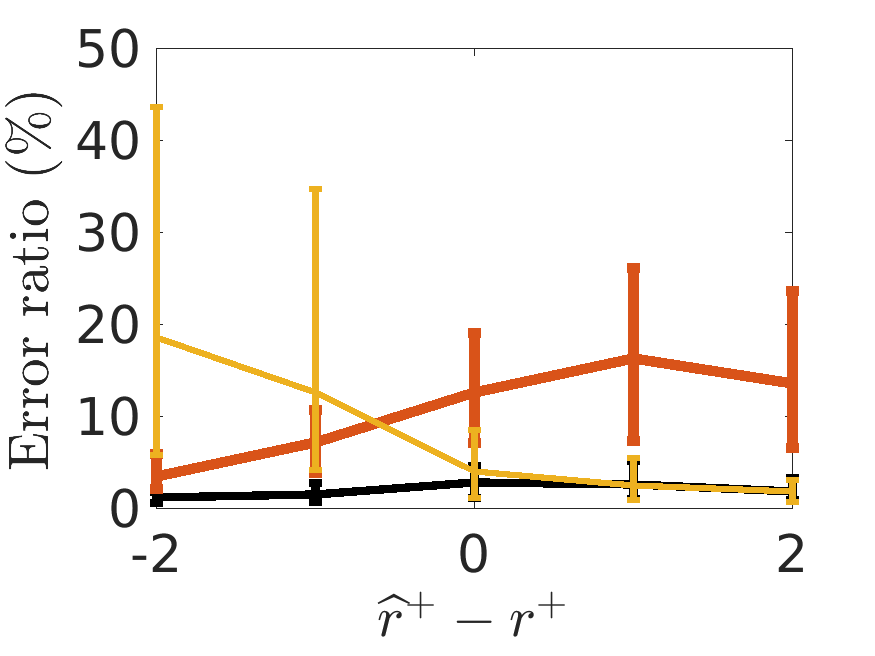}
\includegraphics[width=0.24\linewidth]{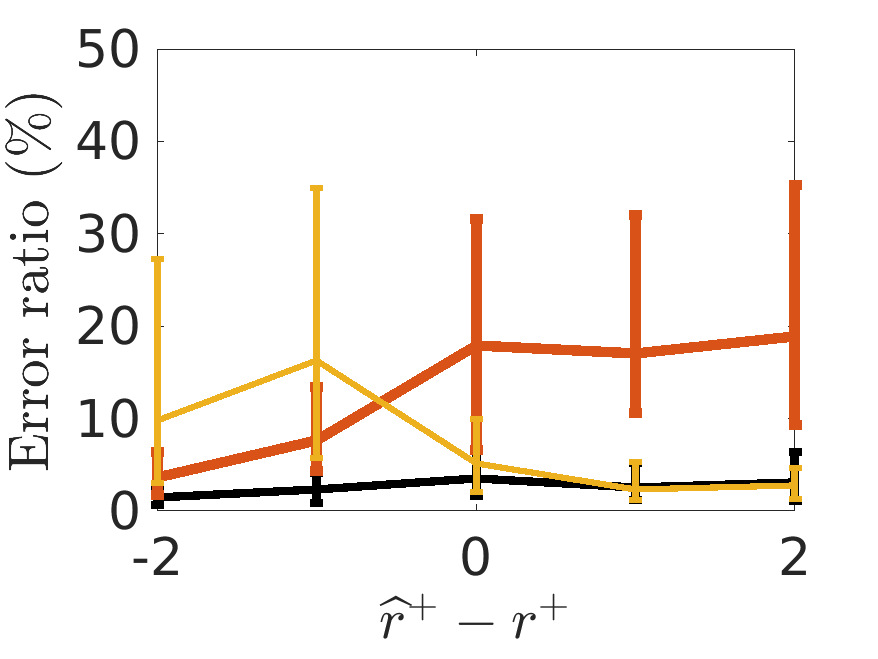}
\end{minipage}

\begin{minipage}{1\textwidth}
\includegraphics[width=0.24\linewidth]{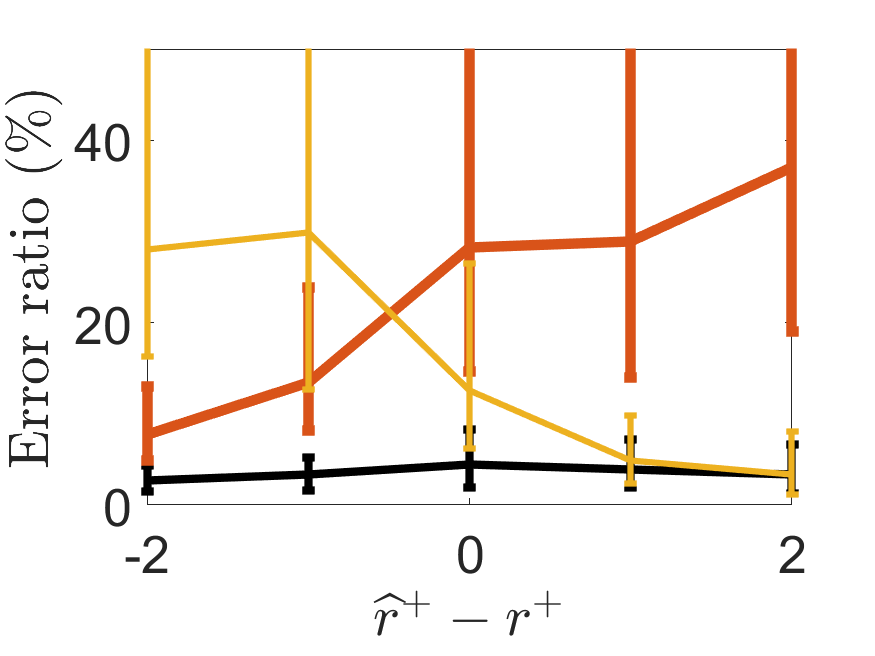}
\includegraphics[width=0.24\linewidth]{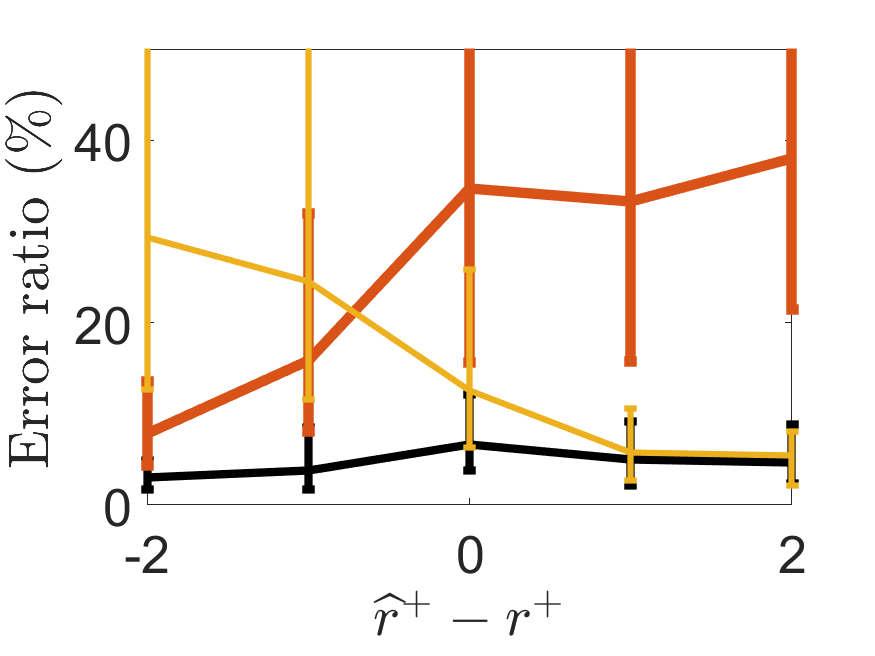}
\includegraphics[width=0.24\linewidth]{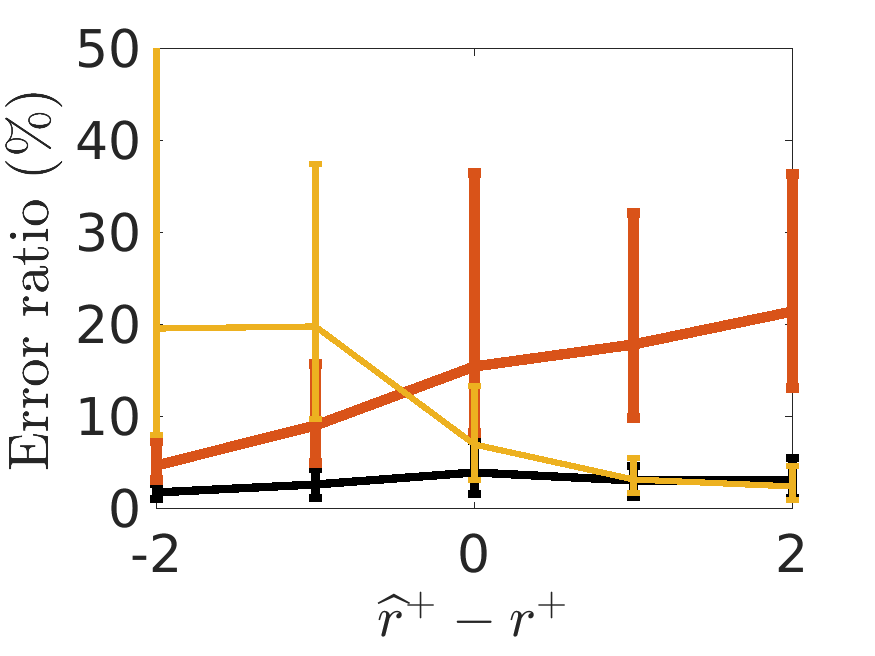}
\includegraphics[width=0.24\linewidth]{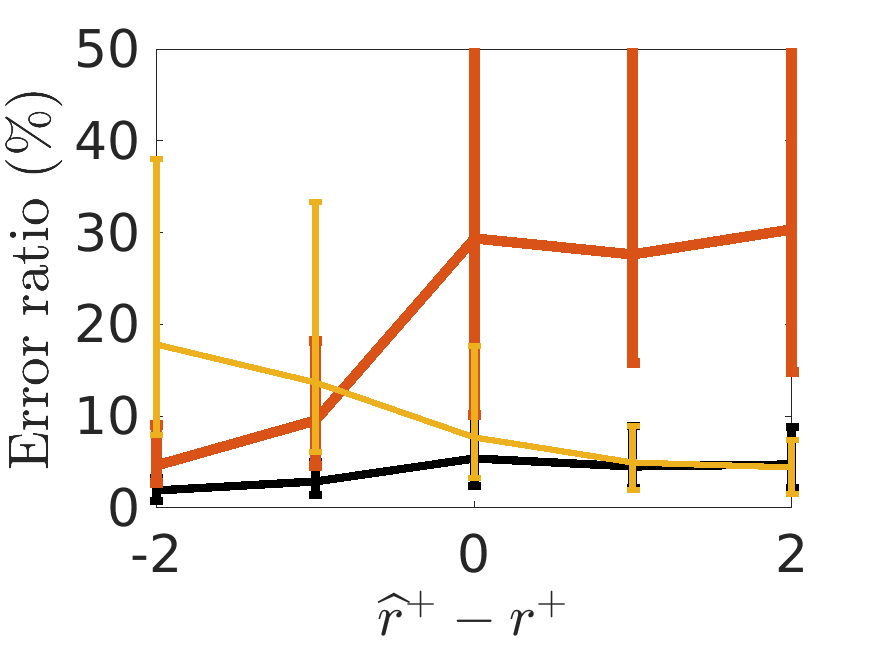}
\end{minipage}

\begin{minipage}{1\textwidth}
\includegraphics[width=0.24\linewidth]{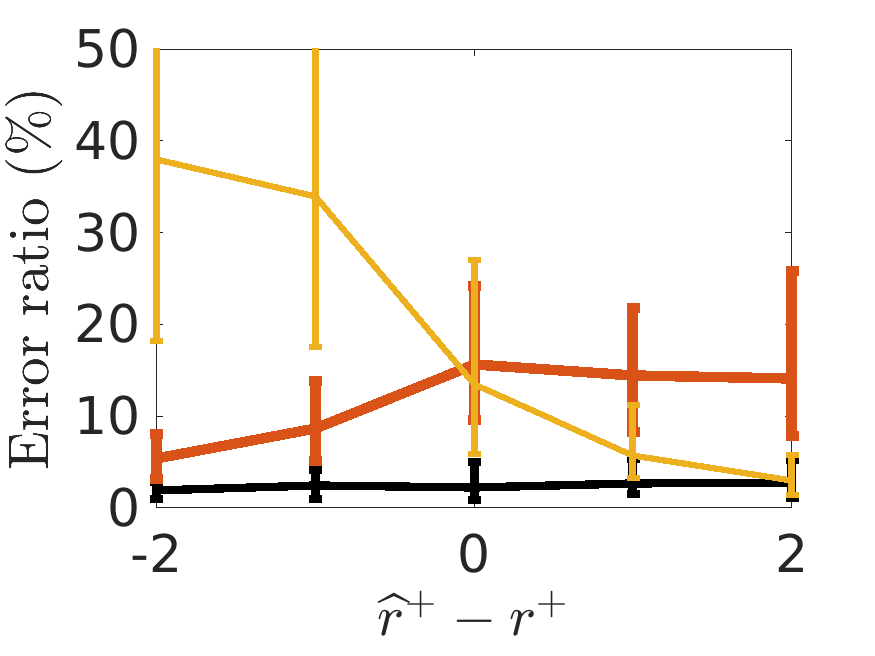}
\includegraphics[width=0.24\linewidth]{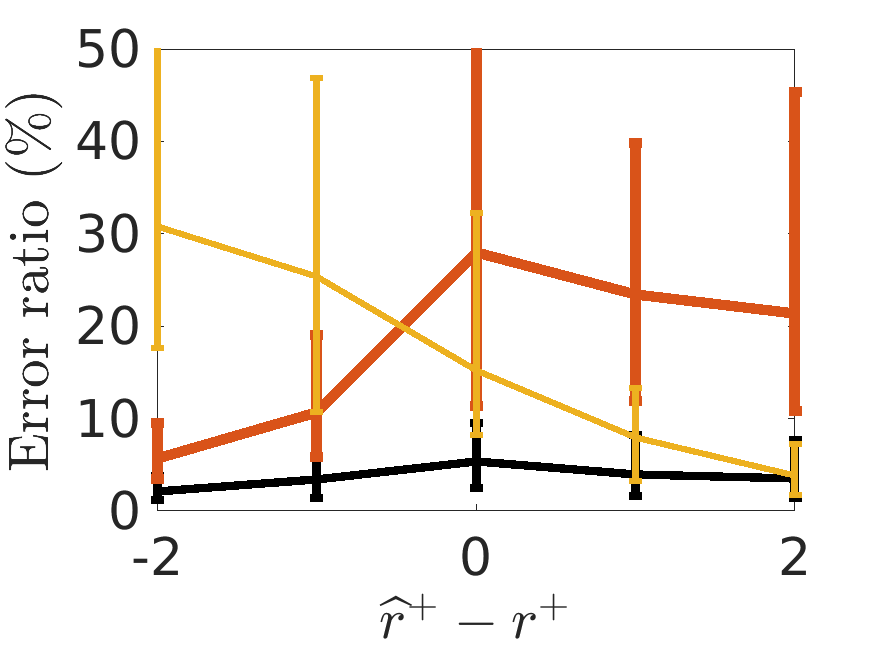}
\includegraphics[width=0.24\linewidth]{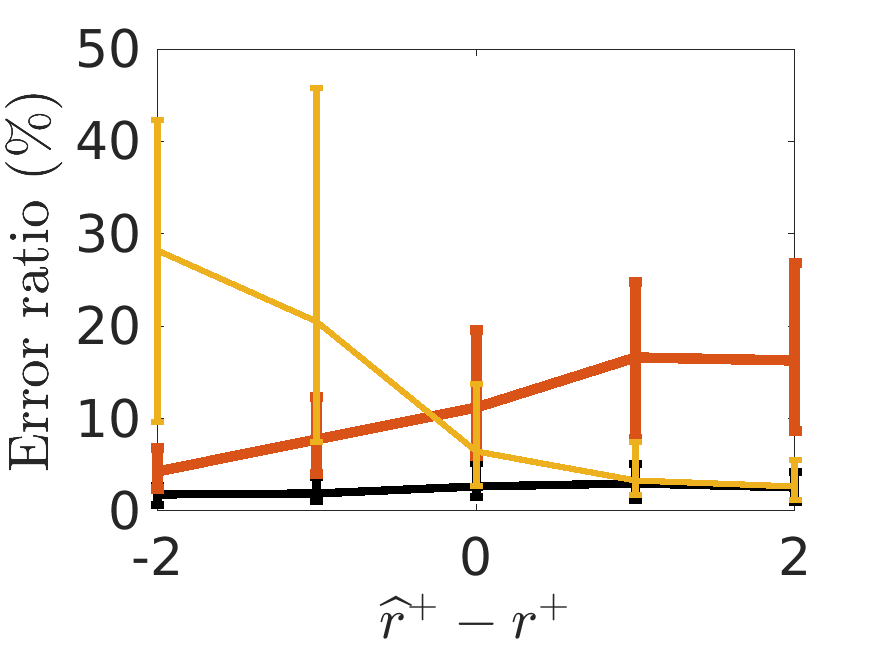}
\includegraphics[width=0.24\linewidth]{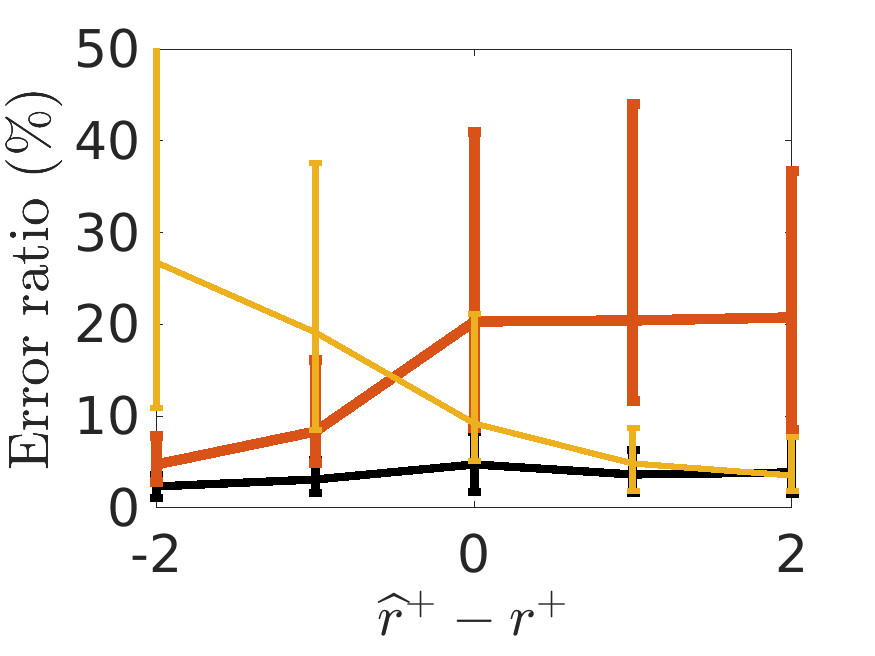}
\end{minipage}
 \caption{\small Interquartile errorbars of error ratios of estimating $\sqrt{a_{1,\min\{r^+,\widehat r^+\}}a_{2,\min\{r^+,\widehat r^+\}}}$ using $\widehat F_{\texttt{e}}(x)$, $\widehat F_{\texttt{T}}(x)$ and $\widehat F_{\texttt{imp}}(x)$, shown in black, yellow, and red lines respectively. If the corresponding error ratio is too high, the associated curve is not totally plotted to enhance the visualization. {From the first row to the third row: TYPE1, TYPE2, and TYPE3 noise. From the first column to the fourth column: $p/n = 0.5$ and $n = 300$, $p/n = 1$ and $n = 300$, $p/n = 0.5$ and $n = 600$, and $p/n = 1$ and $n = 600$.}
}\label{fig_d+_Compare_Fmut2}
\end{figure}

\subsection{Fetal ECG extraction problem}
In our previous work \cite{su2019recovery}, TRAD is the critical step of the algorithm to recover the mECG when we have only one or two ta-mECG channels. The algorithm is composed of two steps. The first step is mainly for two channels, and we ignore it when we only have one channel. The second step is composed of two substeps. Step 2-1 is designed to detect the maternal R peaks from the single channel ta-mECG, which is not the concern of OS. Step 2-2 is mainly illustrated in Figure \ref{Figure1}, where we view fECG as noise and mECG as the signal, and OS is applied to recover mECG from the ta-mECG. As is mentioned in Introduction, fECG, when viewed as noise, is not white, and there is a dependence among segments. Thus, it is natural to consider replacing TRAD in \cite{su2019recovery} by ScreeNOT or eOptShrink.
We consider a semi-real simulated database and a real-world database from {\em 2013 PhysioNet/Computing in Cardiology Challenge} \cite{silva2013noninvasive}, abbreviated as CinC2013. 

\subsubsection{Semi-real simulated database}
\label{Section more numerics simulated fECG}
The semi-real ta-mECG data is constructed from the Physikalisch-Technische Bundesanstalt (PTB) Database \url{https://physionet.org/physiobank/database/ptbdb/}, abbreviated as \texttt{PTBDB} following the same way detailed in \cite{su2019recovery}. The database contains 549 recordings from 290 subjects (one to five recordings for one subject) aged 17 to 87 with the mean age 57.2. 52 out of 290 subjects are healthy. 
Each recording includes simultaneously recorded conventional 12 lead and the Frank lead ECGs. Each signal is digitalized with the sampling rate 1000 Hz. More technical details can be found online. Take 57-second Frank lead ECGs from a healthy recording, denoted as $V_x(t)$, $V_y(t)$ and $V_z(t)$ at time $t\in \mathbb{R}$, as the maternal vectocardiogram (VCG). Take $(\theta_{xy},\theta_z) = (\frac{\pi}{4},\frac{\pi}{4})$, and the simulated mECG is created by $\text{mECG}(t)=(V_x(t)  \cos{\theta_{xy}} + V_y(t)  \sin{\theta_{xy}}) \cos{\theta_{z}} + V_z(t)  \sin{\theta_{z}}$. We create {40} mECGs. The {\em simulated fECG} of healthy fetus are created from another 40 recordings from healthy subjects, where $114$-second V2 and V4 recordings are taken. The simulated and simulated fECG come from different subjects. The simulated fECG then are resampled at $500$ Hz.
As a result, the simulated fECG has about double the heart rate compared with the simulated mECG if we consider both signals sampled at $1000$ Hz. The amplitude of the simulated fECG is normalized to the same level of simulated mECG and then multiplied by $0<R<1$ shown in the second column of Table \ref{table3} to make the amplitude relationship consistent with the usual situation of real ta-mECG signals. We generate {40} simulated healthy fECGs.
The clean simulated ta-mECG is generated by directly summing simulated mECG and fECG. 
We then create a simulated noise starting with a random vector $\textbf{x} = (x_1,x_2,x_3,\ldots)$ with i.i.d entries with student t-10 distribution. The noise is then created and denotes as $\textbf{z}$ with the entries $z_i = (1+0.5\sin((i \mod 500)/500))(x_{i}+x_{i+1}+x_{i+2})$. The final simulated ta-mECG is generated by adding the created noise to the clean simulated ta-mECG according to the desired SNR ratio  shown in the first column of Table \ref{table3}. As a result, we acquire {$40$} recordings of {$57$} seconds simulated ta-mECG signals with the sampling rate $1000$ Hz.

{ Assuming that the covariance structure of the summation of fECG and background noise can be explained by our separable covariance model $Z = A^{1/2}XB^{1/2}$,} we apply our eOptShrink algorithm to each recording in the simulated database in Step 2-2. We then compare its performance with TRAD and OptShrink and report the root mean square error (RMSE) of the recovered mECG by comparing it with the ground truth mECG. 
{ It is claimed that OptShrink works when the effective rank is overestimated \cite{nadakuditi2014optshrink}. Thus, similar to Section \ref{sec_selectionC}, we set the required overestimated rank to $\lfloor n^{c} \rfloor \gg r^+$, where $c = \min(\frac{1}{2.01},\frac{1}{\log(\log(n))})$, the same constant chosen in Algorithm \ref{alg}. Additionally, note that OptShrink only discusses the OS with Frobenius norm loss, so we apply both TRAD and eOptShrink with the Frobenius loss for comparisons. The results are shown in Table \ref{table3}. Clearly, eOptShrink has the smallest RMSE in all scenarios compared with TRAD and OptShrink.}
Overall, a lower SNR results in a higher RMSE across every amplitude ratio $R$ for all OS methods. Moreover, a higher fECG amplitude results in a higher noise level for the recovery of mECG, thus yielding a higher RMSE across all SNRs for all OS methods.

\begin{table}[htb!]
{
\small
\centering
\begin{tabular}{|c||c|c|c|c|}
\multicolumn{5}{c}{\vspace{-3pt}}\\
\cline{1-5}
SNR & $R$ & TRAD  & OptShrink & eOptShrink  \\
\hline
\multirow{3}{*}{$1$ dB} & 1/4 & 21.88 $\pm$ 4.84  &  13.36  $\pm$ 1.82   & 7.88* $\pm$ 1.82\\
& 1/6 & 18.23 $\pm$ 4.15  & 10.94 $\pm$ 1.49 &6.65* $\pm$ 1.46\\

& 1/8  & 16.22 $\pm$ 3.62  & 9.49 $\pm$ 1.29 & 5.92* $\pm$ 1.26 \\
\hline
\hline
\multirow{3}{*}{$0$ dB} & 1/4 & 22.40 $\pm$ 4.99  & 13.79 $\pm$ 1.97 & 8.04* $\pm$ 1.76 \\
&1/6 & 18.52 $\pm$ 4.19  & 11.31 $\pm$ 1.62 & 6.82* $\pm$ 1.50\\
& 1/8 &16.60 $\pm$ 3.89  & 9.82 $\pm$ 1.40 & 6.05* $\pm$ 1.27\\

\hline
\hline
\multirow{3}{*}{$-1$ dB} 
& 1/4& 22.94 $\pm$ 5.23 & 14.31$\pm$ 2.10 & 8.28* $\pm$ 1.78 \\
& 1/6 & 18.97 $\pm$ 4.42 & 11.69 $\pm$ 1.74 & 6.95* $\pm$ 1.50\\
& 1/8 & 16.91 $\pm$ 3.97  & 10.15 $\pm$ 1.50 & 6.21* $\pm$ 1.28\\ 

\hline

\end{tabular}
\caption{\small The comparison of RMSE for the mECG morphology of different algorithms applied to the simulated ta-mECG database. $R$ is the simulated fECG amplitude. All results are presented as mean $\pm$ standard deviation. The asteroid next to mean stand for the statistical significance when comparing eOptShrink to OptShrink and TRAD using the paired t-test. {Optimal shrinkers with respect to Frobenius norm loss is applied for both TRAD and eOptShrink}}\label{table3}

}
\end{table}

\subsubsection{CinC2013 database}
Each recording in CinC2013 comprises four ta-mECG channels and simultaneously recorded directly contacted fECG, all resampled at a sampling rate of 1000 Hz and lasting for 1 minute. For each channel of every recording, we execute the fECG extraction algorithm using various OS algorithms in Step 2-2. In Figure \ref{fig_rfECG}, we compare the recovered mECG and the detected fetal R peak locations using {TRAD, OptShink, and eOptShrink. Similar to the semi-real simulated database, we utilize the required overestimated rank equal to $\lfloor n^{c} \rfloor$ for OptShrink, where $c$ aligns with Algorithm \ref{alg}, and employ both TRAD and eOptShrink with the Frobenius norm loss.} It is evident that eOptShrink produces a more accurate recovery of the {\em morphology} of mECG and consequently the fECG. Notably, {we observe ventricular activity residues from fECG in the estimated mECG recovered by OptShrink, indicated by purple arrows, which are absent in eOptShrink. Consequently, eOptShrink enables better evaluation of the electrophysiological properties of both maternal and fetal hearts.} The clinical implications will be detailed in our forthcoming work.

\begin{figure}[hbt!]
\centering
\includegraphics[trim=20 0 20 0, clip, width=0.325\linewidth]{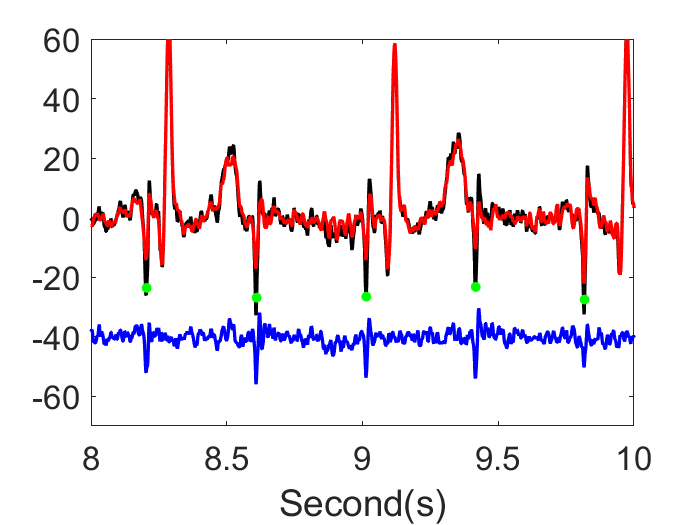}
\includegraphics[trim=20 0 20 0, clip, width=0.325\linewidth]{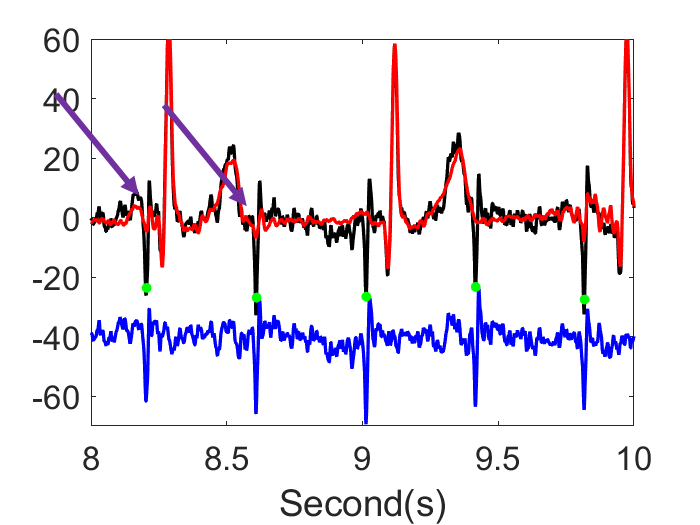}
\includegraphics[trim=20 0 20 0, clip, width=0.325\linewidth]{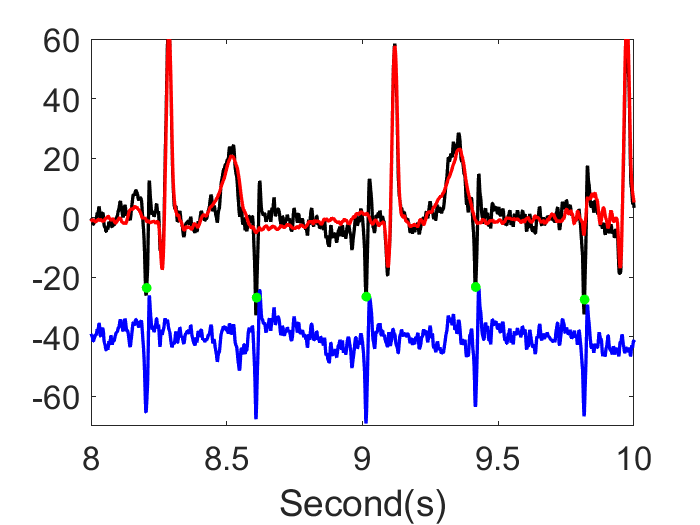}
 \caption{\small A comparison of mECG and fECG recovery using different algorithms (recording 1 and channel 1 from CinC2013) is illustrated. The left panel depicts results obtained by TRAD, the middle one by OptShrink, and the right one by eOptShrink. The original ta-mECG is depicted in black, with the recovered mECG overlaid in red, and the fetal R peaks labeled by experts are indicated as green dots. The difference between the ta-mECG and the recovered mECG is displayed in blue, representing the recovered fetal ECG. The purple arrows highlight ventricular activity residue, a limitation of OptShrink due to the lack of a precise estimator of effective rank.}\label{fig_rfECG}
\end{figure}

\section*{Reference}
\bibliographystyle{plain}
\bibliography{ev}

\clearpage
\appendix

\setcounter{equation}{0}
\setcounter{figure}{0}
\setcounter{table}{0}
\setcounter{page}{1}
\renewcommand{\thepage}{S.\arabic{page}}
\renewcommand{\thetable}{S.\arabic{table}}
\renewcommand{\thefigure}{S.\arabic{figure}}
\renewcommand{\thesection}{S.\arabic{section}}
\renewcommand{\thelemma}{S.\arabic{lemma}}
\renewcommand{\theequation}{S.\arabic{section}.\arabic{equation}}

\section{Background knowledge}
\subsection{Some linear algebra tools}
We recall a perturbation bound for determinants.
\begin{lemma} \label{eq_lemmadeter} Let $A$ and $E$ be two $m \times m$ matrices, where $\sigma_1 \geq \sigma_2 \geq \cdots \geq \sigma_m$ are singular values of $A$. Then  
\begin{equation*}
|\det(A+E)-\det(A)| \leq \sum_{i=1}^m s_{m-i} \norm{E}^i_2\,,
\end{equation*}
where $s_k:=\sum_{1 \leq i_1<i_2<\cdots<i_k \leq m} \sigma_{i_1} \cdots \sigma_{i_k}$ is the $k$-th elementary symmetric functions of singular values of $A$ and $s_0:=1$. 
\end{lemma}
\begin{lemma}\label{lemma det diag formula}
For any $a,b,\alpha_1,\ldots,\alpha_r\in \mathbb{R}$ and $r\in \mathbb{N}$, we have 
\begin{equation}
\det \begin{pmatrix}
a I_r & \textup{diag}\{\alpha_1, \cdots, \alpha_r\} \\
\textup{diag}\{\alpha_1, \cdots, \alpha_r\} & bI_r 
\end{pmatrix}
=\prod_{i=1}^r(ab-\alpha_i^2)\,.
\end{equation}
\end{lemma}
{This lemma can be proved by a direct expansion.} Further, we summarize the Weyl's inequality for the singular values of rectangular matrix. 
\begin{lemma}\cite[p.46]{tao2012topics}\label{lem_weyl} For any $p \times n$ matrices $A$ and $B,$ denote $\sigma_i(A)$ as the $i$-th largest singular value of $A$. Then we have 
\begin{equation*}
\sigma_{i+j-1}(A+B) \leq \sigma_i(A)+\sigma_j(B), \ 1 \leq i,j, i+j-1 \leq p \wedge n. 
\end{equation*}
\end{lemma}

\begin{lemma}\cite[Lemma F.5]{benaych2016lectures}\label{lem_Stiel_perturb} Let $z = E + i\eta \in \mathbb{C}^+$. For  $n \times n$ Hermitian matrices $A$ and $\wt A$, denote $s(z)$ and $\wt s(z)$ to be the Stieltjes transforms of their ESD's. Then we have 
\begin{equation}
    |s(z) - \wt s(z)| \leq \frac{\texttt{rank}(A - \wt A)}{n}  \left( \frac{2}{\eta} \wedge \frac{\|A-\wt A\|}{\eta^2}\right)\,.
\end{equation}
\end{lemma}

\begin{lemma}\label{hanson} \cite[Proposition 6.2]{benaych2011fluctuations}
{Take a centered probability measure $\nu$ on $\mathbb{R}$ with variance one, and assume it satisfies the log-Sobolev inequality.} Then there exists a constant $c>0$ so that for any matrix $A:=(a_{jk})_{1\le j,k\le n}$ with complex entries, any $\delta >0$, and any $g=(g_1,\ldots,g_n)^\top $ with i.i.d. entries with law $\nu$, we have
$$
\mathbb P\left( |\langle g, A g\rangle-\mathbb E\langle g, A g\rangle|>\delta\right)
\leq 4 e^{-c\left(\frac{\delta}{C}\wedge\frac{\delta^2}{C^2}\right)}\,,
$$
where $C=\sqrt{\tr (A{A^*})}$. 
\end{lemma}

\subsection{Some complex analysis tools}
We first provide the following lemma, which is from Lemma 2.6 of \cite{yang2019edge}.
\begin{lemma}\label{lambdar_sqrt}
Assume Assumption \ref{assum_main} holds. Then there exist constants $c_1, c_2>0$ such that for $x\downarrow 0$, where $x$ is real, so that
\begin{align}\label{sqroot3}
\rho_{1c}(\lambda_+ - x) = c_{1} x^{1/2} + \OO(x)\ \mbox{ and }\ 
\rho_{2c}(\lambda_+ - x) = c_{2} x^{1/2} + \OO(x)\,, 
\end{align}
and for $z\in \mathbb{C}^+$ {and $\textup{Im} \,z\geq 0$}, we have
\begin{align}\label{sqroot4}
\quad m_{1c}(z) &\,= m_{1c}(\lambda_+) + \pi c_{1}(z-\lambda_+)^{1/2} + \OO(|z-\lambda_+|)\,, \nonumber \\
\quad m_{2c}(z) &\,= m_{2c}(\lambda_+) + \pi c_{2}(z-\lambda_+)^{1/2} + \OO(|z-\lambda_+|)\,,
\end{align}
where $(z-\lambda_+)^{1/2}$ is taken with the positive imaginary part as the branch cut (the same convention holds in the following). 
These estimates also hold for $\wp_{1c}$,  $\wp_{2c}$, $\mathsf M_{1c}$ and $\mathsf M_{2c}$ with different constants.
Moreover, there exists a constant $c_{\mathcal{T}}>0$ such that for $z\to \lambda_+$ {and $\textup{Im} \,z\geq 0$}, we have
\begin{equation}\label{sqroot5}
{    \mathcal{T}(\lambda_+) - \mathcal{T}(z) = \pi c_\mathcal{T}(z-\lambda_+)^{1/2} + \OO(|z-\lambda_+|)\,.}
\end{equation}
\end{lemma}
\begin{proof}
Equations \eqref{sqroot3} and \eqref{sqroot4} are directly from Lemma 2.6 of \cite{yang2019edge}. By \eqref{sqroot4} and the definition of $\mathcal{T}$ in \eqref{eq_defnt}, we obtain \eqref{sqroot5}.
\end{proof}

For some properly chosen constants $\varsigma_1>0$ and $\varsigma_2>1$, we denote a domain of spectral parameters $z$ as
\begin{equation}\label{eq_S}
S(\varsigma_1,\varsigma_2):= \{z = E+i\eta: \lambda_+ - \varsigma_1 \leq E \leq \varsigma_2\lambda_+,\, 0{\leq}\eta \leq 1\}\,.    
\end{equation} Moreover, for any $z = E + i\eta$, denote
\begin{equation}
    \kappa_z := |E-\lambda_+|\,.
\end{equation}
The following lemma is similar to Lemma 3.4 in \cite{yang2019edge} and Lemma S.3.5 in \cite{DY2019}.
\begin{lemma}\label{s35_DY2019}
Suppose Assumption \ref{assum_main} holds. Then we have
\begin{enumerate}
\item[(i)] for $z\in S(\varsigma_1, \varsigma_2)$, 
\begin{align}
    &|\mathcal{T}(z)| \asymp |m_{1c}(z)|\asymp |m_{2c}(z)| \asymp 1\,,\label{eq_estimm}\\
    &\textup{Im }\mathcal{T}(z)\asymp \textup{Im }m_{1c}(z) \asymp \textup{Im }m_{2c}(z) \asymp 
    \begin{cases}
   \displaystyle\frac{\eta}{\sqrt{\kappa_z+\eta}}  & \textup{if } E\geq \lambda_+\\
   \sqrt{\kappa_z+\eta} & \textup{if } E\leq \lambda_+\,,
  \end{cases}\nonumber
\end{align}
and
\begin{align}
|\textup{Re }\mathcal{T}(z) - \mathcal{T}(\lambda_+) | &\,\asymp |\textup{Re }m_{1c}(z)-m_{1c}(\lambda_+)| \asymp |\textup{Re }m_{2c}(z)-m_{2c}(\lambda_+)|\nonumber\\
      &\,\asymp\begin{cases}
       \sqrt{\kappa_z+\eta} & \textup{if } E\geq \lambda_+ \\
       \displaystyle\frac{\eta}{\sqrt{\kappa_z+\eta}} + \kappa_z & \textup{if } E\leq \lambda_+\,.
    \end{cases}\label{eq_realestimate}
\end{align}
The estimates \eqref{eq_estimm}-\eqref{eq_realestimate} hold for $\mathsf M_{1c}$ and $\mathsf M_{2c}$ {with different constants}.

\item [(ii)] there exists constant $\tau'>0$ such that for any $z \in S(\varsigma_1, \varsigma_2)$,
\begin{equation}\label{Piii}
\min_{j=1,\ldots,n} \vert 1 + \mathsf M_{1c}(z)\sigma_j^b \vert \ge \tau'\ \ \mbox{and}\ \ \min_{i=1,\ldots,p} \vert 1 + \mathsf M_{2c}(z)\sigma_i^a  \vert \ge \tau'\,,
\end{equation}
where $\sigma_i^a$ and $\sigma_j^b$ are eigenvalues of $A$ and $B$ stated in Assumption \ref{assum_main}.
\end{enumerate}
\end{lemma}
\begin{proof}
The estimates for $m_{1c}$ and $m_{2c}$ have been proved in Lemma S.3.5 in \cite{DY2019}. Estimates for $\mathcal{T}$ in \eqref{eq_estimm} and \eqref{eq_realestimate} can be directly derived from \eqref{sqroot5}, and we omit details. 
\end{proof}

The following lemma describes the behavior of $\theta$ and the derivatives of $m_{1c}$, $m_{2c}$ and $\mathcal{T}$ on the real line.

\begin{lemma}\label{s36_DY2019}
Suppose Assumption \ref{assum_main} holds. For {$d\downarrow\alpha$, where $d$ is real,} we have
\begin{equation} \label{eq_s36_1}
     |\theta(d) - \lambda_+| \asymp (d - \alpha)^2.
\end{equation}
For {$x \downarrow \lambda_+$, where $x$ is real,} we have
\begin{equation}\label{eq_s36_2}
    m'_{1c}(x) \asymp \kappa_x^{-1/2}, \quad  m'_{2c}(x) \asymp \kappa_x^{-1/2} \ \mbox{ and }\  \mathcal{T}'(x) \asymp \kappa_x^{-1/2}\,.
\end{equation}
\end{lemma}

\begin{proof}
{ By the definition of $\theta(\cdot)$ in \eqref{eq_functions} and $\alpha$ in \eqref{eq_defnalpha}, $d \downarrow \alpha$ implies $\theta(d) \downarrow \lambda_+$. 
Let $\varsigma>0$ be a small constant.
When $0 <\theta(d)-\lambda_+< \varsigma$, 
by
(\ref{sqroot5}) we obtain that 
$$\alpha^{-2} -d^{-2}=\mathcal{T}(\lambda_+) -  \mathcal{T}(\theta(d)) = C\sqrt{\theta(d)-\lambda_+} + \OO(|\theta(d)-\lambda_+|). $$
The above estimation implies $\alpha^{-2}- d^{-2} \asymp \sqrt{\theta(d)-\lambda_+}$ when $\varsigma$ is sufficiently small. Since $\alpha^{-2}-d^{-2}=(d-\alpha)\frac{d+\alpha}{d^2\alpha^2} \asymp d-\alpha$, we obtain \eqref{eq_s36_1}.}   
The first two statements of \eqref{eq_s36_2} have been proved in Lemma S.3.6 in \cite{DY2019}, and the third statement then can be derived by  
\begin{align}
\mathcal{T}'(x) &\, = m_{1c}(x)m_{2c}(x) + x m'_{1c}(x)m_{2c}(x) + x m_{1c}(x)m'_{2c}(x) \nonumber\\
&\, \asymp C_1 + C_2\kappa_x^{-1/2} \asymp \kappa_x^{-1/2}, 
\end{align}
where we used the first statement of \eqref{eq_s36_2} and \eqref{eq_estimm} in the first $\asymp$ and the fact that $\kappa_x^{-1/2}$ blows up when $x\to \lambda_+$ in the last approximation. 
\end{proof}

Note that by \eqref{eq_s36_1} and \eqref{eq_s36_2}, for {$d\downarrow \alpha$, where $d$ is real,} we have
\begin{equation}\label{eq_s36_3}
    m_{1c}'(\theta(d)) \asymp \frac{1}{d-\alpha}, \quad 
     m_{2c}'(\theta(d)) \asymp \frac{1}{d-\alpha}, \quad  
    \mathcal{T}'(\theta(d)) \asymp \frac{1}{d-\alpha}
\end{equation}
and
\begin{equation}\label{eq_gderivative}
    \theta'(d)\asymp d-\alpha\,.
\end{equation}
{In the next lemma, more quantifications are provided in a more general spectral domain.}

\begin{lemma}\cite[Lemma S.3.7]{DY2019}\label{s37_DY2019} Suppose Assumption \ref{assum_main} holds. Then for any constant $\varsigma>\lambda_+$, there exist constants $\tau_1, \tau_2>0$ such that the following statements hold.
\begin{enumerate}
\item [(i)] $\mathcal{T}$ is a holomorphic homeomorphism on the spectral domain
\begin{equation}
    \mathbf{D}_1(\tau_1,\varsigma):= \{z = E+i\eta: \lambda_+<E<\varsigma,\, -\tau_1<\eta<\tau_1\}.
\end{equation}
As a consequence, the inverse of $\mathcal{T}$ exists and we denote it by $\theta$.

\item [(ii)] $\theta$ is holomorphic homeomorphism on $\mathbf{D}_2(\tau_2, \varsigma)$, where
\begin{equation}\label{eq_d}
    \mathbf{D}_2(\tau_2,\varsigma):= \{\zeta = E+i\eta: \alpha<E<1/\sqrt{\mathcal{T}(\varsigma)},\, -\tau_2<\eta<\tau_2\},
\end{equation}
such that $\mathbf{D}_2(\tau_2, \varsigma) \subset 1/\sqrt{\mathcal{T}(\mathbf{D}_1(\tau_1,\varsigma))}$.

\item [(iii)] For $z \in \mathbf{D}_1(\tau_1,\varsigma)$, we have \begin{equation}\label{eq_iii1}
    |\mathcal{T}(z)-\mathcal{T}(\lambda_+)| \asymp|z-\lambda_+|^{1/2}
\ \mbox{ and }\ 
    |\mathcal{T}'(z)|\asymp|z-\lambda_+|^{-1/2}.
\end{equation}

\item [(iv)] For $\zeta \in \mathbf{D}_2(\tau_2,\varsigma)$, we have \begin{equation}\label{eq_gcomplex}
    |\theta(\zeta)-\lambda_+| \asymp|\zeta-\alpha|^2
\ \mbox{ and }\ 
    |\theta'(\zeta)|\asymp|\zeta-\alpha|.
\end{equation}

\item [(v)] For $z_1,z_2 \in \mathbf{D}_1(\tau_1,\varsigma)$ and $w_1, w_2 \in \mathbf{D}_2(\tau_2,\varsigma)$, we have \begin{equation}\label{Tz1-Tz2 bound}
    |\mathcal{T}(z_1)-\mathcal{T}(z_2)| \asymp \frac{|z_1-z_2|}{\max_{i=1,2}|z_i-\lambda_+|^{1/2}},
\end{equation}
and 
\begin{equation}\label{eq_thetadiff}
    |\theta(w_1)-\theta(w_2)|\asymp |w_1-w_2|\cdot \max_{i=1,2}|w_i-\alpha|.
\end{equation}
\end{enumerate}
\end{lemma}

Note that in \eqref{eq_functions}, $\theta$ is only defined on the real line, but in Lemma \ref{s37_DY2019}(i), we extend its definition to the complex plane. The relationship between $\theta$ and $\mathcal T$ is summarized in Figure \ref{IMG_9221}.

\begin{proof}
Similar results for $m_{1c}$ and $m_{2c}$ and their proofs can be found in  \cite[Lemma S.3.7]{DY2019}. We derive this Lemma specifically using the definition of $\mathcal{T}(z)$ and $\theta(\zeta)$ with the same approach, and we omit the detail here.
\end{proof}

\begin{figure}[t]
\includegraphics[trim=0 50 90 50, width=\textwidth]{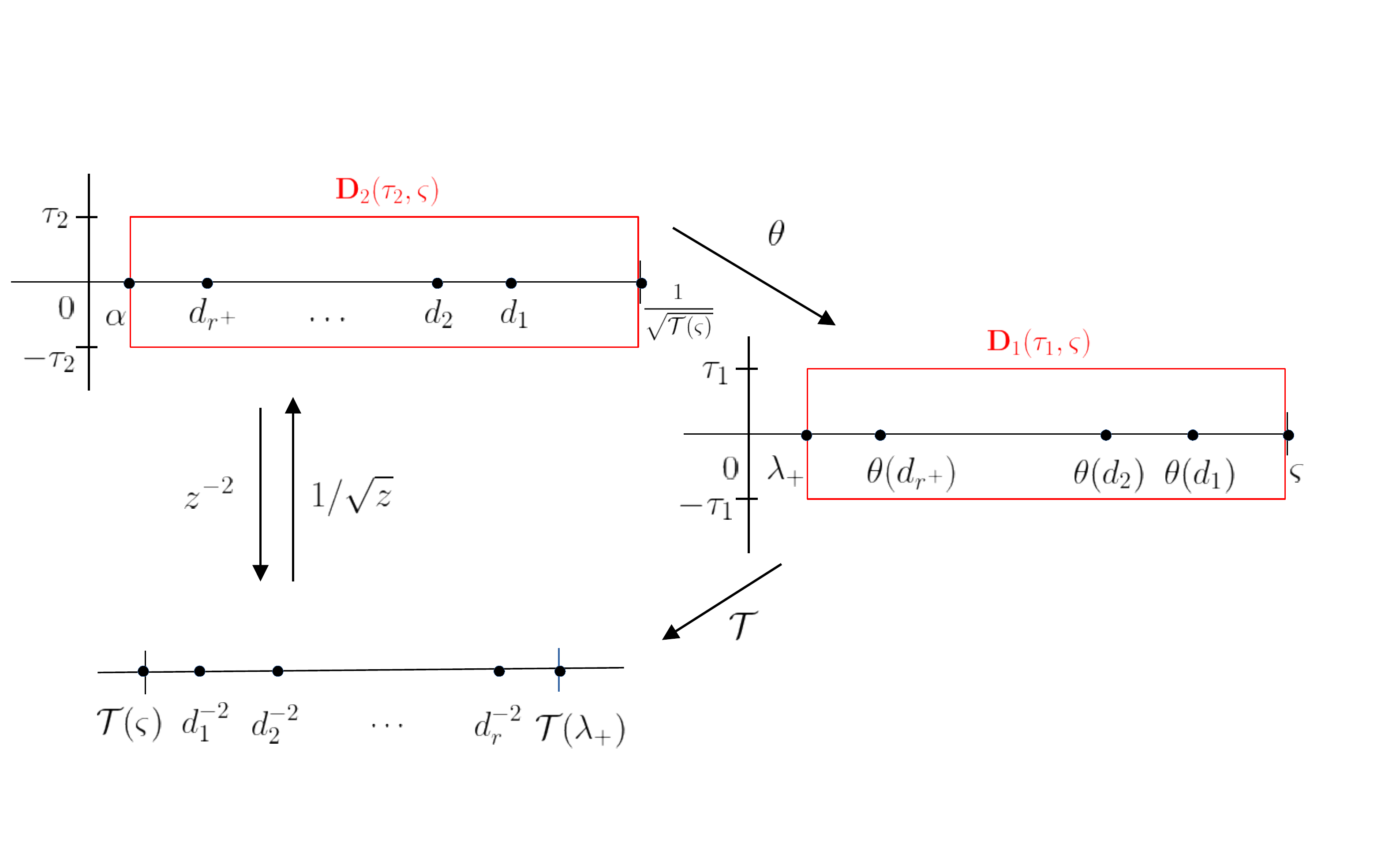}
\caption{\small An illustration of the relationship of $\mathcal{T}$, $\theta$ and several quantities used in the proof. Note that $\theta(\alpha)=\lambda_+$. \label{IMG_9221}}
\end{figure}

\subsection{Some existing random matrix tools and results}

We record some results from the random matrix theorem that we will repeatedly apply in the following proofs. 
Denote \cite[(3.18)]{yang2019edge}
\begin{equation}\label{eq_defPsi}
\Psi(z):= \sqrt{\frac{\text{Im } \mathsf M_{2c}(z)}{n\eta}} +\frac{1}{n\eta}.
\end{equation}
Note that for any $z=E+i\eta\in S(\varsigma_1,\varsigma_2)$, from (i) of Lemma \ref{s35_DY2019} and \eqref{Piii}, we have \cite[(3.19)]{yang2019edge}
\begin{equation}
\Psi(z) \gtrsim n^{-1/2}, \quad \Psi^2(z) \lesssim \frac{1}{n\eta}\,.
\end{equation} 
{Clearly, for any fixed $E$}, $\Psi^2(E+i\eta)+\frac{\phi_n}{n\eta}$ is monotonically decreasing with respect to $\eta$, so for $E\leq \lambda_+$, find the unique $\eta_l(E)$ such that 
\begin{equation}\label{definition eta_l(E)}
n^{1/2} \left[\Psi^2(E+i\eta_l(E))+ \frac{\phi_n}{n\eta_l(E)}\right] =1\,.
\end{equation}
{and set $\eta_l(E):=\eta_l(\lambda_+)$ for $E>\lambda_+$. The following approximation holds \cite[(S.54)]{DY2019}:} 
\begin{equation}\label{etalE}
 \eta_l(E)\asymp n^{-3/4} + n^{-1/2} \left(\sqrt{\kappa_E}+\phi_n\right) 
\end{equation}
{for $E\le \lambda_+$ and 
\[
\eta_l(E)=\OO(n^{-3/4} +  n^{-1/2}\phi_n)
\] 
for $E>\lambda_+$. }
Then we have the following rigidity result. Recall the definition of classical location $\gamma_j$ in \eqref{eq_classical} and the eigenvalues of $ZZ^\top $ are $\lambda_1 \geq \lambda_2 \geq \cdots \geq \lambda_{p \wedge n}$.  

\begin{lemma}\label{lem_rigidty}(Rigidity of eigenvalues {\cite[Lemma S.3.11]{DY2019}}).
Suppose Assumption \ref{assum_main} holds. Then, 
{for a fixed small constant $\varsigma>0$},  
for any $j$ such that $\lambda_+-\varsigma \leq \gamma_j \leq \lambda_+$, we have 
\begin{equation}\label{rigidity}
\vert \lambda_j - \gamma_j \vert \prec n^{-2/3}\left( j^{-1/3}+\mathbbm 1(j \le n^{1/4} \phi_n^{3/2})\right) + \eta_{l}(\gamma_j) + n^{2/3} j^{-2/3} \eta_l^2(\gamma_j)\,,
\end{equation}
where $\eta_{l}(\gamma_j)\asymp  n^{-3/4} + n^{-5/6}j^{1/3}+ n^{-1/2}\phi_n$.
Moreover, if {either (a) 
$\mathbb E x_{ij}^3 = 0$ for all $1\le i \le p$ and $1\le j \le n$,
or} (b) either $A$ or $B$ is diagonal, the bound \eqref{rigidity} is improved to 
\begin{equation}\label{rigidity2}
|\lambda_j-\gamma_j| \prec n^{-2/3}j^{-1/3}. 
\end{equation}
\end{lemma}  

Also, we have the delocalization result for eigenvectors.
\begin{lemma}(Isotropic delocalization of eigenvectors \cite[Lemma S.3.13]{DY2019}) \label{delocal_rigidity}
Suppose Assumption \ref{assum_main} hold. Then, {for a fixed small constant $\varsigma>0$},
for any deterministic unit vectors $\mathbf u \in \mathbb C^{p}$ and $\mathbf v  \in \mathbb C^{ n}$, we have 
\begin{equation}\label{delocal_claim}
 \left|\langle \mathbf u,\bxi_k\rangle \right|^2+\left|\langle \mathbf v , \bzeta_k\rangle \right|^2\prec  n^{-1} + \eta_l(\gamma_k) \left( \frac{k}{n}\right)^{1/3} +\eta_l(\gamma_k) \phi_n
\end{equation}
for all $k$ such that $\lambda_+ - \varsigma \le \gamma_k \le \lambda_+$, where $\eta_l(\gamma_k) {\asymp} n^{-3/4} + n^{-5/6}k^{1/3}+\phi_n n^{-1/2}$. 
{If either (a) 
$\mathbb E x_{ij}^3 = 0$ for all $1\le i \le p$ and $1\le j \le n$,
or (b) either $A$ or $B$ is diagonal, the bound is improved to}
\begin{equation}
\left|\langle \mathbf u,{\bm \xi}_k\rangle \right|^2+\left|\langle \mathbf v ,{\bm \zeta}_k\rangle \right|^2\prec n^{-1}\label{delocals}
\end{equation}
for all $k$ such that $\lambda_+ - \varsigma \le \gamma_k \le \lambda_+$.
\end{lemma}

Next, we discuss some properties of resolvents and local laws.
For any $z \in \mathbb{C}^+,$  denote
\begin{equation*}
H:=H(z)
:=
\begin{pmatrix}
0 & { z^{1/2}} Z \\
z^{1/2}Z^\top  & 0
\end{pmatrix}\in \mathbb{C}^{(p+n)\times (p+n)}
\end{equation*}
and
\begin{equation}
\widetilde{H}:=\widetilde{H}(z):=
\begin{pmatrix}
0 & z^{1/2} \widetilde{S} \\
z^{1/2} \widetilde{S}^\top  & 0
\end{pmatrix}\in \mathbb{C}^{(p+n)\times (p+n)},
\end{equation}
where $z^{1/2}$ is defined with the positive imaginary part as the branch cut.
Note that we have
\begin{equation}
\widetilde{H}(z)=
\begin{pmatrix}
0 & z^{1/2} Z \\
z^{1/2}Z^\top  & 0
\end{pmatrix}
+
\begin{pmatrix}
0 & z^{1/2}S\\
z^{1/2}S^\top  & 0
\end{pmatrix}
=H(z)+\Ub \Db \Ub^\top ,
\end{equation}
where 
\begin{equation}\label{definition D and U}
\Db=
\begin{pmatrix}
0 & {z^{1/2}} D \\
z^{1/2} D & 0
\end{pmatrix}\in {\mathbb{C}^{2r\times 2r}}
\ \mbox{ and }\ 
\Ub=\begin{pmatrix}
U & 0 \\
0 &  V
\end{pmatrix}\in \mathbb{R}^{(p+n)\times 2r}
\end{equation}
such that $D=\textup{diag} \{d_1, \cdots, d_r \}\in \mathbb{R}^{r\times r}$ is a diagonal matrix formed by singular values, and $U=(\ub_1, \cdots, \ub_r)\in \mathbb{R}^{p\times r}$ and $V=(\vb_1, \cdots, \vb_r)\in \mathbb{R}^{n\times r}$ are matrices formed by left and right singular vectors defined in \eqref{eq_model}. 
The eigenvalues of $\widetilde{H}$ coincide with the singular values of $z^{1/2} \widetilde{S}.$ Therefore, it suffices to study the matrix $\widetilde{H}$. Further, we denote the Green functions of $H$ and $\widetilde{H}$ respectively as 
\begin{equation*}
G:=G(z):=(H-zI_{p+n})^{-1}\ \mbox{ and }\  \widetilde{G}:=\widetilde{G}(z):=(\widetilde{H}-zI_{p+n})^{-1}\,,
\end{equation*}
where $z\in \mathbb{C}^+$.
By the Schur's complement, we have 
\begin{align} \label{green2}
G(z) &\,= \left( {\begin{array}{*{20}c}
   { \mathcal{G}_1(z)} &  z^{-1/2}\mathcal{G}_1(z)Z \\
   z^{-1/2}Z^{\T} \mathcal{G}_1(z) & \mathcal{G}_2(z) \\
\end{array}} \right)\nonumber \\
 \wtG(z) &\,= \left( {\begin{array}{*{20}c}
   { \wt{\mathcal{G}}_1(z)} &  z^{-1/2}\wt{\mathcal{G}}_1(z)Z \\
   z^{-1/2}Z^{\T} \wt{\mathcal{G}}_1(z) & \wt{\mathcal{G}}_2(z) \\
\end{array}} \right),
\end{align}
where $\mathcal{G}_1$ and $\mathcal{G}_2$ are defined in \eqref{eq_g1g2}. Further, by the spectral decomposition, we have 
\begin{align}\label{main_representation}
& {G}(z) = \sum_{k=1}^{p \wedge n} \frac{1}{{\lambda}_k-z} \left( {\begin{array}{*{20}c}
   { {\bxi}_k {\bxi}^\top _k} &  z^{-1/2} \sqrt{{\lambda}}_k {\bxi}_k {\bzeta}_k^\top  \\
   z^{-1/2} \sqrt{{\lambda}}_k{\bzeta}_k {\bxi}_k^\top  & {\bzeta}_{k} {\bzeta}_{k}^\top   \\
\end{array}} \right)\in \mathbb{C}^{(p+n)\times (p+n)} \nonumber\\
& \wt{G}(z) = \sum_{k=1}^{p \wedge n} \frac{1}{{\wt\lambda}_k-z} \left( {\begin{array}{*{20}c}
   { {\wt \bxi}_k {\wt \bxi}^\top _k} &  z^{-1/2} \sqrt{{\wt\lambda}}_k {\wt \bxi}_k {\wt \bzeta}_k^\top  \\
   z^{-1/2} \sqrt{{\wt\lambda}}_k{\wt \bzeta}_k {\wt \bxi}_k^\top  & {\wt \bzeta}_{k} {\wt \bzeta}_{k}^\top  \\ 
\end{array}} \right)\in \mathbb{C}^{(p+n)\times (p+n)}\,.
\end{align}
The following lemma characterizes the locations of the outlier eigenvalues of $\widetilde{H}$. 

\begin{lemma}\label{lem_mastereq}\cite[Lemma 6.1]{knowles2013isotropic} 
Assume $\mu \in \mathbb{R} - \text{Spec}(H)$. Then $\mu \in \text{Spec}(\widetilde{H})$ if and only if
\begin{equation} \label{lem_perbubationequation}
\det(\mathbf{U}^\top  G(\mu)\mathbf{U}+\mathbf{D}^{-1})=0\,.
\end{equation}
\end{lemma}

Next lemma provides a link between the Green functions $G(z)$ and $\widetilde{G}(z).$ The proof is straightforward depending on the Woodbury formula and the basic identity $A-A(A+B)^{-1}A=B-B(A+B)^{-1}B$ when $A$, $B$ and $A+B$ are all invertible. 

\begin{lemma}\label{lem_gtitle}{\cite[Lemma 4.8]{ding2020high}}
For $z \in \mathbb{C}^+,$ we have
\begin{equation}\label{defn_greenrep1}
\widetilde{G}(z)=G(z)-G(z)\mathbf{U}(\mathbf{D}^{-1}+\mathbf{U}^\top  G(z)\mathbf{U})^{-1} \mathbf{U}^\top  G(z)\,.
\end{equation}
\end{lemma}
This lemma immediately leads to the following relationship.
\begin{equation} \label{defn_greenrep2}
\mathbf{U}^\top  \widetilde{G}(z) \mathbf{U}=\mathbf{D}^{-1}-\mathbf{D}^{-1}(\mathbf{D}^{-1}+\mathbf{U}^\top  G(z)\mathbf{U})^{-1}\mathbf{D}^{-1}.
\end{equation}
Next, for any $z \in \mathbb{C}^+$, denote   
\begin{equation}\label{eq_Pi}
\Pi (z):=\left( {\begin{array}{*{20}c}
   { \Pi_1}(z) & 0  \\
   0 & { \Pi_2}(z)  \\
\end{array}} \right)\in \mathbb{C}^{(p+n)\times (p+n)}, 
\end{equation}
{where
$$ 
\Pi_1(z):  =-z^{-1}\left(1+\mathsf{M}_{2c}(z)A \right)^{-1}\ \mbox{ and }\  \Pi_2(z):=- z^{-1} (1+\mathsf{M}_{1c}(z)B )^{-1}\,.
$$
}
Denote 
\begin{equation}\label{eq_Pibar}
\overline\Pi(z):=\begin{pmatrix}
m_{1c}(z)I_r & 0 \\
0 & m_{2c}(z) I_r
\end{pmatrix}\in \mathbb{C}^{2r\times 2r}\,.
\end{equation}
Note that {by \eqref{eq_m1cm2c},} $\Pi(z)$ and $\overline\Pi(z)$ are related via
\[
m_{1c}(z) = \frac{1}{p}\tr{\Pi_1}(z),\ \ m_{2c}(z) = \frac{1}{n}\tr{\Pi_2}(z)\,, 
\]
and the relationship between $\Pi(z)$ and $\overline\Pi(z)$ and how $\Ub^\top  \Pi(z)\Ub$ converges to $\overline{\Pi}(z)$ {will be stated below}.
Denote the difference 
\begin{equation}\label{definition omega z}
\Omega(z):= \Ub^\top  G(z)\Ub-\overline{\Pi}(z)\,,
\end{equation}
which will be used to control the noise part in the analysis.
With $\Omega(z)$, we have the following identity  
\begin{align}\label{eq_resid}
    &(\Db^{-1}+\Ub^\top G(z)\Ub)^{-1} \\
    =&\, (\Db^{-1}+\overline\Pi(z))^{-1} + (\Db^{-1}+\overline\Pi(z))^{-1}\Omega(z)(\Db^{-1}+\Ub^\top G(z)\Ub)^{-1}.\nonumber
\end{align}
Iteration leads to the following $\ell$-th resolvent expansion. 

\begin{lemma}(Resolvant identity \cite{benaych2016lectures}) For $\ell\geq 1$, we have the following $\ell$-th order resolvent expansion: 
\begin{align}
    (\Db^{-1}+\Ub^\top G(z)\Ub)^{-1} = & \sum_{k=0}^{\ell-1}
    (\Db^{-1}+\overline\Pi(z))^{-1}[\Omega(z)(\Db^{-1}+\overline\Pi(z))]^k \label{eq_resexp} \\ & +  [(\Db^{-1}+\overline\Pi(z))^{-1}\Omega(z)]^{\ell}[\Db^{-1}+\Ub^\top G(z)\Ub]^{-1}.\nonumber
\end{align}
\end{lemma}
We define the following spectral regions.
Recall the definition of $S(\varsigma_1,\varsigma_2)$ in \eqref{eq_S}, {where $\varsigma_1>0$ and $\varsigma_2>1$, and $\Psi$ defined in \eqref{eq_defPsi}. Fix any constant $\omega>0$.} Denote the following spectral domains of parameter $z$ as
\begin{align}
S_0(\varsigma_1,\varsigma_2,\omega):=&\, S(\varsigma_1,\varsigma_2) \cap \{E+i\eta: \eta\geq n^{-1+\omega}\}\,,\label{tildeS}\\
\wt S_0(\varsigma_1,\varsigma_2,\omega):=&\, S_0(\varsigma_1,\varsigma_2,\omega) \cap \left\{E+ i \eta:  n^{1/2}\left( \Psi^2(z)+\frac{\phi_n}{n\eta}\right) \le n^{-\omega/2}\right\}\nonumber
\end{align}
and 
\begin{equation}\label{eq_sout}
    S_{out}(\varsigma_2,\omega):=\{E+i\eta:\lambda_+ + n^{\omega}(n^{-1/3}\phi_n^2+n^{-2/3}) \leq E \leq \varsigma_2 \lambda_+, \eta \in [0,1]\}.
\end{equation}
Inside these domains, {the following local laws have been established in the literature.}

\begin{theorem}\label{LEM_SMALL}(Local laws ``near'' $\lambda_+$ \cite[Theorem S.3.9]{DY2019}). Suppose that Assumption \ref{assum_main} holds. Fix constant $\varsigma_1>0$ and $\varsigma_2>1$ as those in Lemma \ref{s35_DY2019}. Then for any fixed $\omega>0$, the following estimates hold.
\begin{enumerate}
\item [(1)] {Anisotropic local law}: For any $z \in \wt S_0(\varsigma_1,\varsigma_2,\omega)$ and deterministic unit vectors $\bm{\mu}, \bm{\nu} \in \mathbb{C}^{p+n}$, we have 
\begin{equation}\label{eq_locallaw}
    \left| \langle \bm{\mu}, (G(z)-\Pi(z)) \bm{\nu} \rangle \right| \prec \phi_n + \Psi(z).
\end{equation}

\item [(2)] {Averaged local law}: For any $z \in \wt S_0(\varsigma_1,\varsigma_2,\omega)$, we have
\begin{equation}\label{eq_avlocalaw_1}
  |m_1(z)-m_{1c}(z)| +  |m_2(z)-m_{2c}(z)| \prec \frac{1}{n\eta}\,.
\end{equation}
Moreover, when $z \in \wt S_0(\varsigma_1,\varsigma_2,\omega)\cap \{z = E+i\eta: E \geq \lambda_+, n\eta\sqrt{\kappa_z + \eta}\geq {n^\omega}\}$, we have the following stronger bound
\begin{equation}\label{eq_avlocalaw_2}
 |m_1(z)-m_{1c}(z)| +  |m_2(z)-m_{2c}(z)| \prec \frac{n^{-\omega/4}}{n\eta}+\frac{1}{n(\kappa_z+\eta)}+\frac{1}{(n\eta)^2\sqrt{\kappa_z+\eta}}\,.
\end{equation}
\end{enumerate}
{If either (a) 
$\mathbb E x_{ij}^3 = 0$ for all $1\le i \le p$ and $1\le j \le n$,
or (b) either $A$ or $B$ is diagonal, the bounds \eqref{eq_locallaw}, \eqref{eq_avlocalaw_1} and \eqref{eq_avlocalaw_2} hold for $z\in S_0(\varsigma_1,\varsigma_2,\omega)$.}
\end{theorem}

The term ``anisotropic'' means that the resolvent $G(z)$ is well approximated by a deterministic matrix that is not a multiple of an identity matrix.

\begin{theorem}\label{lem_locallaw} (Anisotropic local law beyond $\lambda_+$ \cite[Theorem S.3.12]{DY2019}). Suppose Assumption \ref{assum_main} holds. For any $0<\omega<2/3$, $z\in S_{out}(\varsigma_2,\omega)$ and deterministic unit vectors $\vb, \wb \in \mathbb{C}^{p+n},$ we have
\begin{equation}
\left| \langle \vb, (G(z)-\Pi(z)) \wb \rangle \right| \prec \phi_n + \sqrt{\frac{\textup{Im} m_{2c}(z)}{n\eta}} \asymp \phi_n + n^{-1/2}(\kappa_z+\eta)^{-1/4}.
\end{equation}
\end{theorem}

Finally, we derive the following lemma to describe the difference between $\Pi(z)$ and $\overline\Pi(z)$ {when $z$ is away from} $\lambda_+$.

\begin{lemma}\label{lem_pipi} 
Under Assumption \ref{assum_main}, for any $z \in \wt S_0(\varsigma_1,\varsigma_2,\omega)\cup S_{out}(\varsigma_2,\omega)$ and $i,j= 1,\ldots,r$, we have 
\begin{align*}
|\langle \ub_i,\,{-}z^{-1}(1+& \mathsf{M}_{2c}(z)A)^{-1}\ub_j\rangle - \delta_{ij} m_{1c}(z) | \\
&\vee |\langle \vb_i,\,{-}z^{-1}(1+ \mathsf{M}_{1c}(z)B)^{-1}\vb_j\rangle - \delta_{ij} m_{2c}(z)|  \prec n^{-1}\,;
\end{align*}
that is,
\begin{equation}
\norm{\Ub^\top  \Pi(z) \Ub - \overline \Pi(z)} \prec n^{-1}\,.
\end{equation}
\end{lemma}

\begin{proof}
When $z \in \wt S_0(\varsigma_1,\varsigma_2,\omega)\cup {S}_{out}(\varsigma_2,\omega)$, by Lemma \ref{hanson}, for any $\delta>0$ we have that
\begin{align*}
&\mathbb P\left( |\langle \ub_i,\, {-}z^{-1}(1+\mathsf{M}_{2c}(z)A)^{-1} \ub_j\rangle-\mathbb \delta_{ij} m_{1c}(z) 
 |>p^{-1}\delta\right)\\
=\,& \mathbb P\left( \Big|\langle \sqrt{p}\ub_i,\, {-}z^{-1}(1+\mathsf{M}_{2c}(z)A)^{-1} \sqrt{p}\ub_j\rangle-\mathbb \delta_{ij}{pm_{1c}(z)}
\Big|>\delta\right)\\
\leq \,&4 e^{-c\left(\frac{\delta}{C}\wedge\frac{\delta^2}{C^2}\right)}
\end{align*}
for some constants $c>0$ {and $C=\sqrt{\tr(|z|^{-2}(1+\mathsf{M}_{2c}(z)A)^{-1})(1+\mathsf{M}_{2c}(z)^*A)^{-1})}<\infty$ since $z \in \wt S_0(\varsigma_1,\varsigma_2,\omega)\cup {S}_{out}(\varsigma_2,\omega)$}, where $p$ is used to normalize $\ub_i$ and $\ub_j$ to fulfill the condition in Lemma \ref{hanson}.
Thus, given any small constant $\epsilon>0$ and large constant $D>0$,
by letting $\delta= n^{\omega}$, we have 
$$\mathbb P\left( |\langle \ub_i, {-}z^{-1}(1+ \mathsf{M}_{2c}(z)A)^{-1} \ub_j\rangle-\delta_{ij} m_{1c}(z)
 |>n^{-1+\epsilon}\right)
\leq  4n^{-\frac{cn^{\omega}}{C}}\leq n^{-D}$$
when $n$ is sufficiently large.
By a similar approach we have the same control for the other term, and hence the proof.
\end{proof}

\section{Proof of Theorem \ref{thm_value}} 

{Fix a sufficiently small constant $\epsilon>0$.} With Lemmas \ref{lem_rigidty}, \ref{LEM_SMALL}, \ref{lem_locallaw} and \ref{lem_pipi} and the triangle inequality, there exists an event $\Xi_\epsilon$ of high probability such that the followings hold when conditional on $\Xi_\epsilon$: 
\begin{enumerate}
\item [(i)] For $z \in \wt S_0(\varsigma_1,\varsigma_2,\epsilon)$, from Lemmas \ref{LEM_SMALL} and \ref{lem_pipi}, and since $r$ is fixed, we have that
\begin{equation}\label{aniso_lawev0}
    \norm{\mathbf{U}^\top  G(z)\mathbf{U}-\overline\Pi(z)} \leq  n^{\epsilon/2}(\phi_n + \Psi(z)). 
\end{equation}

\item [(ii)] For $z \in S_{out}(\varsigma_2,\epsilon),$ from Lemmas \ref{LEM_SMALL} and \ref{lem_pipi}, and since $r$ is {fixed}, we have that
\begin{align}\label{eq_cond0}
  \norm{\mathbf{U}^\top G(z)\mathbf{U}-\overline\Pi(z)}
  \leq   n^{\epsilon/2}(\phi_n + n^{-1/2}(\kappa_z+\eta)^{-1/4})\,.
\end{align} 

\item [(iii)] From Lemma \ref{lem_rigidty}, there exists a large integer $\varpi$, such that for $1 \leq i \leq \varpi,$ we have
\begin{equation}\label{eq_rigi0} 
|\lambda_i-\lambda_+| \leq n^\epsilon(n^{-1/3}\phi_n^2 + n^{-2/3})\,,
\end{equation}
where $\varpi>r$ is a fixed integer. 
\end{enumerate}

Below, the proof {of Theorem \ref{thm_value} is conditioned on $\Xi_\epsilon$} such that \eqref{aniso_lawev0}, \eqref{eq_cond0} and \eqref{eq_rigi0} hold.
\\
\\
\textbf{Step I. Find asymptotic outlier locations $\theta(d_i)$ and the threshold $\alpha$.}
We start with finding the \textit{asymptotic outlier locations} for eigenvalues of $\wt S \wt S^\top $. By combining \eqref{eq_cond0} and the continuity of determinant, $\mu>\lambda_+$ is asymptotically an outlier location if and only if 
\begin{equation}\label{eq_det0}
\lim_{n \to \infty}\det(\mathbf{U}^\top  G(\mu)\mathbf{U}+\mathbf{D}^{-1}) = \det \left(\overline{\Pi}(\mu)+\Db^{-1}\right) = 0 
\end{equation}
{due to Lemma \ref{lem_mastereq}}.
By Lemma \ref{lemma det diag formula}, 
\begin{align}\label{eq_pb}
&\det \left(\overline{\Pi}(\mu)+\Db^{-1}\right) 
= \det \left(\begin{pmatrix}
m_{1c}(\mu) I_r & 0 \nonumber \\
0 & m_{2c}(\mu) I_r
\end{pmatrix}+\Db^{-1}\right)\nonumber \\
 =\,& \mu^{-r} \prod_{i=1}^r ( \mu m_{1c}(\mu) m_{2c}(\mu)-d_i^{-2}) = \mu^{-r} \prod_{i=1}^r ( \mathcal{T}(\mu)-d_i^{-2}) = 0\,.
\end{align}
By the definition of $\theta(\cdot)$, the determinant is zero when $\mu = \theta(d_i)$ for $1\leq i\leq r$. 
Note that since $\mathcal{T}(x)$ is a monotonically decreasing function for $x>\lambda_+$, we have $\theta(d_i)$ as a solution of \eqref{eq_pb} if and only if 
$d_i^{-2} < \lambda_+ m_{1c}(\lambda_+)m_{2c}(\lambda_+){=\mathcal{T}(\lambda_+)}$, which is equivalent to $d_i >\alpha=1/\sqrt{\mathcal{T}(\lambda_+)}$ {defined in \eqref{eq_defnalpha}}.
In other words, $\alpha$ is the desired signal strength threshold.
\newline
\newline
\textbf{Step II. Define the permissible intervals for the spectrum.} Now the strategy is to prove that
with high probability there are no eigenvalues outside a neighbourhood of each $\theta(d_i)$. Define the index set
\begin{equation}\label{dfn_indexset}
    \mathbb{O}_{\epsilon}:= \{i:d_i -\alpha  \geq n^{\epsilon}(\phi_n+n^{-1/3})\} = \{1,2, \ldots, r_{\epsilon}\} \subset \mathbb{O}_+
\end{equation}
for some $r_\epsilon \leq r^+$. {We can understand $d_1,\ldots,d_{r_{\epsilon}}$ as ``strong'' outliers, and $d_{r_{\epsilon}},\ldots,d_{r^+}$ as ``weak'' outliers if $r_\epsilon<r^+$.}
For $1 \leq i \leq {r_\epsilon}$, define the interval
\begin{equation}\label{interval}
I_i:=[\theta(d_i)-n^{\epsilon}\omega(d_i),\, \theta(d_i)+n^{\epsilon}\omega(d_i)]\,,
\end{equation}
where 
\begin{equation}
\omega(d_i):= \phi_n\Delta(d_i)^2+n^{-1/2}\Delta(d_i)\,.
\end{equation}
Also, define 
\begin{align*}
I_0:=[0,\, \lambda_+ + n^{3\epsilon}(\phi_n^2+n^{-2/3})]\quad\mbox{and}\quad
   I:= I_0 \cup \bigcup_{i\in \mathbb{O}_\epsilon} I_i\,.
\end{align*}
{Note that we can choose a sufficiently small $\epsilon$ so that $|I_i|\to 0$ and $n^{3\epsilon}(\phi_n^2+n^{-2/3})\to 0$. Thus, when $n$ is sufficiently large, $I_i\subset S_{out}(\varsigma_2,\epsilon)$ for $i\in \mathbb{O}_\epsilon$.}
\newline
\newline
\textbf{Step III. Show that $I$ contains all eigenvalues of $\wt H$.} 
Note that for $\mu \in S_{out}(\varsigma_2,\epsilon)\cap \mathbb{R}$, by Lemma \ref{lem_mastereq} and \eqref{eq_cond0}, $\mu\in\textup{Spec}(\wt H)$ if and only if
\begin{align}
0=&\,\det(\Ub^\top G(\mu) \Ub+\Db^{-1})\label{eq_Gpi_diff}\\
=&\,\det(\overline\Pi(\mu) +\Db^{-1}) +O(n^{\epsilon/2}(\phi_n+n^{-1/2}\kappa_\mu^{-1/4}) )\nonumber \\
=&\,\mu^{-r} \prod_{i=1}^r ( \mu m_{1c}(\mu)  m_{2c}(\mu)-d_i^{-2})+O(n^{\epsilon/2}(\phi_n+n^{-1/2}\kappa_{\mu}^{-1/4})) \,,\nonumber
\end{align}
where the {second} bound comes from Lemma \ref{eq_lemmadeter} {with the fact that there are} bounded $2r$ eigenvalues of $\overline\Pi(\mu) +\Db^{-1}$. Therefore, to prove $\text{Spec}(\widetilde{H}) \subset I$, by \eqref{eq_Gpi_diff} and the fact that $\lambda_+ = \OO(1)$, it suffices to show that 
\begin{equation}\label{eq_toprove}
    \min_{1\leq i\leq r}| \mu m_{1c}(\mu)m_{2c}(\mu)-d_i^{-2}|\gg n^{\epsilon/2}(\phi_n+n^{-1/2}\kappa_{\mu}^{-1/4})
\end{equation}
when $\mu \notin I$.
Before proving \eqref{eq_toprove}, we claim that for any $\mu \notin I$ and $1\leq i \leq r$,
\begin{equation}\label{eq_diff}
|\mu-\theta(d_i)| > n^{\epsilon}\omega({d_i})
\end{equation}
{when $n$ is sufficiently large.}
\newline
\newline
\textbf{Step IV. Prove Claim \eqref{eq_diff}.}
To show that the claim is true, we consider two cases. \textbf{(i)} When $i\in \mathbb{O}_{\epsilon}$, (\ref{eq_diff}) is true by the definition of $I_i$. \textbf{(ii)} When $i \notin \mathbb{O}_{\epsilon}$, {by the definition of $\mathbb{O}_{\epsilon}$ in} \eqref{dfn_indexset} we have 
\begin{equation}\label{eq_deltabd}
\Delta(d_i)^2 = d_i - \alpha < n^{\epsilon}(\phi_n+n^{-1/3})\,. 
\end{equation}
Thus, by \eqref{eq_s36_1}, we have 
\begin{equation}
    \theta(d_i)- \lambda_+ \asymp (d_i-\alpha)^2< 2n^{2\epsilon}(\phi_n^2 +n^{-2/3 })\,,
\end{equation}
{and hence $\theta(d_i)\in I_0$. As a result,} when $\mu\notin I$, $|\mu-\theta(d_i)| > n^{3\epsilon}(\phi_n^2 + n^{-2/3})$ by the definition of $I_0$. Moreover, by the definition of $\omega(d_i)$ and \eqref{eq_deltabd}, we have
\begin{equation}
    \omega(d_i) < n^{2\epsilon}(\phi_n^{2}+n^{-2/3})
\end{equation}
for $i \notin \mathbb{O}_{\epsilon}$, which leads to $|\mu-\theta(d_i)| > n^{\epsilon}\omega(d_i)$.
We thus obtain the claim.
\newline
\newline
\textbf{Step V. Prove \eqref{eq_toprove}.}
Note that 
\begin{equation}\label{eq_g}
    | \mu m_{1c}(\mu)m_{2c}(\mu)-d_i^{-2}| = | \mathcal{T}(\mu)-d_i^{-2}| = | \mathcal{T}(\mu)-\mathcal{T}(\theta(d_i))|\,.
\end{equation}
We decompose the problem into the following two cases {when $\mu \notin I$}.
\medskip

\textbf{Case (a):} Suppose $\theta(d_i) \in [\mu-c\kappa_{\mu}, \mu+c\kappa_{\mu}]$ for {a positive} constant $c$, which is chosen sufficiently small so that $x - \lambda_+ = \kappa_x \asymp \kappa_\mu$ for $x\in I_i$, {where $i\in \mathbb{O}_\epsilon$}. From {\eqref{eq_gcomplex}}, we have 
\begin{equation}
|\theta(d_i) - \lambda_+| \asymp  \Delta(d_i)^4\,.\label{thetadi-lambda+ control in proof thm 3.2}
\end{equation} 
Together with  $\mathcal{T}'(x) \asymp \kappa_x^{-1/2}$ from {\eqref{eq_iii1}}, for $x \in I_i$, we have 
\begin{equation}\label{eq_simT}
    |\mathcal{T}'(x)| \asymp |\mathcal{T}'(\theta(d_i)))| \asymp \Delta(d_i)^{-2}\,.
\end{equation}
From the claim in (\ref{eq_diff}), when $\mu \notin I_i$, it is either $\mu<\theta(d_i)-n^{\epsilon}\omega(d_i)$ or $\mu> \theta(d_i)+n^{\epsilon}\omega(d_i)$. When $\mu<\theta(d_i)-n^{\epsilon}\omega(d_i)$, we have
\begin{align}\label{eq_g2}
    &\mathcal{T}(\mu)-\mathcal{T}(\theta(d_i)) \nonumber\\
    > &\,  \mathcal{T}(\theta(d_i) - n^{\epsilon}\omega(d_i))-\mathcal{T}(\theta(d_i)) \gtrsim n^{\epsilon}\omega(d_i)\Delta(d_i)^{-2}  \nonumber\\ 
    =&\, n^{\epsilon}\phi_n+n^{-1/2+\epsilon}\Delta(d_i)^{-1} \asymp n^{\epsilon}\phi_n + n^{-1/2+\epsilon}(\theta(d_i) - \lambda_+)^{-1/4} \nonumber\\
     \gg&\, n^{\epsilon/2}(\phi_n+ n^{-1/2}\kappa_{\mu}^{-1/4})\,,
\end{align}
where we use the monotonicity of $\mathcal{T}$ in the first step, the mean value theorem and \eqref{eq_simT} in the second step, the definition of $\omega(d_i)$ in the third step, and {\eqref{thetadi-lambda+ control in proof thm 3.2}} in the fourth step.
Similarly, when $\mu \notin I_i$ such that $\mu> \theta(d_i)+n^{\epsilon}\omega(d_i)$, the same argument leads to
\begin{align}\label{eq_g3}
    \mathcal{T}(\theta(d_i)) - \mathcal{T}(\mu)  
    \gg n^{\epsilon/2}(\phi_n+ n^{-1/2}\kappa_{\mu}^{-1/4})\,.
\end{align}
By \eqref{eq_g2} and \eqref{eq_g3}, the relationship \eqref{eq_g} leads to (\ref{eq_toprove}) when $\theta(d_i) \in [\mu-c\kappa_{\mu}, \mu+c\kappa_{\mu}]$.  
\medskip

\textbf{Case (b):} Suppose $\theta(d_i) \notin [\mu - c\kappa_{\mu}, \mu+ c\kappa_{\mu}]$ for the same constant $c$ in the previous case. For $\theta(d_i)>\mu + c\kappa_{\mu}$, since $\mathcal{T}$ is monotonically decreasing on $(\lambda_+,+\infty)$, we have that 
\begin{align}\label{eq_g0}
    \mathcal{T}(\mu)-\mathcal{T}(\theta(d_i))  
     > \mathcal{T}(\mu) - \mathcal{T}(\mu+c\kappa_{\mu})
    \asymp \kappa_{\mu}^{1/2}\gg n^{\epsilon/2}(\phi_n+n^{-1/2}\kappa_{\mu}^{-1/4}).
\end{align}
where we use {\eqref{eq_iii1}} and the mean value theorem in the second step and $\kappa_{\mu} > n^{3\epsilon}(\phi_n^2+n^{-2/3})$ for $\mu \notin I_0$ by the definition of $I_0$ in the last step. 
By a similar argument, when $\theta(d_i)<\mu-c\kappa_{\mu}$, we have
\begin{align}\label{eq_g1}
    \mathcal{T}(\theta(d_i)) - \mathcal{T}(\mu) 
    \gg n^{\epsilon/2}(\phi_n+n^{-1/2}\kappa_{\mu}^{-1/4})\,.
\end{align}
By \eqref{eq_g0} and \eqref{eq_g1}, we obtained \eqref{eq_toprove} when $\theta(d_i) \notin [\mu - c\kappa_{\mu}, \mu+ c\kappa_{\mu}]$, and hence the proof of \eqref{eq_toprove}. 
\newline
\newline
\textbf{Step VI. Show that each $I_i$ contains exactly the right number of {strong} outliers.} 
{Based on $\text{Spec}(\widetilde{H}) \subset I$, to finalize the proof, we need that each $I_i$ contains exactly the right number of outliers so that the relationship between $\wt\lambda_i$ and $\theta(d_i)$ is established.} We apply the continuity argument used in \cite[Section 6.5]{knowles2013isotropic} to this end. 
Set
\[
\wt S(t)=S(t)+A^{1/2}XB^{1/2}\,, 
\]
where the singular values of $S(t)$ are $(d_1(t),\ldots,d_r(t))$ so that $d_i(t)$ is a continuous function on $[0,1]$ for $i=1,\ldots,r$ satisfying some conditions detailed below. 

We set the conditions for the {continuous} paths $d_i(t)$ for $t\in [0,1]$. Assume $(d_1(1),\ldots,d_r(1))=(d_1,\ldots,d_r)$ satisfy Assumption \ref{assum_main}(iv). Consider $d_1(0), \cdots, d_{r}(0)>0$ so that Assumption \ref{assum_main}(iv) is satisfied but {\em independent} of $n$. 
Assume further that $d_1(0)>d_2(0) > \cdots > d_{r_\epsilon}(0)$ and $d_i(0) - d_{i+1}(0) \gtrsim 1$ for $1\leq i \leq r_{\epsilon}-1$, such that 
\begin{equation}\label{eq_multione}
\theta(d_1(0))>\theta(d_2(0)) > \cdots > \theta(d_{r_\epsilon}(0))\,.
\end{equation}
We require that $d_i(t)$, $i=1,\ldots,r_\epsilon$, {possibly $n$-dependent,} satisfies the following properties:
\begin{itemize}
\item[(i)] For all $t\in [0,1]$, the number $r_\epsilon$ is unchanged. Moreover, we always have the following order of the outliers: 
\[
\theta(d_1(t)) \geq \theta(d_2(t)) \geq \cdots \ge \theta(d_{r_\epsilon}(t))>\lambda_+\,.
\]

\item[(ii)] For all $t\in [0,1]$, denote the permissible intervals as $I_i(t)$,
where 
\[
I_i(t):=[\theta(d_i(t))-n^{\epsilon}\omega(d_i(t)),\, \theta(d_i(t))+n^{\epsilon}\omega(d_i(t))]\,,
\] 
and set 
\[
I(t):= I_0\cup \bigcup_{1\leq i \leq r_\epsilon}I_i(t)\,.
\]
If $I_i(1)\cap I_j(1)=\emptyset$ for $1\leq i < j\leq r_\epsilon$, then {we set $d_i(t)$ so that} $I_i(t)\cap I_j(t)=\emptyset$ for all $t\in [0,1)$. The interval $I_0$ is unchanged along the path. 
\end{itemize}
It is straightforward that such paths $d_i(t)$ exist.
The corresponding continuous path of outliers is denoted as $\{\wt \lambda_i(t)\}_{i=1}^{r_{\epsilon}}$. Since $t \to \wt S(t) \wt S(t)^\top $ is continuous, we find that $\wt\lambda_i(t)$ is continuous in $t \in [0,1]$ for all $i$.
Denote 
\[
\mathbf{x}(t): =(x_1(t), x_2(t), \cdots, x_{r_\epsilon}(t))=(\theta(d_1(t)),\ldots,\theta(d_{r_\epsilon}(t))\,. 
\]
By the continuity of $\theta$, $\mathbf{x}(t)$ is a continuous path over $[0,1]$.

We first claim that when $n$ is sufficiently large, each $I_i(0)$, 
$1\le  i \le r_\epsilon$, contains only the $i$-th eigenvalue of $\wt S(0)\wt S(0)^\top $. To show this, fix any $1\le i \le r_\epsilon$ and choose a small positively oriented closed contour $\mathcal{C} \subset \mathbb{C}\backslash [0, \lambda_+]$ so that 
$\mathcal{C}$ only enclose $I_i(0)$ but no other intervals $I_j(0)$ for $j\neq i$. Define two functions, 
\begin{equation*}
f_0(z):=\det (\Ub^\top  G(z) \Ub+\Db(0)^{-1})\ \mbox{ and }\   g_0(z):=\det (\overline{\Pi}(z)+\Db(0)^{-1})\,,
\end{equation*}
where $\Db(0)$ is defined in the same way as \eqref{definition D and U} with $d_1(0),\ldots,d_r(0)$.
Both $f(z)$ and $g(z)$ are holomorphic on and inside $\mathcal{C}$ by definition. By Lemma \ref{lemma det diag formula} we have $g_0(z)=z^{-r} \prod_{i=1}^r ( z m_{1c}(z)  m_{2c}(z)-d_i(0)^{-2})=z^{-r} \prod_{i=1}^r (\mathcal{T}(z)-\mathcal{T}(\theta(d_i(0))))$. Thus, $g_0(z)$ has precisely one zero at $z = \theta(d_i(0))$ inside $\mathcal{C}.$ So, by a proper choice of the contour $\mathcal{C}$, we have 
\begin{align*}
\min_{z \in \mathcal{C}} |g_0(z)|&\,=\min_{z \in \mathcal{C}}|z|^{-r} \prod_{i=1}^r |\mathcal{T}(z)-\mathcal{T}(\theta(d_i(0)))| \gtrsim   \frac{n^\epsilon \omega(d_i(0))}{\prod_{i=1}^r|\theta(d_i(0))-\lambda_+|^{1/2}}\\
&\,\gtrsim n^\epsilon \omega(d_i(0))\gtrsim n^{3\epsilon}\phi_n(\phi_n+n^{-1/3})^2+n^{2\epsilon}n^{-1/2}(\phi_n+n^{-1/3})\,,
\end{align*}
where first bound comes from the fact that $\lambda_+\leq |z|\leq \tau^{-1}$, the control of $\mathcal{T}(z)-\mathcal{T}(\theta(d_j(0)))$ by \eqref{Tz1-Tz2 bound} for any $1\leq j\leq r$, and the control of $|z-\theta(d_j(0))|\geq |\theta(d_i(0))-\theta(d_j(0))|/2\gtrsim 1$ for any $j\neq i$ by \eqref{eq_thetadiff} and the assumption $d_i(0)-d_{i+1}(0)\gtrsim 1$, the second bound comes from \eqref{eq_s36_1} since $|\theta(d_i(0))-\lambda_+|^{1/2}\asymp |d_i(0)-\alpha|^{1/2}\lesssim \tau^{-1}$, and the last bound comes from the definition of $\omega(d_i(0))$ and the lower bound assumption of $d_i(0)$.
Also, by \eqref{eq_cond0}, we have for any $z\in \mathcal{C}$ that
\begin{equation*}
|f_0(z)-g_0(z)| \lesssim  n^{\epsilon/2}(\phi_n + n^{-1/2}(\kappa_z+\eta)^{-1/4})\,,
\end{equation*}
which is clearly dominated by $\min_{z\in \mathcal{C}}|g_0(z)|$.
Hence, the claim follows from Rouch{\' e}'s theorem. 
In other words, when $t=0$, we have 
$\wt\lambda_i(0) \in I_i(0)$ for $1\leq i \le r_\epsilon$ and $\wt\lambda_i(0) \in I_0$ for $i>  r_\epsilon$.

By the same approach shown in \textbf{Step III} and \textbf{IV}, we can also show that all the eigenvalues $\{\wt\lambda_i(t)\}\subset I(t)$ for all $t \in [0,1]$.

With the above preparation, we can now show that {each $I_i=I_i(1)$ contains exactly the right number of outliers.} We prove this in two cases.
\medskip

\textbf{Case (a):} If $I_1(1),\ldots,I_{r_{\epsilon}}(1)$ are disjoint, then $I_1(t),\ldots,I_{r_{\epsilon}}(t)$ are disjoint for all $t \in [0,1)$ by property (ii). {Together with the results that each $I_i(0)$, $1\le  i \le r_\epsilon$, contains only the $i$-th eigenvalue of $\wt S(0)\wt S(0)^\top$} and the continuity of $\wt\lambda_i(t),$ we conclude that 
$\wt\lambda_i(t) \in I_i(t)$, $1\leq i \leq r_{\epsilon}$ 
for all $t \in [0,1]$, and hence the claim. 
\medskip

\textbf{Case (b):} If some of the intervals are not disjoint at $t=1$, let $\mathcal{B}$ denote the finest partition of $\{1,\cdots, r_{\epsilon}\}$ such that $i$ and $j$ belong to the same block if $I_i (1) \cap I_j(1)\neq\emptyset$. {This is the case when two outliers ``cannot be distinguished''.} Denote by $B_i$ the block of $\mathcal{B}$ that contains  $i$. Note that elements of $B_i$ are sequences of consecutive integers. We now pick any $1\leq i \leq r_{\epsilon}$ so that $|B_i|>1$, and let $j \in B_i$ such that it is not the smallest index in $B_i$. Note that
\begin{equation}\label{deltajj-1}  
x_{j-1}(1) - x_{j}(1) \leq 2n^\epsilon \omega(d_j)
\end{equation}
by assumption. Since the number of elements in $B_i$ is bounded by $r_\epsilon$, we obtain that
\begin{equation}\label{diamBi}
\Big| \bigcup_{j \in B_i} I_j(1) \Big| \leq r_\epsilon n^{ \epsilon} \omega(d_{\min\{j:j\in B_i\} }) = r_\epsilon  n^{ \epsilon} \omega(d_i)\,,
\end{equation}
where $\left| \bigcup_{j \in B_i} I_j(1) \right| $ stands for the length of $\bigcup_{j \in B_i} I_j(1)$. Thus, by the continuity construction, we have 
\begin{equation}
    |\wt\lambda_j(1)- \theta(d_j)| \leq r_\epsilon n^\epsilon \omega(d_i), \quad j \in B_i\,,
\end{equation}
and hence the claim.
\newline
\newline
\textbf{Step VII. Locations of {weak outlier and} non-outlier signals.} First, we fix a configuration $\mathbf x(0)$ satisfying the same setup in \textbf{Step VI}. In this setup, when $t=0$, {for $i=r_\epsilon+1,\ldots,r$, we set}
\begin{equation}\label{eq_upper}
\wt\lambda_i(0) \in I_0\ \ \text{and}\ \  \wt\lambda_i (0) \geq \lambda_+ - n^\epsilon (n^{-1/3}\phi_n^2 + n^{-2/3})\,,
\end{equation} 
and hence 
\begin{equation*}
|\widetilde{\lambda}_i(0)-\lambda_+| \leq {n^{\epsilon}}(\phi_n^2+n^{-2/3})\,.
\end{equation*}
Next we employ a similar continuity argument as that in \textbf{Step VI}. 
For $t \in [0,1],$ by \eqref{lem_rigidty} and Lemma \ref{lem_weyl}, for any $t\in (0,1]$, we always have that
\begin{equation}\label{iterlacing_t}
\wt\lambda_i(t) \geq \lambda_+-n^\epsilon (n^{-1/3}\phi_n^2 + n^{-2/3}), \quad i \notin \mathbb{O_\epsilon}.
\end{equation}
If $I_0$ is disjoint from the other $I_j$'s, then by the continuity of $\wt\lambda_i(t)$ and $\text{Spec}(\wt S\wt S^\top ) \subset I$, we can conclude that $\wt\lambda_i(t) \in I_0(t)$ for all $t \in [0,1]$. Otherwise, we again consider the partition $\mathcal B$ as in \textbf{Step VI}, and let $B_0$ be the block of $\mathcal B$ that contains $i$. With the same arguments, we can prove that 
$$
I_0(1) \cup \left( \bigcup_{j\in B_0} I_{j}(1)\right)\subset [0, \lambda_+ + C{n^{\epsilon}}(\phi_n^2 +n^{-2/3 })]
$$
{for some $C>0$.}
Then using (\ref{eq_upper}), \eqref{iterlacing_t} and the continuity of the eigenvalues along the path, we obtain that for all $r_{\epsilon}< i \leq r$,
\begin{equation}\nonumber
\big|\wt\lambda_i(t) - \lambda_+\big| \le C{n^{\epsilon}}(\phi_n^2 + n^{-2/3})
\end{equation}
for all $t \in [0,1]$. 
{We thus finish the proof.}

\section{Proof of Theorem \ref{thm_eigenvaluesticking}}\label{sec_pf_sticking}
{Similar to the beginning of the proof of Theorem \ref{thm_value},}
by Theorem \ref{thm_value}, Lemma \ref{LEM_SMALL}, Lemma \ref{lem_rigidty}, Theorem \ref{lem_locallaw} and Lemma \ref{delocal_rigidity}, 
for any small constant $\varepsilon>0$, we can choose a high-probability event $\Xi_\varepsilon$ in which \eqref{aniso_lawev0}-\eqref{eq_rigi0} and the following estimates hold:
\begin{equation}\label{eq_stickingrigi}
|\wt\lambda_{i}-\lambda_+| \leq n^{\varepsilon/2}(\phi_n^2 + n^{-2/3}), \quad \text{ for } \ r^+ +1 \le i \le \varpi, 
\end{equation}
for some fixed large integer $\varpi\ge r$ and 
\begin{equation}\label{eq_stickingrigi_st}
\begin{split}
|\lambda_i-\gamma_i| \leq&\,  n^{-2/3+\varepsilon/2}\big( i^{-1/3}+\mathbf 1(i \le n^{1/4} \phi_n^{3/2})\big)  + n^{\varepsilon/2} \eta_l(\gamma_i)\\
&+ n^{2/3+\varepsilon/2} i^{-2/3} \eta_l^2(\gamma_i)
\end{split}
\end{equation}
for $ i \leq \tau p$, {where $\tau>0$ is a small constant.} 
For any $i$, define a set 
\begin{equation}\label{eq_omega}
\begin{split}
\Omega_i :=
\Big\{x \in &\,[\lambda_{i-r-1},\, \lambda_+ + c_0n^{2\varepsilon} (\phi_n^2 +n^{-2/3})]\Big| \\
 &\text{dist} \big (x, \,\text{Spec}({ZZ^\top }) \big )>n^{-1+\varepsilon} \alpha_+^{-1} + n^{\varepsilon}\eta_l(x) \Big\},
\end{split}
\end{equation}
where {$\eta_l(x)$ is defined in \eqref{definition eta_l(E)},} $\lambda_i := \infty $ if $i<1$, $\lambda_i = 0$ if $i >p$, $\text{Spec}(ZZ^\top)$ stands for the spectrum of $ZZ^\top$ and $c_0>0$ is a small constant. Note that {$\Omega_i=\emptyset$ for $i=1,\ldots,r+1$, and} by \eqref{eq_stickingrigi_st}, we have
\begin{equation}\label{eq_st3}
|x-\lambda_+|>n^{-1+\varepsilon} \alpha_+^{-1} + n^{\varepsilon}\eta_l(x)
\end{equation} 
for all $x \in \Omega_i$.
We then have the following Lemm.
\begin{lemma}\label{prop_eigensticking} 
For $\alpha_+ \geq n^\varepsilon (\phi_n + n^{-1/3})$ and $i \leq n^{1-2\varepsilon} \alpha_+^3,$ there exists a constant $c_0>0$ such that the set $\Omega_i$ contains no eigenvalue of $\wt S \wt S^\top .$
\end{lemma}
With Lemma \ref{prop_eigensticking}, the rest of the proof is exactly the same as the proof for \cite[Lemma S.4.5]{DY2019} by letting $s=0$. Thus we omit it here. Below is the proof of Lemma \ref{prop_eigensticking}.

\begin{proof}(Proof of Lemma \ref{prop_eigensticking})
Define
\begin{equation}\label{etax2}
\eta_x:=n^{-1+ \varepsilon} \alpha_+^{-1} + n^{\varepsilon}\eta_l(x) \ \mbox{and} \  z_x:= x+i \eta_x \,,
\end{equation}
where $x \in \Omega_i$. Denote $G_{\mathbf a \mathbf b}(z) = \mathbf a' G(z)\mathbf b$ for $z\in \mathbb{C}^+$ and $\mathbf a, \mathbf b \in \mathbb C^{p+n}$.
Recall \eqref{main_representation} and set 
$\ub = \begin{pmatrix}
     \ub_1\\0
 \end{pmatrix}$ 
 and 
$\vb = \begin{pmatrix}
     0\\\vb_2
 \end{pmatrix}$,  
we have
\begin{align*}
& |G_{\ub \vb}(z_x)-G_{\ub \vb}(x)|   \\
\lesssim &\,\eta_x|G_{\ub \vb}(z_x)| + \sum_{k = 1}^{p\wedge n} \sqrt{\lambda_k}\left|\langle \ub,{ \xi}_k\rangle \langle { \zeta}_k,\vb\rangle\right| \left|\frac{\eta_x}{(\lambda_k-x-i\eta_x)(\lambda_k - x)}\right|\\
\lesssim &\,\sum_{k} \left(\left|\langle \ub,{ \xi}_k\rangle \right|^2+\left|\langle { \zeta}_k,\vb\rangle\right|^2\right) \frac{\eta_x}{(\lambda_k-x)^2+(\eta_x)^2} \\
= &\,\operatorname{Im} G_{\ub \ub}(z_x)+\operatorname{Im} G_{\vb \vb}(z_x),
\end{align*}
where in the second step we use $|x-\lambda_k|\ge \eta_x$ for $x\in \Omega_i$. 
We have similar bounds for $G_{\ub\ub}(\cdot)$, $G_{\vb\ub}(\cdot)$ and $G_{\vb\vb}(\cdot)$. 
Now together with (\ref{aniso_lawev0}), we obtain that 
\begin{align}
& \mathcal{D}^{-1}+\mathbf{U}^* G(x) \mathbf{U}  \label{bound_im}\\
=\,& \mathcal{D}^{-1}+ \mathbf{U}^* G(z_x) \mathbf{U}+\mathbf{U}^*(G(x)-G(z_x)) \mathbf{U}\nonumber\\
=\,& \mathcal{D}^{-1}+ \mathbf{U}^* \overline\Pi(z_x) \mathbf{U}+ \OO\left( n^{\varepsilon/2}\Psi(z_x)+ n^{\varepsilon/2}\phi_n + \im m_{2c}(z_x) \right) \nonumber\\
=\,& \mathcal{D}^{-1}+\mathbf{U}^* \overline\Pi(x) \mathbf{U}+ \OO\left( n^{\varepsilon/2} \im m_{2c}(z_x) +\frac{n^{\varepsilon/2}}{n\eta_x}+ n^{\varepsilon/2}\phi_n\right) ,  \nonumber
\end{align}
where in the last step we used (\ref{eq_estimm}) and 
$$\Psi(z_x) \lesssim \im m_{2c}(z_x) + ({n\eta_x})^{-1}.$$
Therefore, by Lemma \ref{lem_mastereq} and Lemma \ref{eq_lemmadeter}, we conclude that for $x \in \Omega_i$,  $x$ is not an eigenvalue of $\wt Z \wt Z^\top $ if 
\begin{align} \label{sticking_dj}
\min_{1\leq j\leq r}| \mu m_{1c}(x)m_{2c}(x)-d_j^{-2}|
\gg n^{\varepsilon/2} \im m_{2c}(z_x)  +\frac{n^{\varepsilon/2}}{n\eta_x}+ n^{\varepsilon/2}\phi_n\,.
\end{align}
Since $i \leq n^{1-2\varepsilon} \alpha_+^3$, by \eqref{eq_stickingrigi_st} we have
\begin{equation}\label{kappax2}
\begin{split}
 -c_0n^{2\varepsilon}&(\phi_n^2 +n^{-2/3}) \le \lambda_+-x \\
 &\lesssim \left(\frac{i}{n}\right)^{2/3} + n^{-2/3+\varepsilon/2} + n^{\varepsilon/2}\eta_l(\gamma_i) + \frac{n^{2/3+\varepsilon/2}}{ i^{2/3}} \eta_l^2(\gamma_i)\lesssim   n^{-4\varepsilon/3}\alpha_+^2 
 \end{split}
 \end{equation}
 for $x\in \Omega_i$, where we also used $\gamma_i\asymp (i/n)^{2/3}$ and $\alpha_+\ge n^\varepsilon(\phi_n +n^{-1/3})$. Then by (\ref{sqroot4}), we have
$$|m_{2c}(x)-m_{2c}(\lambda_+)| \le C n^{-2\varepsilon/3}{\alpha_+} 
\ll \alpha_+$$
for $x\in \Omega_i \cap \{x: x\le \lambda_+\}$
and 
$$ |m_{2c}(x)-m_{2c}(\lambda_+)| \le C\sqrt{c_0}n^{ \varepsilon}\left(n^{-1/3 }+\phi_n\right) \le C\sqrt{c_0} \alpha_+
$$
for $x\in \Omega_i \cap \{x: x> \lambda_+\}$, where
the constant $C>0$ is independent of $c_0$. Using similar approach, we can derive the same bounds for $m_{1c}(x)$. Plugging the above two estimates into \eqref{sticking_dj} and using $| \lambda_+ m_{1c}(\lambda_+)m_{2c}(\lambda_+)-d_j^{-2}| \asymp |d_j -\alpha|\ge \alpha_+$ \eqref{eq_st3} and the triangle inequality, we obtain that
\begin{equation} \nonumber 
\min_{1\leq j\leq r}| \mu m_{1c}(x)m_{2c}(x)-d_j^{-2}|\gtrsim \alpha_+ 
\end{equation}
as long as $c_0$ is sufficiently small. On the other hand, using \eqref{eq_estimm}, \eqref{etax2} and \eqref{kappax2}, we can verify that for $x\in \Omega_i$ and $x\le \lambda_+$,
$$
n^{\varepsilon/2} \left(\im m_{2c}(z_x) +\frac{1}{n\eta_x} +  \phi_n\right) 
\lesssim n^{\varepsilon/2}\left(\sqrt{\kappa_x+\eta_x} +\frac{1}{n\eta_x} + \phi_n\right)\ll \alpha_+\,,
$$ 
and for $x\in \Omega_i$ and $x> \lambda_+$,
$$
n^{\varepsilon/2} \left(\im m_{2c}(z_x) +\frac{1}{n\eta_x}+  \phi_n\right) \lesssim n^{\varepsilon/2} \left(\frac{\eta_x}{\sqrt{\kappa_x+\eta_x}} +\frac{1}{n\eta_x} +  \phi_n\right) \ll \alpha_+ .$$
This proves \eqref{sticking_dj}, which further concludes the proof of Lemma \ref{prop_eigensticking}.
\end{proof}

\section{Proof of Theorem \ref{thm_noneve}}
\label{sec_pf_evout}

The proof is composed of several steps.
\newline
\newline
\textbf{Step I. Prepare an event.}
Let $\epsilon>0$ be a small positive constant and take $0<\omega<{2/3}$. 
By {Theorems \ref{thm_value} and \ref{lem_locallaw}, and  Lemmas \ref{lem_rigidty}  and \ref{lem_pipi}}, we find that there exists some high probability event $\Xi$ 
such that the followings hold when conditional on $\Xi$. 
\begin{enumerate}
\item [(i)] Recall the definition of $S_{out}$ in \eqref{eq_sout}. {Fix $\varsigma_1>0$ and let $\varsigma_2=(\lambda_+\omega)^{-1}$, 
such that $\wt \lambda_i \in  S(\varsigma_1 , (\lambda_+\omega)^{-1})$ for $1\leq i \leq cn$
and \eqref{eq_estimm} applies for all $z \in S(\varsigma_1, (\lambda_+\omega)^{-1})$. Also,
\begin{align}
   &S_{out}((\lambda_+\omega)^{-1},\omega)\nonumber \\
   =\,& \{E+i\eta: \lambda_+ + n^{\omega}(n^{-2/3}+n^{-1/3}\phi_n^2)\leq E \leq \omega^{-1}, \eta \in [0,1]\}\,.
\end{align}}
{
For all $z \in \wt S_0(\varsigma_1,(\lambda_+\omega)^{-1},\omega)$, from Lemmas \ref{LEM_SMALL} and \ref{lem_pipi}, and since $r$ is fixed, we have that
\begin{equation}\label{aniso_lawev0_2}
    \|\Omega(z)\| = \norm{\mathbf{U}^\top  G(z)\mathbf{U}-\overline\Pi(z)} \leq  n^{\epsilon/2}(\phi_n + \Psi(z)). 
\end{equation}
} 
Alos, for all $z \in S_{out}((\lambda_+\omega)^{-1},\omega)$, {by Theorem \ref{lem_locallaw} and Lemma \ref{lem_pipi}, we have} 
 \begin{equation}\label{eq_2.5localaw}
     \|\Omega(z)\| =\norm{\mathbf{U}^\top  G(z)\mathbf{U}-\overline\Pi(z)}\leq n^\epsilon[\phi_n + n^{-1/2}(\kappa_z+\eta)^{-1/4}]\,.
 \end{equation}

\item [(ii)] From Theorem \ref{thm_value} and the rigidity of eigenvalues in Lemma \ref{lem_rigidty}, we have
\begin{equation}\label{rigid_outXi0}
|\wt\lambda_i-\theta(d_i)| \leq n^\epsilon(n^{-1/2}\Delta(d_i)+ \phi_n \Delta^2(d_i)),\quad 1\leq i\leq r^+,
\end{equation}
\begin{equation}\label{eq_stickingrigi}
|\wt\lambda_{i}-\lambda_+| \leq n^{\epsilon/2}\left(\phi_n^2 + n^{-2/3}\right), \quad  r^+ +1 \le i \le \varpi 
\end{equation}
for some fixed large integer $\varpi\ge r$, and
\begin{align}\label{eq_stickingrigi2}
&|\lambda_i-\gamma_i|\\
 \leq \,& n^{\epsilon/2}\left[n^{-2/3}\left( i^{-1/3}+\mathbbm 1(i \le n^{1/4} \phi_n^{3/2})\right)  +  \eta_l(\gamma_i)+ n^{2/3} i^{-2/3} \eta_l^2(\gamma_i)\right]\nonumber
\end{align}
for $ 1\leq i \leq \varpi$. 
\end{enumerate}
The rest of the proof is restricted to the event $\Xi$.
\newline
\newline
\textbf{Step II. Prepare some quantities.}
We prepare some notations that will be used in the following proofs.
For $i = 1, \cdots, r$, let $\ub_i^e=(\ub_i^\top , \mathbf{0})^\top $ be the embedding of $\ub_i$ in $\mathbb{R}^{p+n}.$  By (\ref{green2}), it is easy to see that 
\begin{equation}\label{eq_gG NEW}
\ub_i^\top  \ctG_1(z) \ub_j=(\ub_i^e)^\top  \widetilde{G}(z) \ub_j^e=e_i^\top  \Ub^\top \widetilde{G}(z)\Ub  e_j\,, 
\end{equation}
where $e_j\in \mathbb{R}^{2r}$ is a unit vector with the {$j$-th} entry $1$. When $1\leq i,j\leq r$, by (\ref{defn_greenrep2}), we have
\[
\ub_i^\top  \ctG_1(z) \ub_j=e_i^\top [\Db^{-1}-\Db^{-1}( \Db^{-1}+\Ub^\top  G(z) \Ub )^{-1}\Db^{-1}]e_j\,.
\]
To study $( \Db^{-1}+\Ub^\top  G(z) \Ub )^{-1}$, by the resolvent expansion from \eqref{eq_resexp}, we will encounter the term $(\mathbf{D}^{-1}+\overline{\Pi}(z))^{-1}$ {and $\Omega(z)$}.
By an elementary computation based on the Schur complement, we have 
\begin{equation} \label{biginverse}
[(\mathbf{D}^{-1}+\overline{\Pi}(z))^{-1}]_{i j}=
\begin{cases}
\delta_{ij}\frac{zm_{2c}(z)}{z m_{1c}(z)m_{2c}(z)-d_i^{-2}}, & 1 \leq i, j \leq r; \\
\delta_{ij}\frac{zm_{1c}(z)}{z m_{1c}(z)m_{2c}(z)-d_{i-r}^{-2}}, & r \leq i, j \leq 2r; \\
-\delta_{\bar{i} j} \frac{z^{1/2}d_{i}^{-1}}{zm_{1c}(z)m_{2c}(z)-d_{i}^{-2}},& 1 \leq i \leq r, \ r \leq j \leq 2r ;\\
-\delta_{i \bar{j}} \frac{z^{1/2}d_j^{-1}}{zm_{1c}(z)m_{2c}(z)-d_j^{-2}} , & r \leq i \leq 2r, \ 1 \leq j \leq r. 
\end{cases}
\end{equation}

{For the following, we will show the approximation of $ \|\Omega(z_i)\|$ for all $ \wt\lambda_i + i \eta \in S(\varsigma_1,(\lambda_+\omega)^{-1})$ and $1\leq i \leq cn$.} We claim that for any {$1\leq i\leq cn$}, the equation
\begin{equation}\label{eq_defneta}
 \eta \operatorname{Im} m_{1c}(\wt\lambda_i + i \eta)=n^{2\epsilon}\phi_n \eta + n^{-1+6\epsilon}
\end{equation}
over $\eta\in [0,1]$ has a unique solution.
Indeed, note that $\eta\operatorname{Im} m_{1c}(\wt\lambda_i + i\eta)$ is a strictly monotonically increasing function of $\eta$. Since {$\wt\lambda_i + i \eta\in S(\varsigma_1,(\lambda_+ \omega)^{-1})$ with a proper selection of $\varsigma_1$ and $\omega$}, the behavior of $\eta\operatorname{Im} m_{1c}(\wt\lambda_i + i\eta)$ is detailed in \eqref{eq_estimm}. First, consider the case $1\leq i\leq r_+$. Since $\wt\lambda_i$ is of order $1$ in this case, $\eta\operatorname{Im} m_{1c}(\wt\lambda_i + i\eta)$ grows quadratically when $\eta\leq n^{3\epsilon}\phi_n$. Thus we find a unique solution $\widehat {\eta}_i$ over the region $0\leq \eta\leq n^{3\epsilon}\phi_n$. When $\eta>n^{3\epsilon}\phi_n$, we cannot find a solution since $\eta^2$ dominates $n^{2\epsilon}\phi_n \eta$.
Second, consider the case $r_++1\leq i\leq \varpi$. In this case, we should consider the case when ${\wt \lambda}_i$ is less than or greater than $\lambda_+$. When ${\wt \lambda}_i\geq\lambda_+$, $\eta\operatorname{Im} m_{1c}(\wt\lambda_i + i\eta)$ grows approximately like $\eta^{3/2}$ when $\eta>n^\epsilon \phi_n$, and we can find a unique solution $\widehat {\eta}_i$ of order $n^{4\epsilon}\phi_n^2$. When $0\leq \eta<n^\epsilon \phi_n$, we cannot find a solution since $n^{2\epsilon}\phi_n \eta$ dominates $\eta\operatorname{Im} m_{1c}(\wt\lambda_i + i\eta)$ in this region. The case ${\wt \lambda}_i<\lambda_+$ can be argued in the same way. We thus obtain the claim.
More precisely,
with \eqref{eq_estimm}, one can check that 
\begin{align}\label{eq_etaestimate}
\widehat\eta_i \asymp \begin{cases}
  n^{4 \epsilon}\left( \phi_n^2 + n^{-2/3}\right) \ &\text{if } \ |\wt\lambda_i - \lambda_+|\le n^{4 \epsilon}\left( \phi_n^2 + n^{-2/3}\right) \\
 n^{2\epsilon}\phi_n \sqrt{\kappa_{\wt\lambda_i}}+ n^{-1/2+3 \epsilon} \kappa_{\wt\lambda_i}^{1/4}  \ &\text{if }\ \wt\lambda_i \geq  \lambda_++n^{4 \epsilon}\left( \phi_n^2 + n^{-2/3}\right)\\
 n^{-1+6 \epsilon}\kappa_{\wt\lambda_i}^{-1/2}    \ &\text{if }\ \wt\lambda_i \leq \lambda_+-n^{4 \epsilon}\left( \phi_n^2 + n^{-2/3}\right)
\end{cases}.
\end{align}

Based on the above claim, in the discussion afterwards, { we fix $z_i=\wt\lambda_i+\ri \eta_i\in  S(\varsigma_1,(\lambda_+\omega)^{-1})$ with a proper selection of $\omega$}, where $\eta_i : = \widehat \eta_i \vee n^{\epsilon}\eta_{l}(\gamma_i)$. 
Next, we claim that for such $z_i$, we have
\[
\|\Omega(z_i)\|=\norm{\mathbf U^\top  G(z_i)\mathbf U-\overline\Pi(z_i)}\lesssim n^{-\epsilon}\operatorname{Im} m_{1c}(z_i)\,.
\]
To show this claim, we discuss two cases: 
(i) $\widehat\eta_i \ge n^\epsilon \eta_l(\gamma_i)$ and (ii) $\widehat\eta_i < n^\epsilon \eta_l(\gamma_i)$. In case (i), { note that by the definition of $\wt S_0(\varsigma_1, (\lambda_+\omega)^{-1},\omega)$ in \eqref{tildeS} and the definition of $\eta(\gamma_i)$ in \eqref{definition eta_l(E)}, we conclude that $z_i\in  {\wt S_0(\varsigma_1,(\lambda_+\omega)^{-1},\omega)}$. Thus, 
\eqref{aniso_lawev0_2} can be applied to $z_i$ and we have}
\begin{equation}\label{ll1}
\begin{split}
\|\Omega(z_i)\|  \leq n^{\epsilon/2}(\phi_n + \Psi(z_i))  \lesssim n^{-3\epsilon/2}\operatorname{Im} m_{1c}(z_i)\,,
\end{split}
\end{equation}
where the last inequality is by \eqref{eq_defneta} and the definition of $\Psi(z_i)$ in \eqref{eq_defPsi}. 
In case (ii), {note that $\eta_l(\gamma_i)\asymp n^{-3/4} + n^{-1/2} \left(\sqrt{\kappa_{\gamma_i}}+\phi_n\right)$ from \eqref{etalE}.}  By comparing $\widehat  \eta_i$ and $n^{\epsilon}\eta_l(\gamma_i)$ using \eqref{eq_etaestimate}, we conclude that $\wt\lambda_i \leq \lambda_+-n^{4 \epsilon}\left( \phi_n^2 + n^{-2/3}\right)$, and hence $\kappa_{\wt\lambda_i} \gtrsim n^{4\epsilon}(\phi_n^2+n^{-2/3})$. 
{Moreover, 
with the interlacing lemma from Lemma \ref{lem_weyl} and the definition of the classical location $\gamma_i$ in \eqref{eq_classical}, we have $\gamma_i \approx\lambda_{i} \leq \wt \lambda_i \leq \lambda_{i-r} \leq \lambda_+$ and thus $\kappa_{\gamma_i}\gtrsim \kappa_{\wt\lambda_i}$. 
By a direct comparison, we derive
\begin{equation}\label{gammaisim0}
n^{\epsilon}\eta_l(\gamma_i) \ll \kappa_{\gamma_i}\,. 
\end{equation}
Also, by a direct comparison using \eqref{eq_etaestimate}, we have
\begin{equation}\label{gammaisim}
\frac{n^{5\epsilon}}{n\eta_l(\gamma_i)} \lesssim \sqrt{\kappa_{\gamma_i}}.
\end{equation}
Next, from \eqref{eq_estimm}, we derive that  
\begin{equation}\label{Imm2cziapprox kappa gammai}
\im m_{1c}(z_i) \asymp \sqrt{\kappa_{\wt\lambda_i} + n^{\epsilon}\eta_l(\gamma_i)}\lesssim \sqrt{\kappa_{\gamma_i}}\,,
\end{equation}
where the last approximation comes from \eqref{gammaisim0} and $\kappa_{\gamma_i}\gtrsim \kappa_{\wt\lambda_i}$.}
By putting the above together, we have 
\begin{equation}\label{ll2}
\begin{split}
&\|\Omega(z_i)\| \leq n^{\epsilon/2}\big(\phi_n +\Psi(\wt\lambda_i + i n^{\epsilon}\eta_l(\gamma_i))\big) \\
\lesssim &\, n^{\epsilon/2}\left(\phi_n +\sqrt{\frac{\sqrt{\kappa_{\gamma_i}}}{n\eta_l(\gamma_i)}} + \frac{1}{n\eta_l(\gamma_i)}\right) 
\le n^{\epsilon/2}\phi_n + n^{-2\epsilon}\sqrt{\kappa_{\gamma_i}} \lesssim n^{-\epsilon}\operatorname{Im} m_{1c}(z_i),
\end{split}
\end{equation}
where the second bound comes from the definition of $\Psi$ and \eqref{Imm2cziapprox kappa gammai}, the third bound comes from \eqref{gammaisim}, and the last inequality comes from applying \eqref{eq_defneta} and \eqref{Imm2cziapprox kappa gammai}.

\textbf{Step III. Prove the theorem.}
With the above bounds, we start the proof. For $j = 1, \cdots, r$, set $\ub_j^e=(\ub_j^\top , \mathbf{0})^\top $ be the embedding of $\ub_j$ in $\mathbb{R}^{p+n}.$
{ We let $z_i = \wt\lambda_i + \eta_i \in S(\varsigma_1, (\lambda_+\omega)^{-1})$ for $1\leq i \leq cn$ as in Step II and thus all the estimations provided can apply.}
With the spectral decomposition \eqref{main_representation}, we have that for $1\leq i \leq cn$,
\begin{equation}
\operatorname{Im}\, \langle \mathbf u_j^e, {\wt G}(z_i)\mathbf u_j^e\rangle=\sum_{k=1}^{p \wedge n} \frac{ \eta_i  \vert \langle \mathbf u_j, \wt{\bm \xi}_k\rangle \vert^2}{(\wt\lambda_k-\wt\lambda_i)^2+\eta^2_i }=\frac{\vert \langle \mathbf u_j, \wt{\bm \xi}_i\rangle \vert^2}{\eta_i }+\sum_{k\neq i} \frac{ \eta_i  \vert \langle \mathbf u_j, \wt{\bm \xi}_k\rangle \vert^2}{(\wt\lambda_k-\wt\lambda_i)^2+\eta^2_i } \,,
\end{equation}
and thus 
\begin{equation}\label{eq_nonspike1}
|\langle \ub_j, \wt\bxi_{i} \rangle|^2 \leq \eta_i \operatorname{Im} \langle \ub_j^e, \widetilde{G}(z_i) \ub_j^e\rangle\,. \end{equation}
By Lemma \ref{lem_gtitle}, we obtain another identity
\begin{equation*}
\langle\mathbf{u}_j^e, \wt G(z_i)\mathbf{u}_j^e\rangle=-\frac{1}{z_id_j^2}[(\mathbf{D}^{-1}+\mathbf{U}^\top  G(z_i)\mathbf{U})^{-1}]_{\bar{j}\bar{j}}.
\end{equation*}
Using the second order resolvent expansion shown in \eqref{eq_resexp} for $(\mathbf{D}^{-1}+\mathbf{U}^\top  G(z_i)\mathbf{U})^{-1}$ and (\ref{biginverse}), we have
\begin{align} \label{bigexpansionnonoutlier}
\langle\ub_j^e, \, & \wt G(z_i)\ub_j^e\rangle 
= -\frac{1}{z_id_j^2} \Big[ \frac{z_im_{1c}(z_i)}{\mathcal{T}(z_i)-d^{-2}_{j}}+\frac{z_if(z_i)}{(\mathcal{T}(z_i)-d_j^{-2})^2} \nonumber \\
& +  \left( [ (\mathbf{D}^{-1}+\overline\Pi(z_i))^{-1} \Omega(z_i) ]^2 (\mathbf{D}^{-1}+ \mathbf{U}^\top G(z_i)\mathbf{U})^{-1} \right)_{\bar{j}\bar{j}}\Big] \,,
\end{align}
where $f(z)=f_1(z)+f_2(z)$ and
\begin{equation*}
f_1(z):=m_{1c}(z)[zm_{1c}(z)\Omega(z)_{\bar{j}\bar{j}}+(-1)^{r}z^{1/2}d_j^{-1}\Omega(z)_{j\bar{j} }],
\end{equation*}
\begin{equation*}
f_2(z):=d_j^{-1}[(-1)^{r}z^{1/2}m_{1c}(z)\Omega(z)_{\bar{j}j}+d_j^{-1}\Omega(z)_{jj}].
\end{equation*}
To estimate the right-hand side of (\ref{bigexpansionnonoutlier}), note that by \eqref{eq_estimm} and the definition of $f$,
\begin{equation}\label{eq_b1}
|f(z_i)| \lesssim \norm{\Omega(z_i)}\,
\end{equation}
and
\begin{equation}\label{eq_b2}
\min_j |\mathcal{T}(z_i)-d_j^{-2}| \geq \operatorname{Im} \mathcal{T}(z_i) \underset{(\ref{eq_estimm})}{\asymp} \operatorname{Im} m_{1c}(z_i) \underset{\eqref{ll1}}{\gg}  \norm{\Omega(z_i)}\,.
\end{equation}
Jointly by \eqref{eq_b1} and \eqref{eq_b2}, the second term on the right-hand side of (\ref{bigexpansionnonoutlier}) is dominated by the first term. For the third term in (\ref{bigexpansionnonoutlier}), we apply the second order resolvent expansion on $(\mathbf{D}^{-1}+\mathbf{U}^\top  G(z_i)\mathbf{U})^{-1}$, then we acquire a similar formula as in (\ref{bigexpansionnonoutlier}) times  $ [ (\mathbf{D}^{-1}+\overline\Pi(z_i))^{-1} \Omega(z_i) ]^2$.  
By \eqref{biginverse} and \eqref{eq_b2}, we have that 
\begin{align}\label{eq_b3}
\left\| (\mathbf D^{-1}+ \mathbf U^\top  G(z_i) \mathbf U)^{-1}\right\| \lesssim \frac{1}{ |\mathcal{T}(z_i)-d_j^{-2}|}  \lesssim \frac{1}{\text{Im} \ m_{1c}(z_i)} \ll \| \Omega(z_i)\|^{-1}\,.
\end{align} 
Inserting bounds from \eqref{eq_b1}-\eqref{eq_b3} into (\ref{bigexpansionnonoutlier}), we obtain that
\begin{equation} \label{eq_reduceim}
\langle\mathbf{u}_j^e, \wt G(z_i)\mathbf{u}_j^e\rangle =\frac{m_{1c}(z_i)}{1-d_j^2\mathcal{T}(z_i)}+O\left(\frac{d_j^2\|\Omega(z_i)\|}{|1-d_j^2 \mathcal{T}(z_i)|^2} \right)\,.
\end{equation}
The next lemma provides a lower bound for $\left|1-d_j^2\mathcal{T}(z_i)\right|$. Its proof is the same as the one for \cite[Lemma 5.6]{ding2020high}, so we omit it.

\begin{lemma}\label{lem_denominator} 
Take $z_i = \wt\lambda_i + i \eta_i$ and $1\leq j \leq r$. For any fixed $\epsilon_0 <1/3$ and $\delta \in [0,1/3-\epsilon_0)$, when
$\wt\lambda_i \in [0, \,\theta(\alpha+(\phi_n+n^{-1/3})n^{\delta+\omega})]$, there exists a constant $c>0$ such that 
\begin{equation*}
|1-d_j^2\mathcal{T}(z_i)| \geq c d_j^2(n^{-2\delta}|d_j^{-2}-\mathcal T(\lambda_+)|+\operatorname{Im} \mathcal{T}(z_i)).
\end{equation*} 
\end{lemma}
Now we fix the $\delta$ in Lemma \ref{lem_denominator}. By (\ref{eq_nonspike1}) and (\ref{eq_reduceim}), we have that
\begin{align}\label{eq_finalbound}
\big|\langle \ub_j, \wt\bxi_{i} \rangle\big|^2 &  \leq \eta_i \left(\operatorname{Im} \left[ \frac{m_{1c}(z_i)}{1-d_j^2\mathcal{T}(z_i)} \right]+\frac{C d_j^2\| \Omega(z_i)\| }{|1-d_j^2 \mathcal{T}(z_i)|^2}\right) \nonumber\\
& = \frac{\eta_i}{|1-d^2_j\mathcal{T}(z_i)|^2} [ \operatorname{Im} m_{1c}(z_i) (1- \re(d_j^2 \mathcal{T}(z_i))) \nonumber\\
&\qquad +d_j^2\operatorname{Re} m_{1c}(z_i) \operatorname{Im}\mathcal{T}(z_i)+ Cd_j^2 \| \Omega(z_i)\| ] \nonumber\\
& = \frac{\eta_i}{|1- d^2_j\mathcal{T}(z_i)|^2} \big[ \operatorname{Im} m_{1c}(z_i) (1- d^2_j\mathcal T(\lambda_+) +d_j^2\re(\mathcal T(\lambda_+) - \mathcal{T}(z_i))) \nonumber \\
&\qquad +d_j^2\operatorname{Re} m_{1c}(z_i) \operatorname{Im}\mathcal{T}(z_i)+ Cd_j^2 \| \Omega(z_i)\| \big]\,.
\end{align}
We next bound the terms in (\ref{eq_finalbound}) one by one. For the first term, by {\eqref{eq_estimm} and \eqref{eq_realestimate}}, we have
\begin{align}\label{eq_671}
|\eta_i\operatorname{Im}& m_{1c}(z_i)[(1-d_j^2 \mathcal{T}(\lambda_+))+ d_j^{2}\operatorname{Re}(\mathcal T(\lambda_+)- \mathcal{T}(z_i))] | \\
& \lesssim |\eta_i\operatorname{Im}m_{1c}(z_i) |\left(|d_j-\alpha|+\sqrt{\kappa_{\wt\lambda_i}+\eta}_i\vee\Big( \frac{\eta_i}{\sqrt{\kappa_{\wt\lambda_i}+\eta_i}}+\kappa_{\wt\lambda_i}\Big)\right)\,.\nonumber
\end{align}
For the second item of (\ref{eq_finalbound}), by Lemma \ref{s35_DY2019} we have
\begin{equation}\label{eq_672}
|\eta_i\operatorname{Re}m_{1c}(z_i) \operatorname{Im} \mathcal{T}(z_i)| \asymp |\eta_i\operatorname{Im}m_{1c}(z_i)|\,.
\end{equation}
Note that by \eqref{eq_defneta}, \eqref{gammaisim}, and \eqref{eq_estimm}, we have
\begin{equation}\label{eq_imestimate}
|\eta_i\operatorname{Im}m_{1c}(z_i)| \lesssim
\begin{cases}n^{2\epsilon}\phi_n\widehat\eta_i+n^{-1+6 \epsilon}, \ &\text{if } \widehat\eta_i \ge n^\epsilon \eta_l(\gamma_i)\\
n^{\epsilon}\eta_l(\gamma_i)\sqrt{\kappa_{\gamma_i}}, \ &\text{if } \widehat\eta_i < n^\epsilon \eta_l(\gamma_i)
\end{cases}.
\end{equation}
For the third term, by \eqref{ll1}, \eqref{ll2} and \eqref{eq_imestimate}, we have
\begin{equation}\label{ll3}
 \|\eta_i\Omega(z_i)\| \leq  \begin{cases} n^{\epsilon}\phi_n\widehat\eta_i +n^{-1+5 \epsilon}, \ &\text{if } \widehat\eta_i \ge n^\epsilon \eta_l(\gamma_i)\\
{\eta_l(\gamma_i)\sqrt{\kappa_{\gamma_i}}}, \ &\text{if } \widehat\eta_i < n^\epsilon \eta_l(\gamma_i)
\end{cases} .
\end{equation}
Inserting the estimates from \eqref{eq_etaestimate}, \eqref{eq_671}, \eqref{eq_672}, \eqref{eq_imestimate} and  \eqref{ll3} into \eqref{eq_finalbound} together, we have
\begin{align}\label{final_estimate}
|\langle \ub_j, \wt\bxi_{i} \rangle|^2 &\, \lesssim \frac{n^{4\epsilon}\phi_n \widehat\eta_i + n^{\epsilon}\eta_l(\gamma_i)\sqrt{\kappa_{\gamma_i}} +n^{6\epsilon}n^{-1}}{ |1-d_i^2\mathcal{T}(z_i)|^2} \\
&\, \lesssim \frac{n^{6\epsilon+\delta}(\phi_n^3 + \eta_l(\gamma_i)\sqrt{\kappa_{\gamma_i}}+ n^{-1})}{ |1-d_i^2\mathcal{T}(z_i)|^2}\,,\nonumber
\end{align}
where in the last inequality we used that  for $\widehat\eta_i \ge n^\epsilon \eta_l(\gamma_i)$,
\begin{align}
\phi_n \widehat \eta_i &\lesssim n^{4\epsilon+\delta}\phi_n (\phi_n^2 + n^{-2/3}) \lesssim n^{4\epsilon+\delta} (\phi_n^3 + n^{-1})\,.
\end{align}
We still need to bound the denominator of \eqref{final_estimate} from below using Lemma \ref{lem_denominator}, which requires a lower bound on $\operatorname{Im}\mathcal{T} (z_i)$. When $i \notin \mathbb O^+$, {with (\ref{eq_estimm}) and (\ref{eq_etaestimate}) and similar approaches in Step II, we find that} $\operatorname{Im}\mathcal{T}(z_i)\asymp\operatorname{Im} m_{1c}(z_i) \gtrsim  \phi_n + \sqrt{\kappa_{\gamma_i}}$. Together with \eqref{final_estimate}, this concludes the proof of (\ref{eq_evebulka}) for $|\langle \ub_j, \wt\bxi_{i} \rangle|^2$ by choosing $\delta = 0$ in Lemma \ref{lem_denominator}. The proof of $|\langle \vb_j, \wt\bzeta_{i} \rangle|^2$ is based on the same steps and we omit details. 
On the other hand, when $i \in \mathbb O^+$ such that (\ref{eq_propohold}) holds, with \eqref{rigid_outXi0} and \eqref{eq_etaestimate} we can verify that 
$\wt\lambda_i \lesssim \theta(d_i) + n^{\epsilon+\wt\tau}(\phi_n^2+n^{-/3})$
and by \eqref{eq_thetadiff} we have
$\theta\left(\alpha+n^{\wt\tau+\epsilon}(\phi_n+n^{-1/3})\right)-\theta(d_i) \asymp n^{2\epsilon+2\wt\tau}(\phi_n^2+n^{-2/3}).$
By the two inequalities, we have
\begin{equation*}
\wt\lambda_i \lesssim \theta\left(\alpha+n^{\wt\tau+\epsilon}(\phi_n+n^{-1/3})\right),
\end{equation*}
and thus by \eqref{eq_iii1} we have
\begin{equation}
|\wt\lambda_i -\lambda_+|\lesssim n^{2\wt\tau+2\epsilon}(n^{-2/3}+\phi_n^2).
\end{equation}
Moreover, together with \eqref{eq_etaestimate} and \eqref{eq_estimm}, we conclude that
\begin{equation}
\operatorname{Im} \mathcal{T}(z_i) \asymp \operatorname{Im} m_{1c}(z_i) \ge n^{2\epsilon+2\wt\tau}(\phi_n +n^{-1/3})\ge n^{2\epsilon-2\wt\tau}(\phi_n + \sqrt{\kappa_{\gamma_i}}). 
\end{equation}
We can therefore conclude the proof of (\ref{eq_nonspikeeg2}) with \eqref{final_estimate} by letting $\delta=\wt\tau-\epsilon$ in Lemma \ref{lem_denominator}.

\section{Proof of Theorem \ref{thm_vector}}

We only show the detailed proof for the control of $| \langle \ub_i, {\mathcal P}_{\textbf{A}}\ub_j \rangle- \delta_{ij}\mathbbm 1(i\in \textbf{A}) a_1(d_i) |$ in \eqref{eq_spikedvector}. The control of the other term is the same and we omit details.
The proof is decomposed into three main parts. First, we prove Theorem \ref{thm_vector} under two stronger assumptions, which is stated in Proposition \ref{prop_outev0}. Then, we remove these two assumptions. 
We start with introducing these two assumptions.

\begin{assumption}\label{assum_nonoverlap}
(Non-overlapping condition). For some fixed constant $\widetilde{\tau}>0$, we assume that for all $i \in \textbf{A}$,
\begin{equation}\label{eq_assnonover}
\nu_i(\textbf{A}) \geq n^{\wt\tau}([\Delta(d_i)]^{-1}n^{-1/2}+\phi_n)= n^{\wt\tau}\psi_1(d_i).
\end{equation}
\end{assumption}
By Assumption \ref{assum_nonoverlap}, an outlier indexed by $\textbf{A}$ does not overlap with an outlier indexed by $\textbf{A}^c$. That is, when $d_i\neq d_j$, they are sufficiently separated if $i\in \mathbf{A}$ and $j\in \mathbf{A}^c$. However, outliers indexed by $\textbf{A}$ can overlap among themselves.

\begin{assumption}\label{assump_strong}
For some fixed small constant $0<\tau'<1/3$, we assume that for $i \in \textbf{A}$,
\begin{equation}\label{eq_assstrong}
   d_i -\alpha \geq n^{\tau'}(\phi_n+n^{-1/3})\,.
\end{equation}
\end{assumption}

The necessary argument to remove this assumption will be given in Section \ref{sec_nonoutliereve} after we complete the proof of Theorem \ref{thm_noneve}, since we need the delocalization bounds there.

\subsection{Proof of Theorem \ref{thm_vector} Under Stronger Assumptions
}\label{pf_thm_vector_strong}

We first provide the following proposition, which is needed for Theorem \ref{thm_vector}. 

\begin{proposition}\label{prop_outev0}
Grant the assumptions and notations in Theorem \ref{thm_vector}. Then under Assumptions  \ref{assum_nonoverlap} and \ref{assump_strong}, we have that for all $i,j = 1, \ldots, r$, 
\begin{align}
\label{eq_spikedvector0}
&|\langle \ub_i, {\mathcal{P}}_{\textbf{A}} \ub_j \rangle-  \delta_{ij}\mathbbm 1(i\in \textbf{A})a_1(d_i) | \\
 \prec &\, \mathbbm 1(i\in \textbf{A},j\in \textbf{A}) \left(\phi_n +    {n^{-1/2}(\Delta(d_i)\Delta(d_j)})^{-1/2}\right) \nonumber\\
&+  n^{-1} \left(\frac1{\nu_{i}(\textbf{A})} + \frac{\mathbbm 1(i\in \textbf{A})}{\Delta^2(d_i)} \right)\left(\frac1{\nu_{j}(\textbf{A})} + \frac{\mathbbm 1(j\in \textbf{A})}{\Delta^2(d_j)} \right) \nonumber\\
& + \phi_n^2  \left[ \left(\frac{\Delta^2(d_i)}{\nu_{i}(\textbf{A})}+1\right) \left(\frac1{\nu_{j}(\textbf{A})} + \frac{\mathbbm 1(j\in \textbf{A})}{\Delta^2(d_j)} \right)\wedge  \left(\frac{\Delta^2(d_j)}{\nu_{j}(\textbf{A})}+1\right) \left(\frac1{\nu_{i}(\textbf{A})} + \frac{\mathbbm 1(i\in \textbf{A})}{\Delta^2(d_j)} \right)\right] \nonumber \\
&+ \mathbbm 1(i\in \textbf{A},j\notin \textbf{A})  \frac{ \psi_1(d_i)\Delta^2(d_i)}{|d_j - d_i|} + \mathbbm 1(i\notin \textbf{A},j\in \textbf{A}) \frac{\psi_1(d_j)\Delta^2(d_j)}{|d_j - d_i|}\,. \nonumber
\end{align}
\end{proposition}

\begin{proof}

Consider the event space $\Xi$ 
introduced in the proof of Theorem \ref{thm_noneve} in Section \ref{sec_pf_evout},
where $\epsilon>0$ is a small positive constant satisfying $0 < \epsilon < \min\{\tau',\widetilde{\tau}\}/10$, and let $\omega <\tau'/2$. 
The rest of the proof is restricted to the event $\Xi$.

For $i \in \textbf{A},$ we define the contour
$\Gamma_i:=\partial B_{\rho_i} (d_i)$, where 
\begin{equation}\label{eq_radius}
\rho_i :=  c_i(\nu_i(\textbf{A})\wedge (d_i-\alpha)) =  c_i(\nu_i(\textbf{A})\wedge \Delta^2(d_i))
\end{equation}
for some small constant $0<c_i<1$. Define
\begin{align}\label{eq_contour}
    \Upsilon: = \cup_{i \in \bf{A}}B_{\rho_i} (d_i) \ \ \mbox{and} \ \ 
    \Gamma: = \cup_{i \in \bf{A}}\Gamma_i\,.
\end{align}
By choosing small enough $c_i$, we can assume that $\Upsilon \subset \mathbf{D}_2(\tau_2,\varsigma)$, where $\mathbf{D}_2$ is defined in Lemma \ref{s37_DY2019}.  The following lemma shows that (i) $\overline{\theta(\Upsilon)}$ is a subset of $S_{out}(\omega)$ so that we can use the estimates of \eqref{eq_2.5localaw}; (ii) $\partial \theta(\Upsilon)  = \theta(\Gamma)$ only encloses the outlier eigenvalues indexed by $\textbf{A}$. The proof will be provided in Section \ref{section_proof of technical lem1}.

\begin{lemma}\label{lem_contour}
Suppose that Assumptions  \ref{assum_nonoverlap} and \ref{assump_strong} hold. The set $\overline{\theta(\Upsilon)}$ lies in the spectral set $S_{out}(\omega)$ as long as the $c_i's$  are sufficiently small. Moreover, by selecting proper $\omega$ and $\epsilon$, we have $\{\widetilde{\lambda}_a\}_{a \in \textbf{A}} \subset \theta(\Upsilon)$ and all the other eigenvalues lie in the complement of $\overline{\theta(\Upsilon)}$ .
\end{lemma}

For $i = 1, \cdots, r$, let $\ub_i^e=(\ub_i^\top , \mathbf{0})^\top $ be the embedding of $\ub_i$ in $\mathbb{R}^{p+n}.$  By (\ref{green2}), it is easy to see that 
\begin{equation}\label{eq_gG}
\ub_i^\top  \ctG_1(z) \ub_j=(\ub_i^e)^\top  \widetilde{G}(z) \ub_j^e\,. 
\end{equation}
By expanding $\widetilde{G}(z)$ by (\ref{main_representation}), Cauchy's integral formula over $\theta(\Gamma)$ and Lemma \ref{lem_contour}, we have
\begin{equation}\label{eq_pa0}
-\frac{1}{2 \pi \ri}  \oint_{\theta(\Gamma)} \langle \ub_i, \ctG_1(z) \ub_j \rangle dz=\langle \ub_i, {\mathcal{P}}_\textbf{A} \ub_j \rangle\,.  
\end{equation} 
We first analyze the condition when  $ 1\leq  i, j \leq r$, and then extend the result to $1\leq i,j\leq p$ afterwards.

Assume $1\leq  i, j \leq r$. By \eqref{eq_gG} and \eqref{eq_pa0}, since $(\ub_i^e)^\top  \widetilde{G}(z) \ub_j^e=e_i^\top  \Ub^\top \widetilde{G}(z)\Ub  e_j$, where $e_j\in \mathbb{R}^{2r}$ is a unit vector with the $i$-th entry $1$, by (\ref{defn_greenrep2}),  we obtain that 
\begin{equation}\label{eq_pa1}
\langle \ub_i, {\mathcal{P}}_\textbf{A} \ub_j \rangle =\frac{1}{ 2\pi \ri d_i d_j } \oint_{\theta(\Gamma)} \left[( \Db^{-1}+\Ub^\top  G(z) \Ub )^{-1}\right]_{\bar{i}\,\bar{j}} \frac{dz}{z}\,,
\end{equation}
where $\bar{i}:=i+r$ and $\bar{j}:=j+r.$ 
By the resolvent expansion from \eqref{eq_resexp}, we have 
\begin{equation}\label{eq_decomrel}
\langle \ub_i, {\mathcal{P}}_\textbf{A} \ub_j \rangle =s_0^{ij}+s_1^{ij}+s_2^{ij},
\end{equation} 
where $s_\ell^{ij}, \ell=0,1,2,$ are respectively defined as
\begingroup
\allowbreak
\begin{align*}
& s_0^{ij}=\frac{1}{2 \pi \ri d_i d_j} \oint_{\theta(\Gamma)} \left[( \Db^{-1}+ \overline{\Pi}(z)  )^{-1}\right]_{\bar{i}\,\bar{j}} \frac{dz}{z}, \\
& s_1^{ij}=\frac{1}{2 \pi \ri d_i d_j} \oint_{\theta(\Gamma)} \left[ (\Db^{-1}+ \overline{\Pi}(z))^{-1} \Omega(z) (\Db^{-1}+ \overline{\Pi}(z))^{-1}  \right]_{\bar{i}\,\bar{j}} \frac{dz}{z} ,\\
& s_2^{ij}=\frac{1}{2 \pi \ri d_i d_j} \oint_{\theta(\Gamma)}\left[ ({\mathbf{D}^{-1}+ \overline{\Pi}(z)})^{-1}\Omega(z) (\mathbf{D}^{-1}+\overline{\Pi}(z))^{-1} \Omega(z) (\mathbf{D}^{-1}+\mathbf{U}^\top  G(z)\mathbf{U})^{-1} \right]_{\bar{i}\,\bar{j}}\frac{dz}{z}.
\end{align*}
\endgroup 
From now on, we write $s_\ell^{ij}$ as $s_\ell$ to simplify the notation. We start with $s_0$. By using the expansion in \eqref{biginverse},
we have
\begin{align}\label{estimate_s0finish}
s_0 & =  \delta_{ij}\frac{1}{2 \pi \ri d_i^2} \oint_{\theta(\Gamma)}  \frac{m_{1c}(z)}{\mathcal{T}(z)-d_{i}^{-2}}  dz  =  \delta_{ij}\frac{1}{2 \pi \ri  d_i^2} \oint_{\Gamma}  \frac{m_{1c}(\theta(\zeta))}{\zeta^{-2}-d_{i}^{-2}}  \theta'(\zeta) d\zeta \nonumber \\
&= -\delta_{ij}\frac{d_i m_{1c}(\theta(d_i))\theta'(d_i)}{ 2 } = \delta_{ij}\frac{m_{1c}(\theta(d_i))}{  d_i^2 \mathcal{T}'(\theta(d_i))} = \delta_{ij} a_1(d_i),
\end{align}
where in the second equality we use the change of variable $z = \theta(\zeta)$, in the third equality of use the residual theorem, and in the fourth equality we simply use $\theta'(d_i) = 2d_i^{-3} (\mathcal{T}^{-1})'(d_i^{-2})= 2d_i^{-3}/\mathcal{T'}(\theta(d_i))$.

Next, we control $s_1$. Similarly, for $s_1$ we can write it as
\begin{align}\label{defn_sij1}
s_1 & =   \frac{1}{2 \pi \ri d_i d_j} \oint_{\theta(\Gamma)} \left[ ( \Db^{-1}+ \overline{\Pi}(z))^{-1} \Omega(z) (\Db^{-1}+ \overline{\Pi}(z))^{-1}  \right]_{\bar{i}\bar{j}} \frac{dz}{z} \nonumber\\ & = 
\frac{d_i d_j}{2\pi\ri} \oint_{\Gamma} \frac{f(\zeta)\theta'(\zeta)\zeta^4}{(\zeta^{2}-d_i^{2})(\zeta^{2}-d_j^{2})}  d\zeta,
\end{align}
where $f(\zeta)=f_1(\zeta)+f_2(\zeta)$ and
\begin{equation*}
f_1(\zeta):=m_{1c}(\theta(\zeta))[\theta(\zeta)m_{1c}(\theta(\zeta))\Omega(\theta(\zeta))_{\bar{i}\bar{j}}+(-1)^{\bar{i}+j}\theta(\zeta)^{1/2}d_j^{-1}\Omega(\theta(\zeta))_{i\bar{j} }],
\end{equation*}
\begin{equation*}
f_2(\zeta):=d_i^{-1}[(-1)^{i+\bar{j}}\theta(\zeta)^{1/2}m_{1c}(\theta(\zeta))\Omega(\theta(\zeta))_{\bar{i}j}+(-1)^{i+j+\bar{i}+\bar{j}}d_i^{-1}\Omega(\theta(\zeta))_{ij}].
\end{equation*}
To continue to bound $s_1$, we prepare a bound. Denote 
\begin{equation}
f_{ij}(\zeta)=\frac{f(\zeta)\theta^{\prime}(\zeta) \zeta^4}{(d_i+\zeta)(d_j+\zeta)}.
\end{equation}
We know $f(\zeta)$ is holomorphic inside the contour $\Gamma$ by the assumption on $\epsilon$ and $\omega$ and $\theta'(\zeta)$ is holomorphic as well by Lemma \ref{s37_DY2019}. So, by Cauchy's differentiation formula, we have
\begin{equation}\label{fcauchydiff}
f_{ij}^{\prime}(\zeta)=\frac{1}{2 \pi i} \int_{\mathcal{C}} \frac{f_{ij}(\xi)}{(\xi-\zeta)^2} d\xi, 
\end{equation}
where $\mathcal{C}$ is the circle of radius $|\zeta-\alpha|/2$ centered at $\zeta.$
For $\zeta = E + \ri \eta \in \Gamma$, 
\begin{align} \label{eq_hjjbound}
    |f(\zeta)\theta'(\zeta)\zeta^4| 
    & \lesssim  n^{\epsilon}(\phi_n+n^{-1/2}|\kappa_E + \eta|^{-1/4})|\zeta - \mathcal{T}(\lambda_+)| \nonumber \\
    &\lesssim n^{\epsilon}(\phi_n+n^{-1/2}|\theta(\zeta) - \lambda_+|^{-1/4})|\zeta - \mathcal{T}(\lambda_+)| \nonumber \\
    & \lesssim n^{\epsilon}(\phi_n|\zeta - \mathcal{T}(\lambda_+)|+n^{-1/2}|\zeta - \mathcal{T}(\lambda_+)|^{1/2})\,,
\end{align}
where in the first inequality we use the bound of $\norm{\Omega}$ provided in (\ref{eq_2.5localaw}) and $\theta'(\zeta) \asymp|\zeta-\alpha|$ given by \eqref{eq_gcomplex},  in the second inequality we use $(\kappa_E+\eta)|_{z = \theta(\zeta)}\gtrsim |\theta(\zeta)-\lambda_+|$, and in the third inequality we use $|\theta(\zeta)-\lambda_+| \asymp|\zeta-\alpha|^2$ given in \eqref{eq_gcomplex}.  
As a consequence, we conclude that 
\begin{equation}\label{eq_hjjbound'}
|f'_{ij}(\zeta)| \leq C n^{\epsilon}(\phi_n+n^{-1/2}|\zeta - \mathcal{T}(\lambda_+)|^{-1/2})\,.
\end{equation}
We consider three different cases. 
\begin{enumerate}
\item [(i)] Suppose that $i\in \textbf{A}$ and $j\in \textbf{A}$. If $d_i \ne d_j$, we have
\begin{equation}
\label{estimate_s1finish1}
\begin{split}
|s_1| & \leq  \left|\frac{d_i d_j}{2\pi\ri}\right| \left|\oint_{\Gamma} \frac{f_{ij}(\zeta)}{(\zeta-d_i)(\zeta-d_j)}  d\zeta \right| \leq C\left|\oint_{\Gamma} \frac{1}{d_i-d_j}\left\{\frac{f_{ij}(\zeta)}{\zeta - d_i} - \frac{f_{ij}(\zeta)}{\zeta - d_j}\right\}d\zeta\right|\\
&= C\left| \frac{f_{ij}(d_i)-f_{ij}(d_j)}{d_i-d_j}\right| \leq \frac{C}{|d_i-d_j|}\left|\int_{d_i}^{d_j}|f'_{ij}(\zeta)|d\zeta \right| \\
& \underset{\eqref{eq_hjjbound'}}{\leq} Cn^\epsilon\left[\phi_n +\frac{n^{-1/2}}{{\Delta(d_i)+\Delta(d_j)}}\right] \leq  Cn^\epsilon\left[\phi_n +\frac{n^{-1/2}}{\sqrt{\Delta(d_i)\Delta(d_j)}}\right]
\end{split}
\end{equation}
for some constant $C>0$, where we use the arithmetic and geometric means in the last inequality. If $d_i = d_j$, then by the application of the residue's theorem we get a similar bound
\begin{equation}
\label{estimate_s1finish1_2}
\begin{split}
|s_1| & \leq  \left|\frac{d_i^2}{2\pi\ri}\right| \left|\oint_{\Gamma} \frac{f_{ii}(\zeta)}{(\zeta-d_i)^2}  d\zeta \right| \leq C\left| f'_{ij}(d_i)\right| \leq  Cn^\epsilon\left[\phi_n +\frac{n^{-1/2}}{\Delta(d_i)}\right]\,.
\end{split}
\end{equation} 

\item [(ii)] Suppose $i\in \textbf{A}$ and $j\notin \textbf{A}$. Then we get from 
\eqref{eq_hjjbound} that
\begin{equation}\label{estimate_s1finish2}
|s_1|\le C\frac{|f_{ij}(d_i)|}{|d_i-d_j|} \leq Cn^\epsilon \frac{n^{-1/2 }\Delta(d_i) + \phi_n \Delta^2(d_i)}{|d_i-d_j|} = C n^{\epsilon} \frac{\psi_1(d_i)\Delta^2(d_i)}{|d_i-d_j|}.
\end{equation}
We have a similar estimate if $i\notin \textbf{A}$ and $j\in \textbf{A}$. 

\item [(iii)] If $i\notin \textbf{A}$ and $j \notin \textbf{A}$, we have $s_1=0$ by Cauchy's residue theorem since there is no poles inside the contour. We thus conclude the bound of $s_1$.
\end{enumerate}

It remains to estimate the second order error $s_2$. Recall (\ref{eq_contour}) that $\Gamma=\bigcup_{i\in \textbf{A}} \Gamma_i$
We have the following basic estimates on each of these components, whose proof is given in Section \ref{section_pflem}. 
\begin{lemma}\label{lem_distance}
For any $k \in \textbf{A}$, $1\leq \ell \leq r$ and $\zeta \in \Gamma_k$, we have 
\begin{equation}\label{zeta_rhoi}
    |\zeta -d_\ell| \asymp \rho_k + |d_k - d_\ell|\,.
\end{equation}
\end{lemma}

Now we finish the estimate of $s_2$. By a trivial bound, we have
\begin{align}\label{eq_s2bd}
|s_2| \leq  \frac{1}{2 \pi \ri d_i d_j} \oint_{\theta(\Gamma)}\left|\left( ({\mathbf{D}^{-1}+ \overline{\Pi}(z)})^{-1}\Omega(z) (\mathbf{D}^{-1}+\overline{\Pi}(z))^{-1} \Omega(z) ({\mathbf{D}^{-1}+\mathbf{U}^\top  G(z)\mathbf{U}})^{-1} \right)_{\bar{i}\bar{j}}\frac{1}{z}\right|dz.
\end{align}
By the bound of $\norm{\Omega}$ in (\ref{eq_2.5localaw}), $(\kappa_E+\eta)|_{z = \theta(\zeta)}\gtrsim |\theta(\zeta)-\lambda_+|$ and $|\theta(\zeta)-\lambda_+| \asymp|\zeta-\alpha|^2$ given by \eqref{eq_gcomplex}, the entries of $({\mathbf{D}^{-1}+ \overline{\Pi}(z)})^{-1}$ in \eqref{biginverse} together with \eqref{eq_estimm}, and a simple change of variable, \eqref{eq_s2bd} is bounded by
\begin{align}
&  C \oint_{\Gamma} \frac{n^{2\epsilon}(\phi_n^2+n^{-1}|\zeta-\alpha|^{-1})}{|\zeta - d_i||\zeta-d_j|}   \times \left\| \left( \mathbf{D}^{-1}+ \mathbf{U}^\top  G(\theta(\zeta)) \mathbf{U}\right)^{-1} \right\| |\theta'(\zeta)|d \zeta
\end{align}
for some constant $C>0$, which is further bounded by
\begin{align}
C \oint_{\Gamma} \frac{n^{2\epsilon}(\phi_n^2|\zeta-\alpha|+n^{-1})}{|\zeta-d_i||\zeta-d_j|}  \left\| \left( \mathbf{D}^{-1}+ \mathbf{U}^\top  G(\theta(\zeta)) \mathbf{U}\right)^{-1} \right\|  d \zeta\label{eq_s3bound}
\end{align}
by using the fact that $|\theta'(\zeta)|\asymp|\zeta-\alpha|$ from \eqref{eq_gcomplex}.
To continue to bound  \eqref{eq_s3bound}, we use the same method as that for \eqref{eq_b3}. Assume $\zeta \in \Gamma_k$, we can bound $\norm{ \Omega(\theta(\zeta))} $ using (\ref{eq_2.5localaw}) and obtain
\begin{equation}\label{eq_localcontrol}
\norm{\Omega(\theta(\zeta))} \lesssim  n^\epsilon \left[\phi_n +n^{-1/2 } \Delta(d_k)^{-1} \right].
\end{equation}
By \eqref{eq_iii1} and \eqref{eq_gcomplex}, we have for any $1\le \ell \le r$, 
\begin{align}\label{sij2bound2}
|\mathcal{T}(\theta(\zeta))-d_\ell^{-2}|  & = |\zeta^{-2} - d_\ell^{-2}| \gtrsim  |\zeta-d_\ell| \geq |\zeta-d_k| \nonumber \\ & = \rho_k \geq n^{\wt\tau} (\phi_n+n^{-1/2}\Delta(d_k)^{-1})\,,
\end{align}
where the last bound comes from Assumption \eqref{assum_nonoverlap}.
By \eqref{eq_localcontrol} and \eqref{sij2bound2}, we have $ |\mathcal{T}(\theta(\zeta))-d_\ell^{-2}| \gg \norm{\Omega(\theta(\zeta))}$ since we have assumed $\wt \tau>\epsilon$. Hence, as we derive the bound in \eqref{eq_b3}, by the resolvent expansion and \eqref{biginverse}, we have
\begin{equation} \label{sij2matrixnormbound}
 \left\| \left( \mathbf{D}^{-1}+ \mathbf{U}^\top  G(\theta(\zeta)) \mathbf{U}\right)^{-1} \right\| \lesssim 1/|\mathcal{T}(\theta(\zeta))-d_\ell^{-2}| \lesssim \frac{1}{\rho_k},
\end{equation}
where the last bound comes from \eqref{sij2bound2}. Together with Assumptions  \ref{assum_nonoverlap} and \ref{assump_strong}, and the fact that $\Gamma_k$ has length $2\pi \rho_k$, we obtain 
\begin{align}\label{estimate_s2}
|s_2| &\,\lesssim \sum_{k \in \textbf{A}} \sup_{\zeta \in \Gamma_k}  \frac{n^{-1+2\epsilon}+n^{2\epsilon}\phi_n^2\Delta^2(d_k) }{|\zeta - d_i||\zeta - d_j| } \\
&\,\asymp \sum_{k \in \textbf{A}}   \frac{n^{-1+2\epsilon}+n^{2\epsilon}\phi_n^2\Delta^2(d_k) }{(\rho_k + |d_k - d_i|)(\rho_k + |d_k - d_j|) }, \nonumber
\end{align} 
where we apply Lemma \ref{lem_distance} in the last inequality.
Finally, we bound  \eqref{estimate_s2}. First, by triangle inequality we have
$$ 
\Delta^2(d_k) = |d_k - \alpha|   \leq |d_i-\alpha|+|d_k-d_i| = \Delta^2(d_i) +|d_k - d_i|.
$$
For $i\notin \textbf{A}$, $k\in \textbf{A}$, we have
\begin{align*}
\frac{1}{(\rho_k + |d_k - d_i|) }  \le \frac{1}{|d_k - d_i| } \le \frac{1}{\nu_i(\textbf{A})} \,.
\end{align*}
For $i,k\in \textbf{A}$, by triangle inequality, we have $$\rho_k + |d_k - d_i|\geq\rho_i.$$ Then we have 
$$ 
\frac{1}{(\rho_k + |d_k - d_i|) }  \leq  \frac{1}{\rho_i}\lesssim \frac{1}{\nu_{i}(\textbf{A})} + \frac{1}{\Delta(d_i)^2}\,,
$$
where the last inequality is by the definition of $\rho_i$.
Plugging the above estimates into \eqref{estimate_s2}, we get that 
\begin{align}\label{estimate_s2finish}
|s_2| \lesssim&\, n^{-1+2\epsilon}\left(\frac1{\nu_{i}(\textbf{A})} + \frac{\mathbbm 1(i\in \textbf{A})}{\Delta(d_i)^2} \right)\left(\frac1{\nu_{j}(\textbf{A})} + \frac{\mathbbm 1(j\in \textbf{A})}{\Delta(d_j)^2} \right)\\
&+ n^{2\epsilon}\phi_n^2  \left[\left(\frac{\Delta(d_i)^2}{\nu_{i}(\textbf{A})} +1 \right)\left(\frac1{\nu_{j}(\textbf{A})} + \frac{\mathbbm 1(j\in \textbf{A})}{\Delta(d_j)^2} \right)\wedge  \left(\frac{\Delta(d_j)^2}{\nu_{j}(\textbf{A})} +1 \right)\left(\frac1{\nu_{i}(\textbf{A})} + \frac{\mathbbm 1(i\in \textbf{A})}{\Delta(d_i)^2} \right) \right] \,.\nonumber
\end{align}
Combining \eqref{estimate_s0finish} for the bound of $s_0$, \eqref{estimate_s1finish1}, \eqref{estimate_s1finish1_2} {and} \eqref{estimate_s1finish2} for the bound of $s_1$, and \eqref{estimate_s2finish} for the bound of $s_2$, we obtain (\ref{eq_spikedvector0}) for $1\le i,j \le r$ since $\epsilon$ can be arbitrarily small.

Finally, we extend the above results to $1\leq i,j\leq p$.
Define $\mathcal R:=\{1,\cdots, r\}\cup \{i,j\}$. Then we define a perturbed model with SVD as
\begin{equation}\nonumber S_{\wt\epsilon} := \sum_{i\in \mathcal R} \wt d_i \ub_i \vb_i^\top ,
\end{equation}
where $\wt \epsilon>0$ and 
$\wt d_{k}= d_k$ when  $1\le k\le r$ and $\wt d_k=
\wt\epsilon$ when $k>r$ and $k\in \mathcal R$.
Then all the previous proof goes through for the perturbed model. Taking $\wt\epsilon\downarrow 0$ and using continuity, we get (\ref{eq_spikedvector0}) for general $i,j \in \{1,\cdots, p\}$.

\end{proof}

Note that Proposition \ref{prop_outev0} is essentially Theorem \ref{thm_vector} with stronger assumptions.
Thus, to finish the proof of Theorem \ref{thm_vector}, we need to remove Assumptions \ref{assum_nonoverlap} and \ref{assump_strong}, and this is done in the following two subsections.

\subsection{Removing Assumption \ref{assum_nonoverlap}}\label{sec_remove}

\begin{proof}
Recall the constants $\tau'$ in Assumption \ref{assump_strong} and $\wt\tau$ in Assumption \ref{assum_nonoverlap} and set $\wt\tau<\tau'/4$. Recall \eqref{dfn_indexset}, we write an index set $\mathbb O_{\tau'/2} = \{i:d_i - \alpha \geq n^{\tau'/2}(\phi_n+n^{-1/3})\} $.
We say that $\fa, \fb\in \mathbb{O}_{\tau'/2}$, $\fa \neq \fb$, overlap if 
\begin{align}\label{eq_overlap}
|d_\fa - d_\fb| \leq  n^{\wt\tau}(\psi_1(d_\fa)\vee \psi_1(d_\fb))\,.
\end{align}
For $\textbf{A}$ satisfying Assumption \ref{assump_strong}, we define sets $L_1(\textbf{A})$, $L_2(\textbf{A})\subset \mathbb{O}_{\tau'/2}$, such that $L_1(\textbf{A})\subset \textbf{A}\subset L_2(\textbf{A})$. 
$L_1(\textbf{A})$ is constructed by successively removing $k \in \textbf{A}$, such that $k$ overlaps with an index of $\textbf{A}^c$. This process is repeated until no such $k$ exists. In other words,
$L_1(\textbf{A})$ is the largest subset of $\textbf{A}$ that do not overlap with $L_1(\textbf{A})^c$. on the other hand, $L_2(\textbf{A})$ is constructed by successively adding $k \in \mathbb{O}_{\tau'/2}\backslash \mathbf{A}$ into $\mathbf{A}$, where $k$ overlaps with an index
of $\mathbf{A}$. This process is repeated until no such k exists. In other words, $L_2(\textbf{A})$ is the smallest subset of $\mathbb{O}_{\tau'/2}$ that do not overlap with $L_2(\textbf{A})^c$. See Figure $5.2$ in \cite{principal} for an illustration of construction of $L_1(\textbf{A})$ and $L_2(\textbf{A})$. It is easy to see that $L_1(\textbf{A})$ and $L_2(\textbf{A})$ exist and are unique. The main reason for defining these two sets is that (\ref{eq_spikedvector0}) now holds under Assumption \ref{assump_strong} with the parameter sets as $(\tau'/2,L_1(\textbf{A}))$ or $(\tau'/2,L_2(\textbf{A}))$. 
Now we are ready to prove \eqref{eq_spikedvector}. There are four cases, (a)-(d), to consider.

\begin{enumerate}
\item [(a)] $i,j\notin \textbf{A}$ and $i=j$. If $i\notin L_2(\textbf{A})$, then using $r$ is bounded, we see that $\nu_{i}(\textbf{A})\asymp \nu_{i}(L_2(\textbf{A}))$. Then using Proposition \ref{prop_outev0} and the definition of $\psi_1$, we have
\begin{align}\label{proofout_a1}
 &\langle \ub_i, {\mathcal{P}}_\textbf{A}\ub_i\rangle \le\langle \ub_i, {\mathcal{P}}_{L_2(\textbf{A})}\ub_i\rangle \\
 \underset{\eqref{eq_spikedvector0}}{\prec} & \frac1{n\nu^2_{i}(L_2(\textbf{A}))} + \phi_n^2 \frac{\Delta^2(d_i)+ \nu_{i}(L_2(\textbf{A}))}{\nu_{i}^2(L_2(\textbf{A}))}\nonumber\\
\lesssim &\, \frac{\psi_1^2(d_i)\Delta^2(d_i)}{\nu^2_{i}(\textbf{A}) } +  \frac{\phi_n^2}{\nu_{i}(\textbf{A})} \,.\nonumber
\end{align}
If $i\in L_2(\textbf{A})$, which implies that $L_2(\textbf{A})\setminus \textbf{A}\ne \emptyset$. Since $\textbf{A}$ overlaps with $L_2(\textbf{A})$, this gives that 
\begin{equation}\label{proofout_a2}
\nu_{i}(\textbf{A})\underset{\eqref{eq_overlap}}{\leq} \left[\Delta(d_i)\right]^{-1} n^{-1/2+\wt\tau} + n^{\wt\tau}\phi_n  = n^{\wt\tau}\psi_1(d_i) 
\underset{\eqref{eq_assnonover}}{\le} \nu_{i}(L_2(\textbf{A})),
\end{equation}
Then Proposition \ref{prop_outev0} gives that
\begin{equation}\label{proofout_a3}
\begin{split}
&\left|\langle \ub_i, {\mathcal{P}}_\textbf{A}\ub_i\rangle- a_1(d_i) \right|\le \langle \ub_i, {\mathcal{P}}_{L_2(\textbf{A})}\ub_i\rangle  +  a_1(d_i) \\
\underset{\eqref{eq_spikedvector0}}{\prec}&\, \frac{m_{1c}(\theta(d_i))}{d_i^2\mathcal T'(\theta(d_i))} + \phi_n +   \frac{1}{n^{1/2} \Delta(d_i)}+ \frac1{n\nu^2_{i}(L_2(\textbf{A}))} \\
&+ \frac{\phi_n^2 \Delta^2(d_i)}{\nu_{i}^2(L_2(\textbf{A}))} + \frac{\phi_n^2}{\Delta^2(d_i)}
{\prec}
 \Delta^2(d_i) \underset{\eqref{proofout_a2}}{\le} \frac{n^{2\wt\tau}\psi_1^2(d_i)\Delta^2(d_i)}{\nu^2_{i}(\textbf{A})} ,
\end{split}
\end{equation}
where in the second step we used \eqref{eq_s36_3} such that $1/\mathcal T'(\theta(d_i)) \asymp (d_i - \alpha) = \Delta^2(d_i)$, and \eqref{proofout_a2} such that $ \frac1{n\nu^2_{i}(L_2(\textbf{A}))} + \frac{\phi_n^2 \Delta^2(d_i)}{\nu_{i}^2(L_2(\textbf{A}))} \leq \Delta^2(d_i)$, and \eqref{eq_assstrong} for the rest terms. From \eqref{proofout_a1} and \eqref{proofout_a3}, we conclude
\begin{equation}\label{proofout_a4}
\begin{split}
\left|\langle \ub_i, {\mathcal{P}}_\textbf{A}\ub_i\rangle- a_1(d_i) \right| \prec \frac{n^{2\wt\tau}\psi_1^2(d_i)\Delta^2(d_i)}{\nu^2_{i}(\textbf{A})}, \quad i \notin \textbf{A}.
\end{split}
\end{equation}

\item [(b)] $i,j\in \textbf{A}$ and $i=j$. We first consider the case $i\in L_1( \textbf{A})$. We can write
\begin{equation}\label{proofout_b1}
\langle \ub_i, {\mathcal{P}}_ \textbf{A}\ub_i\rangle = \langle \ub_i, {\mathcal{P}}_{L_1( \textbf{A})}\ub_i\rangle +  \langle \ub_i, {\mathcal{P}}_{ \textbf{A}\setminus L_1( \textbf{A})}\ub_i\rangle .
\end{equation}
Using \eqref{eq_spikedvector0} and the fact that $\nu_{i}( \textbf{A})\asymp \nu_{i}(L_1( \textbf{A}))$ (because $i$ do not overlap with either $\textbf{A}^c$ or $L_1(\textbf{A})^c$), we can estimate the first term as 
\begin{align}
&\left| \langle \ub_i, {\mathcal{P}}_{L_1( \textbf{A})}\ub_i\rangle- a_1(d_i) \right| \nonumber\\
\prec &\,\psi_1(d_i) + \psi^2_1(d_i) \Delta^2(d_i) \left(\frac1{\nu^2_{i}( L_1(\textbf{A}))} + \frac{1}{\Delta^4(d_i)} \right)\label{proofout_b2} \\
\prec  &\, \psi_1(d_i) + \psi^2_1(d_i) \Delta^2(d_i) \left(\frac1{\nu^2_{i}( \textbf{A})} + \frac{1}{\Delta^4(d_i)} \right)
\prec \psi_1(d_i) +  \frac{\psi^2_1(d_i) \Delta^2(d_i) }{\nu^2_{i}( \textbf{A})} \,,\nonumber
\end{align}
where we used that $\nu_i(\textbf A)\le \Delta^2(d_i) $ (by Assumption \ref{assump_strong}) in the last step. For the second term in \eqref{proofout_b1}, it suffices to assume that $\textbf{A}\setminus L_1(\textbf{A})\ne \emptyset$ (otherwise it is equal to zero). Then we observe that $\nu_{i}(\textbf{A}) \asymp \nu_{i}(\textbf{A}\setminus L_1(\textbf{A}))$. By \eqref{eq_spikedvector0}, similar to \eqref{proofout_a1} with $\textbf{A}$ replaced by $\textbf{A}\setminus L_1(\textbf{A})$, we obtain that
\begin{equation}\label{proofout_b3}
\langle \ub_i, {\mathcal{P}}_{\textbf{A}\setminus L_1(\textbf{A})}\ub_i\rangle \prec  \frac{\psi_1^2(d_i)\Delta^2(d_i)}{\nu^2_{i}(\textbf{A}) } +  \frac{\phi_n^2}{\nu_{i}(\textbf{A})} \prec  \phi_n +  \frac{\psi^2_1(d_i) \Delta^2(d_i) }{\nu^2_{i}(\textbf{A})}\,,
\end{equation}
where the last step comes from the fact that $i$ does not overlap with its complement since $i \in L_1(\textbf{A})$.
Next, for the case $i\notin L_1(\textbf{A})$, this implies $\textbf{A}\setminus L_1(\textbf{A})\ne \emptyset$, $\textbf{A}$ overlaps with its complement, and thus  $L_2(\textbf{A})\setminus \textbf{A}\ne \emptyset$. With these conditions, \eqref{proofout_a2} holds by the same arguments, and with similar steps in deriving \eqref{proofout_a3}, we get
\begin{equation}\label{proofout_b4}
\left|\langle \ub_i, {\mathcal{P}}_\textbf{A}\ub_i\rangle- a_1(d_i) \right|\le \langle \ub_i, {\mathcal{P}}_{L_2(\textbf{A})}\ub_i\rangle  +  a_1(d_i) \prec \frac{n^{2\wt\tau}\psi_1^2(d_i)\Delta^2(d_i)}{\nu^2_{i}(\textbf{A})}.
\end{equation}
Combining \eqref{proofout_a4} and \eqref{proofout_b2}-\eqref{proofout_b4}, we conclude that 
\begin{equation}\label{conclude_i=j}
\begin{split}
&\left| \langle \ub_i, {\mathcal{P}}_\textbf{A}\ub_i\rangle- \mathbbm 1(i\in \textbf{A})a_1(d_i) \right| \prec n^{2\wt\tau}R(i, \textbf{A}).
\end{split}
\end{equation}
This concludes \eqref{eq_spikedvector} for the $i=j$ case since $\wt \tau$ can be chosen arbitrarily small. 

\item [(c)] $i\ne j$ and $i\notin \textbf{A}$ or $j\notin  \textbf{A}$. Using Cauchy-Schwarz inequality such that
\begin{equation}\label{CS new}
\left| \langle \ub_i, {\mathcal{P}}_ \textbf{A}\ub_j\rangle\right|^2 \le \langle \ub_i, {\mathcal{P}}_ \textbf{A}\ub_i\rangle\langle \ub_j, {\mathcal{P}}_ \textbf{A}\ub_j\rangle ,
\end{equation}
and combine with \eqref{proofout_a4} and \eqref{conclude_i=j},
we find that in this case \eqref{eq_spikedvector0} holds since $\wt \tau$ can be chosen arbitrarily small. 

\item [(d)] $i\ne j$ and $i,j\in \textbf{A}$. Our goal is to prove that  
\begin{equation}\label{eq_pfgoal}
\left| \langle \ub_i, {\mathcal{P}}_\textbf{A}\ub_j\rangle \right|\prec n^{2\wt\tau} \left[\psi_{1}^{1/2}(d_i)+  \frac{\psi_{1}(d_i)\Delta(d_i) }{\nu_{i}(\textbf{A})} \right] \left[\psi_{1}^{1/2}(d_j)+  \frac{\psi_{1}(d_j)\Delta(d_j) }{\nu_{j}(\textbf{A})}  \right]\,.
\end{equation}  
We again split ${\mathcal{P}}_\textbf{A}$ into
\begin{equation}\label{proofout_d1}
\langle \ub_i, {\mathcal{P}}_\textbf{A}\ub_j\rangle = \langle \ub_i, {\mathcal{P}}_{L_1(\textbf{A})}\ub_j\rangle +  \langle \ub_i, {\mathcal{P}}_{\textbf{A}\setminus L_1(\textbf{A})}\ub_j\rangle .
\end{equation}
There are four cases: (i) $i,j\in L_1(\textbf{A})$; (ii) $i\in L_1(\textbf{A})$ and $j\notin L_1(\textbf{A})$; (iii) $i\notin L_1(\textbf{A})$ and $j\in L_1(\textbf{A})$; (iv) $i,j\notin L_1(\textbf{A})$. 
In case (i), we can bound the first term in \eqref{proofout_d1} using Proposition \ref{prop_outev0} and the estimates that $\nu_{i}(\textbf{A})\asymp \nu_{i}(L_1(\textbf{A}))$ and $\nu_{j}(\textbf{A})\asymp \nu_{j}(L_1(\textbf{A}))$. The second term in \eqref{proofout_d1} can be bounded as in case (c) above (with $\textbf{A}$ replaced by $\textbf{A}\setminus L_1(\textbf{A})$) together with the estimates $\phi_n\le \nu_{i}(\textbf{A}) \le C\nu_{i}(\textbf{A}\setminus L_1(\textbf{A}))$ and $\phi_n\le \nu_{j}(\textbf{A}) \le C\nu_{j}(\textbf{A}\setminus L_1(\textbf{A}))$.
In case (ii), we have
\begin{equation}\label{revisionadd}
\begin{split}
&\nu_{i}(\textbf{A})\asymp \nu_{i}(L_1(\textbf{A})) \asymp \nu_{i}(\textbf A\setminus L_1(\textbf{A})),\quad \nu_{i}(\textbf{A}) \le C|d_i-d_j|, \\
 &\nu_{j}(\textbf{A})\lesssim \nu_{j}(\textbf{A}\setminus L_1(\textbf{A})) \lesssim  n^{\wt\tau}\psi_1(d_j) 
 \lesssim \nu_{j}(L_1(\textbf{A})).
 \end{split}
 \end{equation}
 Then with Proposition \ref{prop_outev0}, we can bound the first term in \eqref{proofout_d1} as
\begin{align*}
&\left| \langle \ub_i, {\mathcal{P}}_{L_1(\textbf{A})}\ub_j\rangle\right| \\
\prec &\,  \frac{1}{n\nu_{i}(L_1(\textbf{A}))\nu_{j}(L_1(\textbf{A}))} + \frac{1 }{n\nu_{j}(L_1(\textbf{A}))\Delta^2(d_i)} \\
& + \phi_n^2 \Delta(d_i)\Delta(d_j) \left[\left(\frac{1}{\nu_{i}(L_1(\textbf{A})) } +\frac1{\Delta^2(d_i)}   \right) \left(\frac{1}{\nu_{j}(L_1(\textbf{A})) } +\frac1{\Delta^2(d_j)}   \right) \right] \\
& + \frac{\psi_1(d_i)\Delta^2(d_i)}{|d_i-d_j|}\\
\lesssim &\, \left[\psi_{1}^{1/2}(d_i)+ \frac{\psi_{1}(d_i)\Delta(d_i) }{\nu_{i}(\textbf{A})}   \right]  \left[\psi_{1}^{1/2}(d_j)+\frac{\psi_{1}(d_j)\Delta(d_j) }{\nu_{j}(\textbf{A})} \right] + \frac{\psi_1(d_i)\Delta^2(d_i)}{|d_i-d_j|}\,.
\end{align*}
For the last term, we first assume that $d_j\le d_i$ and $ d_i -\alpha \le 2|d_i-d_j|$. Then
\begin{align*}
  \frac{\psi_1(d_i)\Delta^2(d_i)}{|d_i - d_j|}\leq 2\psi_1(d_i) \le \sqrt{\psi_1(d_i)\psi_1(d_j)}. 
 \end{align*}
On the other hand, if $d_j\ge d_i$ or $ d_i -\alpha \ge 2|d_i-d_j|$, we have $\Delta(d_i) \lesssim \Delta(d_j)$. Hence using \eqref{revisionadd}, we get
\begin{align*}
  \frac{\psi_1(d_i)\Delta^2(d_i)}{|d_i-d_j|}\lesssim n^{\wt\tau} \frac{\psi_1(d_i)\Delta(d_i)\psi_1(d_j)\Delta(d_j) }{\nu_{i}(\textbf{A})\nu_{j}(\textbf{A})}. 
 \end{align*}
 The above estimates show that $ | \langle \ub_i, {\mathcal{P}}_{L_1(\textbf{A})}\ub_j\rangle | $ can be bounded by the right-hand side of \eqref{eq_pfgoal}. 
 The second term in \eqref{proofout_d1} can be bounded as in case (c) above (with $\textbf{A}$ replaced by $\textbf{A}\setminus L_1(\textbf{A})$) together with the estimates in \eqref{revisionadd} such that
$$\nu_{i}(\textbf{A})\asymp \nu_{i}(\textbf{A}\setminus L_1(\textbf{A}))\gtrsim n^{\wt\tau}\phi_n, \ \ \nu_{j}(\textbf{A}) \lesssim \nu_{j}(\textbf{A}\setminus L_1(\textbf{A})) \lesssim n^{\wt\tau}\psi_1(d_j). $$
Then we get that
\begin{equation}\nonumber
\begin{split}
&\left| \langle \ub_i, {\mathcal{P}}_{\textbf{A}\setminus L_1(\textbf{A})}\ub_j\rangle\right|  \\
\prec&\, n^{2\wt\tau} \left[  \frac{\phi_n}{\nu^{1/2}_{i}(\textbf{A}\setminus L_1(\textbf{A}))}+\frac{\psi_{1}(d_i)\Delta(d_i) }{\nu_{i}(\textbf{A}\setminus L_1(\textbf{A}))} \right] \left[\psi_{1}^{1/2}(d_j)+  \frac{\psi_{1}(d_j)\Delta(d_j) }{\nu_{j}(\textbf{A}\setminus L_1(\textbf{A}))}  \right] \\
\prec &\, n^{2\wt\tau} \left[ \psi_{1}^{1/2}(d_i) +\frac{\psi_{1}(d_i)\Delta(d_i) }{\nu_{i}(\textbf{A} )} \right]\left[\psi_{1}^{1/2}(d_j)+  \frac{\psi_{1}(d_j)\Delta(d_j) }{\nu_{j}(\textbf{A})} \right].
\end{split}
\end{equation}
This concludes the proof of \eqref{eq_pfgoal} for case (ii). The case (iii) can be handled in the same way as case (ii) by interchanging $i$ and $j$. Finally, we deal with case (iv). Again we again split ${\mathcal{P}}_{\textbf A}$ as \eqref{proofout_d1}. For the first term in \eqref{proofout_d1}, we have
$$\nu_{i}(\textbf{A})\lesssim \nu_{i}(L_1(\textbf{A})), \quad \nu_{i}(L_1(\textbf{A})) \gtrsim \psi_1(d_i), $$ 
and similar estimates for the $j$ case. Then using Proposition \ref{prop_outev0}, we can obtain that 
\begin{align}
&\left| \langle \ub_i, {\mathcal{P}}_{L_1(\textbf{A})}\ub_j\rangle \right| \prec  \frac{1}{n\nu_{i}(L_1(\textbf{A}))\nu_{j}(L_1(\textbf{A}))} \nonumber \\
&\qquad+  \phi_n^2  \left[\left(\frac{\Delta^2(d_i)}{\nu_{i}(L_1(\textbf{A}))} +1 \right) \frac1{\nu_{j}(L_1(\textbf{A}))} \right] \wedge  \left[\left(\frac{\Delta^2(d_j)}{\nu_{j}(L_1(\textbf{A}))} +1 \right) \frac1{\nu_{i}(L_1(\textbf{A}))} \right]\nonumber\\ 
\lesssim  &\, \frac{\psi_1(d_i)\psi_1(d_j)\Delta(d_i) \Delta(d_j)}{ \sqrt{\nu_{i}(L_1(\textbf{A}))\nu_{j}(L_1(\textbf{A}))}}   \left[\left(\frac{1}{\nu_{i}(L_1(\textbf{A}))} +\frac1{\Delta^2(d_i)} \right) \left(\frac{1}{\nu_{j}(L_1(\textbf{A}))} +\frac1{\Delta^2(d_j)} \right)   \right] ^{1/2}\nonumber \\
\lesssim &\, \left[ \psi_{1}^{1/2}(d_i)+ \frac{\psi_{1}(d_i)\Delta(d_i) }{\nu_{i}(\textbf{A})} \right]  \left[ \psi_{1}^{1/2}(d_j)+\frac{\psi_{1}(d_j)\Delta(d_j) }{\nu_{j}(\textbf{A})}  \right] .\nonumber
\end{align}
For the second term in \eqref{proofout_d1}, we use the estimate
$$\nu_{i}(\textbf{A})\lesssim \nu_{i}(\textbf{A}\setminus L_1(\textbf{A}))\lesssim n^{\wt\tau}\psi_1(d_i) $$
and the same discussion in case (b) to get that
\begin{align*}
\langle \ub_i, {\mathcal{P}}_{\textbf{A}\setminus L_1(\textbf{A})}\ub_i\rangle & \prec \Delta^2 (d_i)+\psi_{1}(d_i)+ n^{2\wt\tau}\left( \psi_{1} (d_i)+\frac{\psi_{1}^2(d_i)\Delta^2(d_i) }{\nu^2_{i}(\textbf{A}\setminus L_1(\textbf{A}))}  \right)\\
&\lesssim  n^{2\wt\tau}\left(\psi_{1} (d_i)+\frac{\psi_{1}^2(d_i)\Delta^2(d_i) }{\nu^2_{i}(\textbf{A} )} \right).
\end{align*}
A similar estimate holds for $\langle \ub_j, {\mathcal{P}}_{\textbf{A}\setminus L_1(\textbf{A})}\ub_j\rangle$. Then we conclude that 
\begin{align*}
\left| \langle \ub_i, {\mathcal{P}}_\textbf{A}\ub_j\rangle\right| &\le \langle \ub_i, {\mathcal{P}}_\textbf{A}\ub_i\rangle^{1/2}\langle \ub_j, {\mathcal{P}}_\textbf{A}\ub_j\rangle^{1/2}  \\
& \prec n^{2\wt\tau}  \left[\psi_{1}^{1/2}(d_i)+  \frac{\psi_{1}(d_i)\Delta(d_i) }{\nu_{i}(\textbf{A})} \right] \left[\psi_{1}^{1/2}(d_j)+  \frac{\psi_{1}(d_j)\Delta_1(d_j) }{\nu_{j}(\textbf{A})} \right].
\end{align*}
This proves \eqref{eq_pfgoal} for case (iv), and hence concludes the proof for case (d).

\end{enumerate}

Combining cases (c) and (d), we conclude \eqref{eq_spikedvector} for the $i\ne j$ case since $\wt \tau$ can be chosen arbitrarily small. This finishes the proof of Theorem \ref{thm_vector} under Assumption \ref{assump_strong} together with \eqref{conclude_i=j}.
\end{proof}

\subsection{Removing Assumption \ref{assump_strong}}\label{sec_nonoutliereve}
By (iv) of Assumption \ref{assum_main}, for all $i \in \textbf{A}\subset \mathbb{O}_+$, we have
\begin{equation}\label{eq_spikeweak}
\Delta^2(d_i) = d_i -\alpha \geq \phi_n+n^{-1/3}.
\end{equation}
Recall in Assumption \ref{assump_strong} we had a stronger condition that $d_i - \alpha \geq n^{\tau'}(\phi_n + n^{-1/3})$ for $0<\tau'<1/3$, and now we remove it. 

\begin{proof}
The proof is devoted to showing that \eqref{eq_spikedvector} holds for $\textbf{A}\subset \mathbb O^+$. 
Fix a small constant $0<\epsilon<1/3$. Note that the following {\em gap property} can be checked by contradiction; that is, there exists some $x_0 \in [1, r]$ so that for all $k$ such that $d_k> \alpha+x_0 n^{\epsilon}(\phi_n+n^{-1/3})$, we have $d_k\geq \alpha+(x_0+1)n^{\epsilon}(\phi_n+n^{-1/3})$. Following the idea in \cite[Section 6.2]{principal}, for such $x_0$, we split $\textbf{A}=S_{0} \cup S_1$ such that $d_k \le \alpha +x_0 n^{\epsilon}(\phi_n+n^{-1/3})$ for $k \in S_0$, and $d_k \geq \alpha+(x_0+1)n^{\epsilon}(\phi_n+n^{-1/3})$ for $k\in S_1$. Note that Assumption \ref{assump_strong} fit in $S_1$ by letting $\tau' = \epsilon$, thus Theorem \ref{thm_vector} is valid when $\mathbf{A} = S_1$ as we proved in Section \ref{sec_remove}. Therefore, without loss of generality, we assume that $S_0\ne \emptyset$.
There are totally six cases: (a) $i,j\in S_0$; (b) $i\in S_0$ and $j\in S_1$; (c) $i\in S_0$ and $j\notin \textbf{A}$; (d) $i,j\in S_1$; (e) $i\in S_1$ and $j\notin \textbf{A}$; (f) $i,j\notin \textbf{A}$.  

\begin{enumerate}
\item [(a)] $i,j\in S_0$. We have the splitting 
\begin{equation}\label{split_general}
\langle \ub_i, {\mathcal{P}}_\textbf{A}\ub_j\rangle = \langle \ub_i, {\mathcal{P}}_{S_0}\ub_j\rangle +  \langle \ub_i, {\mathcal{P}}_{S_1}\ub_j\rangle.
\end{equation}
Applying Cauchy-Schwarz inequality as in \eqref{CS new} and (\ref{eq_nonspikeeg2}) to the first term, and Theorem \ref{thm_vector} to the second term, we get that
\begin{align*}
&\left |\langle \ub_i, {\mathcal{P}}_\textbf{A}\ub_j\rangle-\delta_{ij}a_1(d_i) \right| \\ 
\prec &\, \frac{n^{3\epsilon}\left( \phi_n^3+n^{-1}\right)}{\Delta^2(d_i) \Delta^2(d_j)} + \left(  \frac{\phi_n}{ \nu^{1/2}_{i}(S_1)}+\frac{\psi_{1}(d_i)\Delta(d_i) }{\nu_{i}(S_1)} \right)\left(  \frac{\phi_n}{ \nu^{1/2}_{j}(S_1)}+\frac{\psi_{1}(d_j)\Delta(d_j) }{\nu_{j}(S_1)} \right)\\
\lesssim &\,
n^{4\epsilon}\psi_1^{1/2}(d_i)\psi_1^{1/2}(d_j)\,,
\end{align*}
where in the first step we use $\eta_l (\gamma_i) \sqrt{\kappa_{\gamma_i}}  \lesssim n^{-1} +\phi_n n^{-5/6}\lesssim n^{-1}+\phi_n^3$ for $i\in \mathbf{A}$ derived by \eqref{etalE} and $\kappa_{\gamma_i} \asymp n^{-2/3}$, and 
$$
\frac{n^{3\epsilon}\left( \phi_n^3+n^{-1}\right)}{\Delta^2(d_i) \Delta^2(d_j)} \lesssim n^{3\epsilon}(\phi_n + n^{-1/3}) \lesssim n^{4\epsilon}\psi_1^{1/2}(d_i)\psi_1^{1/2}(d_j),
$$ 
and 
$$
\frac{\phi_n}{ \nu^{1/2}_{i}(S_1)}+\frac{\psi_{1}(d_i)\Delta(d_i) }{\nu_{i}(S_1)} \lesssim \frac{\psi_1(d_i)}{\nu_i^{1/2}(S_1)} \lesssim \phi_n^{1/2} + \frac{n^{-1/3}}{\Delta(d_i)} \lesssim (\phi_n + n^{-1/3})^{1/2} \lesssim \psi_1^{1/2}(d_i)
$$
in the second step, since for $d= d_i$ or $d = d_j$, we have $\phi_n+n^{-1/3} \le \Delta^2(d) \lesssim n^\epsilon(\phi_n+n^{-1/3}) \lesssim \nu_{i}(S_1) $ and $\psi_1(d) = \phi_n+n^{-1/2}/\Delta(d) \gtrsim \phi_n + \phi_n^{-1/2}n^{-1/2-\epsilon/2}+n^{-1/3-\epsilon/2} \gtrsim n^{-\epsilon/2}(\phi_n + n^{-1/3})$.

\item [(b)] $i\in S_0$ and $j\in S_1$. Similar to the previous case, by applying Cauchy-Schwarz and Theorem \ref{thm_noneve} to the first term in \eqref{split_general}, we get that
\begin{align}
 &|\langle \ub_i, {\mathcal{P}}_{S_0}\ub_j\rangle|\prec \frac{n^{3\epsilon}(n^{-1}+\phi_n^3)}{\Delta^2 (d_i)\Delta^2(d_j)} \lesssim n^{4\epsilon} \psi_1^{1/2}(d_i)\psi_1^{1/2}(d_j)\,.
\end{align}
For the second term, we first let Assumption \ref{assum_nonoverlap} hold and apply Proposition \ref{prop_outev0} and get that
\begin{align}
&|\langle \ub_i,\, {\mathcal{P}}_{S_1}\ub_j\rangle| \nonumber\\
\prec &\, \frac{\psi_1(d_j)\Delta^2(d_j)}{|d_j - d_i|} +\psi_1^2(d_i)\Delta^2(d_i)  \left(\frac{1}{\nu_{i}(S_1)}+ \frac{1}{\Delta^2(d_i)}\right)\left(\frac1{\nu_{j}(S_1)} + \frac{1}{\Delta^2(d_j)} \right)\nonumber\\
\lesssim &\,\left[\psi_{1}^{1/2}(d_i)+  \frac{\psi_{1}(d_i)\Delta(d_i) }{\nu_{i}(\textbf{A})}   \right] \left[\psi_{1}^{1/2}(d_j)+  \frac{\psi_{1}(d_j)\Delta(d_j) }{\nu_{j}(\textbf{A})}  \right] ,
\end{align}
where we use
\begin{align*}
\nu_{i}(S_1)  \gtrsim \Delta^2(d_{i}),\ \  \nu_{j}(S_1)\gtrsim \Delta^2(d_{j}) \wedge \nu_{j}(\textbf{A}),\ \  \psi_1(d_j) \lesssim \psi_1(d_i)\,,\\
|d_j - d_i|\gtrsim \Delta^2(d_j) \gtrsim \Delta^2(d_i),\ \  \psi_1(d_j) \Delta(d_j)\gtrsim \psi_1(d_i)\Delta(d_i)\,.
\end{align*}
This concludes the proof of case (b) if the non-overlapping condition hold. Otherwise, we can remove the non-overlapping condition as in Section \ref{sec_remove}. 

\item [(c)] Since (e) and (f) follow the same proof, we prove them together here. Note that in all cases $j \notin \mathbf{A}$, and we have $\nu_{j}(\textbf{A})\le \nu_{j}(S_1)$. In case (c) with $i\in S_0$ and $j \notin \textbf{A}$, we use the splitting in \eqref{split_general} and apply Cauchy-Schwarz, Theorem \ref{thm_noneve}, and Proposition \ref{prop_outev0} to the first term to obtain that
\begin{align}
 |\langle \ub_i, {\mathcal{P}}_{S_0}\ub_j\rangle| 
  \lesssim n^{5\epsilon}  \psi_1^{1/2}(d_i)  \left[  \frac{\phi_n}{ \nu^{1/2}_{j}(\textbf{A})}+\frac{\psi_{1}(d_j)\Delta(d_j) }{\nu_{j}(\textbf{A})} \right] ,
\end{align}
and then we use Theorem \ref{thm_vector} to the second term and obtain that  
\begin{align}
 |\langle \ub_i, {\mathcal{P}}_{S_1}\ub_j\rangle| &\prec \left[  \frac{\phi_n}{ \nu^{1/2}_{i}(S_1)}+\frac{\psi_{1}(d_i)\Delta(d_i) }{\nu_{i}(S_1)} \right] \left[  \frac{\phi_n}{ \nu^{1/2}_{j}(S_1)}+\frac{\psi_{1}(d_j)\Delta(d_j) }{\nu_{j}(S_1)} \right]\nonumber \\
&  \lesssim    \psi_1^{1/2}(d_i) \left[  \frac{\phi_n}{ \nu^{1/2}_{j}(\textbf{A})}+\frac{\psi_{1}(d_j)\Delta(d_j) }{\nu_{j}(\textbf{A})} \right]\,,
\end{align}
where in the last step we also used the last inequality in solving case (a).
In case (e) with $i\in S_1$ and $j\notin \textbf{A}$, $ |\langle \ub_i, {\mathcal{P}}_{S_0}\ub_j\rangle|$ can be bounded in the same way as case (c). On the other hand, by Theorem \ref{thm_vector} we have
\begin{align*}
|\langle \ub_i, {\mathcal{P}}_{S_1}\ub_j\rangle|  \prec&\,  \Delta(d_i)\left[  \frac{\phi_n}{ \nu^{1/2}_{j}(S_1)}+\frac{\psi_{1}(d_j)\Delta(d_j) }{\nu_{j}(S_1)} \right] \\
&+ \left[   \psi^{1/2} (d_i)+\frac{\psi_{1}(d_i)\Delta(d_i) }{\nu_{i}(S_1)}  \right] \left[  \frac{\phi_n}{ \nu^{1/2}_{j}(S_1)}+\frac{\psi_{1}(d_j)\Delta(d_j) }{\nu_{j}(S_1)} \right] \\
\lesssim &\, \Delta(d_i)\left[  \frac{\phi_n}{ \nu^{1/2}_{j}(\textbf{A})}+\frac{\psi_{1}(d_j)\Delta(d_j) }{\nu_{j}(\textbf{A})} \right]  \\
&+\left[   \psi^{1/2} (d_i)+\frac{\psi_{1}(d_i)\Delta(d_i) }{\nu_{i}(\textbf{A})} \right] \left[  \frac{\phi_n}{ \nu^{1/2}_{j}(\textbf{A})}+\frac{\psi_{1}(d_j)\Delta(d_j) }{\nu_{j}(\textbf{A})} \right] ,
\end{align*}
where we used $ \nu_{i}(S_1)\gtrsim \Delta^2(d_i) \wedge \nu_{i}(\textbf{A})$ in the second bound. In case (f) with $i,j\notin \textbf{A}$, by Theorem \ref{thm_noneve} we obtain that
\begin{align*}
 |\langle \ub_i, {\mathcal{P}}_{S_0}\ub_j\rangle| &\prec \frac{n^{3\epsilon}(n^{-1}+\phi_n^3)}{\left(\Delta^2(d_i)+\phi_n+n^{-1/3}\right) \left(\Delta^2(d_j)+\phi_n+n^{-1/3}\right)} \\
&  \lesssim n^{5\epsilon}  \left[  \frac{\phi_n}{\nu_{i}^{1/2}(\textbf{A})}+\frac{\psi_{1}(d_i)\Delta(d_i) }{\nu_{i}(\textbf{A})} \right]   \left[  \frac{\phi_n}{\nu^{1/2}_{j}(\textbf{A})}+\frac{\psi_{1}(d_j)\Delta(d_j) }{\nu_{j}(\textbf{A})} \right] , 
\end{align*}
where in the second bound we use the fact that for some $k \in S_0$, 
\begin{align*}
\nu_{i}(\textbf{A})&\,\lesssim |d_i-\alpha|+|d_k-\alpha|\lesssim n^\epsilon (\Delta^2(d_i) +\phi_n + n^{-1/3} ),\\
\nu_{j}(\textbf{A})&\,\lesssim |d_j-\alpha|+|d_k-\alpha|\lesssim n^\epsilon (\Delta^2(d_j) +\phi_n + n^{-1/3} ),
\end{align*}
and
\begin{align*}
\Delta(d_i)\psi_1(d_i) + \phi_n v_i(\mathbf{A})^{1/2}&\,\gtrsim \phi_n^{3/2}+\phi_n n^{-1/6}+n^{-1/2} +\phi_n v_i^{1/2}(\mathbf A) \gtrsim (\phi_n^3 + n^{-1})^{1/2},\\
\Delta(d_j)\psi_1(d_j) + \phi_n v_j(\mathbf{A})^{1/2}&\,\gtrsim \phi_n^{3/2}+\phi_n n^{-1/6}+n^{-1/2} +\phi_n v_j^{1/2}(\mathbf A) \gtrsim (\phi_n^3 + n^{-1})^{1/2},
\end{align*}
by $\Delta(d_k)^2 = d_k - \alpha \gtrsim \phi_n+n^{-1/3}$, for $k = 1,\ldots r$. 
For the ${\mathcal{P}}_{S_1}$ term, we have
\begin{align*}
 |\langle \ub_i, {\mathcal{P}}_{S_1}\ub_j\rangle| &\prec \left[  \frac{\phi_n}{\nu_{i}^{1/2}(S_1)}+\frac{\psi_{1}(d_i)\Delta(d_i) }{\nu_{i}(S_1)} \right]   \left[  \frac{\phi_n}{\nu^{1/2}_{j}(S_1)}+\frac{\psi_{1}(d_j)\Delta(d_j) }{\nu_{j}(S_1)} \right]  \\
& \le  \left[  \frac{\phi_n}{\nu_{i}^{1/2}(\textbf{A})}+\frac{\psi_{1}(d_i)\Delta(d_i) }{\nu_{i}(\textbf{A})} \right]   \left[  \frac{\phi_n}{\nu^{1/2}_{j}(\textbf{A})}+\frac{\psi_{1}(d_j)\Delta(d_j) }{\nu_{j}(\textbf{A})} \right] ,
\end{align*}
where we use $\nu_{i}(\textbf{A})\le \nu_{i}(S_1)$ and $\nu_{j}(\textbf{A})\le \nu_{j}(S_1)$ in the second bound.

\item [(d)] $i,j\in S_1$.  Again, using \eqref{split_general}, Theorem \ref{thm_noneve}, Theorem \ref{thm_vector}, and similar steps in case (a), we get that 
\begin{align*}
& \left|\langle \ub_i, {\mathcal{P}}_\textbf{A}\ub_j \rangle - \delta_{ij}a_1(d_i)\right|\\
\prec &\, n^{4\epsilon}\psi_1^{1/2}(d_i)\psi_1^{1/2}(d_j) +\left[   \psi^{1/2} (d_i)+\frac{\psi_{1}(d_i)\Delta(d_i) }{\nu_{i}(S_1)} \right]\left[   \psi^{1/2} (d_j)+\frac{\psi_{1}(d_j)\Delta(d_j) }{\nu_{j}(S_1)}\right] \\
\prec &\, n^{4\epsilon} \left[   \psi^{1/2} (d_i)+\frac{\psi_{1}(d_i)\Delta(d_i) }{\nu_{i}(\textbf{A})}\right]\left[   \psi^{1/2} (d_j)+\frac{\psi_{1}(d_j)\Delta(d_j) }{\nu_{j}(\textbf{A})} \right]\,,
\end{align*}
where we used $ \nu_{i}(S_1)\gtrsim \Delta^2(d_{i})\wedge \nu_{i}(\textbf{A})$ and $ \nu_{j}(S_1)\gtrsim \Delta^2(d_{j})\wedge \nu_{j}(\textbf{A})$ in the second step. 

\end{enumerate}

Combining all the above six cases, we conclude that even without Assumption \ref{assump_strong}, the estimate \eqref{eq_spikedvector} still holds with an additional factor $n^{6\epsilon}$ multiplying on the right hand side. Since $\epsilon$ can be arbitrarily small, we conclude the proof.
\end{proof}

\subsection{Proof of Lemmas \ref{lem_contour} and \ref{lem_distance}}\label{section_proof of technical lem1}

\begin{proof}[{Proof of Lemma \ref{lem_contour}}]\label{section_pflem}  
Let $\zeta \in \Gamma$. We first show that there exists $\wt c_0 :=  \wt c_0(c_i)$, such that when $c_i$ is sufficiently small, then $\zeta$ satisfies (1) $ \re \zeta \ge \alpha$, (2) $|\text{Im} \zeta | \leq \wt c_0 (\re \zeta-\alpha)$ and (3) $|\zeta| \leq C$, and then there exists a constant $\wt c_1 := \wt c_1(\wt c_0, C)$ such that 
\begin{equation}\label{claim_reg2c}
\re \theta(\zeta) \geq \lambda_+ + \wt c_1 (\re \zeta-\alpha)^2.
\end{equation}
By Assumption \ref{assump_strong} and the definition of $\rho_i$, (1) and (3) are satisfied. For (2), because we have that for all $\zeta \in \Gamma_i$, $$|\operatorname{Im} \zeta| \le \rho_i \le c_i (d_i -\alpha)$$
and
$$\re \zeta-\alpha \geq d_i - \rho_i-\alpha \geq d_i - c_i(d_i - \alpha) - \alpha = (1-c_i)(d_i - \alpha),$$
which together lead to $$|\text{Im}\zeta| \leq \frac{c_i}{1-c_i}(\re \zeta - \alpha).$$
Thus (1), (2) and (3) are fulfilled.  To show \eqref{claim_reg2c}, by (1) and (3), we have
$0 \leq \re \zeta-\alpha \leq c_0 $ for some constant $c_0>0$, then \eqref{claim_reg2c} follows from (\ref{eq_gcomplex}) that
$$\re \theta(\zeta) -\lambda_+ \asymp \re( \zeta-\alpha)^2 \asymp (\re \zeta-\alpha)^2\,,
$$
where the second $\asymp$ comes from (2) shown above. 
The claim \eqref{claim_reg2c} then follows by first choosing a sufficiently small constant $ \wt c_0$ and then choosing an appropriate constant $\wt c_1$.

Now we can finish the proof of the first statement in the Lemma. By (\ref{eq_gcomplex}), we have $|\theta(\zeta)| \leq \omega^{-1}$ for all $\zeta \in \Gamma$ as long as $\omega$ is sufficiently small. Also, using \eqref{claim_reg2c}, we can find the lower bound of $\Re \theta(\zeta)$. Thus we can conclude that $\theta(\Gamma_i) \subset S_{out}((\lambda_+\omega)^{-1},\omega)$ as long as $c_i$ is sufficiently small, so as $\overline{\theta(\Upsilon)}$.

To prove the second statement, it suffices to show that: 
\begin{itemize}
\item[(i)] $\wt\lambda_{i} \in \theta(\Upsilon_i)$ for all  $ i \in  \textbf{A}$;

\item[(ii)] $\wt\lambda_j \notin  \theta(\Upsilon_i)$ for all $ j \notin  \textbf{A}$ and  $ i \in  \textbf{A}$. 
\end{itemize}
To prove (i), we notice that under Assumptions \ref{assum_nonoverlap},
\begin{equation}\label{eq_rhobound0}
\rho_i \geq c_i(\left[\Delta(d_i)\right]^{-1}n^{-1/2}+\phi_n)n^{\wt\tau}\,. 
\end{equation}
Together with \eqref{eq_gderivative}, by mean value theorem we get that
$$\left|\theta\left(d_i + \rho_i\right) - \theta\left(d_i \right) \right| \gtrsim (\Delta(d_i)n^{-1/2}+\phi_n\Delta^2(d_i))n^{\wt\tau}$$
and 
$$\left|\theta\left(d_i - \rho_i\right) - \theta\left(d_i \right) \right|  \gtrsim (\Delta(d_i)n^{-1/2}+\phi_n\Delta^2(d_i))n^{\wt\tau}$$
for $i\in \textbf{A}.$ Then we conclude (i) using \eqref{rigid_outXi0}.
In order to prove (ii), we consider the two cases: (1) $ j\in \mathbb O^+ \setminus \textbf{A}$; (2) $j\notin \mathbb O^+$. In case (1), if $d_j > d_i$, we have $$\wt\lambda_j - \theta(d_i) > \theta(d_j-\rho_j) - \theta(d_i) \geq \theta'(d_j)(d_j - d_i - \rho_j) \gtrsim (\Delta(d_j) n^{-1/2} + \Delta^2(d_j)\phi_n)n^{\wt\tau},$$ where we used that $\theta$ is monotone in the first inequality, mean value theorem in the second inequality, and the definition of $\rho_i$, \eqref{eq_assnonover}, and \eqref{eq_gderivative} in the third inequality. Similarly, when $d_j<d_i$, we have  $$\theta(d_i)-\wt\lambda_j  \gtrsim (\Delta(d_j) n^{-1/2} + \Delta^2(d_j)\phi_n)n^{\wt\tau}.$$ In conclusion, we have  $$|\wt\lambda_j-\theta(d_i)|  \gtrsim (\Delta(d_j) n^{-1/2} + \Delta^2(d_j)\phi_n)n^{\wt\tau}.$$ Together with \eqref{rigid_outXi0} and that $\epsilon < \wt \tau$, case (1) is proved.
For case (2), the claim follows from \eqref{eq_stickingrigi} and \eqref{eq_gcomplex}. This concludes the proof. 
\end{proof}

\begin{proof}[{Proof of Lemma \ref{lem_distance}}] 
The upper bound in \eqref{zeta_rhoi} follows from the triangle inequality and the definition of $\rho_k$ in \eqref{eq_radius}:
\begin{equation*}
|\zeta-d_\ell|  \leq \rho_k+|d_k-d_\ell|.
\end{equation*}
It remains to prove a lower bound. For $\ell \notin \textbf{A}$, again by the definition of $\rho_k$, we trivially have $|d_k-d_\ell| \geq 2\rho_k$, from which we obtain that
\begin{equation*}
|\zeta-d_\ell| \geq |d_k-d_\ell|-\rho_k \geq \rho_k.
\end{equation*} 

Next we consider the case $\ell\in \textbf{A}$. Define $\delta:=|d_k-d_\ell|-\rho_\ell-\rho_k$, which is the
distance between $B_{\rho_k}
(d_k)$ and $B_{\rho_\ell}
(d_\ell)$.  First suppose that $C_0 \delta >|d_k-d_\ell|$ for some constant $C_0>1$. It then follows that $\rho_k+\rho_\ell \leq \frac{C_0-1}{C_0} |d_k-d_\ell|.$  
As a consequence, we obtain 
\begin{align}
|\zeta-d_\ell| &\, \geq \left| d_k-d_\ell \right|-\rho_k \geq \frac{1}{C_0} \left| d_k-d_\ell \right| + \rho_\ell \\
&\, >\frac{1}{C_0} \left| d_k-d_\ell \right| \gtrsim \rho_k+ \left| d_k-d_\ell \right|\,. \nonumber
\end{align}
Suppose now that $C_0 \delta \leq  \left| d_k-d_\ell \right|$. Then we have
\begin{equation*}
 \left| d_k-d_\ell \right| \leq \frac{C_0}{C_0-1}(\rho_k+\rho_\ell) .
\end{equation*}
We claim that for a sufficiently large constant $C_0>0$, there exists a constant $\wt C(c_k,c_\ell,C_0)>0$ such that   
\begin{equation}\label{comparablE_rho}
\wt C^{-1}\rho_k \leq \rho_\ell \leq \wt C \rho_k.
\end{equation}
If \eqref{comparablE_rho} holds, then we have
\begin{equation*}
|\zeta-d_\ell| \geq \rho_\ell \gtrsim \rho_k+\rho_\ell \gtrsim \rho_k +  \left| d_k-d_\ell \right| .
\end{equation*} 
This concludes \eqref{zeta_rhoi}.

It remains to prove \eqref{comparablE_rho}. Recall the definition of $\rho_k$ in \eqref{eq_radius}. Consider the following two cases. (i) If $\rho_\ell = c_\ell \left| d_\ell-d_s \right|$ for some $s$ such that $s\notin \textbf{A}$, we have
\begin{equation}\label{proof_example}
\frac{\rho_\ell}{c_\ell}  = \left| d_\ell-d_s \right| \le \left| d_k-d_s \right| + \left| d_s-d_\ell \right| \le \frac{\rho_k}{c_k} + \frac{C_0}{C_0-1}(\rho_k+\rho_\ell) .
\end{equation}
Thus as long as $c_\ell$ and $C_0$ is chosen such that $c_\ell^{-1} > \frac{C_0}{C_0-1}$, we obtain the upper bound in \eqref{comparablE_rho}. (ii) If $\rho_k = c_k (d_k - \alpha)$, the proof is the same as that in case (i), and we omit details. 
\end{proof}

\section{Proof of Theorems in Section \ref{section proposed algo}}
\label{section:more proof}
We have the following remarks that guide us toward the proof.
Suppose Assumption \ref{assum_main} and \eqref{alpha+} hold. Give any constant $0<\epsilon\leq\varepsilon$, we find that there exists an event $\Xi$ of high probability such that the followings hold when conditional on $\Xi$: 
\begin{enumerate}
\item By Theorem \ref{thm_value}, for $1\leq i \leq r^+ $, 
\begin{equation}\label{eq_adapt_evoutlier}
|\widetilde{\lambda}_i-\theta(d_i)| \leq  n^{\epsilon}(\phi_n \Delta^2(d_i)+n^{-1/2}\Delta(d_i)) \,. 
\end{equation} 
{
Together with \eqref{eq_s36_1} and \eqref{alpha+}, for $1\leq i \leq r^+ $ we have 
\begin{equation}\label{eq_adapt_ev}
    \wt\lambda_i - \lambda_+ \asymp  \Delta(d_i)^4\,.
\end{equation}
Moreover,} for a fixed integer $\varpi>r$, when $r^++1 \leq i \leq \varpi$ we have 
\begin{equation}\label{eq_evbulk_adapt}
|\widetilde{\lambda}_i-\lambda_+| \leq n^{\epsilon}(\phi_n^2+ n^{-2/3}).
\end{equation}

\item By Theorem \ref{thm_vector} {and \eqref{alpha+}, let $\textbf{A} = \mathbb O^+$}, we have that for {$i = 1\ldots r^+$ and for $j=1, \cdots, r$,}  
\begin{align}
&\left| \langle \ub_i, {\mathcal{P}}_\textbf{A}\ub_j\rangle- \delta_{ij}\mathbbm 1(i\in \textbf{A}) a_1(d_i)\right| \vee \left| \langle \vb_i, {\mathcal{P}}'_\textbf{A}\vb_j\rangle- \delta_{ij}\mathbbm 1(i\in \textbf{A})a_2(d_i)\right|\nonumber\\
 \leq &\, n^{\epsilon}(\phi_n+ n^{-1/2}/\Delta(d_i))\,.\label{eq_adapt_spikedvector}
\end{align}
{Moreover, by Theorem \ref{thm_noneve},} for  $ r_++1\leq i \leq r$, we have for $j = 1, \ldots, r$
\begin{equation}\label{eq_adapt_evebulka}
 |\langle \ub_j, \wt\bxi_{i}  \rangle|^2 \vee |\langle \vb_j, \wt\bzeta_{i}  \rangle|^2 \leq  n^{\epsilon}(\phi_n + n^{-1/3}).
\end{equation}

\item From Theorem \ref{lem_locallaw} and that $r$ is bounded, we have that for $z \in S_{out}(\varsigma_2,\epsilon)$
\begin{equation}\label{eq_ada_local}
|m_1(z) - m_{1c}(z)| +  |m_2(z) - m_{2c}(z)| \leq n^{\epsilon}(\phi_n + n^{-1/2}(\kappa_z+\eta)^{-1/4}).
\end{equation}

\item By Lemma \ref{lem_rigidty}, for {for a fixed small constant $\varsigma>0$}, for any $j$ such that $\lambda_+-\varsigma \leq \gamma_j \leq \lambda_+$, we have 
\begin{equation}\label{rigidity_adapt}
\vert \lambda_j - \gamma_j \vert \leq n^{\epsilon} [n^{-2/3}\left( j^{-1/3}+\mathbbm 1(j \le n^{1/4} \phi_n^{3/2})\right) + \eta_{l}(\gamma_j) + n^{2/3} j^{-2/3} \eta_l^2(\gamma_j)]\,,
\end{equation}
\item By Theorem \ref{thm_eigenvaluesticking}, for any sufficiently small constant $\tau>0.$ We have that for $ 1 \leq i \leq \tau n$,
\begin{equation} \label{eq_stickingeq_adapt}
\left|\wt\lambda_{i+ r^+ }-\lambda_i\right| \leq n^{\epsilon}(\frac{1}{n \alpha_+} + n^{-3/4} + i^{1/3}n^{-5/6} + n^{-1/2}\phi_n +  i^{-2/3}n^{-1/3}\phi_n^2 ). 
\end{equation}

\end{enumerate}

Hereafter, we restrict our discussion in the event space $\Xi$.

\subsection{Proof of Theorem \ref{thm_numbershrink} and \ref{thm_numbershrink2}}

\begin{proof}[Proof of Theorem \ref{thm_numbershrink}]
Fix a constant $0<c<1$. Note that $n^c\gg r$ holds when $n$ is sufficiently large, since $r$ is fixed.
From the definition of $\rho_{2c}$, we have  
\begin{align}\label{eq00}
n\int_{\gamma_{2n^c-r^++1}}^{\gamma_{n^c-r^++1}}\rho_{2c}(x)dx & = n\int_{\gamma_{n^c-r^++1}}^{\lambda_{+}}\rho_{2c}(x)dx +r^+= n^c.
\end{align}
Together with \eqref{sqroot3} from Lemma \ref{lambdar_sqrt} to approximate $\rho_{2c}(x)$ as $x\downarrow 0$, {after taking the integration} we have 
\begin{align}\label{eq1}
2(\lambda_{+}-\gamma_{n^c-r^++1})^{3/2}=  (\lambda_+-\gamma_{2n^c-r^++1})^{3/2}+ \OO((\lambda_+ - \gamma_{2n^c-r^++1} )^2+n^{-1})\,.
\end{align}
{Similarly, again from \eqref{sqroot3} and $n^c \gg r^+$,} we have
\begin{equation}\label{eq0}
\lambda_+ - \gamma_{2n^c-r^++1} =  \OO(n^{2(c-1)/3}).
\end{equation}
Thus, we have
\begin{align}\label{eq2}
\lambda_+ 
= \,& \gamma_{n^c-r^++1} + \frac{1}{2^{2/3}-1}(\gamma_{n^c-r^++1}-\gamma_{2n^c-r^++1}) +  \OO(n^{4(c-1)/3}+n^{-1})\\
 = \,&  \lambda_{n^c-r^++1} + \frac{1}{2^{2/3}-1}(\lambda_{n^c-r^++1}-\lambda_{2n^c-r^++1}) +  \OO(n^{4(c-1)/3}+n^{-1}) + \OO(\lambda_{n^c}-\gamma_{n^c})\nonumber\\
= \,&  \wt\lambda_{n^c+1} + \frac{1}{2^{2/3}-1}(\wt\lambda_{n^c+1}-\wt\lambda_{2n^c+1}) +  \OO(n^{4(c-1)/3}+n^{-1}) + \OO(\lambda_{n^c}-\gamma_{n^c})+\OO(\wt\lambda_{n^c+1}-\lambda_{n^c-r^++1})\nonumber\\
= \,&  {\widehat\lambda_+ + \OO\Big(n^{\epsilon}\big[\frac{1}{n\alpha^+}+n^{-\min\{\frac{2}{3}+\frac{c}{3},\frac{3}{4},\frac{5}{6}-\frac{c}{3},\frac{4}{3}-\frac{4c}{3}\}}+ n^{-1/2}\phi_n+n^{-1/3-2c/3}\phi_n^2 \big]\Big)}\,,\nonumber
\end{align}
{where in the first equation we used \eqref{eq1} and \eqref{eq0}, and in the last equality we used \eqref{rigidity_adapt} and \eqref{eq_stickingeq_adapt}.}
\end{proof}
\begin{proof}[Proof of Theorem \ref{thm_numbershrink2}] 
{Note that $n^{\varepsilon}(\phi_n+n^{-1/3}) > n^{-1/6}$ by assumption of $\varepsilon$. }
When $0<c<1/2$, {together with \eqref{eq2}}, with high probability we have that
\begin{equation}\label{eq9}
    |\widehat\lambda_+-\lambda_+| \leq n^{\epsilon}(\phi_n^2n^{-1/3} + n^{-2/3})\,.
\end{equation}
{Moreover, \eqref{eq9}, \eqref{alpha+} and \eqref{eq_adapt_ev} lead to 
\begin{equation}
    \wt\lambda_i-\widehat\lambda_+ \gtrsim n^{2\varepsilon}(\phi_n^2+n^{-2/3}) > n^{-1/3}\,,
\end{equation}
}
for all $1\leq i \leq r^+$. 
Also, by \eqref{eq9} and \eqref{eq_stickingeq_adapt}, we have
\begin{equation}
    |\wt\lambda_i-\widehat\lambda_+| \leq n^{\epsilon}(\phi_n^2n^{-1/3}+n^{-2/3}) \leq n^{-1/3}
\end{equation}
for $i=r^++1,\ldots,r$ when $\epsilon$ is chosen small enough. We thus conclude that $\Xi(r^+)$ is an event with high probability.
\end{proof}

\subsection{Proof of Theorem \ref{thm_shrink}} 
{In the event space $\Xi(r^+)$ we have $\widehat r^+ = r^+$, so from now on we will replace $\widehat r^+$ with $r^+$ in equation \eqref{eq_hatlanbda2} for the definition of $\widehat \lambda_i$, $i = 1 \ldots \lfloor n^c \rfloor$.} 
Denote the CDF of $ZZ^\top $ as $F_1(x) := \frac1p \sum_{j=1}^p \mathbbm 1(\lambda_{j}\leq x)$. Recall the following Lemma that compares $F_{e}$ and $F_1$.

\begin{lemma}\label{lem:main:approx-FZ}\cite[Lemma 4]{donoho2020screenot}
Suppose $k \ge r$. Then the Kolmogorov-Smirnov (KS) distance between $F_{e}$ and $F_1$ is controlled by 
\[
d_{\texttt{KS}}(F_{e},\,F_{1}) = \sup_{x} \left| F_{e}(x)-F_1(x) \right| \le \frac{k}{p}\,.
\]
\end{lemma}

The proof is then divided into several steps for clearance.
\newline\newline
\textbf{Step 1: Bound
$|m_{1}(\wt\lambda_i) - \widehat m_{e,1,i}|$ and $|m_{2}(\wt\lambda_i) - \widehat m_{e,2,i}|$:}
As the bound derived in \cite[page 30]{donoho2020screenot}, if $F_a$ and $F_b$ are CDFs of probability measures supported on an interval $I$ and $g:I\to \mathbb R$ is bounded and continuously differentiable, with an integral by part we have 
\begin{align}\label{eq_bddCDF}
\left|  \int g(t)( dF_a(t)-dF_b(t) ) \right|
& = \left|  g(t)( F_a(t)-F_b(t) )|_{I} +  \int g'(t)( F_a(t)-F_b(t) )dt  \right|\nonumber \\
& \le  \|g'\|_{L^1(I)} d_{\texttt{KS}}(F_a,\,F_b)
\end{align}
Fix $1\leq i\leq r^+$. Note that $|\lambda_1 -\lambda_+| \leq n^{\epsilon}(\phi_n^2 + n^{-2/3})$ {by \eqref{rigidity_adapt}} and $|\widehat\lambda_{r^++1} - \lambda_+| \leq n^{\epsilon}(\phi_n^2 + n^{-2/3})$ { by the same approach for deriving \eqref{eq9}}. Together with \eqref{eq_adapt_ev}, we have $\lambda_1\vee\widehat\lambda_{r^++1}< \wt{\lambda}_i$.
Thus, if we take $F_a = F_{e}$, $F_b = F_1$, $g(t) = {1}/(t - \wt\lambda_i)$ and $I = [0, \lambda_1 \vee \widehat \lambda_{r^++1}]$ in \eqref{eq_bddCDF}, we have 
{
\begin{equation}
|m_{1}(\wt\lambda_i) - \widehat m_{e,1,i}| \lesssim \frac{1}{n(|\wt\lambda_i - \lambda_1|\vee |\wt\lambda_i -\widehat \lambda_{r^++1}|)}\,,
\end{equation}
}
where the last bound comes from Lemma \ref{lem:main:approx-FZ}.
{Together with \eqref{eq_adapt_ev} and $|\widehat\lambda_{r^++1} - \lambda_+| \leq n^{\epsilon}(\phi_n^2 + n^{-2/3})$,} we have 
\begin{equation}\label{eq_mdiff4}
    |m_{1}(\wt\lambda_i) - \widehat m_{e,1,i}| \lesssim  \frac{1}{n\Delta^4(d_i)}.
\end{equation}
With the same approach, we have 
\begin{equation}\label{eq_mdiff5}
    |m_{2}(\wt\lambda_i) - \widehat m_{e,2,i}| \lesssim  \frac{1}{n\Delta^4(d_i)}.
\end{equation}
\newline\newline
\textbf{Step 2: Bound $|\widehat m_{e,1,i} -  m_{1c}(\wt\lambda_i)|$ and $|\widehat m_{e,2,i} - m_{2c}(\wt\lambda_i)|$:}
By \eqref{eq_adapt_ev} and \eqref{eq_ada_local}, we have that for $1\leq i \leq r^+,$ 
\begin{align}\label{eq_mdiff10}
&|m_1(\wt\lambda_i)-m_{1c}(\wt\lambda_i)| \leq n^{\epsilon}(\phi_n + n^{-1/2}/\Delta(d_i)), \\ 
&|m_2(\wt\lambda_i)-m_{2c}(\wt\lambda_i)| \leq n^{\epsilon}(\phi_n + n^{-1/2}/\Delta(d_i)). \nonumber
\end{align}
By \eqref{eq_mdiff4}, \eqref{eq_mdiff5} and \eqref{eq_mdiff10}, 
we conclude that for $1 \leq i \leq r^+,$
\begin{align}\label{eq_mdiff_final}
&|\widehat{m}_{e,1,i}-m_{1c}(\wt\lambda_i)| \leq n^{\epsilon}(\phi_n + n^{-1/2}/\Delta(d_i)),\\  
& |\widehat{m}_{e,2,i}-m_{2c}(\wt\lambda_i)| \leq n^{\epsilon}(\phi_n + n^{-1/2}/\Delta(d_i)). \nonumber
\end{align}
\newline\newline
\textbf{Step 3: Finish the claim:}
With the above preparation and that $\wt{\lambda}_i$, $m_{1c}(\wt{\lambda}_i)$, $m_{2c}(\wt{\lambda}_i)$,  $\widehat {m}_{e,1,i}$, and $\widehat {m}_{e,2,i}$ are all $\asymp 1$ by \eqref{eq_estimm} and \eqref{eq_mdiff_final}, we immediately have 
\begin{align}
&|\widehat{\mathcal{T}}_{e,i}-\mathcal{T}(\wt\lambda_i)| \nonumber\\
=&\,  
|\wt\lambda_i(\widehat m_{e,1,i}\widehat m_{e,2,i}- m_{1c}(\wt\lambda_i)m_{2c}(\wt\lambda_i))| \nonumber \\ 
=&\, |\wt\lambda_i((\widehat m_{e,1,i}-m_{1c}(\wt\lambda_i))\widehat m_{e,2,i} + m_{1c}(\wt\lambda_i)(\widehat m_{e,2,i}-m_{2c}(\wt\lambda_i)))| \nonumber \\
 \lesssim &\, n^{\epsilon}(\phi_n + n^{-1/2}/\Delta(d_i))\,. \label{eq_diff_t}
\end{align}
Together with \eqref{eq_estimm} we also have $\widehat{\mathcal{T}}_{e,i}\asymp \mathcal{T}(\wt\lambda_i)\asymp 1$.
Recall the definition of $\widehat d_i$ in \eqref{eq_a1a2}. We have 
\begin{align}\label{eqf_1}
\left|\widehat{d}_{e,i} -1/\sqrt{\mathcal{T}(\wt\lambda_i)}\right| &\, \asymp \left|\sqrt{\mathcal{T}(\wt\lambda_i)} -\sqrt{ \widehat{\mathcal{T}}_{e,i}}\right|\\
&\,\asymp \left|\mathcal{T}(\wt\lambda_i) - \widehat{\mathcal{T}}_{e,i}\right| \lesssim n^{\epsilon}(\phi_n + n^{-1/2}/\Delta(d_i))\,.\nonumber 
\end{align}
Also, with \eqref{eq_adapt_evoutlier}, \eqref{eq_gderivative}, and the mean value theorem, for some $d$ between $d_i$ and $1/\sqrt{\mathcal{T}(\wt\lambda_i)}$ we have that
\begin{equation}\label{eqf_2}
    |d_i - 1/\sqrt{\mathcal{T}(\wt\lambda_i)}| =  |\theta(d_i)-\wt\lambda_i|/\theta'(d) \lesssim n^{\epsilon}(\phi_n+n^{-1/2}/\Delta(d_i)).
\end{equation}
Thus, by combining \eqref{eqf_1} and \eqref{eqf_2}, 
and the triangle inequality, we have
\begin{equation}\label{eq_dest}
    |\widehat{d}_{e,i} - d_i| \lesssim n^{\epsilon}(\phi_n + n^{-1/2}/\Delta(d_i)) \,, 
\end{equation}
which concludes the case of the operator norm. 

For the case of the Frobenius norm, by triangle inequality, we can show that for $1\leq i \leq r^+$
\begin{equation}\label{eq_diff_a_1}
    |a_{1,i}-\widehat{a}_{e,1,i}|  \leq |a_{1,i}-a_1(d_i)|+|a_1(d_i)-\widehat{a}_{e,1,i}|.
\end{equation}
Now we need to bound the two terms on the right hand side. For the first term, by \eqref{eq_adapt_spikedvector} we have
\begin{equation}\label{eq_diff_a_2}
 |a_{1,i}-a_1(d_i)| \leq n^{\epsilon}(\phi_n+ n^{-1/2}/\Delta(d_i)).
\end{equation}
For the second term, again by triangle inequality, we have
\begin{align}\label{eq_diff_a_3}
|a_1(d_i)-\widehat{a}_{e,1,i}| &\, =  \left|\frac{m_{1c}(\theta(d_i))}{d_i^2\mathcal T'(\theta(d_i))} -\frac{\widehat m_{e,1,i}}{\widehat {d}_i^2 \widehat{\mathcal T'_{e,i}}}\right| \\
&\,\leq \left|\frac{m_{1c}(\theta(d_i))}{d_i^2\mathcal T'(\theta(d_i))}-\frac{m_{1c}(\wt\lambda_i)}{d_i^2\mathcal T'(\wt\lambda_i)}\right|  + \left|\frac{m_{1c}(\wt\lambda_i)}{d_i^2\mathcal T'(\wt\lambda_i)}  -\frac{\widehat m_{e,1,i}}{\widehat {d}_i^2 \widehat{\mathcal T'_{e,i}}}\right| \,.\nonumber
\end{align}
For some $s$ between $\theta(d_i)$ and $\wt\lambda_i$, by mean value theorem and the order of $m_{1c}'$, $\mathcal T'$ and $\mathcal T''$ from \eqref{eq_s36_3}, together with the bound of $|\theta(d_i) - \wt\lambda_i|$ in \eqref{eq_adapt_evoutlier} we have
\begin{align}\label{eq_diff_a_4}
    \left|\frac{m_{1c}(\theta(d_i))}{d_i^2\mathcal T'(\theta(d_i))}-\frac{m_{1c}(\wt\lambda_i)}{d_i^2\mathcal T'(\wt\lambda_i)}\right| &\,= \left|\frac{m'_{1c}(s)T'(s) - m_{1c}(s)T''(s)}{d_i^2 \mathcal (T'(s))^2}\right||\theta(d_i) - \wt\lambda_i| \nonumber\\
    &\, \lesssim n^{\epsilon}(\phi_n \Delta^2(d_i)+n^{-1/2}\Delta(d_i))\,.
\end{align}
Also, again by triangle inequality,  the order of $\mathcal T'$ from \eqref{eq_s36_3}, and the error bounds from {\eqref{eq_imp_bd2}, \eqref{eq_mdiff_final} and \eqref{eq_dest}}, we have 
\begin{align}
&\left|\frac{m_{1c}(\wt\lambda_i)}{d_i^2\mathcal T'(\wt\lambda_i)}  -\frac{\widehat m_{e,1,i}}{\widehat {d}_i^2 \widehat{\mathcal T'_{e,i}}}\right|\label{eq_diff_a_5}\\ 
\leq &\,\left|\frac{m_{1c}(\wt\lambda_i)}{d_i^2\mathcal T'(\wt\lambda_i)} -\frac{\widehat m_{e,1,i}}{d_i^2 {\mathcal T'(\wt\lambda_i)}}\right|+ \left|\frac{\widehat m_{e,1,i}}{d_i^2\mathcal T'(\wt\lambda_i)} -\frac{\widehat m_{e,1,i}}{\widehat {d}_i^2 {\mathcal T'(\wt\lambda_i)}}\right|+ \left|\frac{\widehat m_{e,1,i}}{\widehat {d}_i^2 {\mathcal T'(\wt\lambda_i)}} -\frac{\widehat m_{e,1,i}}{\widehat d_i^2 \widehat{\mathcal T'_{e,i}}}\right| \nonumber\\
 \lesssim &\,n^{\epsilon}(\phi_n+n^{-1/3})\Delta^2(d_i)\,.\nonumber
\end{align}
Combine \eqref{eq_diff_a_1}-\eqref{eq_diff_a_5} 
and the upper bound from (iv) of Assumption \ref{assum_main}, we conclude that 
\begin{equation}\label{eq_adest1}
    |a_{1,i}-\widehat{a}_{e,1,i}|\lesssim n^{\epsilon}(\phi_n + n^{-1/2}/\Delta(d_i)).
\end{equation}
Also, with similar approach, we have 
\begin{equation}\label{eq_adest2}
    |a_{2,i}-\widehat{a}_{e,2,i}|\lesssim n^{\epsilon}(\phi_n + n^{-1/2}/\Delta(d_i)).
\end{equation}
Thus, with \eqref{eq_dest},\eqref{eq_adest1}, and \eqref{eq_adest2}, we have 
\begin{align}
   & |d_i\sqrt{a_{1,i}a_{2,i}} - \widehat d_i\sqrt{\widehat a_{1,i}\widehat a_{2,i}}|\\
    \leq &\, |d_i\sqrt{a_{1,i}a_{2,i}} - \widehat d_i\sqrt{ a_{1,i} a_{2,i}}|+|\widehat  d_i\sqrt{a_{1,i}a_{2,i}} - \widehat d_i\sqrt{\widehat a_{1,i} a_{2,i}}|\nonumber\\
    &+|\widehat d_i\sqrt{\widehat a_{1,i}a_{2,i}} - \widehat  d_i\sqrt{\widehat a_{1,i}\widehat a_{2,i}}|
    \leq n^{\epsilon}(\phi_n + n^{-1/2}/\Delta(d_i))\,,\nonumber
\end{align}
which concludes the case of the Frobenius norm. The proof for the nuclear norm follows the similar approach and we omit the detail.

\section{More simulated results}\label{section supp more numerical simulation}
Following the same approach in Section \ref{section simulation signal noise}, in this section we provide a more extensive numerical study with more $D_A$ and $D_B$ using the distributions proposed in \cite{donoho2020screenot}. We create three  types of {\em one-sided noises} so that $D_B$ is the identity matrix $I_n$ and $D_A$ follows the eigenvalue distribution of Mix2, Unif[1,10], or Fisher3n, where {Mix2} stands for an equal mixture of $1$ and $10$ as eigenvalues, Unif[1,10] stands for sampling eigenvalues uniformly from [1, 10], and {Fisher3n} stands generating eigenvalues from the eigenvalues of $W^\top W$, where $W \in \mathbb R^{3p\times p}$ is a random matrix with i.i.d Gaussian entries with mean $0$ and variance $1/(3p)$.  
Moreover, by the same approach mentioned in Section \ref{section simulation signal noise}, the resulting $\widehat \alpha$ for Mix2, Unif[1,10], or Fisher3n are $1.9160\pm 0.0211$, $1.2477\pm 0.0182$, and $1.2237 \pm 0.0156$ respectively, where we show the mean $\pm$ standard deviation over 100 realizations with $n'=10000$.
We also generate three types of {\em two-sided noises}, where $D_A$ and $D_B$ follow Mix2 and Unif[1,10], Mix2 and Fisher3n, and Unif[1,10] and Fisher3n respectively, and the resulting $\widehat \alpha$ has mean $\pm$ standard deviation as $2.0784\pm 0.0369$, $2.0105\pm 0.0446$, and $1.3860 \pm 0.0173$ respectively, where we show the mean $\pm$ standard deviation over 100 realizations with $n'=10000$. The signal matrix is designed in the same way as that in Section \ref{section simulation signal noise}. 

As in Section \ref{sec_compareTSE}, in Figure \ref{fig:rankesterror_sup}, we compare the estimated rank using \eqref{eq_adrk0}, the ScreeNOT rank, and the rank estimated by TRAD when $p/n = 1$ with different $n$. Similar to the results of {TYPE2 and TYPE3 noises}, TRAD always overestimates the rank for all combinations of noise types, and ScreeNOT rank often underestimates the rank with a larger error compared to eOptShrink.

\begin{figure}[hbt!]
\centering
\begin{minipage}{1\textwidth}
\begin{minipage}{0.3\textwidth}
\includegraphics[width=0.99\linewidth]{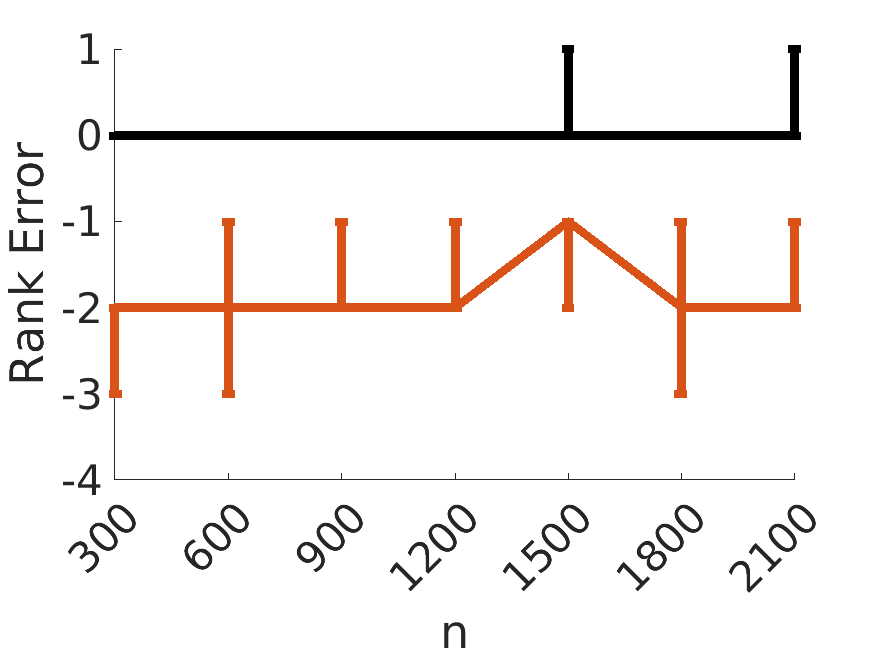}
\end{minipage}
\begin{minipage}{0.3\textwidth}
\includegraphics[width=0.99\linewidth]{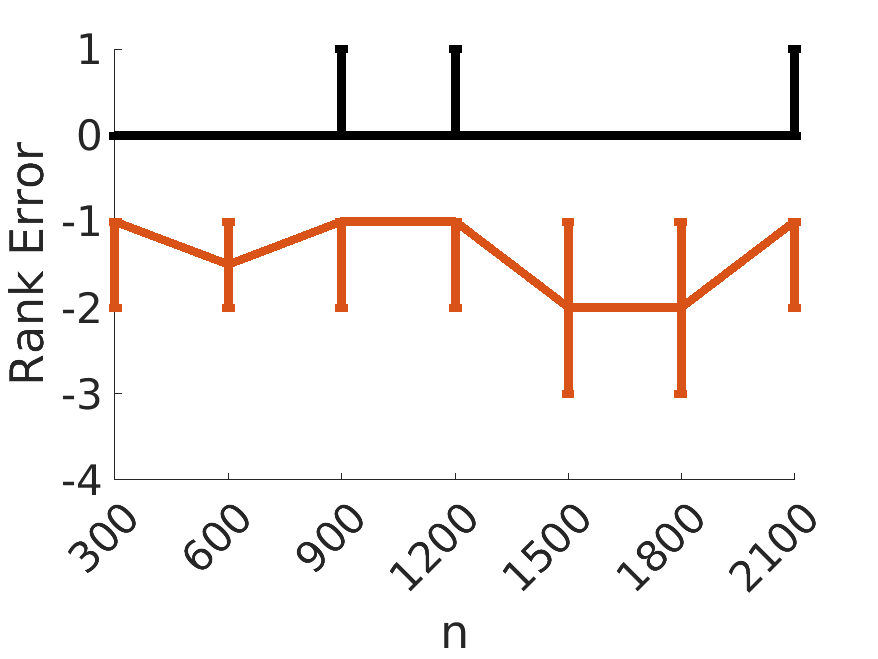}
\end{minipage}
\begin{minipage}{0.3\textwidth}
\includegraphics[width=0.99\linewidth]{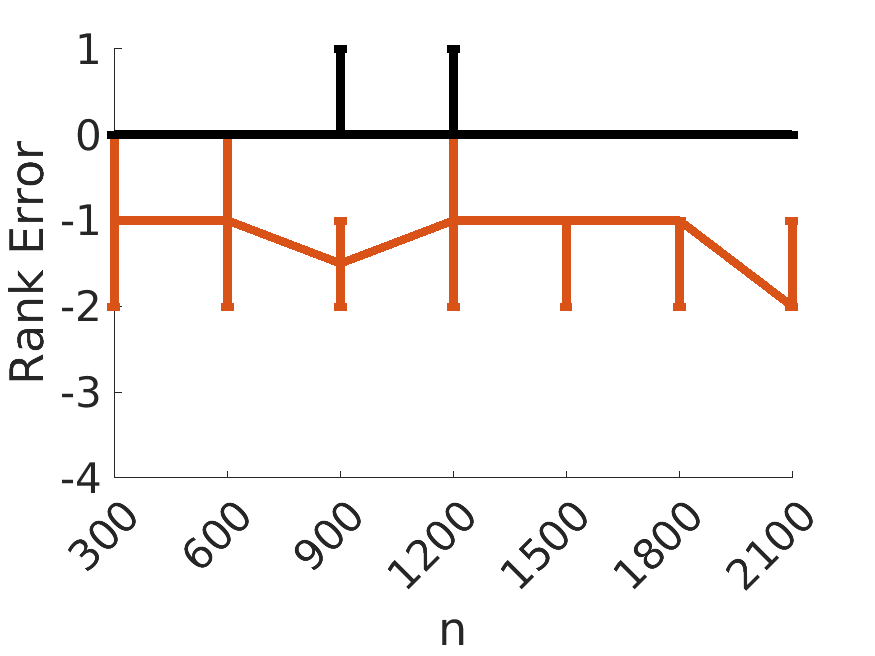}
\end{minipage}
\end{minipage}

\begin{minipage}{1\textwidth}
\begin{minipage}{0.3\textwidth}
\includegraphics[width=0.99\linewidth]{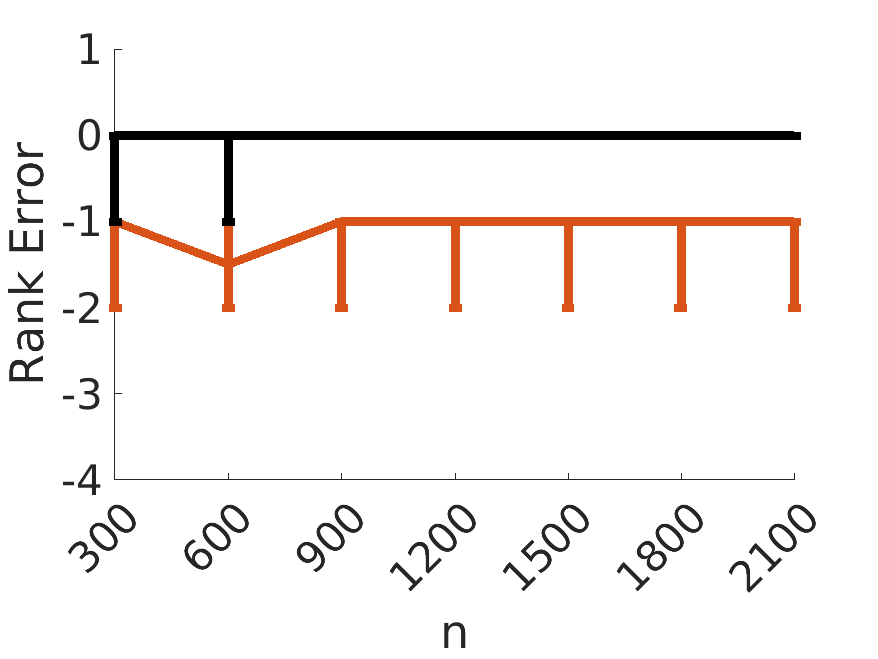}
\end{minipage}
\begin{minipage}{0.3\textwidth}
\includegraphics[width=0.99\linewidth]{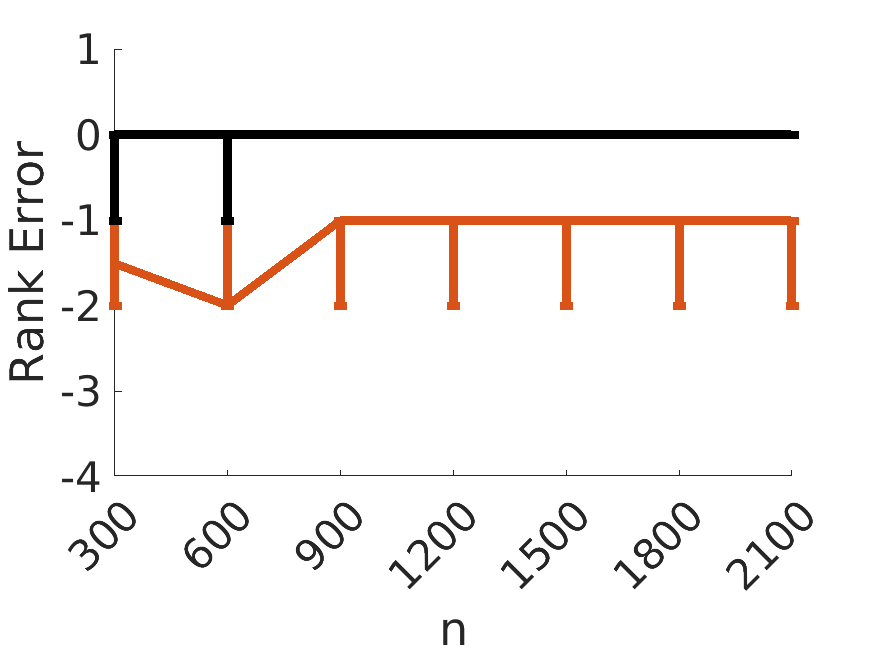}
\end{minipage}
\begin{minipage}{0.3\textwidth}
\includegraphics[width=0.99\linewidth]{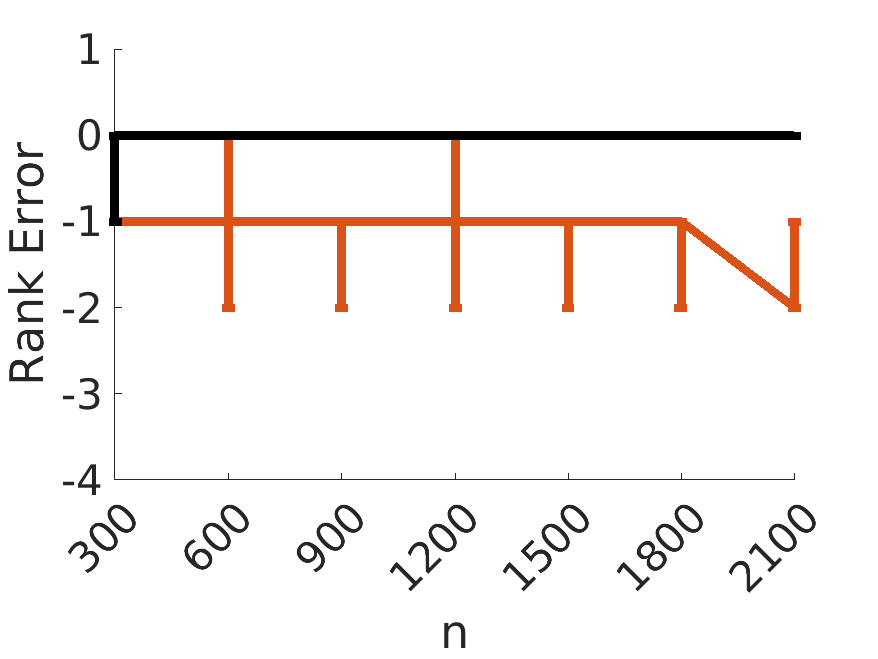}
\end{minipage}
\end{minipage}
\caption{\small A comparison of different rank estimators when $p/n=1$, where we show $\widehat r^+ - r^{+}$. The black (red and yellow respective) lines are errors of rank estimator from our rank estimator \eqref{eq_adrk0} (ScreeNot and TRAD respectively). 
The one-sided noises are shown in the first row, where from left to right are the results when $D_A$ follows Mix2, Unif[1,10], and Fisher3n respectively. The two-sided noises are shown in the second row, where from left to right are the results when $D_A$ and $D_B$ follow Mix2 and Unif[1,10], Mix2 and Fisher3n, and Unif[1,10] and Fisher3n respectively. If the corresponding error ratio is too high, the associated curve is not totally plotted to enhance the visualization. 
}\label{fig:rankesterror_sup}
 \end{figure}

In Figure \ref{fig_dr+error_sup} and \ref{fig_ar+error_sup}, we compare the error ratio of estimating $d_{\min\{r^+,\widehat r^+\}}$ and $\sqrt{a_{1,\min\{r^+,\widehat r^+\}}a_{2,\min\{r^+,\widehat r^+\}}}$ with different pseudo distributions as that in Section \ref{section:subsection:compare F} for different combinations of $A$ and $B$. We fix $n=600$ and $p/n = 1$. The black, yellow, and red lines indicate the estimator using $\widehat F_{\texttt{e}}(x)$, $\widehat F_{\texttt{T}}(x)$, and $\widehat F_{\texttt{imp}}(x)$ respectively. Clearly, our $\widehat F_{\texttt{e}}(x)$ always has a lower error ratio over every $\widehat r^+$ with statistical significance over every $\widehat r^+$ while comparing each estimator. This result shows that eOptShrink is robust to a slightly erroneous rank estimation.

\begin{figure}[hbt!]

\begin{minipage}{1\textwidth}
\begin{minipage}{0.3\textwidth}
\includegraphics[width=0.99\linewidth]{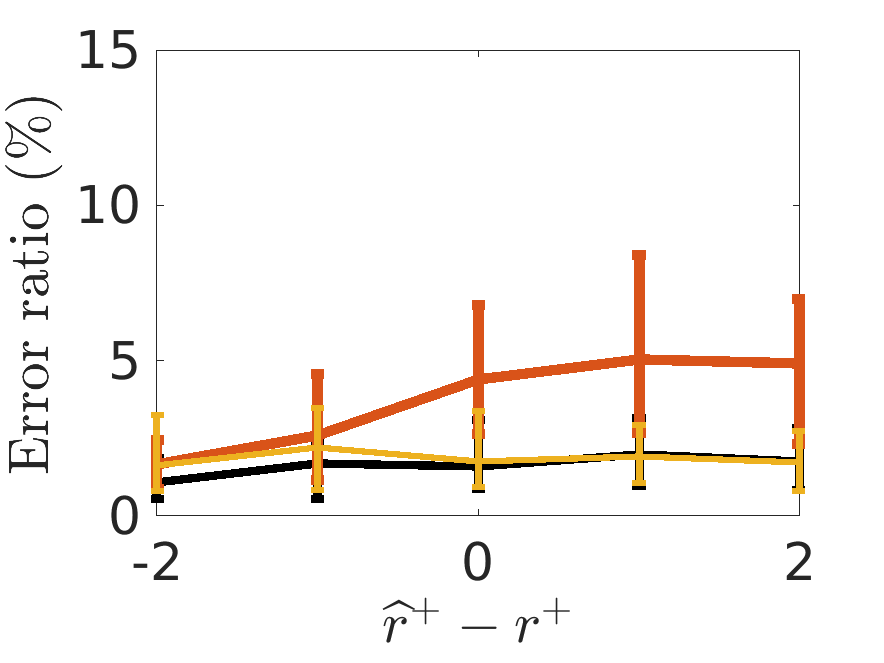}
\end{minipage}
\begin{minipage}{0.3\textwidth}
\includegraphics[width=0.99\linewidth]{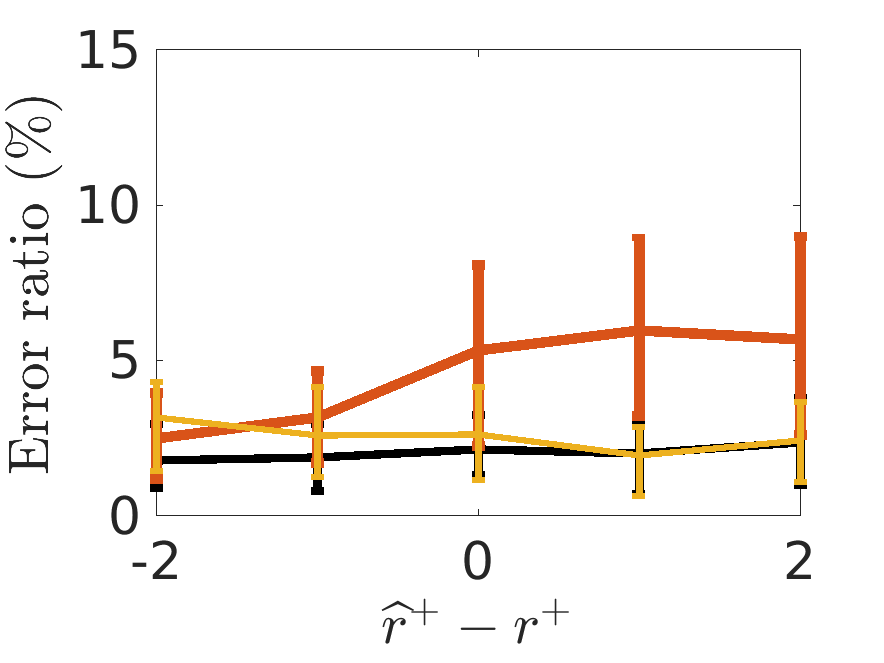}
\end{minipage}
\begin{minipage}{0.3\textwidth}
\includegraphics[width=0.99\linewidth]{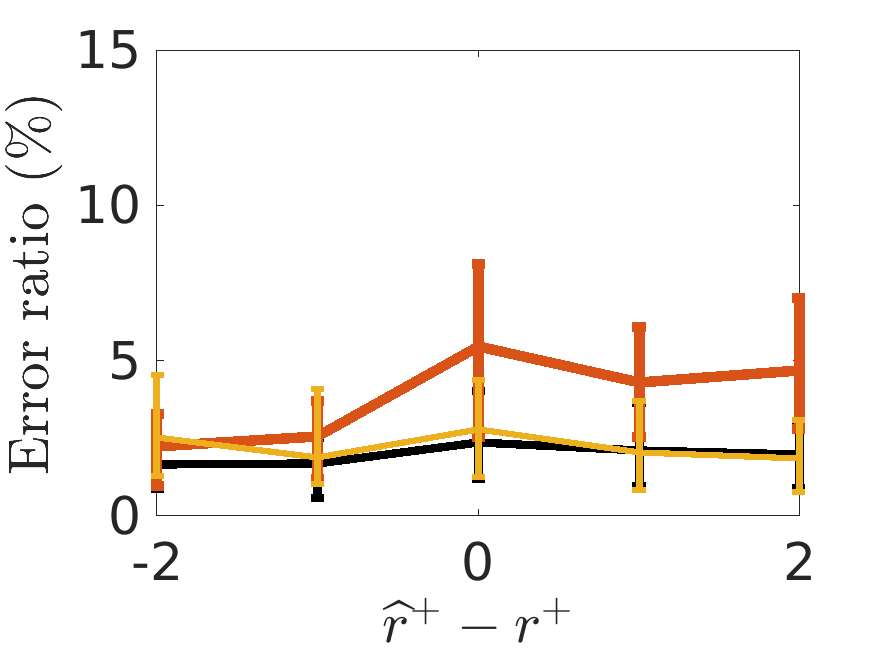}
\end{minipage}
\end{minipage}

\begin{minipage}{1\textwidth}
\begin{minipage}{0.3\textwidth}
\includegraphics[width=0.99\linewidth]{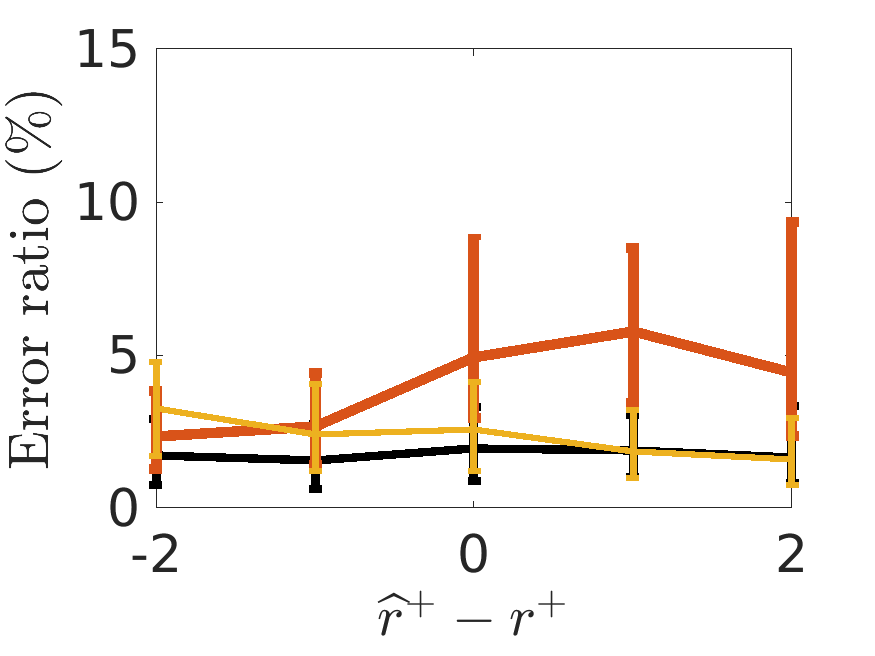}
\end{minipage}
\begin{minipage}{0.3\textwidth}
\includegraphics[width=0.99\linewidth]{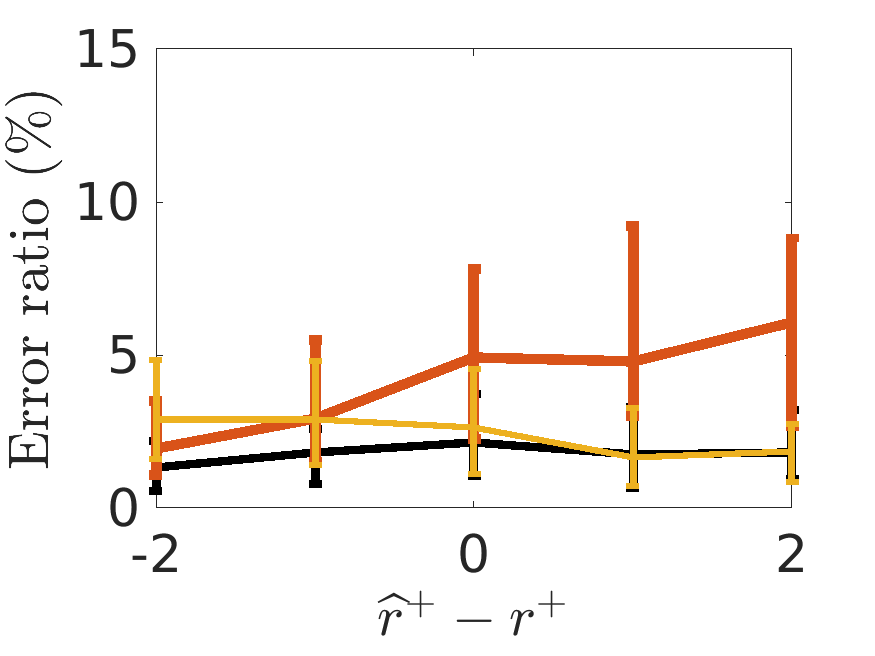}
\end{minipage}
\begin{minipage}{0.3\textwidth}
\includegraphics[width=0.99\linewidth]{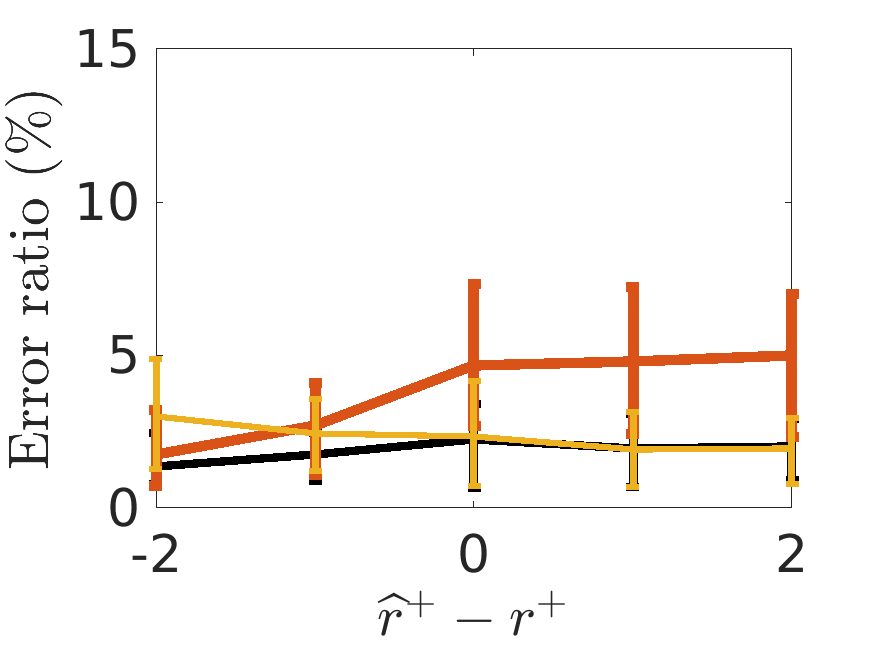}
\end{minipage}
\end{minipage}
\caption{\small Interquartile errorbars of error ratios of estimating  $d_{\min\{r^+,\widehat r^+\}}$ using $\widehat F_{\texttt{e}}(x)$, $\widehat F_{\texttt{T}}(x)$, and $\widehat F_{\texttt{imp}}(x)$, shown in black, yellow and red lines respectively. 
The one-sided noises are shown in the first row, where from left to right are the results when $D_A$ follows Mix2, Unif[1,10], and Fisher3n respectively. The two-sided noises are shown in the second row, where from left to right are the results when $D_A$ and $D_B$ follow Mix2 and Unif[1,10], Mix2 and Fisher3n, and Unif[1,10] and Fisher3n respectively.
We fix $n = 600$ and $p/n=1$.
\label{fig_dr+error_sup}}
\end{figure}

\begin{figure}[hbt!]

\begin{minipage}{1\textwidth}
\begin{minipage}{0.3\textwidth}
\includegraphics[width=0.99\linewidth]{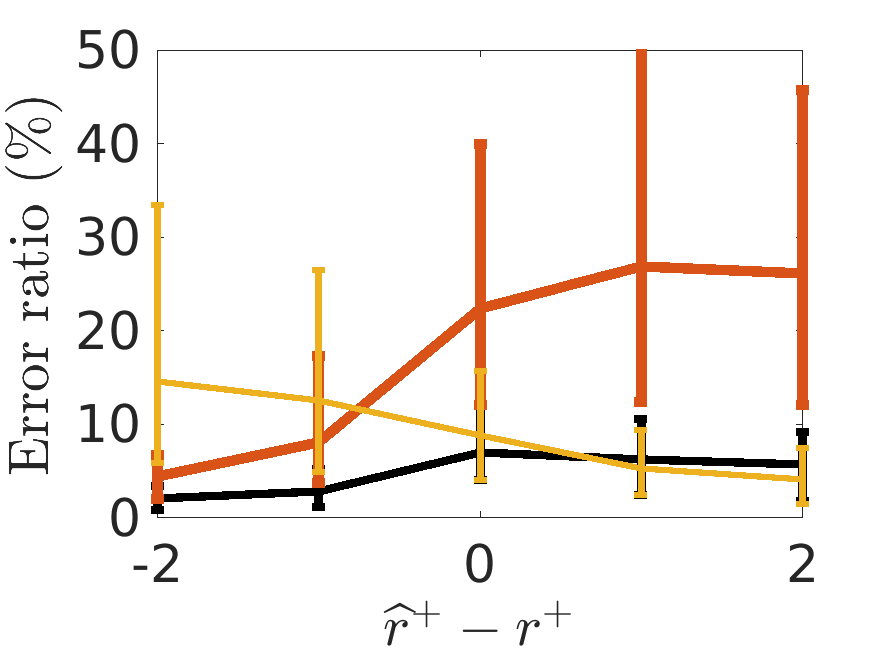}
\end{minipage}
\begin{minipage}{0.3\textwidth}
\includegraphics[width=0.99\linewidth]{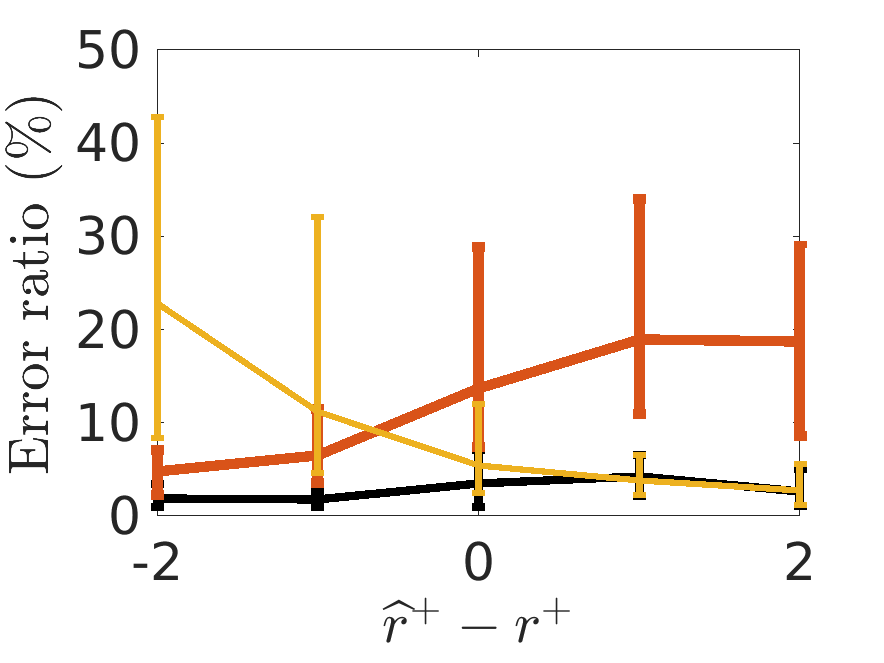}
\end{minipage}
\begin{minipage}{0.3\textwidth}
\includegraphics[width=0.99\linewidth]{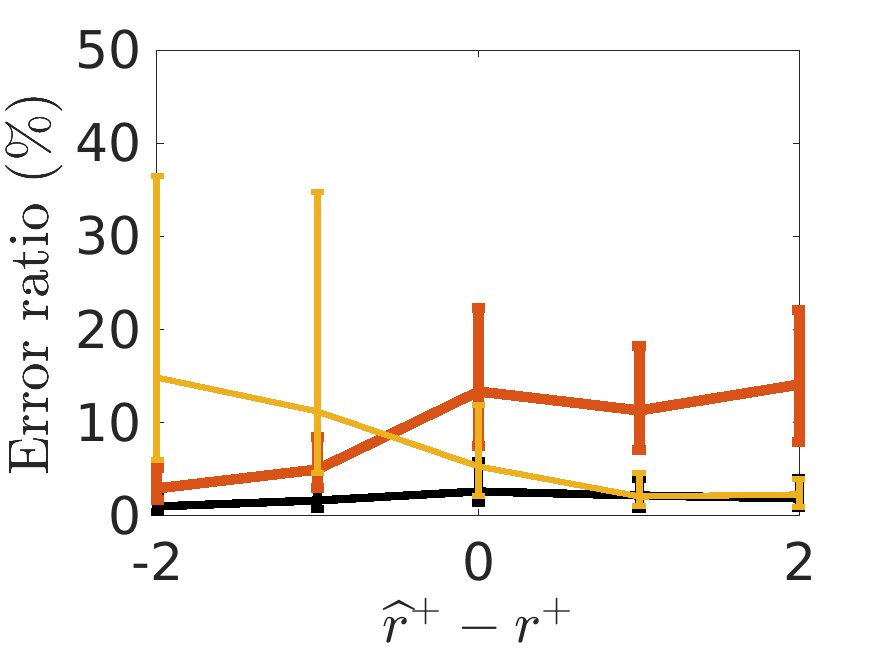}
\end{minipage}
\end{minipage}

\begin{minipage}{1\textwidth}
\begin{minipage}{0.3\textwidth}
\includegraphics[width=0.99\linewidth]{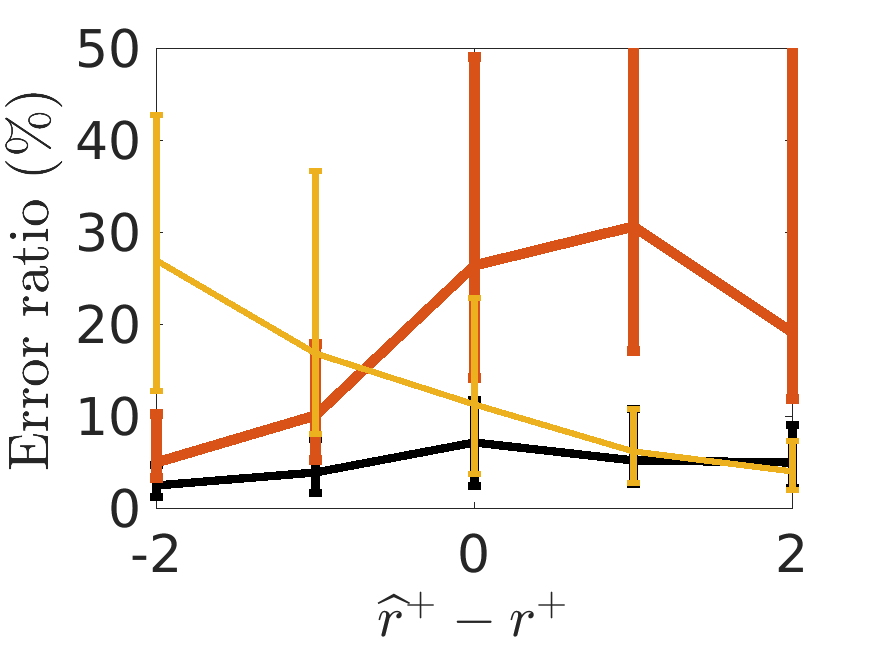}
\end{minipage}
\begin{minipage}{0.3\textwidth}
\includegraphics[width=0.99\linewidth]{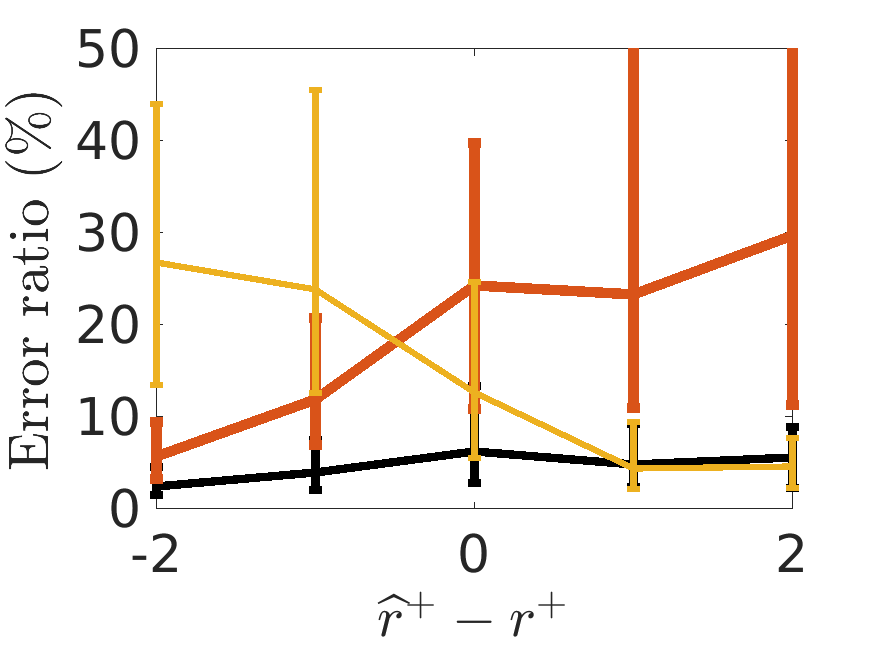}
\end{minipage}
\begin{minipage}{0.3\textwidth}
\includegraphics[width=0.99\linewidth]{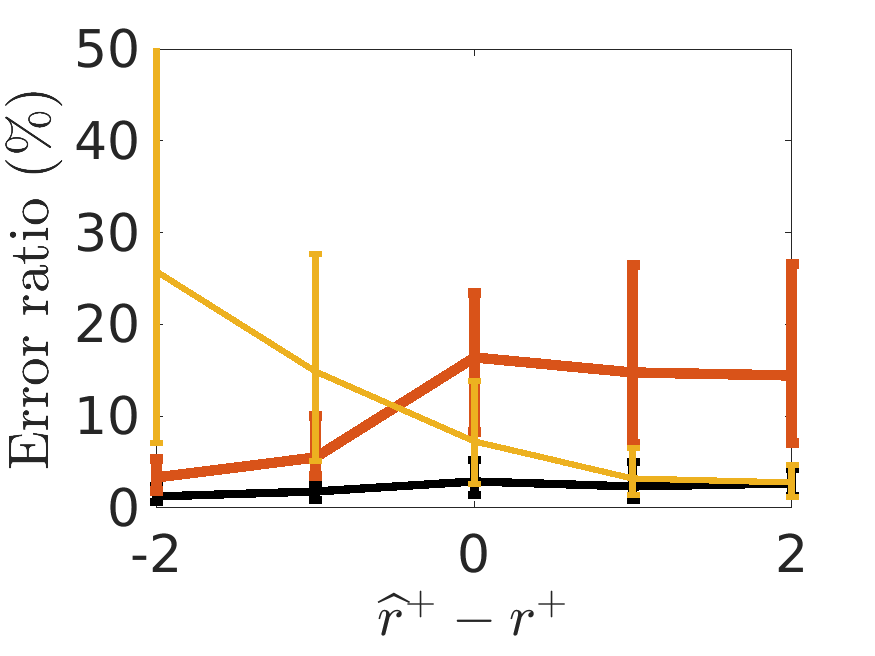}
\end{minipage}
\end{minipage}
\caption{\small Interquartile errorbars of error ratios of estimating  $\sqrt{a_{1,\min\{r^+,\widehat r^+\}}a_{2,\min\{r^+,\widehat r^+\}}}$ using $\widehat F_{\texttt{e}}(x)$, $\widehat F_{\texttt{T}}(x)$, and $\widehat F_{\texttt{imp}}(x)$, shown in black, yellow and red lines respectively. 
The one-sided noises are shown in the first row, where from left to right are the results when $D_A$ follows Mix2, Unif[1,10], and Fisher3n respectively. The two-sided noises are shown in the second row, where from left to right are the results when $D_A$ and $D_B$ follow Mix2 and Unif[1,10], Mix2 and Fisher3n, and Unif[1,10] and Fisher3n respectively.
We fix $n = 600$ and $p/n=1$.
\label{fig_ar+error_sup}}
\end{figure}

\clearpage
\begin{table}[bht!]
 \scriptsize
\caption{ List of default notations, part 1.}\label{tab:notations}
 \begin{tabular}{  p{2cm}  p{10cm}  p{2cm} }
 \hline
 $S$  &   clean data matrix of size $p \times n$ & p.2\\
 $\beta_n$,$\beta$ & $\beta_n:= p/n \to \beta$ as $p,\ n \to \infty$ & p.7, p.2\\
 $r$ & rank of $S$ & p.2\\
 $\{d_i\}_{i=1}^r$ & singular values of $S$, $d_1\geq d_2\geq \ldots \geq d_r$ & p.2\\
 $\{\ub_i, \ \vb_i\}_{i=1}^r$ & left and right singular vectors of $S$ & p.2\\
 $Z$  &  $Z = A^{1/2}XB^{1/2}$,  noise matrix of size $p \times n$ & p.4\\
 $X$     &  noise matrix with i.i.d. entries & p.4\\ 
 $A$,$B$ & colorness and dependence for noise & p.4\\
 $\{\sigma_i^a\}_{i=1}^p$ & eigenvalues of $A$ & p.9\\
 $\{\sigma_i^b\}_{i=1}^n$ & eigenvalues of $B$ & p.9\\ 
 $\{\lambda_i\}_{i=1}^p$ & eigenvalues of $Z Z^\top $  & p.7\\
 $\wt S$ & $\wt S = S+Z$, noisy data matrix  & p.2\\ 
 $\{\wt\lambda_i\}_{i=1}^p$ & eigenvalues of $\wt S \wt S^\top $  & p.2\\
 $\{\wt\bxi_i\}_{i=1}^p$ & left singular vectors of $\wt S$  & p.2\\ 
 $\{\wt\bzeta_i\}_{i=1}^n$ & right singular vectors of $\wt S$  & p.2\\ 
 $\varphi$, $\varphi^*$ & the shrinker and optimal shrinkger  & p.2\\
 $\widehat  S_{\varphi}$ & Estimator constructed by shrinker $\varphi$.  & p.2\\
$\vartheta_{\texttt{SN}}$, $\widehat{r}_{\texttt{SN}}$ & The hard threshold and estimated rank from ScreeNot & p.5\\
 $\pi^{(n)}_H$ &  empirical spectral distribution (ESD) of  an $n \times n $ symmetric matrix $H$  & p.7\\
 $m_\nu(z)$ & Stieljes transform for probability measure $\nu$ for $z \in \mathbb C^+$ & p.7\\
 
 $\mathcal{G}_1, \ \mathcal{G}_2$ & Green functions for $ZZ^\top $ and $Z^\top  Z$ & p.7 \\
 $m_1, \ m_2$ & Stieltjes transforms of ESD of $ZZ^\top $ and $Z^\top  Z$  & p.7\\
 $\wt{\mathcal{G}}_1, \ \wt{\mathcal{G}}_2$ & Green functions for $\wt S \wt S^\top $ and $\wt S^\top  \wt S$ & p.7\\
 $\wt m_1, \ \wt m_2$ & Stieltjes transforms of ESD of $\wt S \wt S^\top $ and $\wt S^\top  \wt S$ & p.7\\
$\rho_{A\infty},\rho_{B\infty}$ & $\pi_A \to \rho_{A\infty}$ and  $\pi_B \to \rho_{B\infty}$ weakly & p.7\\
 $\mathsf{M}_{1\infty}(z),\ \mathsf{M}_{2\infty}(z)$ & $\beta \int\frac{x}{-z\left[1+x\mathsf{M}_{2\infty}(z) \right]} \rho_{A\infty}(d x)$,
 $\int\frac{x}{-z\left[1+x\mathsf{M}_{1\infty}(z) \right]} \rho_{B\infty}(d x)$ & p.8\\
$\wp_{1\infty}$, $\wp_{2\infty}$ & Corresponding densities derived from $\mathsf{M}_{1\infty}(z)$ and   $\mathsf{M}_{2\infty}(z)$ & p.8\\
$m_{1\infty}(z),\ m_{2\infty}(z)$ &   $\int\frac{1}{-z\left[1+x\mathsf{M}_{2\infty}(z) \right]} \rho_{A\infty}(d x)$,
 $\int\frac{1}{-z\left[1+x\mathsf{M}_{1\infty}(z) \right]} \rho_{B\infty}(d x)$ & p.8\\
$ \rho_{1\infty},\ \rho_{2\infty}$ & 
  Corresponding densities derived from $m_{1\infty}(z)$ and   $m_{2\infty}(z)$ & p.8\\
 
 $\mathsf{M}_{1c}(z)$, $\mathsf{M}_{2c}(z)$ & $\beta_n \int\frac{x}{-z\left[1+x\mathsf{M}_{2c}(z) \right]} \pi_A(d x)$,
 $\int\frac{x}{-z\left[1+x\mathsf{M}_{1c}(z) \right]} \pi_B(d x)$ & p.9\\
$\wp_{1c}$, $\wp_{2c}$ & Corresponding densities derived from $\mathsf{M}_{1c}(z)$ and   $\mathsf{M}_{2c}(z)$ & p.9\\
$m_{1c}(z)$, $m_{2c}(z)$ &   $\int\frac{1}{-z\left[1+x\mathsf{M}_{2c}(z) \right]} \pi_A(d x)$,
 $\int\frac{1}{-z\left[1+x\mathsf{M}_{1c}(z) \right]} \pi_B(d x)$ & p.9\\
$ \rho_{1c},\ \rho_{2c}$ & 
  Corresponding densities derived from $m_{1c}(z)$ and   $m_{2c}(z)$ & p.9\\ 
  \hline
\end{tabular}
\end{table}

\begin{table}[bht!]
 \scriptsize
\caption{\label{tab:notations2} List of default notations, part 2.}
 \begin{tabular}{  p{2cm}  p{10cm} p{2cm}  }
 \hline

 $\lambda_+$ & the right most edge of support of $\wp_{1c}$, $\wp_{2c}$,  $\rho_{1c}$, and $\rho_{2c}$ on $(0,\infty)$ & p.10\\
 $\gamma_j$ & $\gamma_j:= \sup_x\{\int_x^{+\infty}\rho_{1c}(x)dx > \frac{j-1}{n}\}$, the classical locations & p.10\\
 $\mathcal{T}(z)$ & $\mathcal{T}(z) = zm_{1c}(z)m_{2c}(z)$, the $D$-transform of $\rho_{1c}$ & p.10\\
 $\alpha$ & threshold $\alpha = 1/\sqrt{\mathcal T(\lambda_+)}$ & p.10\\
 $\phi_n$ & bound for entries of $X$ & p.10\\
  $r^+$ & effective rank, $d_k -\alpha > \phi_n + n^{-1/3}$ if and only if $1\leq k\leq r^+$ & p.12\\
   $G_\ub^n$, $G_\vb^n $ & independent matrices for generating $\ub_i$ and $\vb_i$  & p.12\\
    $\nu$ & probability measure of entires in $G_\ub^n$ and $G_\vb^n $  & p.12\\
$\Delta(d_i)$ & $(d_i-\alpha)^{1/2}$ & p.15\\
 $\mathbb O_+$ & the set $\{1,\cdots,r^+\}$ & p.15\\ 
$\nu_i(\textbf{A})$ & defined as $\begin{cases}
   \min_{j \notin \textbf{A}} |d_j - d_i|, & \textup{ if } i \in \textbf{A}, \\
   \min_{j \in \textbf{A}} |d_j - d_i|, & \textup{ if } i \notin \textbf{A}
\end{cases}$ & p.16\\
${\mathcal{P}}_{\bf A},\ {\mathcal{P}}'_{\bf A}$ & the random projections for  $\textbf{A} \subset \mathbb O^+$, ${\mathcal{P}}_{\bf A}:=\sum_{k\in \bf A} \wt{\xi}_k\wt{\xi}_k^\top $ and ${\mathcal{P}}'_{\bf A}:=\sum_{k\in \bf A} \wt{\zeta}_k\wt{\zeta}_k^\top  $ & p.17\\
$\widehat{\lambda}_{+}$ &  $\widehat\lambda_+ := \wt\lambda_{\lfloor n^c\rfloor+1} + \frac{1}{2^{2/3}-1}\left( \wt\lambda_{\lfloor n^c\rfloor +1}-\wt\lambda_{2\lfloor n^c\rfloor+1} \right)$, the estimator for $\lambda_+$, & p.20\\ 
$\widehat r^+$ & $\widehat r^+: = \big|\{\wt\lambda_i| \wt\lambda_i>\widehat\lambda_++n^{-1/3}\} \big|$, the estimator of $r^+$, & p.20\\
$\widehat{\lambda}_{j}$ & $\widehat \lambda_j:=\wt\lambda_{\lfloor n^c\rfloor + \widehat r^++1} + \frac{1-\left(\frac{j-\widehat r^+-1}{\lfloor n^c\rfloor}\right)^{2/3}}{2^{2/3}-1}\left( \wt\lambda_{\lfloor n^c\rfloor + \widehat r^++1}-\wt\lambda_{2\lfloor n^c\rfloor + \widehat r^++1} \right)$, the estimator of $\lambda_j$  & p.20\\ 
$\widehat F_{e}$ &$\widehat F_{\texttt{e}}(x) := \frac{1}{p-\widehat r^+} \left( \sum_{j=\widehat r^++1}^{\lfloor n^c\rfloor+\widehat r^+} \mathbbm 1(\widehat\lambda_{j} \le x)+\sum_{j=\lfloor n^c\rfloor+\widehat r^++1}^p \mathbbm 1(\wt\lambda_{j}\leq x)\right)$,  the estimated CDF of $\pi_{ZZ^{\T}}$  & p.21 \\
$\widehat m_{e,1,i}$, $\widehat m_{e,2,i}$ & Estimators of $m_{1c}(\wt\lambda_i)$ and $m_{2c}(\wt\lambda_i)$ & p.21\\
$\widehat m_{e,1,i}'$, $\widehat m_{e,2,i}'$ & Estimators of $m_{1c}'(\wt\lambda_i)$ and $m_{2c}'(\wt\lambda_i)$ & p.21\\
$\widehat \varphi_{e,i}$ & estimator of $\varphi^*_i = \varphi^*(\wt\lambda_i)$ by eOptShrink & p.22\\
\hline
\multicolumn{2}{l} {$\theta(x) := \mathcal T^{-1}(x^{-2}), \ a_1(x):=\frac{m_{1c}(\theta(x))}{x^2 \mathcal{T}'(\theta(x))}, \   a_2(x):=\frac{m_{2c}(\theta(x))}{x^2 \mathcal{T}'(\theta(x))}$}  &p.14\\
\multicolumn{2}{l}{$\varkappa_i:=i^{2/3}n^{-2/3},$ for $i = 1, \ldots, n$} & p.16\\ 
\multicolumn{2}{l}{$\eta_i:=n^{-3/4}+n^{-5/6} i^{1/3} +  n^{-1/2}\phi_n$} & p.16\\

 \multicolumn{2}{l}{$\psi_1(d_i) := \phi_n + \frac{n^{-1/2}}{\Delta(d_i)}$} & p.17\\
\multicolumn{2}{l}{$\widehat{\mathcal{T}}_{e,i}= \wt\lambda_i \widehat{m}_{e,1,i} \widehat{m}_{e,2,i}, \quad \widehat{\mathcal{T}}'_{e,i}=\widehat{m}_{e,1,i} \widehat{m}_{e,2,i}+\wt\lambda_i \widehat{m}'_{e,1,i} \widehat{m}_{e,2,i}+\wt\lambda_i \widehat{m}'_{e,2,i} \widehat{m}_{e,1,i}$} & p.21\\
\multicolumn{2}{l} {$\widehat{d}_{e,i}=\sqrt{\frac{1}{\widehat{\mathcal{T}}_{e,i}}}, \quad \widehat{a}_{e,1,i}=\frac{\widehat{m}_{e,1,i}}{\widehat{d}_{e,i}^2 \widehat{\mathcal{T}}'_i} \ \mbox{ and } \  \widehat{a}_{e,2,i}=\frac{\widehat{m}_{e,2,i}}{ \widehat{d}_{e,i}^2\widehat{\mathcal{T}}'_i}$} & p.21\\
\multicolumn{2}{l}{$S(\varsigma_1,\varsigma_2):= \{z = E+i\eta: \lambda_+ -\varsigma_1 \leq E \leq \varsigma_2\lambda_+, 0<\eta \leq 1\}$} & p.S.2\\
\multicolumn{2}{l}{$\kappa_z:= |E - \lambda_+|$, for $z=E+i\eta\in \mathbb{C}^+$} & p.S.2\\
\multicolumn{2}{l}{$\Psi(z):= \sqrt{\frac{\text{Im } \mathsf m_{2c}(z)}{n\eta}} +\frac{1}{n\eta}$}& P.S.4\\
\multicolumn{2}{l}{$\qquad  \text{dist} \Big (x, \text{Spec}({ZZ^\top }) \Big )>n^{-1+\epsilon} \alpha_+^{-1} + n^{\epsilon}\eta_l(x) \Big\}$} \\
\multicolumn{2}{l}{ $   \mathbf{D}_1(\tau_1,\varsigma):= \{z = E+i\eta: \lambda_+<E<\varsigma,\, -\tau_1<\eta<\tau_1\}$} & p.S.4\\
\multicolumn{2}{l}{$ \mathbf{D}_2(\tau_2,\varsigma):= \{\zeta = E+i\eta: \alpha<E<1/\sqrt{\mathcal{T}(\varsigma)},\, -\tau_2<\eta<\tau_2\}$} & p.S.4\\
\multicolumn{2}{l}{$\eta_\ell(E)\mbox{ satisfies }n^{1/2} \left[\Psi^2(E+i\eta_l(E))+ \frac{\phi_n}{n\eta_l(E)}\right] =1$} & p.S.5\\
\multicolumn{2}{l}{$H(z)
:=
\begin{pmatrix}
0 & z^{1/2} Z  \\
z^{1/2}Z^\top  & 0
\end{pmatrix}, \quad   
\widetilde{H}(z):=
\begin{pmatrix}
0 & z^{1/2} \widetilde{S} \\
z^{1/2} \widetilde{S}^\top  & 0
\end{pmatrix}$}& p.S.6\\
\multicolumn{2}{l}{
$G(z):=(H-zI_{p+n})^{-1},\quad \widetilde{G}(z):=(\widetilde{H}-zI_{p+n})^{-1},\quad\Db=
\begin{pmatrix}
0 & z^{1/2} D \\
z^{1/2} D & 0
\end{pmatrix}, \quad
\Ub=\begin{pmatrix}
U & 0 \\
0 &  V
\end{pmatrix}$} & p.S.6\\
\multicolumn{2}{l}{
$\Pi (z):=\left( {\begin{array}{*{20}c}
   { \Pi_1}(z) & 0  \\
   0 & { \Pi_2}(z)  \\
\end{array}} \right),\quad \Pi_1(z):  =-z^{-1}\left(1+\mathsf{M}_{2c}(z)A \right)^{-1},\quad  \Pi_2(z):=- z^{-1} (1+\mathsf{M}_{1c}(z)B )^{-1}.$} & p.S.7\\
\multicolumn{2}{l}{$S_0(\varsigma_1,\varsigma_2,\omega):= S(\varsigma_1,\varsigma_2) \cap \{z = E+i\eta: \eta\geq n^{-1+\omega}\}$} & p.S.8\\

\multicolumn{2}{l}{$\wt S_0(\varsigma_1,\varsigma_2,\omega):=S_0(\varsigma_1,\varsigma_2,\omega) \cap \left\{z = E+ i \eta:  n^{1/2}\left( \Psi^2(z)+\frac{\phi_n}{n\eta}\right) \le n^{-\omega/2}\right\}$} & p.S.8\\
\multicolumn{2}{l}{$ S_{out}(\varsigma_2,\omega):=\{E+i\eta:\lambda_+ + n^{\omega}(n^{-2/3}+n^{-1/3}\phi_n^2) \leq E \leq \varsigma_2 \lambda_+, \eta \in [0,1]\}$} & p.S.8\\

\multicolumn{2}{l}{$\overline\Pi(z):=\begin{pmatrix}
m_{1c}(z)I_r & 0 \\
0 & m_{2c}(z) I_r
\end{pmatrix},\quad \Omega(z):= \Ub^\top  G(z)\Ub-\overline{\Pi}(z)$} & p.S.8\\
\multicolumn{2}{l}{$I_i:=[\theta(d_i)-n^{\epsilon}\omega(d_i), \theta(d_i)+n^{\epsilon}\omega(d_i)]$} & p.S.10\\
\multicolumn{2}{l}{$\omega(d_i):= \phi_n\Delta^2(d_i)+n^{-1/2}\Delta(d_i)$}& p.S.11\\
\multicolumn{2}{l}{$I_0:=[0, \lambda_+ + n^{3\epsilon}\phi_n^2+n^{-2/3+3\epsilon}]$} & p.S.11\\
\multicolumn{2}{l}{ $I:= I_0 \cup \bigcup_{i\in \mathbb{O}_\epsilon} I_i$} & p.S.11\\
\multicolumn{2}{l}{
$\Omega_i:=
\Big\{x \in [\lambda_{i-r-1}, \lambda_+ + c_0n^{2\epsilon} (\phi_n^2 +n^{-2/3})]: $} & p.S.15\\
\hline
\end{tabular}
\end{table}

\end{document}